\documentclass[11pt]{article}

\usepackage[margin=1in]{geometry}

\usepackage{amsmath,amssymb,amsthm,graphicx,stmaryrd}
\usepackage{verbatim}
\usepackage{enumitem}
\usepackage{xcolor}
\usepackage[enableskew]{youngtab}
\usepackage{tikz,tikz-cd,genyoungtabtikz,pgfplots, bbm, mathtools}
\usepackage{caption,subcaption}
\usetikzlibrary{patterns}
\usetikzlibrary{patterns.meta}
\usetikzlibrary{arrows.meta}
\usetikzlibrary{decorations.markings}
\usepackage[colorlinks=true, allcolors=purple]{hyperref}

\usepackage[utf8]{inputenc}
\usepackage{graphicx}
\usepackage{bbm}
\usepackage{url}
\usepackage{amsfonts}
\usepackage{mathtools}
\usepackage[T1]{fontenc}
\usepackage{enumitem}

\usepackage{lmodern}
\usepackage{authblk}

\usetikzlibrary{intersections}

\newcommand{\bplus}{%
  \tikz[baseline=-0.5ex] \draw[line width=2pt] (0,-0.1) -- (0,0.1) ( -0.1,0 ) -- (0.1,0);
}

\newcommand{\wplus}{%
  \tikz[baseline=-0.5ex]{
  \draw[line width=2pt] (0,-0.1) -- (0,0.1) ( -0.1,0 ) -- (0.1,0);;
    \draw[line width=1pt, white] (0,-0.08) -- (0,0.08) ( -0.08,0 ) -- (0.08,0);
    }
}

\newcommand{\bminus}{%
  \tikz[baseline=-0.5ex] \draw[line width=2pt] (-0.1,0) -- (0.1,0);
}

\newcommand{\wminus}{%
  \tikz[baseline=-0.5ex]{
  \draw[line width=2pt] (-0.1,0) -- (0.1,0);
  \draw[line width=1pt, white] (-0.08,0) -- (0.08,0);
  }
}

\newcommand{\cross}{\mathchoice
	{\vcenter{\hbox{\includegraphics[width=.8em]{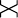}}}}
	{\vcenter{\hbox{\includegraphics[width=.8em]{cross.pdf}}}}
	{\vcenter{\hbox{\includegraphics[width=.6em]{cross.pdf}}}} 
	{\vcenter{\hbox{\includegraphics[width=.5em]{cross.pdf}}}} 
}
\newcommand{\dbar}{\mathchoice
	{\vcenter{\hbox{\includegraphics[width=.8em]{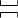}}}}
	{\vcenter{\hbox{\includegraphics[width=.8em]{dbar.pdf}}}}
	{\vcenter{\hbox{\includegraphics[width=.6em]{dbar.pdf}}}} 
	{\vcenter{\hbox{\includegraphics[width=.5em]{dbar.pdf}}}} 
}

\pgfdeclarepattern{
	name=hatch,
	parameters={\hatchsize,\hatchangle,\hatchlinewidth},
	bottom left={\pgfpoint{-.1pt}{-.1pt}},
	top right={\pgfpoint{\hatchsize+.1pt}{\hatchsize+.1pt}},
	tile size={\pgfpoint{\hatchsize}{\hatchsize}},
	tile transformation={\pgftransformrotate{\hatchangle}},
	code={
		\pgfsetlinewidth{\hatchlinewidth}
		\pgfpathmoveto{\pgfpoint{-.1pt}{-.1pt}}
		\pgfpathlineto{\pgfpoint{\hatchsize+.1pt}{\hatchsize+.1pt}}
		\pgfpathmoveto{\pgfpoint{-.1pt}{\hatchsize+.1pt}}
		\pgfpathlineto{\pgfpoint{\hatchsize+.1pt}{-.1pt}}
		\pgfusepath{stroke}
	}
}

\tikzset{
	hatch size/.store in=\hatchsize,
	hatch angle/.store in=\hatchangle,
	hatch line width/.store in=\hatchlinewidth,
	hatch size=5pt,
	hatch angle=0pt,
	hatch line width=.5pt,
}

\renewcommand{\d}{\mathrm{d}}
\newcommand{\ising}{\mathrm{ising}}
\newcommand{\XY}{\mathrm{XY}}
\newcommand{\spin}{\mathrm{spin}}
\newcommand{\horiz}{\mathrm{horiz}}
\renewcommand{\vert}{\mathrm{vert}}

\newcommand{\mir}{\mathrm{mir}}

\newcommand{\NW}{\mathrm{NW}}
\newcommand{\N}{\mathrm{N}}
\newcommand{\W}{\mathrm{W}}
\newcommand{\NE}{\mathrm{NE}}
\newcommand{\AT}{\mathrm{AT}}

\newcommand{\ATRC}{\mathrm{ATRC}}
\newcommand{\wind}{\mathrm{wind}}
\newcommand{\perc}{\mathrm{perc}}
\newcommand{\grad}{\mathrm{grad}}
\newcommand{\hf}{\mathrm{hf}}
\newcommand{\poi}{\mathrm{poi}}
\newcommand{\lp}{\mathrm{loop}}

\newcommand{\f}{\mathrm{f}}
\newcommand{\sv}{\mathrm{6v}}
\newcommand{\eiv}{\mathrm{8v}}
\newcommand{\eivRC}{\mathrm{8vRC}}

\newcommand{\ul}{\underline}

\newcommand{\Var}{\mathrm{Var}}

\newcommand{\CC}{\mathbb{C}}
\newcommand{\EE}{\mathbb{E}} 
\newcommand{\FF}{\mathbb{F}}
\newcommand{\GG}{\mathbb{G}}

\newcommand{\LL}{\mathbb{L}}
\newcommand{\NN}{\mathbb{N}}
\newcommand{\PP}{\mathbb{P}} 
\newcommand{\RR}{\mathbb{R}} 

\renewcommand{\SS}{\mathbb{S}}
\newcommand{\TT}{\mathbb{T}} 
\newcommand{\VV}{\mathbb{V}}

\newcommand{\ZZ}{\mathbb{Z}}

\renewcommand{\c}[1]{\mathcal{#1}}
\newcommand{\cA}{\mathcal{A}}
\newcommand{\cB}{\mathcal{B}} 

\newcommand{\cC}{\mathcal{C}}
 
\newcommand{\cF}{\mathcal{F}}
\newcommand{\cE}{\mathcal{E}}

\newcommand{\cH}{\mathcal{H}}
\newcommand{\cM}{\mathcal{M}}

\newcommand{\cO}{\mathcal{O}}

\newcommand{\cP}{\mathcal{P}} 

\newcommand{\cW}{\mathcal{W}}

\newcommand{\cZ}{\mathcal{Z}}

\newcommand{\D}{\Delta} 
\newcommand{\G}{\mathcal{G}}

\newcommand{\lam}{\lambda}

\newcommand{\Om}{\Omega}
\newcommand{\om}{\omega}

\newcommand{\eps}{\varepsilon}

\newcommand{\odd}{\mathrm{odd}}
\newcommand{\even}{\mathrm{even}}

\newcommand{\cyl}{\mathrm{cyl}}

\newcommand{\inn}{\mathrm{in}}
\newcommand{\outt}{\mathrm{out}}

\newcommand{\oo}{\infty}

\newcommand{\es}{\varnothing}

\newcommand{\p}{\mathbf{p}}

\newcommand{\id}{\mathrm{id}}

\newcommand{\Tr}{\mathrm{Tr}}

\newcommand{\one}{\hbox{\rm 1\kern-.27em I}}

\newcommand{\be}{\begin{equation}}
\newcommand{\ee}{\end{equation}}
\newcommand{\bes}{\begin{equation*}}
\newcommand{\ees}{\end{equation*}}

\newtheoremstyle{slthm}
{}
{\baselineskip}
{\slshape}
{\parindent}
{\scshape}
{.}
{ }
{}

\theoremstyle{slthm}
\newtheorem{definition}{Definition}[section]
\newtheorem{theorem}[definition]{Theorem}
\newtheorem{proposition}[definition]{Proposition}
\newtheorem{lemma}[definition]{Lemma}
\newtheorem{corollary}[definition]{Corollary}
\newtheorem{remark}[definition]{Remark}

\usepackage[
    backend=biber,
    style=numeric,
    maxnames=50
]{biblatex}

\renewbibmacro{in:}{}
\begin{document}
	\title{The spin-1/2 Heisenberg XXZ chain and the Lorentz mirror model with loop weight 2}
	\author{Kieran Ryan}
	\affil{Technical University of Vienna \\
    \href{mailto:kieran.ryan@tuwien.ac.at}{kieran.ryan@tuwien.ac.at}}
    
	\date{}
	\maketitle
	
	\begin{abstract}

    We prove that for the spin-$\tfrac{1}{2}$ Heisenberg XXZ chain in the range $\Delta\in[-1,\tfrac{1}{2}]$, the ground state on the torus of length $L$ converges to an infinite volume ground state $\langle\cdot\rangle$ as $L\to\infty$, and that the spin-spin correlation $\langle S_0^{(1)}S_x^{(1)}\rangle$ decays polynomially fast in $x$. In the range $\Delta\in[-1,0]$ we have the stronger results: that the convergence holds for several finite-volume, finite temperature states, that $L$ and $\beta=1/T$ can be taken to infinity in any order, that the convergence holds on (Euclidean) dynamic correlators, and that $\langle S_0^{(1)}S_x^{(1)}(t)\rangle$ decays to zero in $|(x,t)|$, and cannot decay exponentially fast.

    The convergence of the ground state in infinite volume is known rigorously by Bethe Ansatz methods, while our correlation decay results (apart from the points $\Delta=0,-1$), the convergence on dynamic correlators and the interchangeability of limits are new at the rigorous level. Moreover our methods are new and do not use any Bethe Ansatz techniques, using only the following related probabilistic models.
    
    In the Lorentz mirror model with loop weight 2, a model of random loops on $\ZZ^2$, in a large range of parameters we prove that connection probabilities tend to 0, and do so polynomially fast in the symmetric case of the model. Our proof uses two couplings of this mirror model with the six-vertex model, one of which is new. The main input is the delocalisation of the height function of the six-vertex model proved by several authors. 

    We further prove that the height function of a certain space-time version of the six-vertex model delocalises and then prove through analogous couplings to those mentioned above that connection probabilities converge to 0 in the loop representation of the XXZ model introduced by Ueltschi. 
    
    We also give the exact rate of decay of $\langle S_0^{(3)}S_x^{(3)}\rangle$ for $\Delta\in[-1,-\tfrac{1}{2}]$ as a corollary of recent work on the six vertex model by Duminil-Copin, Kozlowski, Lammers and Manolescu and of $\langle S_0^{(1)}S_x^{(1)}(t)\rangle$ for $\Delta=0$ as a corollary of work of Li and Mahfouf.

	\end{abstract}

\tableofcontents

\part{Introduction, results and background}

\section{Introduction and results}
This paper deals with the problem of classifying the infinite volume ground states of the spin-$\tfrac{1}{2}$ Heisenberg XXZ chain. The Hamiltonian of the model is given by 
\be\label{eq:ham-xxz}
    H_\Lambda=-\sum_{xy\in\c E(\Lambda)}
    \big[
    S^{(1)}_xS^{(1)}_y + S^{(2)}_xS^{(2)}_y + \Delta S^{(3)}_xS^{(3)}_y
    \big],
    \qquad
    \mathrm{acting \ on \ } (\CC^2)^\Lambda,
\ee
with anisotropy parameter $\Delta\in\RR$, where $\Lambda$ is a finite graph (here either an interval length $L$ in $\ZZ$ or a finite 1D torus length $L$), $E(\Lambda)$ its set of edges, and $S^{(i)}_x$, $i=1,2,3$, are the usual spin operators at a site $x\in\Lambda$, given by $\tfrac{1}{2}\sigma^{(i)}_x$, with $\sigma^{(i)}, i=1,2,3,$ the Pauli matrices. 
There are the special cases $\Delta=1$ (the
Heisenberg ferromagnet), $\Delta=-1$ (the Heisenberg antiferromagnet),
and $\Delta=0$ (the XY model). 
	
For $\Delta\in[-1,1)$, the XXZ chain is ``critical'' (also called massless), in the
sense that there is expected to be a unique, gapless ground
state with polynomially decaying correlations. This paper's aim to make progress
towards this conjecture, and in particular without appealing to the Bethe ansatz.\\

In finite volume, the ground state was shown to be unique by Yang and Yang \cite{yang-yang-1}, and Lieb, Schultz and Mattis \cite{lieb-schultz-mattis}
showed that the spectral gap is at most $\mathrm{const.}/L$,
with $L$ the length of the system. Note that this does not, however, prove that the infinite volume system is gapless. To the author's knowledge,
the only rigorous proof of a unique ground state in infinite
volume is for the XY model ($\Delta=0$) by Araki and Matsui
\cite{araki-matsui-XY}. Affleck and Lieb
\cite{affleck-lieb-XXZ} proved that a unique ground state
implies there is no spectral gap; their result extends to all
half-odd-integer spins and is some evidence for the Haldane
conjecture. Polynomial decay of correlations which are slow enough (in this case of order distance$^{-\alpha}$, $\alpha<1$)  also implies no
spectral gap, see for example Problem 6.1.a in
\cite{tasaki-book}. Let us also mention the probabilistic work of Duminil-Copin, Li and
Manolescu \cite{dc-li-mano-FK}, who prove that the associated loop model of the antiferromagnet
($\Delta=-1$) (see Section \ref{sec:loop-results}), has a unique
infinite volume Gibbs measure.

The XXZ chain is
widely studied using exact solutions methods; see for example
the textbooks \cite{KBI-book, sutherland-book, takahashi-book}. We will mention a selection of results which are relevant to our results; see futher references within these works for more detail. Firstly, the finite volume ground state was proved to converge to an infinite volume ground state (ie. convergence on all local observables) for $\Delta\in[0,1)$ and $\Delta<\Delta_0<-1$ by Dorlas and Samsonov \cite{dorlas-samsonov} and then for all $\Delta<1$ by Kozlowski \cite{karol1}, via proving condensation of the Bethe roots. Other proofs of condensation are given $\Delta<-1$ in \cite{dc-discont-q>4} and for all $\Delta<1$ in \cite{dc-free-energy}. 

Formulae in terms of multiple integrals for general local observables in the thermodynamic limit were previously given by Jimbo, Miki, Miwa and Nakayashiki \cite{jimbo-miki-miwa-nakayashiki} for $\Delta<-1$ and Jimbo and Miwa \cite{jimbo-miwa}, and later by Kitanine, Maillet and Terras \cite{kit-maillet-terras-form-factors, kit-maillet-terras-corr-fns}.

The large-distance asymptotics of the spin-spin correlation functions of the XXZ chain have been conjectured to great precision. The critical exponents of the chain in the critical regime $\Delta\in[-1,1)$ were first conjectured by Luther and Peschel \cite{luther-peschel}:
\bes
    \langle S_0^{(1)}S_x^{(1)}\rangle \sim  |x|^{-\frac{1}{\pi} \arccos(\Delta)},
    \qquad
    \langle S_0^{(3)}S_x^{(3)}\rangle =  \frac{-1}{4\pi\cdot\arccos(\Delta)}|x|^{-2} + o(|x|^{-2}),
\ees
where by $\sim$ we mean that $\langle S_0^{(1)}S_x^{(1)}\rangle \cdot |x|^{\frac{1}{\pi} \arccos(\Delta)}$ converges to a constant in $(0,\infty)$ as $|x|\to\infty$, and later were conjectured in a magnetic field by Haldane \cite{haldane-1, haldane-2, haldane-3}. Higher order terms known as amplitudes were given in works of Lukyanov and Zamolodchikov \cite{lukyanov-zamadolo}, Lukyanov \cite{lukyanov} and Lukyanov and Terras \cite{lukyanov-terras}. 

Alternative multiple integral representations of correlation functions which allow one to analyse asymptotics were derived by Kitanine, Kozlowski, Maillet, Slavnov and Terras \cite{KKMST}, and the works of Kozlowski \cite{karol2, karol3, karol4, karol5} bring these formulae to a single hypothesis that a certain series converges. While the formulae obtained by the Bethe ansatz methods are much more precise than ours, as far as the author is aware, our derivation of polynomial decay is the first rigorous proof of such behaviour.\\

Note that outside the critical regime there are more complete works on classifying the infinite volume ground states.	
For $\Delta<-1$ the $S^{(3)}$ term dominates and one expects behaviour like the antiferromagnetic Ising model ($\Delta=-\infty$). For $\Delta<-1$ and $|\Delta|$ sufficiently large, Matsui \cite{matsui-XXZ-antiferro} proved there are exactly two extremal ground states, and for all $\Delta<-1$, Aizenman, Duminil-Copin and Warzel \cite{ADCW} proved the existence of two distinct ground states (which should be the only two extremal ones) exhibiting Néel order. 
	
For $\Delta\ge1$ one has ferromagnetic behaviour. For
$\Delta>1$, the $S^{(3)}$ term dominates once again, and one
has the behaviour of the ferromagnetic Ising model ($\Delta=+\infty$). Here there are two
translation-invariant ground states (all spins up and all
spins down), and an infinite number of
non-translation-invariant ground states (all spins left of
$x\in\ZZ$ up (resp.\ down) and all spins right of $x$ down
(resp.\ up)), known as kink (or anti-kink) states. This was
proved to be a complete list of all the extremal ground states
by Matsui \cite{matsui-XXZ-ferro}, a result extended to all
spins by Koma and Nachtergaele
\cite{koma-nacht-XXZ-ferro}. The kink and antikink states were
discovered by Alcaraz, Salinas, 
and Wreszinski \cite{ASW}, and Gottstein and Werner \cite{gottstein-werner}. For $\Delta=1$ (the Heisenberg ferromagnet), the Ising behaviour disappears and for all spins, all ground states are translation-invariant \cite{koma-nacht-XXZ-ferro} (in fact for spin-$\tfrac{1}{2}$ they are exactly all of the permutation-invariant states). \\

\subsection{Results for the Heisenberg XXZ chain}
We consider the nearest-neighbour spin-$\tfrac{1}{2}$ quantum XXZ model with anisotropy parameter $\Delta\in\RR$, whose Hamiltonian is \eqref{eq:ham-xxz}. For finite temperature $T=\tfrac{1}{\beta}$, the Gibbs state $\langle\cdot\rangle_{\beta,\Lambda}$ mapping $L((\CC^2)^\Lambda)$ to $\CC$ is defined by 
\be
\langle A \rangle_{\beta,\Lambda} =
    \frac{1}{Z_{\beta,\Lambda}} \Tr[Ae^{-\beta H_{\Lambda}}],
\ee
where $Z_{\beta,\Lambda}=\Tr[e^{\beta H_{\Lambda}}]$. For a vector $\Psi\in(\CC^2)^\Lambda$, we define the seeded state as:
\be
\langle A \rangle_{\beta,\Lambda}^{\Psi} =
    \frac{\langle \Psi | e^{-\tfrac{1}{2}\beta H_{\Lambda}} A e^{-\tfrac{1}{2}\beta H_{\Lambda}} | \Psi \rangle}
    {\langle \Psi | e^{\beta H_{\Lambda}} | \Psi \rangle}.
\ee

From hereon in we consider the model in dimension $d=1$. For $L\in\NN$, we work with the finite graphs $\Lambda_L:=\{-L+1,\dots,L\}\subset\ZZ$ and the discrete torus $\TT_L$ given by adding an extra edge connecting $-L+1$ and $L$ in $\Lambda_L$. We further define
\be\label{eq:psi_1}
    \Psi^\bullet_L = 
        \frac{1}{\sqrt{n}} \sum_{\substack{i\in\Lambda_L \\ i \ \mathrm{even}}} \sum_{a=1}^2 | a,a \rangle_{i-1, i}
    \qquad \text{and} \qquad
    \Psi^\circ_L = 
        \frac{1}{\sqrt{n}} \sum_{\substack{i\in\Lambda_L \\ i \ \mathrm{even}}}   \sum_{a=1}^2 | a,a \rangle_{i, i+1},
\ee
where we interpret $L+1$ as $-L+1$ and $L-2$ as $L$ when they appear in the subscripts. The vector $\Psi^\bullet$ is the vector $\Psi^\circ_L$ shifted by one step to the right on the torus. We write $\Psi_L$ for $\Psi^\bullet_L$ when $L$ is even and $\Psi^\circ_L$ when $L$ is odd; this is the choice so that the pair $L,-L+1$ does not appear in the subscripts.

For $\Lambda$ a finite graph and $A$ a local observable, define $A(t)=A^\Lambda(t)=e^{-tH_\Lambda}A e^{tH_{\Lambda}}$. In the case that $A^\Lambda(t)$ converges in an infinite volume limit, we write its limit as simply $A(t)$.

\begin{theorem}\label{thm:main-xxz}
    Consider the spin-$\tfrac{1}{2}$ Heisenberg XXZ chain with parameter $\Delta$.
    \begin{enumerate}
        \item 
            Let $\Delta\in[-1,\tfrac{1}{2}]$. There exists an infinite volume ground state $\langle\cdot\rangle$, and we have the convergence: 
            \bes
                \langle \cdot \rangle = 
                    \lim_{L\to\infty}\lim_{\beta\to\infty} 
                    \langle \cdot \rangle_{\beta,\TT_L}
                    = \lim_{L\to\infty}\lim_{\beta\to\infty} 
                    \langle \cdot \rangle_{\beta,\TT_L}^{\Psi^*_L},
            \ees
            for $*=\bullet,\circ$. 
            The ground state $\langle\cdot\rangle$ is extremal and $\ZZ$-translation invariant. 
            Further, there exist $\alpha_1,\alpha_2,C_1,C_2 >0$ such that
            \be\label{eq:xxz-poly-decay}
                C_1|x|^{-\alpha_1}
                \le 
                \langle S^{(1)}_0 S^{(1)}_x\rangle 
                \le
                C_2|x|^{-\alpha_2}
                .
            \ee
            (Note $\langle S^{(1)}_0 S^{(1)}_x\rangle \ge |\langle S^{(3)}_0 S^{(3)}_x \rangle|$).
        \item 
            Let $\Delta\in[-1,0]$. We further have the convergence
            \bes
            \begin{split}
                \langle \cdot \rangle &= 
                \lim_{L\to\infty}\lim_{\beta\to\infty} \langle \cdot \rangle_{\beta,\TT_L}
                =
                \lim_{\beta,L\to\infty} \langle \cdot \rangle_{\beta,\TT_L}^{\Psi^*_L}
                =
                \lim_{\beta,L\to\infty} \langle \cdot \rangle_{\beta,\Lambda_L}
                =
                \lim_{\beta,L\to\infty} \langle \cdot \rangle_{\beta,\Lambda_L}^{\Psi_L},
            \end{split}
            \ees
            for $*=\bullet,\circ$ and where here, for all but the $\langle \cdot \rangle_{\beta,\TT_L}$ case, the limits in $\beta$ and $L$ can be taken in either order or simultaneously. Further, the convergence holds on all observables of the form
            \bes
                \langle \prod_{k=1}^r A_k(t_k)\rangle_\Lambda 
                \to 
                \langle \prod_{k=1}^r A_k(t_k)\rangle,
            \ees
            where here we write $\langle\cdot\rangle_\Lambda$ for any of the finite volume, finite temperature states for which the convergence above holds. The state $\langle \cdot \rangle$ is $\ZZ\times\RR$ invariant and
            we have for any local observables $A,B$ that $\langle A \cdot \tau_xB(t)\rangle \to \langle A \rangle \langle B \rangle$, where $\tau_x$ is the shift in $\ZZ$ by $x$, giving
            \bes
                \langle S^{(1)}_0 S^{(1)}_x(t) \rangle \to 0,
            \ees
            as $||(x,t)||\to\infty$. Finally, for all $\eps>0$, we have that for any fixed $(x_0,t_0)\in\ZZ\times\RR$,
            \be\label{eq:q:xxz-no-exp-decay}
                \sum_{n=1}^\infty n \cdot \langle S^{(i)}_0 S^{(i)}_{nx_0}(nt_0)\rangle^{1-\eps} =\infty.
            \ee
            \end{enumerate}
\end{theorem}

As noted in the introduction above, the convergence in Part 1 of the theorem is known rigorously for all $\Delta<1$ due to \cite{dorlas-samsonov, karol1}. The polynomial decay \eqref{eq:xxz-poly-decay} is known for $\Delta=-1$ (the antiferromagnet) by \cite{dc-li-mano-FK} and known for $\Delta=0$ (the XY model) by \cite{mccoy68xy} with $\alpha_1=\alpha_2=\tfrac{1}{2}$ (see below). As far as the author is aware, the remaining correlation decay results and the results from Part 2 of the theorem that convergence holds on dynamic correlators and that the limits $\beta,L\to\infty$ can be interchanged are new.

Note that while we do not prove that the limits can be interchanged in the case of the periodic boundary conditions state $\langle \cdot \rangle_{\beta,\TT_L}$, this should be proveable. In the six-vertex model, the fact that the measure on the torus converges in infinite volume is a consequence of the recent delocalisation results and the work of Sheffield \cite{sheffield-thesis} (see Corollary \ref{cor:6v-unique-TI-gibbs} for more detail). One would therefore need to extend the work of Sheffield to suitable space-time models and combine it with our delocalisation result to obtain the interchangeability.

The XXZ model is the most general two-body $O(2)$-invariant spin system, where by invariant we mean that $g^{\otimes \Lambda} H_{\Lambda} (g^{-1})^{\otimes \Lambda} = H_{\Lambda}$ for all $g\in O(2)$, the complex orthogonal group. The general $O(n)$-invariant chain (here $n=2S+1$ where $S$ is the spin number) is expected to have different behaviour for $n\ge3$, indeed in the analogous range for $n\ge3$ dimerization is expected: two extremal ground states which are translations by 1 of each other. In fact the states $\langle \cdot \rangle_{\beta,\TT_L}^{\Psi^*_L}$ for $*=\bullet,\circ$ are expected to converge to the two extremal states, and this is proved for some regions of the parameters in \cite{ADCW, nacht-uelt, BMNU, bjornberg-ryan-dimerization}. See \cite{bjornberg-ryan-dimerization} for an overview of the expected behaviour for $n\ge3$. Our Theorem \ref{thm:main-xxz} shows in particular that this dimerization does not occur for $n=2$, as expected. Note also that the loop models studied in this paper can be seen as manifestations of the $O(2)$-invariance of the system.\\

Note a ground state is defined rigorously as follows. Let $\mathfrak{U}_{\mathrm{loc}}$ be the algebra of all local operators (linear operators on $(\CC^2)^{X}$ for all finite $X\subset \ZZ$) and let $\mathfrak{U}$ be the completion with respect to the operator norm. A state on $\mathfrak{U}$ is a linear map $\rho:\mathfrak{U}\to\CC$ such that $\rho(\id)=1$ and $\rho(A^\dagger A)=\ge 0$ for any $A\in\mathfrak{U}$. A ground state for the system \eqref{eq:ham-xxz} is then a state $\rho$ such that $\rho(A^\dagger [H_\Lambda,A])\ge 0$ for all $A\in\mathfrak{U}_{\mathrm{loc}}$ and for all $\Lambda$ containing the support of $A$, which is the set of $x\in\ZZ$ such that $A$ acts non-trivially on the copy of $\CC^2$ corresponding to $x$. It is an exercise to show that if it exists, the limit of a finite-volume ground state $\langle\cdot\rangle^{\mathrm{grd}}=\langle\Psi^{\mathrm{grd}} \cdot \Psi^{\mathrm{grd}}\rangle / \langle\Psi^{\mathrm{grd}} | \Psi^{\mathrm{grd}} \rangle$ for $\Psi^{\mathrm{grd}}$ the eigenvector of the lowest eigenvalue of $H_{\Lambda}$, as $\Lambda\to\ZZ$, is indeed an infinite volume ground state. We will not use this definition except for this fact.

\begin{theorem}\label{thm:xxz-decay-rates}
    Consider the XXZ model ground state $\langle\cdot\rangle$.
    \begin{enumerate}
        \item 
             Let $\Delta\in[-1,-\tfrac{1}{2}]$. We have that
            \be\label{eq:q:xxz-hugo}
                \langle S^{(3)}_0 S^{(3)}_x \rangle 
                =
                \frac{-1}{4\pi\cdot\arccos(\Delta)}
                |x|^{-2}
                + 
                o(|x|^{-2})
            \ee
            as $|x|\to\infty$.
        \item 
            In the special case $\Delta=0$, we have
            \be\label{eq:q:xy-decay-rate}
            	\langle S^{(1)}_0 S^{(1)}_x(t) \rangle \sim ||(x,t)||^{-\tfrac{1}{2}}
            \ee
            as $||(x,t)||\to\infty$.
    \end{enumerate}
\end{theorem}
The decay \eqref{eq:q:xxz-hugo} is a consequence of the recently proved convergence of the height function of the six-vertex model to the GFF, for $\Delta\in[-1,-\tfrac{1}{2}]$ \cite{dc-etal-GFF-convergence}, by Duminil-Copin, Kozlowski, Lammers, and Manolescu. Note that this input is the only place where Bethe ansatz results are used in this paper. 

The decay \eqref{eq:q:xy-decay-rate} is a consequence of \cite{XY-ising-jullien-fields} (also our new coupling) which writes a 1D quantum XY model as two independent 1D quantum Ising models, and of the decay proved for the 1D quantum Ising model in \cite[Theorem 3.9]{li-mahfouf-ising}. Note that by the same method one could show the decay for multi-point functions. Note the decay \eqref{eq:q:xy-decay-rate} is classical \cite{mccoy68xy} for $t=0$. Real-time correlations are studied in \cite{mccoy-perk-shrock-1, mccoy-perk-shrock-2}. \\

\subsubsection{Proof ideas}

As noted above, our proofs do not use the Bethe ansatz. Instead they arise from our study of related probabilistic models, which we describe rigorously below.

We first give an brief overview of the ideas of the proofs for the XXZ model. There are two key identities which we establish. The first is showing that arbitrary static observables in the XXZ chain can we written as ratios of partition functions in the six-vertex model, in particular as disorder operators (see Section \ref{sec:xxz-conv}). One should think of these disorders as having a finite set of edges being split in two, with both half-edges being directed inwards or both outwards. The second key identity is that the spin-spin correlation function $\langle S^{(1)}_0 S^{(1)}_x\rangle$ can be written as the probability that two edges are connected by a loop in the Lorentz mirror model with loop weight 2 (see Lemma \ref{lem:spin-spins-are-mirror-connections})
\bes
    \langle S^{(1)}_0 S^{(1)}_x\rangle
    =
    \lim_{V\nearrow\ZZ^2} \mu^{\mir}_{V,\ul{p}}[e_0 \leftrightarrow e_x],
\ees
where $e_0$, $e_x$ are vertical edges of $\ZZ^2$ on the same horizontal row at distance $x$. Both of these identities make use of the fact that the Hamiltonian $H_{\Lambda}$ \eqref{eq:ham-xxz} and the transfer matrix of the six-vertex model share a unique ground state vector.

Both of these discrete models (six-vertex and mirror) can be thought of as discretised versions of 1+1D models (discrete in one direction and continuous in the other), which in this paper we call space-time models. There are two analogous identities to above: we show first that arbitrary local observables (including (Euclidean) dynamic ones) in the XXZ model can be written as disorder operators in the space-time six-vertex model (see Section \ref{sec:xxz-conv-cts}). Secondly, $\langle S^{(1)}_0 S^{(1)}_x(t)\rangle$ can be written as the probability that two points are connected by a loop in the space-time loop model (this was proved by Ueltschi \cite{ueltschi}): 
\bes
    \langle S^{(1)}_0 S^{(1)}_x(t) \rangle 
    =
    \frac{1}{4}\lim_{F\nearrow\cF}\nu^{\lp}_{\cF,\Delta}[(0,0)\leftrightarrow (x,t)].
\ees

Given these identities, the challenge is then to show that these disorder variables converge in infinite volume, and that the connection probabilities decay appropriately. Our proof exploits fully the recent delocalisation result in the six-vertex model by several authors, and the delocalisation for the space-time six-vertex model that we prove in Section \ref{sec:delocalisation}, and in particular the percolation process $\xi^\bullet$ and its domain Markov property. Much of the proof is then devoted to proving that even with disorders present, one still sees infinitely many circuits of $\xi^\bullet$. Note that the percolation process $\xi^\bullet$ was instrumental in the latest proof of delocalisation due to Glazman and Lammers \cite{glazman-lammers} and we use a version of it in our proof for the space-time model. The process appears in several papers \cite{lis-spins, lis-deloc, glaz-peled, lis-lopez-heeney, dc-etal-GFF-convergence, ray-spinka-2, chayes-mckellar-winn} under different names - see below \eqref{eq:sampling-xi} for more detail.

\subsection{The Lorentz mirror model with loop weight 2}	


The (unwieghted) Lorentz mirror model is a model of random loops on $\ZZ^2$. At each vertex, independently, one places one of three configurations: a north-west oriented mirror placed at $45^\circ$ to the lattice with probability $p_{\NW}$, a north-east oriented mirror with probability $p_{\NE}$, or no mirror with probability $p_{\es}$, where $p_\NW+ p_\NE+ p_\es=1$. Such a configuration at each site deterministically gives a sets of loops, which are set of edges of $\ZZ^2$: the loops reflect off the mirrors when they are present and pass straight through the vertex when there is no mirror. The main question of interest is finding the probability that there exist infinite loops. 

\begin{figure}[h]
\centering
    \resizebox{0.4\textwidth}{!}{%
        \includegraphics[]{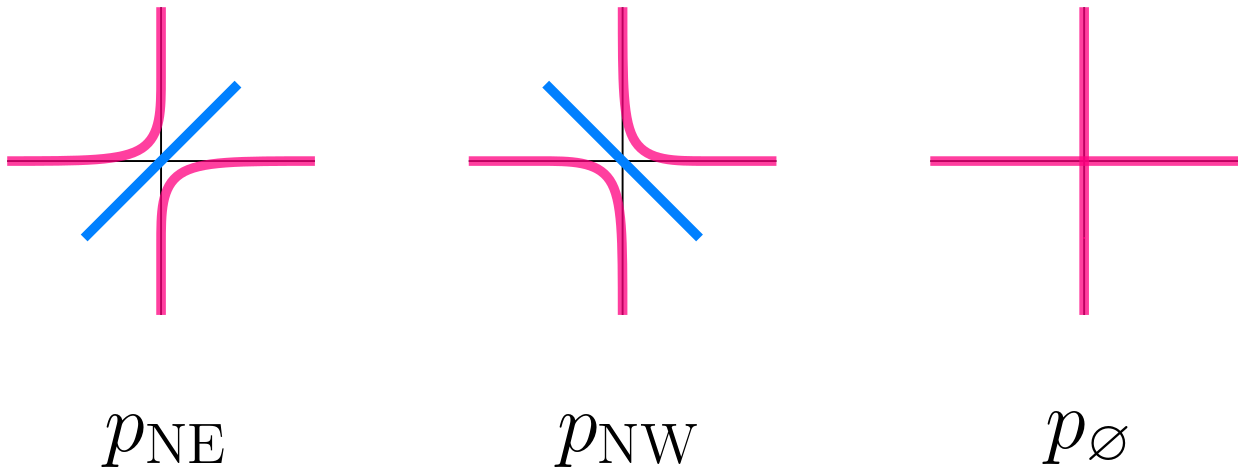}}
    \caption{The three configurations possible at a vertex in the mirror model}
    \label{fig:mirror}
\end{figure}

While the model is very simple, this problem is largely open: for $p_\es=0$ the model is equivalent to critical Bernoulli percolation on $\ZZ^2$, with loops being the boundaries between primal and dual clusters, and so it is known that there are no infinite loops almost surely. For $p_\es>0$ the problem is open. The main other rigorous result on the model is due to Kozma and Sidoravicius \cite{koz-sid}, that the probability that the origin is connected by a loop to the edge of the box $\Lambda_n=[-n,n]^2$ of size $n$ decays slowly: $\PP[0\leftrightarrow\partial\Lambda_n]\ge 1/(2n+1)$. 

Li \cite{li-mirror-poly} showed that in the model on the cylinder of width $n$, loops do not travel distance more than $n^{10}$, while in \cite{ryan-mirrors} the author showed that if the probability of a mirror ($p=p_{\NE}+p_{\NW}$) small (less than $1/n$), then the loops are at most length $p^{-2}$. Li \cite{li-manhattan-21} extended the range of parameters for which there are no infinite loops in the related Manhattan model. The mirror model has higher-dimensional analogues; for $d\ge3$ it is conjectured that there are infinite loops for $p$ small enough. In \cite{elboim-et-al-mirrors} it is proved that for $d\ge4$ and $p$ small enough, the trajectory of the loop at the origin remains diffusive for at least $p^{-M}$ steps, for any fixed $M>0$. \\


In the present paper we deal with a variant of the Lorentz mirror model, that is, a loop-weighted version. One defines the model on a finite portion of $\ZZ^2$ with some given boundary conditions and then re-weighs the measure by $n^{\# \text{loops}}$; $n=1$ recovers the original model. This model was studied in \cite{nienhuis-rietman, martins-nienhuis-rietman}. Here we conjecture a dichotomy in $n$: for all values of $p_\NW, p_\NE, p_\es$,
\be
\begin{cases}
    \PP[0\leftrightarrow x] \ \text{decays polynomially fast} & n\in[1,2]\\
    \PP[0\leftrightarrow x] \ \text{decays exponentially fast} & n>2,
\end{cases}
\ee
where $\PP$ is any infinite volume Gibbs measure of the model. This conjecture is proved for $p_{\es}=0$ \cite{dc-sid-tass-cont,dc-discont-q>4}; here the model is equivalent to critical planar FK percolation with parameter $q=n^2$, with the loops being the boundaries of the primal and dual clusters. For $n$ large the conjecture is provable using the techniques of \cite{dc-etal-exp-decay}, see \cite{bjornberg-ryan-dimerization} for the statement. Note that the torus limit for $p_{\es}=0, n\in(0,1]$ is proved to have no infinite loops by \cite{glazman-lammers}; it may be that the conjecture fully extends to $n\in(0,1]$. Apart from these cases, the conjecture is open.

Our result below studies the model with $n=2$, and proves that if $p_\NW,p_\NE\ge p_\es$, then connection probabilities tend to zero and do not decay exponentially fast, and further in the symmetric case $p_\NW=p_\NE$, the decay is polynomial. In particular this is the first result that shows connection probabilities go to 0 for this model outside of the ranges $\p_\es=0$ and $n$ large, and the first to show polynomial decay outside of the range $\p_\es=0$.\\

\subsubsection{Mirror model results}

Let $\ZZ^2=(\VV,\EE)$ be the 2-dimensional square lattice. Let the faces $\FF$ of $\ZZ^2$ have a black/white chessboard colouring, with the face centred at $(\tfrac{1}{2},\tfrac{1}{2})$ coloured black. Let $V$ be a simply connected (with respect to the graph connectivity), finite subset of $\VV$, such that $\VV\setminus V$ is connected. Let $\Om^{\mir}_V$ be the set of configurations which assign each vertex in $V$ either a north-west (NW) oriented mirror, a north-east (NE) oriented mirror, or no mirror; formally functions $m:V\to\{\NW,\NE,\es\}$. For $m\in\Om^{\mir}_V$, write $m_{\NW}$ (resp. $m_{\NE}$, $m_\es$) for the set of vertices assigned a NW mirror (resp. NE mirror, no mirror). 

Let $\LL^\bullet=(\FF^\bullet, \EE^\bullet)$ be the graph isomorphic to $\ZZ^2$ with vertex set given by the black faces of $\ZZ^2$, and edges given by nearest neighbours. Define $\LL^\circ=(\FF^\circ,\EE^\circ)$ similarly with the white faces. Each mirror can be interpreted as an edge in either $E^\bullet$ or $E^\circ$, we write $m^\bullet$ (resp. $m^\circ$) for the set of vertices assigned a mirror in $E^\bullet$ (resp. $E^\circ$).

A mirror configuration on $\VV$ gives a set of (closed) loops and bi-infinite paths - each edge of $\ZZ^2$ is on exactly one loop or path, and at vertices, loops/paths reflect off the mirror if it is present, and otherwise pass straight through the vertex. See Figure \ref{fig:mirror}.

Write $\cP$ for the set of vectors $\ul{p}=(p_\NW, p_\NE, p_\es)$ such that 
\be\label{eq:p-sum-to-1}
    p_\NW+ p_\NE+ p_\es=1,
\ee 
and $0<p_\NW, p_\NE<1$ and $0\le p_\es<1$. For $\ul{p}\in\cP$ and a configuration $m'\in\Om^\mir_\VV$, define a probability measure $\mu^{\mir;m'}_{V,\ul{p}}$ on $\Om^{\mir}_V$ as
\be\label{eq:mir-measure}
\mu^{\mir;m'}_{V,\ul{p}}[m]
\propto
p_\NW^{|m_\NW|} p_\NE^{|m_\NE|} p_\es^{|m_\es|}2^{l(m;m')},
\ee
where $l(m,m')$ is the number of loops or bi-infinite paths of the configuration $m|_V \cup m'|_{\VV\setminus V}$ passing through points of $V$.

We write $\mu^{\mir;\bullet}_{V,\ul{p}}$ for the measure when $m'$ has $(m')^\bullet=\VV$, that is, $m'$ has mirrors everywhere, such that the mirrors always join two black faces (ie. they lie in $E^\bullet$); we write $\mu^{\mir;\circ}_{V,\ul{p}}$ for when $(m')^\circ=\VV$. Write $\partial_E V=\{e=uv\in\EE \ : \ v\in V, u\notin V\}$. Write $\mu^{\mir}_{V,\ul{p}}$ for the model where in $m'$, all loops which after leaving $V$ through edges of $\partial_E V$, never return to $V$. One should think of this as free boundary conditions; one obtains the same measure by replacing $l(m;m')$ by the number of connected components of loops or bi-infinite paths passing through vertices of $V$. 

We write $\mu_{\TT_{L,M},\ul{p}}^{\mir}$ for the measure where $V={\TT}_{L,M}$ is the torus $\TT_{L,M}$ with vertex set $\{-L+1,\dots,L\}\times\{-M+1,\dots,M\}$, and $2^{l(m;m')}$ is replaced by $2^{l(m)}$, $l(m)$ being simply the number of loops in the configuration. Let $\cyl_{L,M}$ be the finite cylinder (periodic in the horizontal direction) also with vertex set $\{-L+1,\dots,L\}\times\{-M+1,\dots,M\}$; write $\partial^+=\{-L+1,\dots,L\}\times\{M\}$ and $\partial^+=\{-L+1,\dots,L\}\times\{-M+1\}$.
Write $\mu_{\cyl_{L,M},\ul{p}}^{\mir,\bullet}$ for the measure on the cylinder $\cyl_{L,M}$ with boundary conditions given by $v\in m^\bullet$ for all $v\in\partial^+\cup\partial^-$, and again $2^{l(m;m')}$ is replaced by $2^{l(m)}$. Define $\mu_{\cyl_{L,M},\ul{p}}^{\mir,\circ}$ similarly.

A measure $\mu$ on $\Om^\mir_{\VV}$ is called Gibbs if, for all finite $V\subset\VV$, 
\bes
\mu[m|_V=\eta \ | \ m|_{\VV\setminus V}=m']
=
\mu^{\mir;m'}_{V,\ul{p}}[\eta].
\ees

\begin{theorem}[The mirror model with loop weight 2]\label{thm:main-mirror} 
Consider the mirror model with parameters $\ul{p}$ and loop weight 2, where $\ul{p}\in\cP$ such that $p_\NW,p_\NE\ge p_\es$. 
\begin{enumerate}
    \item 
        The measures $\mu_{V,\ul{p}}^{\mir}$, $\mu^{\mir;\bullet}_{V,\ul{p}}$, $\mu^{\mir;\circ}_{V,\ul{p}}$, $\mu_{\TT_{L,M},\ul{p}}^{\mir}$, $\mu_{\cyl_{L,M},\ul{p}}^{\mir,\bullet}$, $\mu_{\cyl_{L,M},\ul{p}}^{\mir,\circ}$ converge weakly to the common limit $\mu_{\ul{p}}^\mir$ as $V\nearrow\ZZ^2$ or $L,M\to\infty$ (in either order or simultaneously). The limit is tail-trivial, $\ZZ^2$-ergodic and translation-invariant.
    \item 
        We have that for edges $e_1,e_2$ of $\ZZ^2$
        \bes
            \mu_{\ul{p}}^{\mir}[e_1\leftrightarrow e_2]
            \le 
           \lim_{V\nearrow\ZZ^2}
            \mu_{V,\ul{p}}^{\mir}
            [e_1\leftrightarrow e_2]
            \to 0
        \ees 
        as $||e_1-e_2||\to\infty$, where $e_1\leftrightarrow e_2$ is the event that $e_1$ and $e_2$ lie on the same loop or bi-infinite path. Further, for any edges $e_1,e_2$ of $\ZZ^2$, and any $x\in\ZZ^2$, we have
        \bes
    	   \sum_{n=1}^\infty n\cdot
           \lim_{V\nearrow\ZZ^2}
            \mu_{V,\ul{p}}^{\mir}
            [e_1\leftrightarrow \tau_{nx}e_2]
           = \infty,
        \ees 
        where $\tau_x$ is the translation by $x$ in $\ZZ^2$. In particular $\lim_{V\nearrow\ZZ^2}
            \mu_{V,\ul{p}}^{\mir}
            [e_1\leftrightarrow e_2]$ 
        cannot decay exponentially fast in the distance $||e_1-e_2||$.
    \item 
        If $p_\NW=p_\NE\ge p_\es$, there exist $\alpha_1,\alpha_2, C_1,C_2 >0$ such that for all edges $e_1,e_2$ of $\ZZ^2$ with $||e_1-e_2||=n$, we have 
        \be\label{eq:mir-poly-decay}
        \begin{split}
        	\mu_{\ul{p}}^\mir[e_1\leftrightarrow e_2]
        	&\le
        	C_2 n^{-\alpha_2},\\
        	C_1 n^{-\alpha_1}
        	\le
            \lim_{V\nearrow\ZZ^2}
            \mu_{V,\ul{p}}^{\mir}
            [e_1\leftrightarrow e_2]
            &\le
        	C_2 n^{-\alpha_2}.
        \end{split}
        \ee
        Fix a face $f_0$ and let $\tilde{e}_n$ denote the $n^{th}$ vertical edge to the right of $f_0$. Then \eqref{eq:mir-poly-decay} holds for $e_1=\tilde{e}_1$ and $e_2=\tilde{e}_n$ in the general case $\ul{p}\in\cP$ with $p_\NW,p_\NE\ge p_\es$.
    \item 
        In the case $p_\NW\cdot p_\NE= p_\es$, we also have \eqref{eq:mir-poly-decay} for all pairs of edges, with $\alpha_1=\alpha_2=\tfrac{1}{2}$. 
    \item 
        Finally, if $p_\NW\cdot p_\NE\le p_\es$, then the lower bound of \eqref{eq:mir-poly-decay} holds for all pairs of edges, and the measure $\mu_{\ul{p}}^\mir$ is Gibbs, and moreover is the unique Gibbs measure and the unique thermodynamic limit, that is, for all $m'\in\Om^\mir_{\ZZ^2}$, we have the weak convergence
        \be
        	\lim_{V \nearrow \ZZ^2} \mu^{\mir;m'}_{V,\ul{p}} = \mu_{\ul{p}}^\mir.
        \ee
\end{enumerate}
\end{theorem}
\begin{remark}
    Note that while we do have that connection probabilities decay to 0 by Part 2 of Theorem \ref{thm:main-mirror}, we do not prove that the limiting measure $\mu^{\mir}_{\ul{p}}$ exhibits no infinite loops almost surely (although this is indeed expected). This is due to a lack of monotonicity in the model.
\end{remark}

We actually provide two proofs that connection probabilities in this mirror model tend to 0 in the distance. The first is a consequence of the XXZ convergence result; in particular the proof of extremality in Section \ref{sec:extremality}. The second uses our new coupling; one takes the delocalisation result through the Ashkin-Teller model and eight-vertex models to the mirror model - this proof is a consequence of Section \ref{sec:8v-AT}.

\subsubsection{Proof ideas}

The main input for proving Theorem \ref{thm:main-mirror} is the recent result of several authors that the height function of the six-vertex model delocalises. The six-vertex model is a measure on arrow configurations on $\ZZ^2$ such that each vertex has exactly two incoming and two outgoing arrows. There are six possible configurations at a given vertex (see Figure \ref{fig:six-vertex}), given weights $a,b,c>0$. This can be seen as a random height function (a map $h:F\to\ZZ$, $F$ the faces) where $h(y)=h(x) +1$ if when moving from the face $x$ to the adjacent face $y$, one crosses an arrow oriented left, and $h(y)=h(x) -1$ if it is oriented right. If one defines the model on a finite domain to have flat boundary conditions (height 0 or 1), the delocalisation states that the variance of this height function at any given face tends to infinity as the domain tends to $\ZZ^2$. See Theorem \ref{thm:6v-deloc} for a formal statement.

\begin{figure}[h]
\centering
    \resizebox{0.7\textwidth}{!}{%
        \includegraphics[]{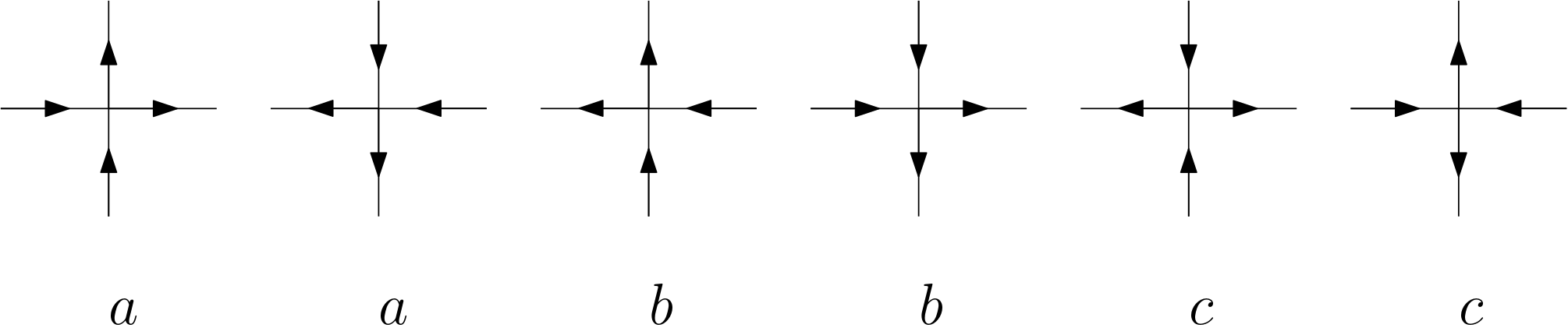}}
    \caption{The six possible configurations at a vertex in the six-vertex model, with given weights}
    \label{fig:six-vertex}
\end{figure}

Delocalisation is expected for all $|a-b| < c \le a+b$ (this is known as the disordered regime of the six-vertex model). The statement was proved for $c=a+b$ by Glazman and Peled \cite{glaz-peled}, for $a=b=c$ by Chandgotia, Peled, Sheffield and Tassy \cite{sheffield-square-ice} (later the stronger logarithmic delocalisation, that is, that the variance grows logarithmically in the diameter of the domain, was proved by Duminil-Copin, Harel, Laslier, Raoufi and Ray \cite{dc-square-ice}). 

Lis \cite{lis-spins} showed that delocalisation can be proved if one shows that a certain related percolation process has no infinite clusters. He then applied this to prove delocalisation for $(2+\sqrt{2})a=(2+\sqrt{2})b\le c \le a+b$ in \cite{lis-deloc}. Shortly afterwards Duminil-Copin, Karrila, Manolescu and Oulamara \cite{dckmo-6v-deloc} proved the statement for $a=b\le c\le a+b$ using Bethe ansatz methods (they also proved logarithmic delocalisation in this range). Finally Glazman and Lammers \cite{glazman-lammers} then proved the statement for the range $a,b\le c \le a+b$, recovering logarithmic delocalisation for $a=b\le c\le a+b$, using the percolation technique. Duminil-Copin, Kozlowski, Lammers and Manolescu have since proved the much stronger statement that the height function converges to the Gaussian Free Field for $\Delta\in[-1,-\tfrac{1}{2}]$. 

Our work takes the last of the delocalisation results, that of Glazman and Lammers, as its sole major input. Note that this result, as well as the present paper, makes no use of the Bethe ansatz. Our proof of delocalisation of the space-time six-vertex model follows the percolation proof of Glazman and Lammers closely.  \\

Crucial to the proof of Theorem \ref{thm:main-mirror} are two couplings of the Lorentz mirror model with the six-vertex model. The first has been known at least since the work of Chayes and Stengel \cite{chayes-shtengel}, and is extremely simple: from the loop model, independently and uniformly orient each loop; this gives an orientation to each edge and taking the marginal on these orientations one has the six-vertex model. Going in the other direction, one takes a six-vertex configuration and samples mirrors at each vertex with weights proportional to $p_\NW,p_\NE, p_\es$, conditional that the loops formed are consistently oriented. See Figure \ref{fig:mir-6v}.

\begin{figure}[h]
\centering
    \resizebox{0.9\textwidth}{!}{%
        \includegraphics[]{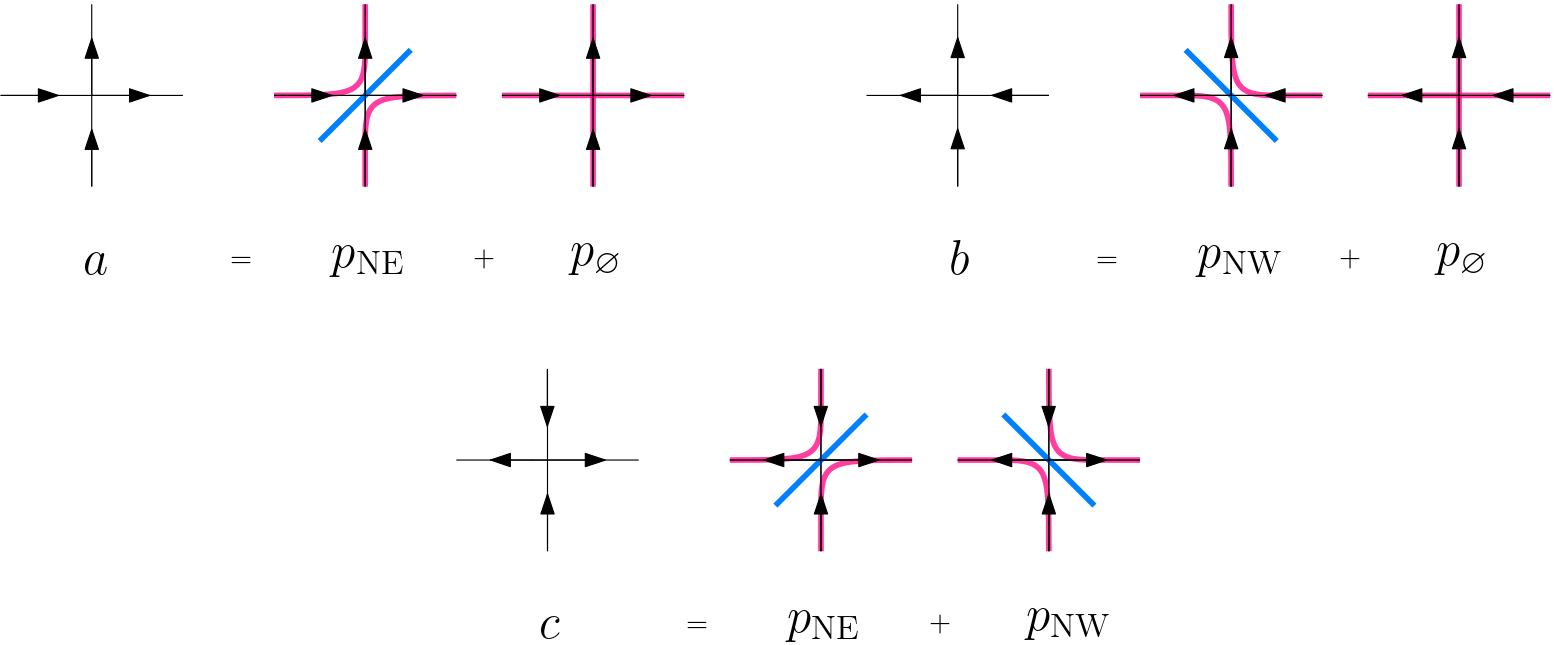}}
    \caption{The basic coupling between the mirror and six-vertex models}
    \label{fig:mir-6v}
\end{figure}

One obtains the relationship between the weights of the models: $a=p_\es+p_\NE,
	b=p_\es+p_\NW,
	c=p_\NW+p_\NE,$
where $\ul{p}\in\cP$. Note that the weights $a,b,c$ can be written in this form with $\ul{p}\in\cP$ if and only if they lie in the disordered regime of the six-vertex model, that is: $|a-b|< c \le a+b$. The existence of the coupling seems therefore to be a genuine feature of the disordered regime.

Note that the Baxter-Kelland-Wu coupling between the six-vertex model and FK-percolation also forms oriented loops from a six-vertex configuration, but this coupling is distinct; indeed, it produces only non-crossing loops, while for $p_\es>0$ the mirror model allows loops which cross one another. The two couplings coincide only when $p_\es=0$ ($c=a+b$ in six-vertex), where the loops are the loops separating primal and dual clusters of critical FK-percolation with parameter $q=4$, and the associated Ashkin-Teller model has $J=U$ and is a 4-state Potts model.

This coupling between the mirror model and six-vertex model is how we prove the convergence statements of \ref{thm:main-mirror} and also the fact that connection probabilities do not decay exponentially fast. Moreover it is used in our proof of convergence of the XXZ model states to the infinite volume ground state in Theorem \ref{thm:main-xxz}.\\

\begin{figure}[h]
\centering
    \resizebox{0.4\textwidth}{!}{%
        \includegraphics[]{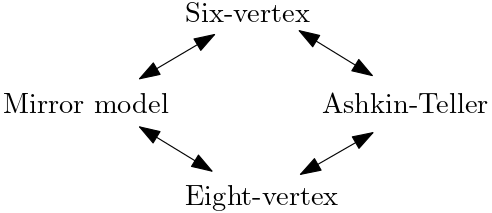}}
    \caption{The four models and the pairs that can be coupled}
    \label{fig:couplings-diagram}
\end{figure}

The second coupling goes through two other related models: an eight-vertex model and the Ashkin-Teller model. To be precise, our new coupling is one between the mirror model and the eight-vertex model, which in turn is naturally coupled with the Ashkin-Teller model, which in turn has a coupling with the six-vertex model. Our new coupling therefore completes an interesting circuit of couplings of the four models. See Figure \ref{fig:couplings-diagram} for an illustration.

\begin{figure}[h]
\centering
    \resizebox{0.9\textwidth}{!}{%
        \includegraphics[]{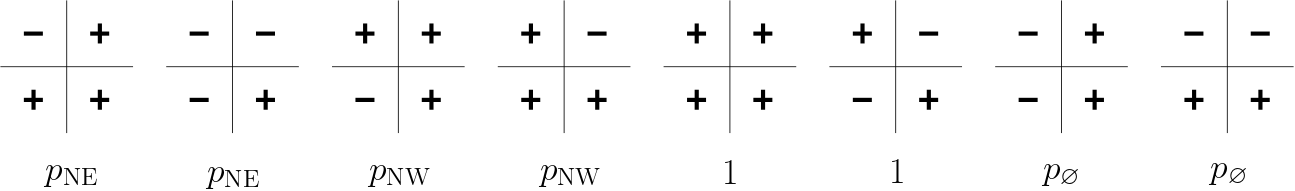}}
    \caption{The eight possible configurations at a vertex in the eight-vertex model (up to overall sign flip), where we fix the bottom right face to be +.}
    \label{fig:eight-vertex}
\end{figure}

The eight-vertex model is presented as two interacting Ising models $\pi^\bullet$ and $\pi^\circ$, one on the black faces of $\ZZ^2$ and one on the white faces. The weights at a vertex are given by the configuration in the four neighbouring faces. See Figure \ref{fig:eight-vertex} for all eight possibilities.

\begin{figure}[h]
\centering
    \resizebox{0.7\textwidth}{!}{%
        \includegraphics[]{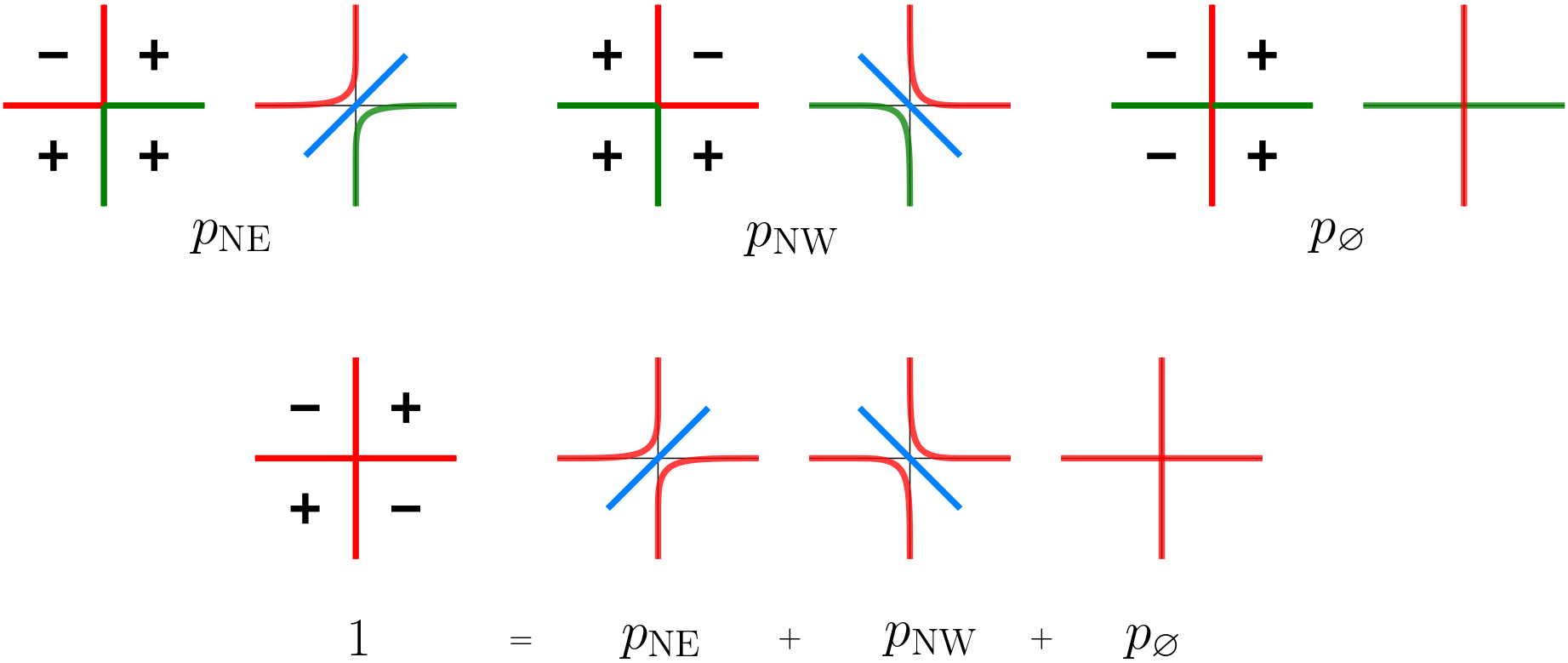}}
    \caption{The new coupling between the mirror and eight-vertex models}
    \label{fig:mir-8v}
\end{figure}

The new coupling is also very simple: instead of orienting the loops, one colours them (say, red or green) independently and uniformly. The red loops are then the domain walls of a spin configuration $\pi=(\pi^\bullet,\pi^\circ)$ on the faces of $\ZZ^2$. To go the other way, one colours the domain walls of $\pi$ red and other edges green, and then samples mirrors at each vertex such that loops are consistently coloured (note that some vertices have only one choice for their mirror configuration). See Figure \ref{fig:mir-8v}.

The main advantage of our new coupling is that the observable that two edges are connected in the mirror model can be analysed in the eight-vertex model. the connection probability is then a four-point correlation function in the eight-vertex model:
\bes
    \begin{split}
    \mu^{\mir}_{V,\ul{p}}[e_1\leftrightarrow e_2]
    =
    \mu^{\eiv}_{V,K,W}
    [\pi^\bullet_{x_1}\pi^\circ_{y_1}\pi^\bullet_{x_2}\pi^\circ_{y_2}],
    \end{split}
\ees
where $x_i,y_i$ are the (resp. black, white) faces of $\ZZ^2$ either side of the edges $e_i$. See Lemma \ref{lem:8v-connection-probs}. This is the crucial tool which allows us to prove the polynomial decay of connection probabilities in Theorem \ref{thm:main-mirror}, as well as the polynomial decay of the spin-spin correlations $\langle S^{(1)}_0 S^{(1)}_x\rangle$ in the XXZ chain.

The final step is to express the eight-vertex four-point function as a connection probability in a certain random cluster representation of the Ashkin-Teller model:
\bes
    \begin{split}
    \mu^{\eiv}_{V,K,W}
    [\pi^\bullet_{x_1}\pi^\circ_{y_1}\pi^\bullet_{x_2}\pi^\circ_{y_2}]
    =
    \mu^{\eivRC}_{\ul{p}}[x_1\xleftrightarrow{\eta^\bullet} x_2,
	y_1\xleftrightarrow{\eta^\circ} y_2].
    \end{split}
\ees
Here $\mu^{\eivRC}_{\ul{p}}$ is a percolation model with $\eta^\bullet$ a percolation on the graph with vertices given by the black faces of $\ZZ^2$, and $\eta^\circ$ one on the white faces. We can then apply a dichotomy argument similar to \cite{dc-sid-tass-cont, dc-tass-renorm} to obtain the fact that this connection probability converges to 0 and indeed does so polynomially fast.

For readers familiar with the Ashkin-Teller random cluster model, if one takes $(\om,\om')$ the representation by expanding in the two single spins $\tau,\tau'$, then $(\eta^\bullet,\eta^\circ)=(\om,(\om')^*)$, where $^*$ denotes the dual configuration. Moreover the eight-vertex model itself is obtained by fixing one spin, $\tau$, and applying duality in $\tau'$. Note that in contrast the six-vertex model is obtained from Ashkin-Teller by fixing the \textit{product} of spins $\tau\tau'$ and applying duality in $\tau$.


\subsection{Results for space-time models}

Throughout our work on space-time models we will make use of the parameter $u$, satisfying
\bes
    2u-1=\Delta.
\ees
Note $u\in[0,1)$ whenever $\Delta\in[-1,1)$.\\

Quantum spin systems have several representations as probabilistic models all having a common geometric feature that they are continuous in one direction: to be precise they are models on $\Lambda\times[-\beta/2,\beta/2]$ where $\Lambda$ is the graph on which the spin system is defined and $\beta$ is the inverse temperature of the spin system. The ground states of a quantum spin chain can therefore be studied via two-dimensional models on $\ZZ\times\RR$. In this work we refer to this general class of models as ``space-time'' models.

One of the most widely studied of these space-time representations is a model of random loops. We present results on this model in the case of the XXZ chain below in Section \ref{sec:loop-results}. First, we present a delocalisation result for a space-time version of the six-vertex model. 

Note that space-time representations on $\Lambda\times[-\beta/2,\beta/2]$ can generically be described as limits of truly discrete classical models on $\Lambda\times\ZZ_m$ as $m$ is taken to infinity (and thereby $\ZZ_m$ to $[-\beta/2,\beta/2]$) and certain parameters of the model are tuned. Our space-time six-vertex model is indeed such a limit of the ordinary six-vertex model, and the above-mentioned loop model is such a limit of the Lorentz mirror model with loop weight 2 (see Appendix \ref{appendix:FKG}). In fact, every feature of the expansive discrete 2d story of interconnected models including the six-vertex model, FK percolation, the Ising and Potts models, the Ashkin-Teller models, etc, should find its space-time analogue. In this paper we only cover the four models involved in the couplings in Figure \ref{fig:couplings-diagram}: The mirror model, the six-vertex model, the Ashkin-Teller model and the eight-vertex model (see Section \ref{sec:quantum-AT-8v} for our statements on the space-time Ashkin-Teller and eight-vertex models).

While the fact that versions of such a limiting procedure connects the six-vertex and XXZ models (for example in terms of transfer matrices) has been known for decades, to the best of author's knowledge this space-time six-vertex model has only been studied before in the work of Aizenman, Duminil-Copin and Warzel \cite{ADCW} for $\Delta<-1$, and has otherwise not been explicitly studied for $\Delta\in[-1,1)$ before.

\subsubsection{The space-time six vertex model}\label{sec:space-time-6v-results}
To define our space-time models we will need some geometric definitions. We define vertices, faces and black and white edges to be a certain limit of those in the discrete $\ZZ^2$ case described above. In particular, we rotate $\ZZ^2$ by $\pi/4$, scale and shift it so that its vertices (and the centres of faces) lie on $(\ZZ+\tfrac{1}{2})\times\RR$, with the (centres of the) black faces lying in $(2\ZZ+\tfrac{1}{2})\times\RR$. We then take a scaling limit in the vertical direction. With this in mind, we have the following definitions.

Let $\cF=(\ZZ+\tfrac{1}{2})\times\RR$. Let $\cF^\bullet=(2\ZZ+\tfrac{1}{2})\times\RR$ and $\cF^\circ = (2\ZZ+\tfrac{3}{2})\times\RR$. We define a graph connectivity on $\cF$ by saying that points are connected as usual within the copies of $\RR$, and also each $v\in\cF$ is connected to $x+(1,0)$ and $x+(-1,0)$. Define a metric on $\cF$ by restricting the $||\cdot||_1$ norm to $\cF$.


Let $V\subset\cF$ be simply connected with respect to the above connectivity, bounded, composed of a finite union of disjoint open intervals in $\cF$, such that for all $(x,t)\in\cF$, we have that if $(x-(1,0),t), (x+(1,0),t)\in V$ then $(x,t)\in V$. Let $F$ be the closure of $(\overline{V}\cup (V+(1,0)) \cup (V-(1,0)))$, where $\overline{V}$ is the closure of $V$ in $\cF$ (this is the continuous analogue of the faces with at least one corner in $V$). Define $\partial_\horiz F$ to be the points $x$ of $F$ which have $x+(1,0)$ or $x-(1,0)$ not in $V$ (note $\partial_\horiz F$ is a union of intervals). Let $\partial_\vert F= \overline{V}\setminus V$ (note $\partial_\vert F$ is a finite set of discrete points). Write $\partial F = \partial_\horiz F \cup \partial_\vert F$ (this is the continuous analogue of faces in $F$ with at least one corner not in $V$).

We say the pair $(V,F)$ constructed in this way is a domain, and we usually refer to just $F$ as the domain.\\

We define the height function of the space-time six-vertex model. The model can also be expressed in terms of arrow configurations like in the discrete model; see Section \ref{sec:delocalisation}. 

For a domain $F$, let $\ul{\Om}^{\sv-\hom}_{F}$ be the set of functions $h:F\to\ZZ$ which are piecewise constant with a finite number of points of discontinuity in any compact interval, and right continuous in each copy of $\RR$, such that $|h(v)-h(v+(1,0))|=1$ for all pairs $v,v+(1,0)\in F$, and that $h$ takes even values on $\cF^\bullet$ (and therefore odd on $\cF^\circ$). Notice that at each point of discontinuity in the vertical direction, the height function changes by $\pm2$. 

Write $\nu^{\poi}_{F,1}(\chi)$ for a Poisson point process of rate 1 on $F$, where $\chi$ is the random variable given by the set of Poisson points. For $\tilde{h}\in\ul{\Om}^{\sv-\hom}_{\cF}$, define the measure
\be\label{eq:q:6v-hom-measure}
\d \nu^{\sv-\hom;\tilde{h}}_{F,\Delta}
\propto
\d \nu^{\poi}_{F,1}(\chi)
\mathbbm{1}_{\{h\sim\chi\}}
\exp\left[u\left|\om^{b}[h]\right| + (1-u)\left|\om^{c}[h]\right|\right]
\mathbbm{1}_{h|_{\partial F} = \tilde{h}}
,
\ee
where $h\sim\chi$ if $\chi$ is precisely the points of discontinuity of $h$ (one should think of these points as the $a$ vertices of the six-vertex model). Here $\om^{b}[h]$ is the set of $v\in V$ such that $h(v+(1,0)) = h(v-(1,0))\pm2$ (one should think of these points as the $b$ vertices of the six-vertex model) and $\om^{c}[h]$ is the set of $v\in V$ such that $h(v+(1,0)) = h(v-(1,0))$ (these are the $c$ vertices). Write $\nu^{\sv-\hom;0,1}_{F,\Delta}$ for the same measure where $\tilde{h}$ takes only the values 0 or 1. Note that we can define the measure \eqref{eq:q:6v-hom-measure} for $\tilde{h}$ defined only on $\partial F$, provided $\tilde{h}$ is admissible, that is, there exists a valid height function $h\in\ul{\Om}^{\sv-\hom}_{F}$ with $h|_{\partial F}=\tilde{h}$.

\begin{theorem}[Delocalisation in the space-time six-vertex model]\label{thm:Q:delocalisation-basic}
	Let $\Delta\in[-1,0]$, so $u\in[0,\frac{1}{2}]$. For any fixed $x\in\cF=(\ZZ+\tfrac{1}{2})\times\RR$,
	\be\label{eq:Q:delocalisation}
	\lim_{F\nearrow \cF}\Var_{\nu^{\sv-\hom;0,1}_{F,\Delta}}
	\left[h(x)\right]=\infty.
	\ee
\end{theorem}

The space-time six-vertex model can also be written in terms of arrow configurations, as well as in terms of a spin model of two interacting space-time Ising models. We give a more comprehensive statement including the behaviour of these representations in Theorem \ref{thm:Q:delocalisation}.

\subsubsection{Ueltschi's loop model}\label{sec:loop-results}

As noted above, one of the most widely studied space-time representation of quantum spin systems is the following loop model. The model is sometimes known as the Interchange process (with or without reversals); it is a model of random loops lying on $\Lambda\times [0,\beta]$ determined by two Poisson point processes, one of ``crosses'' and one of ``double bars'' on $E(\Lambda)\times[0,\beta]$ and then by re-weighing by $n^{\# \mathrm{loops}}$. See below for the formal definition.

The model with only crosses and $n=2$ was introduced by Tóth \cite{toth} as a representation of the spin-$\tfrac{1}{2}$ quantum Heisenberg ferromagnet ($\Delta=1$); the long-standing open problem of proving a phase transition in this model in $d\ge3$ is equivalent to proving there are infinite loops in the loop model. The model with only double bars was then studied by Aizenman and Nachtergaele \cite{aiz-nacht} as a representation of the antiferromagnet ($\Delta=-1$). They observed that in the dimension 1 case in the ground state (where the loop model is on $\ZZ\times\RR$), that the model is equivalent to a space-time FK percolation - the loops are the loops separating the primal and dual clusters (this is the analogue of the $p_{\es}=0$ case of the mirror model). They proved a dichotomy; that correlations decay exponentially fast or slowly in the sense: $\sum_{x\in\ZZ}|x|\langle \bf{S}_0\cdot \bf{S}_x\rangle$. This dichotomy was resolved by work on the space-time FK percolation model by Duminil-Copin, Li and Manolescu \cite{dc-li-mano-FK}.

Later Ueltschi \cite{ueltschi} introduced the more general model which combines that of Tóth and Aizenman-Nachtergaele. The model is a representation of general $O(n)$-invariant quantum spin systems ($n=2S+1$ here where $S$ is the spin number), where the spin-$\tfrac{1}{2}$ ($n=2$) model is precisely the XXZ model. The loop model for XXZ is probabilistic so long as $\Delta\in[-1,1]$; otherwise it has negative weights. Since Ueltschi's work, where he used the loop model to prove long-range order in these models in dimensions $d\ge3$ and general spin number, the loop model has been studied (including in its own right) in myriad settings; see references in \cite{bjornberg-ryan-dimerization,, BMNU, ryan-orth}. As far as the author is aware, our work is the first to cover the $d=1$ spin-$\tfrac{1}{2}$ (XXZ) case away from the point $\Delta=-1$, the antiferromagnet.\\

Let us define the loop model carefully. Recall the definition of a domain $(V,F)$, which we usually denote by just $F$. We say $F$ is a black domain (resp. white domain) if $(\partial_\horiz F \setminus V) \subset \cF^\bullet$ (resp. $(\partial_\horiz F \setminus V)\subset \cF^\circ$). Recall $E=(V+(\tfrac{1}{2},0))\cup (V-(\tfrac{1}{2},0))$.


The loop model is defined as follows. Let $\nu^{\poi}_{V,u}(m_{\cross})$ and $\nu^{\poi}_{V,1-u}(m_{\dbar})$ be two independent Poisson point processes on $V$ of rates $u$ and $1-u$ respectively, where we denote the points of the first process by the symbol $\cross$ and call them crosses, and those of the second by $\dbar$ and call them double bars. Write $\ul{\Om}^{\lp}_F$ for the set of all possible realisations (call them configurations) of this joint process. One can realise each configuration as a simple point measure on $F\times \{ \cross, \dbar \}$; let $\FF^{\lp}_F$ be the sigma algebra generated by the $w^\#$ topology generated by the metric $d^\#$ \cite[Section A2.6]{daley-verejones-1}. 

Each configuration $m=(m_{\cross}, m_{\dbar})$ produces a set of loops and paths, defined as follows. First, $m$ naturally partitions $E=(V+(\tfrac{1}{2},0))\cup (V-(\tfrac{1}{2},0))$ into a collection of intervals $y\times[a,b)$ such that there is a point of $m$ at $(y,a)+(\tfrac{1}{2},0)$ or $(y,a)-(\tfrac{1}{2},0)$, and similarly there is a point of $m$ at $(y,b)+(\tfrac{1}{2},0)$ or $(y,b)+(\tfrac{1}{2},0)$, and no point $c\in(a,b)$ has a point of $m$ at $(y,c)+(\tfrac{1}{2},0)$ or $(y,c)+(\tfrac{1}{2},0)$. 

Each point $(y+\tfrac{1}{2},b)$ in the configuration $m$ is then incident to four intervals: $y\times[a_0,b)$, $(y+1)\times[a_1,b)$, $y\times[b,c_0)$ and $(y+1)\times[b,c_1)$. We define a connectivity as follows:
\begin{enumerate}
    \item If $(y+\tfrac{1}{2},b)\in m_{\cross}$ then connect $y\times[a_0,b)$ to $(y+1)\times[b,c_1)$ and connect $(y+1)\times[a_1,b)$ to $y\times[b,c_0)$.
    \item If $(y+\tfrac{1}{2},b)\in m_{\dbar}$ then connect $y\times[a_0,b)$ to $(y+1)\times[a_1,b)$ and connect $(y+1)\times[b,c_0)$ to $y\times[b,c_1)$.
\end{enumerate} 

A loop or path is a maximal union of these intervals which are connected by this connectivity; we call this union a loop if it is a cycle under the connectivity and a path otherwise. We write $\{(x_1,t_1)\leftrightarrow (x_2,t_2)\}$ for the event that the two points $(x_1,t_1)$ and $(x_2,t_2)$ lie on the same loop or path. \\

Let $l(m)$ be the number of loops or paths of $m$. The loop measure is defined by re-weighing the Poisson point process by $2^{l(m)}$, that is, we define
\be\label{eq:quantum-loop-measure}
    \d\nu^{\lp}_{F,\Delta}(m) 
    \propto  
    \d\nu^{\poi}_{V,u}(m_{\cross}) \d\nu^{\poi}_{V,1-u}(m_{\dbar})  
    \cdot 2^{l(m)}.
\ee
For a configuration $m'\in\ul{\Om}^{\lp}_{\cF}$, let the measure with boundary conditions $m'$ be defined the same, but with $l(m)$ replaced with $l(m;m')$, where the latter is defined as the number of connected components of loops or paths of the infinite configuration $m|_F\sqcup m'|_{\cF\setminus F}$ which intersect $F$. Call this measure $\nu^{\lp,m'}_{F,\Delta}$. 

We define the measure $\nu^{\lp,\bullet}_{F,\Delta}$ to be the measure $\nu^{\lp,m'}_{F,\Delta}$ in the case that $F$ is a black domain, $m'=\dbar$ at every point of $\partial_\vert F$ in $\cF^\circ$ and $m'$ has no other points. We similarly define $\nu^{\lp,\circ}_{F,\Delta}$ to be the same with the roles of $\cF^\bullet$ and $\cF^\circ$ reversed. 

We define the model with periodic boundary conditions, firstly in the vertical direction. Let $L\in\NN$ and $\beta>0$. Let $
v=\cF\cap([-L+1,L]\times[-\tfrac{\beta}{2},\tfrac{\beta}{2}])$ 
and add a further connectivity on the intervals defining the number of loops: $(a,\tfrac{\beta}{2}]$ is connected to the interval $[-\tfrac{\beta}{2},c]$. Let $\nu^{\sv}_{\cyl^h_{L,\beta},\Delta}$ be the measure where loops are defined using this altered connectivity.

Finally now consider the model on $V=\cF\cap([-L+1,L]\times[-\tfrac{\beta}{2},\tfrac{\beta}{2}])$ but periodic in the horizontal direction, meaning we identify $L+1$ with $-L-1$ (and $-L-2$ with $L$). Let $\nu^{\sv,\bullet}_{\cyl_{L,\beta},\Delta}$ be this measure conditioned on there being points of $m_{\dbar}$ at the points $(x+\tfrac{1}{2},\pm \tfrac{\beta}{2})$ for all $x$ odd; let $\nu^{\sv,\circ}_{\cyl_{L,\beta},\Delta}$ be the same where for all $x$ even.\\


\begin{theorem}\label{thm:main:quantum-loop}
    Consider Ueltschi's loop model for the XXZ chain with $\Delta\in[-1,0]$, so $u\in[0,\tfrac{1}{2}]$.
\begin{enumerate}
    \item 
    As $F\nearrow \cF$, respectively as $\beta\to\infty$ and $L\to\oo$ either simultaneously or in either order,
        \be\label{eq:q:loop-measures-converge}
            \nu^{\lp}_{F,\Delta}, \nu^{\lp,*}_{F,\Delta},
            \nu^{\lp}_{\cyl^h_{L,\beta},\Delta}, 
            \nu^{\sv,*}_{\cyl_{L,\beta},\Delta},
            \to \nu^{\lp}_\Delta,
        \ee
    weakly in $(\ul{\Om}^{\lp}_{\cF}, \FF^{\lp,\#}_{\cF})$ in the $w^\#$ topology, for $*=\bullet,\circ$.
    The measure $\nu^{\lp}_\Delta$ is tail-trivial, $\ZZ\times\RR$-invariant and ergodic.
    \item 
    We have that
    \be\label{eq:q:loop-connections-go-to-0}
            \nu^{\lp}_\Delta[(0,0)\leftrightarrow (x_0,t_0)] 
            \le
            \lim_{F\nearrow \cF} \nu^{\lp}_{F,\Delta}[(0,0)\leftrightarrow (x_0,t_0)] 
            \to 0
        \ee
    as $||(x,t)||\to\infty$. Further, for $(x_0,t_0)\in\ZZ\times\RR$, we have
        \be\label{eq:q:loop-connections-no-exp-decay}
            \sum_{n=1}^{\infty}
            n\cdot
            \lim_{F\nearrow \cF} \nu^{\lp}_{F,\Delta}
            [(0,0)\leftrightarrow (nx_0,nt_0)] 
            = \infty,
        \ee
    and there exist $\alpha_1, \alpha_2 \in (0,\infty)$ such that 
    \be\label{eq:q:loop-connections-poly}
    \begin{split}
        \nu^{\lp}_{\Delta}[(0,0)\leftrightarrow (n,0)] &\le n^{-\alpha_2}\\
        n^{-\alpha_1} \le 
        \lim_{F\nearrow \cF} \nu^{\lp}_{F,\Delta}[(0,0)\leftrightarrow (n,0)] &\le n^{-\alpha_2}.
    \end{split}
    \ee 
    \item
    Finally, if $u=\tfrac{1}{2}$, $\nu^{\lp}_\Delta$ is Gibbs, and in fact is the unique Gibbs measure and the unique thermodynamic limit, that is, for all $m'\in\ul{\Om}^\lp_{F\nearrow \cF}$, we have as $F\nearrow(F\nearrow \cF)$ that
    \be\label{eq:q:unique-loop-thermo-limit}
        \nu^{\lp,m'}_{F,\Delta}
        \to
        \nu^\lp_\Delta,
    \ee
    and further we have that 
	    \be\label{eq:q:xy-corr-decay-loop}
	    \lim_{F\nearrow \cF} 
        \nu^{\lp}_{F,\Delta}[(0,0)\leftrightarrow(x,t)] 
        \sim ||(x,t)||^{-\tfrac{1}{2}}.
	    \ee
\end{enumerate}
\end{theorem}

While the proof of our space-time delocalisation result follows closely that of \cite{glazman-lammers} for the ordinary six-vertex model, our proof of Theorem \ref{thm:main:quantum-loop} closely follows that of the mirror model, Theorem \ref{thm:main-mirror}. In particular, we introduce analogies of the four couplings in the discrete, which all function in essentially the same was as the discrete counterparts. We obtain the circuit of four models which can be pairwise coupled; see Figure \ref{fig:couplings-diagram-space-time}. See Sections \ref{sec:quantum-loop-proofs} and \ref{sec:quantum-AT-8v}.

\begin{figure}[h]
\centering
    \resizebox{0.6\textwidth}{!}{%
        \includegraphics[]{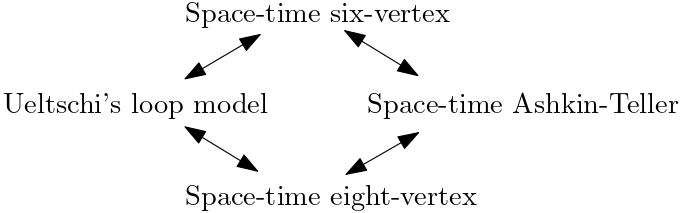}}
    \caption{The four space-time models and the pairs that can be coupled}
    \label{fig:couplings-diagram-space-time}
\end{figure}

\subsection{Organisation of the paper}
The paper is divided into three remaining parts. 

In Part \ref{part:convergence} we prove the convergence results for the mirror model and for XXZ with $\Delta\in[-1,\tfrac{1}{2}]$. This begins in Section \ref{sec:1st-coupling} where we recall the delocalisation results of the six-vertex model, show the first coupling between six-vertex and the mirror model, and then prove the convergence for the mirror model. In Section \ref{sec:xxz-conv} we prove the XXZ convergence and extremality.

In Part \ref{part:new-coupling} we prove the polynomial decay in the mirror and XXZ models. We introduce our new coupling in Section \ref{sec:new-coupling} and prove results for the Ashkin-Teller and eight-vertex models in Section \ref{sec:8v-AT}. In Section \ref{sec:symmetric-case} we prove the polynomial decay.

In Part \ref{part:space-time} we deal with the space-time models and the XXZ results for $\Delta\in[-1,0]$. We prove the delocalisation of the space-time six-vertex model in Section \ref{sec:delocalisation}, results for the loop model and convergence for XXZ in Section \ref{sec:quantum-loop-proofs}, and results for the space-time Ashkin-Teller and eight-vertex models in Section \ref{sec:quantum-AT-8v}.

\section*{Acknowledgements}
Many thanks to Marcin Lis, Lorca Heeney, Jakob Björnberg, Daniel Ueltschi, Moritz Dober and Sasha Glazman for helpful discussions, and to Bruno Nachtergaele, Piet Lammers and Karol Kozlowski for help understanding existing literature. Many thanks to Sasha Sodin for introducing me to the Lorentz mirror model. This work was supported by the FWF Standalone grants ``Spins, Loops and
Fields'' P 36298 and ``Order-disorder phase transition in 2D lattice
models'' P 34713,  the FWF SFB ``Discrete Random Structures'' 10.55776/F1002, and
the Academy of Finland Centre of Excellence Programme grant number
346315 ``Finnish centre of excellence in Randomness and STructures''
(FiRST).



\part{XXZ and mirror model convergence}\label{part:convergence}



\section{The mirror and six vertex models: the first coupling}\label{sec:1st-coupling}
In this section we take the existing results for the six-vertex model and prove the following proposition through the known, direct coupling between the six-vertex model and mirror model. 
\begin{proposition}\label{prop:mir-conv-from-6v}
    Let $\ul{p}\in\cP$ with $p_\NW,p_\NE\ge p_\es$. The measures $\mu_{V,\ul{p}}^{\mir}$, $\mu^{\mir;\bullet}_{V,\ul{p}}$, $\mu^{\mir;\circ}_{V,\ul{p}}$, $\mu_{\TT_{L,M},\ul{p}}^{\mir}$, $\mu_{\cyl_{L,M},\ul{p}}^{\mir,\bullet}$, $\mu_{\cyl_{L,M},\ul{p}}^{\mir,\circ}$ converge weakly to a common limit $\mu_{\ul{p}}^\mir$ as $V\nearrow\ZZ^2$ or $L,M\to\infty$ (in either order or simultaneously). The measure $\mu_{\ul{p}}^\mir$ is tail trivial, $\ZZ^2$-ergodic and translation-invariant. Further, for $e_1,e_2$ edges of $\ZZ^2$ and $x\in\ZZ^2$, 
    \be\label{eq:mir-escape-probs}
	\sum_{n=1}^\infty n\cdot
    \lim_{V\nearrow\ZZ^2}\mu_{V,\ul{p}}^{\mir}[e_1\leftrightarrow \tau_{nx}e_2] = \infty.
    \ee
\end{proposition}

\subsection{Definitions and existing results for six-vertex}
Recall $\LL^\bullet=(\FF^\bullet, \EE^\bullet)$ is the graph isomorphic to $\ZZ^2$ with vertex set given by the black faces of $\ZZ^2$, and edges given by nearest neighbours, and $\LL^\circ=(\FF^\circ,\EE^\circ)$ similar for the white faces. Let $V$ be a simply connected, finite subset of $\VV$, let $E$ be the set of edges in $\EE$ incident to at least one vertex in $V$, and let $F$ be the set of faces of $\ZZ^2$ incident to a vertex in $V$.

Write $G^\bullet=(F^\bullet,E^\bullet)$ for the graph with vertex set $F^\bullet$ the set of black faces in $F$, and edge set given by nearest neighbours whose edge passes through a vertex of $V$. Define $G^\circ=(F^\circ,E^\circ)$ similarly. Note $V$, $E^\bullet$ and $E^\circ$ are in bijection; each vertex has one edge of $E^\bullet$ and one edge of $E^\circ$ passing through it. For an edge $e\in E^\bullet$, we write $e^*$ for the corresponding edge in $E^\circ$, and say $e$ and $e^*$ are the dual of each other. 

We say $\partial F^\bullet$ is the set of black faces in $F^\bullet$ with at least one corner not in $V$, $\partial F^\circ$ similar, and $\partial F=\partial F^\bullet\cup\partial F^\circ$. We say $\partial E^\bullet$ is the set of edges in $\EE^\bullet$ between two faces in $\partial F^\bullet$. Note $\partial E^\bullet$ is a (possibly non-self-avoiding) circuit, with simply connected interior, since $V$ is simply connected.\\ 

The six-vertex model on $V$ is defined below. We give four different formulations of the measure. We start with the most well-known: in terms of arrow configurations.

Let $\Om^{\sv,\rightarrow}_V$ be the set of configurations $\kappa$ assigning an orientation to the edges in $E$, such that at each vertex in $V$, the number of incoming edges (and therefore the number of outgoing edges) is always 2 (this is the ice rule). The six-vertex model's measure on $\Om^{\sv}_V$ (with free boundary conditions) is given by
\bes
	\mu^{\sv,\rightarrow}_{V,a,b,c}[\kappa]
	\propto
	a^{n_1+n_2}b^{n_3+n_4}c^{n_5+n_6},
\ees
where $n_i$ is the number of vertices $v\in V$ with local configuration of $\kappa$ given by the $i^{th}$ from the left in the following Figure.

\begin{figure}[h]
\centering
    \resizebox{0.8\textwidth}{!}{%
        \includegraphics[]{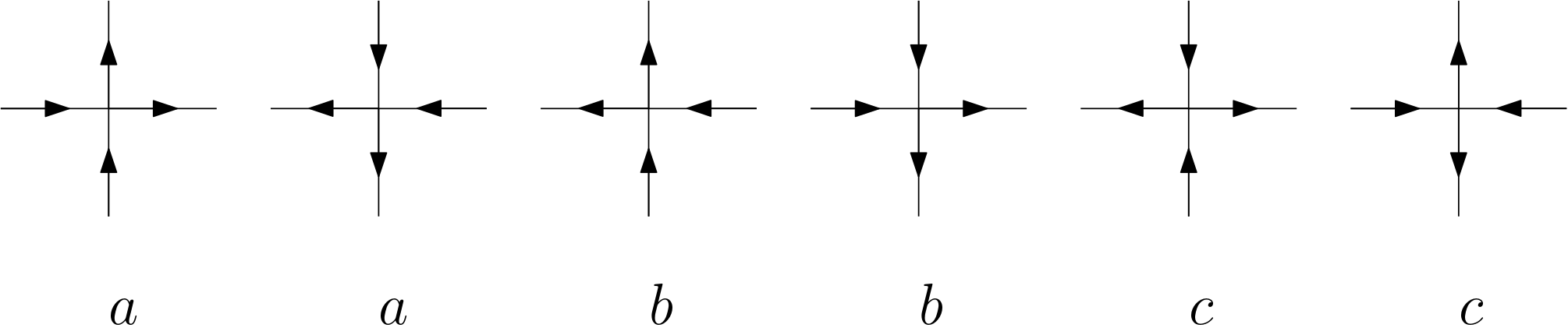}}
    \caption{The six-vertex model configurations and weights at a vertex}
    \label{fig:six-vertex-repeat}
\end{figure}

Let $\partial_EV$ is the set of edges with exactly one endpoint in $V$. For an arrow configuration $\kappa'$ on $\partial_EV$ which is the restriction of at least one valid configuration in $\Om^{\sv,\rightarrow}_V$, define
\bes
	\mu^{\sv,\rightarrow,\kappa'}_{V,a,b,c}[\kappa]
	\propto
	a^{n_1+n_2}b^{n_3+n_4}c^{n_5+n_6} \mathbbm{1}_{\{\kappa=\kappa' \ \mathrm{on} \ \partial_EV\}}.
\ees
Note that these measures satisfy a domain Markov property. Let $V'$ be a domain with $V'\subset V$. Then
\be\label{eq:6v-domain-markov-arrows}
    \mu^{\sv,\rightarrow}_{V,a,b,c}[ \ \cdot \ | \kappa|_{\partial_{E'}V'}=\kappa'] |_{V'}
    =
    \mu^{\sv,\rightarrow,\kappa'}_{V',a,b,c}.
\ee

The measure can be viewed as a measure on height functions (more specifically graph homomorphisms). Fix some $x_0\in F^\bullet$. For a configuration $\kappa\in\Om^{\sv}$, define a height function $h:F\to\ZZ$ as follows: set $h(x_0)=0$, and then if two faces $x,y$ are separated by an edge oriented by $\kappa$ to the left (when viewing from $x$ to $y$), set $h(y)=h(x)+1$. It is straightforward to check that this defines a unique height function.  

Write $\Om^{\sv-\hom}_V$ for the set of functions $h:F\to\ZZ$ with $h(x_0)=0$ and such that for any adjacent faces $x,y$, $|h(x)-h(y)|=1$. Write $\mu^{\sv-\hom}_{V,a,b,c}$ for the pushforward of the measure $\mu^{\sv,\rightarrow}_{V,a,b,c}$ to $\Om^{\sv-\hom}_V$ via the map defined above. 
More generally, let $h'\in\Om^{\sv-\hom}_{\VV}$; then
\bes
	\mu^{\sv-\hom,h'}_{V,a,b,c}[h]
	\propto
	a^{n_1+n_2}b^{n_3+n_4}c^{n_5+n_6} \mathbbm{1}_{\{h=h' \ \mathrm{on} \ \partial F\}}.
\ees
We write $\mu^{\sv-\hom,0,1}_{V,a,b,c}$ for the measure above where $h'$ takes values only in $\{0,1\}$ (it must be 0 on black faces and 1 on white faces).
\\

The spin representation of the six-vertex model is defined as follows. Let $\Om^{\sv-\spin}_V$ be the set of pairs of functions $\sigma^\bullet:F^\bullet\to\{\pm1\}$, and $\sigma^\circ:F^\circ\to\{\pm1\}$. Define the set of domain walls of $\sigma^\bullet$ as 
\be
    \om[\sigma^\bullet] = \{e^*\in E^\circ \ : \ \sigma^\bullet(x)\neq\sigma^\bullet(y) \ \mathrm{where} \ e=xy\},
\ee
and $\om[\sigma^\circ]$ similar. Write $\om_\NE$ for the edges in $\om[\sigma^\bullet]\cup\om[\sigma^\circ]$ oriented north-east; similarly define $\om_\NW$. Let $\om^c$ denote $V\setminus(\om[\sigma^\bullet]\cup\om[\sigma^\circ])$. 

The six-vertex model (spin) measure (with free boundary conditions) on $\Om^{\sv}_V$ is given by
\be\label{eq:6v-measure}
\begin{split}
	\mu^{\sv}_{V,a,b,c}[(\sigma^\bullet,\sigma^\circ)]
	\propto
	a^{|\om_\NE|}
	b^{|\om_\NW|}
	c^{|\om^c|}
	\mathbbm{1}\{\om[\sigma^\bullet] \cap \om[\sigma^\circ]=\es\}.
\end{split}
\ee

Note the spin representation $\mu^{\sv}_{V,a,b,c}$ is obtained from the measure $\mu^{\sv-\hom}_{V,a,b,c}$ by setting $\sigma^\bullet_{x_0}=\pm1$ uniformly on some $x_0\in F^\bullet$, and then for all neighbouring faces $x,y$ with $x\in F^\bullet$, $y\in F^\circ$, setting 
\be\label{eq:sv-spin-hf}
    \sigma^\bullet_x\sigma^\circ_y=h(x)-h(y).
\ee 

We define $\mu^{\sv,\bplus,\wplus}_{V,a,b,c}$ to be the measure $\mu^{\sv}_{V,a,b,c}$ conditioned on the event $\sigma^\bullet,\sigma^\circ\equiv1$ on $\partial F$. Similarly one can define $\mu^{\sv,\bplus}_{V,a,b,c}$, where only $\sigma^\bullet\equiv1$ on $\partial F^\bullet$, and $\mu^{\sv,\wplus}_{V,a,b,c}$, where only $\sigma^\circ\equiv1$ (or $h\equiv1$) on $\partial F^\circ$.\\ 

Finally, we recall the percolation process $\xi^\bullet$:
Given $\sigma\in\Om^{\sv}_{V}$, we define a percolation configuration $\xi\in\Om^{\perc}_{G^\bullet}$ as follows. For $e\in E^\bullet$, $e^*\in E^\circ$ its dual edge,
\be\label{eq:sampling-xi}
\begin{cases}
	\xi^\bullet(e)=1 \ \mathrm{deterministically} 
	& \mathrm{if} \ e\in\om(\sigma^\circ)\\
	\xi^\bullet(e)=0  \ \mathrm{deterministically}
	& \mathrm{if} \ e^*\in\om(\sigma^\bullet)\\
	\xi^\bullet(e)=1 \ \mathrm{w.p.} \ \tfrac{b}{c}=\tfrac{p_\NW+p_\es}{p_\NW+p_\NE}
	& \mathrm{if} \ e\in E^\bullet_\NW, \ e\notin\om(\sigma^\circ), \ e^*\notin\om(\sigma^\bullet)\\
	\xi^\bullet(e)=1 \ \mathrm{w.p.} \ \tfrac{a}{c}=\tfrac{p_\NE+p_\es}{p_\NW+p_\NE}
	& \mathrm{if} \ e\in E^\bullet_\NE, \ e\notin\om(\sigma^\circ), \ e^*\notin\om(\sigma^\bullet).
\end{cases}
\ee
Here w.p. is short for ``with probability''. We sometimes enlarge our probability space and write $\mu^{\sv}_{V,a,b,c}$ for the joint measure on pairs $(\sigma^\bullet,\sigma^\circ,\xi^\bullet)$. Note that the process $\xi^\bullet$ appears in several works under different names. In \cite{lis-spins, lis-deloc} its dual $(\xi^\bullet)^*$ is written as $\om$; in \cite{lis-lopez-heeney,dc-etal-GFF-convergence,chayes-mckellar-winn} $\xi^\bullet$ itself is written as $\om$; in \cite{ray-spinka-2} $(\xi^\bullet)^*$ is written as $\eta$; in \cite{glaz-peled} $(\xi^\bullet)^*$ is written as $\xi$.

Observe that $\xi$ can also be sampled from an arrow configuration (or indeed a height function) via the identification:
\be\label{eq:sample-xi-from-arrows}
\begin{split}
    \om[\sigma^\circ] &= \{ e \in E^\bullet \ | \ |h(x)-h(y)=2|, \ xy=e^*\in E^\circ\}\\
    \om[\sigma^\bullet] &= \{ e \in E^\circ \ | \ |h(x)-h(y)=2|, \ xy=e^*\in E^\bullet\}\\
\end{split}
\ee
It will be useful to note that one can write the marginal on $(\kappa,\xi^\bullet)$ as a product of local weights: $\mu^{\sv}_{V,a,b,c}(\kappa,\xi^\bullet) \propto \prod_{v\in V} w_v[\kappa,\xi^\bullet]$, where for $v\in V$ and $e,e^*$ its corresponding edges in $E^\bullet$, $E^\circ$ respectively,
\be\label{eq:product-weights-kappa-xi}
\begin{split}
    w_v[\kappa,\xi^\bullet]
    &=
    \mathbbm{1}_{\{\om[\sigma^\circ]_e\le\xi^\bullet_e\}} 
    \mathbbm{1}_{\{\om[\sigma^\bullet]_{e^*}\le(\xi^\bullet)^*_{e^*}\}}\\
    &\qquad\cdot
    \left[
    \mathbbm{1}_{\{e\in E^\bullet_{\NE}\}} \left(
    a^{\mathbbm{1}_{\{e\in\xi^\bullet\}}} 
    (c-a)^{\mathbbm{1}_{\{e^*\in(\xi^\bullet)^*\setminus\om[\sigma^\bullet]\}}} 
    b^{\mathbbm{1}_{\{e^*\in\om[\sigma^\bullet]\}}}
    \right) \right.\\
    &\qquad + \left.
    \mathbbm{1}_{\{e\in E^\bullet_{\NW}\}} \left(
    b^{\mathbbm{1}_{\{e\in\xi^\bullet\}}} 
    (c-b)^{\mathbbm{1}_{\{e^*\in(\xi^\bullet)^*\setminus\om[\sigma^\bullet]\}}} 
    a^{\mathbbm{1}_{\{e^*\in\om[\sigma^\bullet]\}}}
    \right)
    \right]
\end{split}
\ee
where $\om[\sigma^\circ]_e\in\{0,1\}$, thinking of $\om[\sigma^\circ]$ as a percolation configuration on $E^\bullet$, and similar for $\om[\sigma^\bullet]$ on $E^\circ$. 

We denote $\xi^{\bplus}$ the edges of $\xi^\bullet$ on whose endpoints one has $\sigma^\bullet=+1$. The process $\xi^\bullet$ satisfies a domain Markov property \cite[Lemma 4.8]{glazman-lammers}. Let $\gamma$ be a circuit in $E^\bullet$ and let $V(\gamma)$ be the domain given by all vertices of $\ZZ^2$ lying strictly inside $\gamma$. Then
\be\label{eq:6v-domain-markov-xi}
    \mu^{\sv}_{V,a,b,c}[\ \cdot \ | \ \gamma\subset\xi^{\bplus}]
    =
    \mu^{\sv,\bplus}_{V(\gamma),a,b,c},
\ee
and conditional on $\gamma\subset\xi^\bullet$, the spins strictly inside $\gamma$ and those strictly outside $\gamma$ are independent (the same statement holds for arrow configurations).\\

We work with the weights 
\be\label{eq:6v-weights}
\begin{split}
	a&=p_\es+p_\NE,\\
	b&=p_\es+p_\NW,\\
	c&=p_\NW+p_\NE,
\end{split}
\ee
where $\ul{p}\in\cP$. \\

\begin{theorem}[\cite{dckmo-6v-deloc,glazman-lammers,lis-deloc}]\label{thm:6v-deloc}
    Let $0< a,b\le c\le a+b$. For any fixed $x\in\FF$,
    \be\label{eq:delocalisation}
        \lim_{V\nearrow\ZZ^2}\Var_{\mu^{\sv-\hom;0,1}_{V,a,b,c}}
        \left[h(x)\right]=\infty.
    \ee    
    The measure $\mu^{\sv,\rightarrow}_{V,a,b,c}\to\mu^{\sv,\rightarrow}_{a,b,c}$ 
    weakly. Moreover, on the triple $(\sigma^\bullet,\sigma^\circ,\xi^\bullet)$,
    the measures $\mu^{\sv}_{V,a,b,c}$, $\mu^{\sv,\bplus,\wplus}_{V,a,b,c}$, $\mu^{\sv,\bplus}_{V,a,b,c}$, and $\mu^{\sv,\wplus}_{V,a,b,c}$ all converge to a common limit 
    $\mu^{\sv}_{a,b,c}$ weakly as $V\nearrow\ZZ^2$ and as $L,M\to\infty$. The limits $L,M$ can be taken in either order or simultaneously.
    
    The limit $\mu^{\sv}_{a,b,c}$ is Gibbs, $\ZZ^2$ translation-invariant and ergodic, and tail-trivial, and under this measure, the percolation processes $\xi^\bullet$ and $(\xi^\bullet)^*$ both do not percolate almost surely. As a consequence,
    \be\label{eq:6v-decay-correlations}
    \begin{split}
        \lim_{||x_1-x_2||\to\infty} \mu^\sv_{a,b,c}[\sigma^\bullet_{x_1}\sigma^\bullet_{x_2}]=0;
        \qquad
        \lim_{\substack{||x_1-x_2||\to\infty \\ ||y_1-y_2||\to\infty}}
        \mu^\sv_{a,b,c}[\sigma^\bullet_{x_1}\sigma^\bullet_{x_2}
        \sigma^\circ_{y_1}\sigma^\circ_{y_2}]=0.
    \end{split}
    \ee
\end{theorem}

Note that Theorem \ref{thm:6v-deloc} holds for weights \eqref{eq:6v-weights} precisely for $\ul{p}\in\cP$ satisfying $p_\NW,p_\NE\ge p_\es$.\\

We also need that the convergence holds for the measures on the cylinder and torus. Starting with the cylinder $\cyl_{L,M}$ recall the vertices are $V=\{-L+1,\dots,L\}\times\{-M+1,\dots,M\}$, and the faces $F$ are those with at least one corner in $V$. Define $\partial^+_F,\partial^-_F$ to be the topmost (resp. bottommost) row of faces in the cylinder. Let $\mu^{\sv,\bullet}_{\cyl_{L,M},a,b,c}$ be the measure on spin configurations such that the spins on the black faces of $\partial^+$ are constant, and the spins on the black faces of $\partial^-$ are also constant (although perhaps different from those on $\partial^+$).

\begin{lemma}\label{lem:cylinder-6v-converges}
    The measure $\mu^{\sv,\bullet}_{\cyl_{L,M},a,b,c}$ converges to $\mu^{\sv}_{a,b,c}$ as $L,M\to\infty$ in any order.
\end{lemma}
A similar statement holds when one uses the white faces of $\partial^+, \partial^-$, or a mixture of the two, or if one specifies the value of the constant spin on either or both boundaries, but we will not need these statements.
\begin{proof}
    Notice that under $\mu^{\sv,\bullet}_{\cyl_{L,M},a,b,c}$, the height function defined by \eqref{eq:sv-spin-hf} is well defined (up to overall shifts), that is, the height difference accumulated along any cycle (contractible or not) is 0. This is enough that the proof of FKG of the spins and percolation of \cite[Proposition 4.9]{glazman-lammers} can be applied to this case (see also Appendix \ref{appendix:FKG}). One then has for any contractible domain $V$ in $\cyl_{L,M}$,
    \bes
        \mu^{\sv,\bminus}_{V,a,b,c} \preceq_\bullet \mu^{\sv,\bullet}_{\cyl_{L,M},a,b,c}|_V
        \preceq_\bullet \mu^{\sv,\bplus}_{V,a,b,c},
    \ees
    where the stochastic domination $\preceq_\bullet$ is that of \cite{glazman-lammers}. The convergence follows.
\end{proof}

We now turn to the measure on the torus, which can be shown to converge to the same measure as above as a consequence of delocalisation and results of Sheffield \cite{sheffield-thesis}. We summarise how this is done. 

The six-vertex model measure can be interpreted as a gradient measure, with a potential which is convex (in the language of \cite{sheffield-thesis}, simply attractive) in the range $a,b\le c$. A result of Sheffield \cite[Lemma 8.7.1]{sheffield-thesis} then states that if there is a translation-invariant ergodic gradient Gibbs measure $\mu$ of zero slope (meaning $\mu(h(x)-h(y))=0$ for all faces $x\neq y$) which is rough (meaning it cannot be written as the gradient of a height function Gibbs measure), then $\mu$ is the unique translation-invariant ergodic gradient Gibbs measure, and is extremal. For the range $a,b\le c$, the model also satisfies an absolute-value FKG property (ie. $|h|$ is FKG), which implies that this roughness definition is equivalent to the delocalisation of the form:
\bes
    \lim_{V\nearrow\ZZ^2}\Var_{\mu^{\sv-\hom,0,1}_{V,a,b,c}}
        \left[h(x)\right]=\infty.
\ees
See \cite[Theorem 2.7]{ott-lammers}. Hence Theorem \ref{thm:6v-deloc} and these results imply that $\mu^{\sv,\rightarrow}_{a,b,c}$ is the unique translation-invariant ergodic gradient Gibbs measure. 

One can use this to prove that the measure on the torus converges to this measure.
\begin{lemma}\label{cor:6v-unique-TI-gibbs}
    For the six-vertex model with $a,b\le c\le a+b$ we have
    \bes
    \mu^{\sv,\rightarrow}_{\TT_{L,M},a,b,c} \to \mu^{\sv,\rightarrow}_{a,b,c}, 
    \qquad 
    \text{and so}
    \qquad
    \mu^{\sv}_{\TT_{L,M},a,b,c} \to \mu^{\sv}_{a,b,c}
    \ees
    as $L,M\to\infty$ (in either order or simultaneously).
\end{lemma}
To prove this, take a subsequential limit $\nu$ of the measures on the torus and observe that it must be Gibbs (as the interaction is local), translation invariant (as we're on the torus) and of zero slope (symmetry under flipping all arrows). It remains to show that the limit $\nu$ is ergodic. 

One first proves that the specific free energy (SFE) (see Section 2 of \cite{sheffield-thesis}) of the limit $\nu$ must be equal to that of $\mu^{\sv}_{a,b,c}$. One has that $SFE(\mu^{\sv}_{a,b,c})=\sigma(0)$, where $\sigma(0)$ is the surface tension at slope 0, defined as the infimum of $SFE(\mu)$ among translation-invariant gradient measures of slope 0. Hence $SFE(\nu)=\sigma(0)$. One then notes that the SFE is affine, so that when one decomposes $\nu$ into ergodic components: $\nu=\int w_\nu(\mu) \d\mu$, where the integral is over ergodic translation invariant gradient measures, then $SFE(\nu) = \int w_\nu (\sigma(S(\mu))) \d\mu$, where $S(\mu)$ is the slope of $\mu$. See \cite[Theorem 6.3.1 (Variational principle)]{sheffield-thesis}. Finally, the surface tension is known to be strictly convex \cite[Theorem 8.6.2]{sheffield-thesis}, hence all of the ergodic components $\mu$ of $\nu$ must have zero slope and must have $SFE(\mu)=\sigma(0)$, otherwise one would have $SFE(\nu)>\sigma(0)$. They are all therefore also Gibbs by \cite[Theorem 6.3.1]{sheffield-thesis}. But since there is a unique ergodic translation-invariant gradient Gibbs measure of zero slope, we have that $\nu$ must be it.  

\subsection{Mirror - six-vertex coupling}
The mirror model and six-vertex model can be coupled as follows. Take a mirror model configuration and spend the factor $2^{l(m;m')}$ by orienting the loops independently, uniformly in either direction. This produces a joint configuration of mirrors and oriented edges, such that the oriented edges give a consistent orientation to each loop. The marginal on the oriented edges is the six-vertex model. To go the other way, start with a six-vertex configuration $\kappa$ and sample mirrors at each vertex with appropriate weights and such that the resulting loops are consistently oriented:
\be\label{eq:sample-mirror-from-6v}
    \begin{cases}
        m_v \in m_\NE \ \mathrm{w.p.} \ \frac{p_{\NE}}{p_{\NE}+p_{\es}}=\frac{p_{\NE}}{a}, 
        \ \ \ \ \ \ \
        m_v \in m_\es \ \mathrm{w.p.} \ \frac{p_{\es}}{p_{\NE}+p_{\es}} =\frac{p_{\es}}{a}, 
        \ \ \ \  
        & \mathrm{if} \ v\in\kappa_a\\
        m_v \in m_\NW \ \mathrm{w.p.} \ \frac{p_{\NW}}{p_{\NW}+p_{\es}}=\frac{p_{\NW}}{b}, 
        \ \ \ \ \  
        m_v \in m_\es \ \mathrm{w.p.} \ \frac{p_{\es}}{p_{\NW}+p_{\es}}
        =\frac{p_{\es}}{b}, 
        \ \ \ \    
        & \mathrm{if} \ v\in\kappa_b\\
        m_v \in m_\NW \ \mathrm{w.p.} \ \frac{p_{\NW}}{p_{\NW}+p_{\NE}} =\frac{p_{\NW}}{c}, 
        \ \ \ \  
        m_v \in m_\NE \ \mathrm{w.p.} \ \frac{p_{\NE}}{p_{\NW}+p_{\NE}} =\frac{p_{\NE}}{c}, 
        \ \ \ \   
        & \mathrm{if} \ v\in\kappa_c,\\
    \end{cases}
\ee
where $\kappa_a$ is the set of vertices in $V$ with an $a$-type vertex in $\kappa$, and similar for $\kappa_b$ and $\kappa_c$. 
This gives back the joint measure, and the marginal on unoriented loops is the mirror model. 
We make this rigorous below.

\begin{proposition}\label{prop:mir-6v-coupling}
    Let $\ul{p}\in\cP$ and let $a,b,c$ satisfy \eqref{eq:6v-weights} and let $V\subset\ZZ^2$ be a domain. 
    \begin{enumerate}
    	\item If one takes $m\sim\mu^{\mir}_{V,\ul{p}}$ and uniformly and independently orients the loops, and then the marginal on oriented edges is $\mu^{\sv,\rightarrow}_{V,a,b,c}$. 
    	\item If one takes $\kappa\sim\mu^{\sv}_{V,a,b,c}$ and at each vertex independently samples mirrors according to the rule \eqref{eq:sample-mirror-from-6v}, then the marginal on mirrors is $m\sim\mu^{\mir}_{V,\ul{p}}$.
    	\item There exists an analogous coupling between $\mu^{\mir;\bullet}_{V,\ul{p}}$ and $\mu^{\sv,\bplus}_{V,a,b,c}$, the only difference being that when sampling $(\sigma^\bullet,\sigma^\circ)$ from $\kappa$, one sets $\sigma^\bullet\equiv +1$ on all of $\partial^\bullet_F$ in part 1; in part 2 one sets $m_v\in m^\bullet$ deterministically for all $v\in\partial V$. One defines analogous couplings between the pair $\mu^{\mir;\circ}_{V,\ul{p}}$ and $\mu^{\sv,\wplus}_{V,a,b,c}$, the pair $\mu^{\mir;\bullet}_{\cyl_{L,M},\ul{p}}$ and $\mu^{\sv,\bplus}_{\cyl_{L,M},a,b,c}$ and the pair $\mu^{\mir}_{\TT_{L,M},\ul{p}}$ and $\mu^{\sv}_{\TT_{L,M},a,b,c}$.
	\end{enumerate}
\end{proposition}

The proof is elementary; we leave the details to the reader.\\

We're now able to prove the convergence of mirror model measures.

\begin{proof}[Proof of Proposition \ref{prop:mir-conv-from-6v}]
    Take the infinite volume measure $\mu^{\sv,\rightarrow}_{a,b,c}$ on arrow configurations. At each vertex of $\ZZ^2$, sample mirrors independently according the rules \eqref{eq:sample-mirror-from-6v}, and call the resulting marginal on mirror configurations $\mu^\mir_{\ul{p}}$, and the coupling $\mu^{\sv-\mir}_{\ul{p}}$. 

    We show that all the measures in the proposition converge to $\mu^{\mir}_{\ul{p}}$. Let $A$ be a cylinder event on mirrors, dependent only on the mirror configuration on vertices $V_0$. Let $E_0$ be the set of edges of $\ZZ^2$ with at least one endpoint in $V_0$. We have
    \be\label{eq:mir-conv-from-6v}
    \begin{split}
        \lim_{V\nearrow\ZZ^2} \mu^{\mir}_{V,\ul{p}}[A] 
        &= 
        \lim_{V\nearrow\ZZ^2} \sum_{\tilde{\kappa}} 
        \mu^{\sv-\mir}_{V,\ul{p}}[A \ | \ \kappa|_{E_0}=\tilde{\kappa}] 
        \mu^{\sv-\mir}_{V,\ul{p}}[\kappa|_{E_0}=\tilde{\kappa}]  \\
        &=
        \sum_{\tilde{\kappa}} 
        \mu^{\sv-\mir}_{\ul{p}}[A \ | \ \kappa|_{E_0}=\tilde{\kappa}]
        \lim_{V\nearrow\ZZ^2} 
        \mu^{\sv-\mir}_{V,\ul{p}}[\kappa|_{E_0}=\tilde{\kappa}]\\
        &=
        \sum_{\tilde{\kappa}} 
        \mu^{\sv-\mir}_{\ul{p}}[A \ | \ \kappa|_{E_0}=\tilde{\kappa}]
        \mu^{\sv,\rightarrow}_{\ul{p}}[\kappa|_{E_0}=\tilde{\kappa}]\\
        &=
        \mu^{\mir}_{\ul{p}}[A],
    \end{split}
    \ee
    where the second equality is from the fact that the sampling of mirror configurations given an arrow configuration is local and only dependent on the arrows on neighbouring edges, and the third equality is the six-vertex of Theorem \ref{thm:6v-deloc}. The proof is the same for the other measures; note for the torus and cylinder measures we use Corollary \ref{cor:6v-unique-TI-gibbs} and Lemma \ref{lem:cylinder-6v-converges} respectively.\\

    The limit $\mu^{\mir}_{\ul{p}}$ is translation-invariant, as the six-vertex measure from which it is sampled is itself translation-invariant. For any event $A$ in the mirror model dependent the configuration on a set of vertices $V'$, since the mirrors are sampled locally from the six-vertex configuration we have that $\mu^{\mir}_{\ul{p}}[A] = \mu^{\sv}_{a,b,c}[F_A]$  where $F_A$ is a random variable which is a function of the six-vertex configuration, dependent only on the arrow configuration on $V'$ (or the spin configuration on $F'$, the faces with at least one corner in $V'$). We can prove the following condition which is equivalent to tail-triviality \cite[Proposition 7.9]{georgii}: for all local events $A$,
    \be\label{eq:mir-tail-triv}
        \lim_{\Lambda\subset\ZZ^2} \sup_{B\in\tau_\Lambda} | \mu^{\mir}_{\ul{p}}[A\cap B] - \mu^{\mir}_{\ul{p}}[A]\mu^{\mir}_{\ul{p}}[B] | =0,
    \ee
    where $\tau_\Lambda$ is the sigma algebra generated by the configuration outside of $\Lambda$. Indeed, let $A$ be local and let $B\in\tau_\Lambda$. We have $\mu^{\mir}_{\ul{p}}[A\cap B] = \mu^{\sv}_{a,b,c}[F_A F_B]$. Let $\Lambda_n$ contain the support of $A$ and let $\cO_{n,N}$ be the event that there exists a circuit of $\xi^\bullet$ in the box $\Lambda_N$, surrounding $\Lambda_n$, where $N>n$ and $\Lambda_N\subset\Lambda$. We know that $\mu^{\sv}_{a,b,c}[\cO_{n,N}] \to1$ as $N\to\infty$ by Theorem \ref{thm:6v-deloc}, and by the domain Markov property \eqref{eq:6v-domain-markov-xi}, for any realisation $\gamma$ of the circuit, conditional on $\gamma\subset\xi^\bullet$, the random variables $F_A$ and $F_B$ are independent. This gives the condition \eqref{eq:mir-tail-triv}. Translation invariance and tail-triviality readily imply that $\mu^{\mir}_{\ul{p}}$ is ergodic.\\

    It remains to prove \eqref{eq:mir-escape-probs}. The proof here is inspired by that of \cite[Proposition 9]{eng-lis}, and we thank Marcin Lis for suggesting this method of proof. 
    
    Take a mirror configuration $m$ and an arrow configuration $\kappa$ such that $m\sim\kappa$ ($\kappa$ provides a valid orienting of the loops of $m$). Height differences are well-defined by $\kappa$. Let $x=(x_1,x_2)\in\ZZ^2$ and let $f_0$ be a fixed face, and  $f_n=\tau_{nx}f_0$. Let $E_0$ be the set of primal edges of $\ZZ^2$ dual to a path of length $x_1+x_2$ from $f_0$ to $f_1$ in the dual lattice. Let $E_n=\tau_{nx}E_0$ ($E_n$ joins $f_n$ and $f_{n+1}$). 
    
    Observe from the coupling that the height difference $h(f_0)-h(f_n)$ is equal to the total flux of arrows through the edges $\cup_{i=1}^n E_i$ (the number of right oriented arrows minus the number of left oriented arrows one encounters when moving from $f_0$ to $f_n$). Observe further that for each edge $e$ in $\cup_{i=1}^n E_i$, the loop through $e$ contributes 1 (or -1) to the flux at $e$ only if it has non-zero winding around at least one of $f_0$ and $f_n$, in which case it must pass through $\bigcup_{i<0; i\ge n}E_i$. Hence we have that
    \bes
    \begin{split}
        \lim_{V\nearrow\ZZ^2}\mu^{\sv}_{V,a,b,c}[|h(f_0)-h(f_n)|]
        &\le 
        (x_1+x_2) \sum_{i=0}^{n-1} 
        \lim_{V\nearrow\ZZ^2}\mu^{\mir}_{V,\ul{p}} 
        [ E_i\leftrightarrow
        \bigcup_{i<0; i\ge n}E_i
        ]\\
        &\le 
        (x_1+x_2) \sum_{i=0}^{n-1}
        \left(
        \sum_{j<0} 
        \lim_{V\nearrow\ZZ^2}\mu^{\mir}_{V,\ul{p}} [ E_i\leftrightarrow E_j]
        +
        \sum_{j\ge n} 
        \lim_{V\nearrow\ZZ^2}\mu^{\mir}_{V,\ul{p}} [ E_i\leftrightarrow E_j]
        \right)\\
        &\to
        2(x_1+x_2)\sum_{j=1}^\infty j \cdot 
        \lim_{V\nearrow\ZZ^2}\mu^{\mir}_{V,\ul{p}} [ E_0\leftrightarrow E_j],
    \end{split}
    \ees
    as $n\to\infty$. Here we used the translation invariance of the measure $\mu^{\mir}_{\ul{p}}$. From Theorem \ref{thm:6v-deloc}, the left hand side of the above equation tends to infinity as $n\to\infty$, so so does the right.

    Finally, we have that for $e_1,e_2$ edges of $\ZZ^2$,
    \bes
    \lim_{V\nearrow\ZZ^2}\mu^{\mir}_{\ul{p}} [ E_0\leftrightarrow E_j]
    \ge 
    C \cdot 
    \lim_{V\nearrow\ZZ^2}\mu^{\mir}_{\ul{p}} [ e_1\leftrightarrow \tau_{jx} e_2],
    \ees
    where $C$ depends on $e_1, e_2$ and $x$ but is independent of $j$. This is an application of finite energy. Indeed, on the event $E_0\leftrightarrow E_j$, let $l_1$ and $l_2$ be the loops connecting $E_0$ and $E_j$ which come closest to $e_1$, $\tau_{jx}e_2$, respectively. From those closest points, modify the mirror configuration in as straight a line as possible so that the loop $l_1$ passes through $e_1$ and the loop $l_2$ passes through $\tau_{jx}e_2$. This comes at a cost of $C'(\tfrac{1}{2}\min\{p_\NE,p_\NW,p_\es\})^{d(e_1,E_0)+d(e_2,E_0)}$ for some constant $C'$, which is uniform in $V$. This completes the proof of \eqref{eq:mir-escape-probs} and the proof of Proposition \ref{prop:mir-conv-from-6v}.
\end{proof}

\section{Convergence of the XXZ model in the range $\Delta\in[-1,1/2]$}\label{sec:xxz-conv}

In this section we prove Part 1 of Theorem \ref{thm:main-xxz}. The main remaining tool to introduce is the transfer matrix of the six-vertex model and how it relates to the Hamiltonian of the XXZ model. 

\begin{proof}[Proof of Part 1 of Theorem \ref{thm:main-xxz}]

Consider the six-vertex model on the torus $\TT_{L,M}$. Let $a,b,c>0$. Recall the parameter
\be\label{eq:delta}
    \Delta
    =
    \frac{a^2+b^2-c^2}{2ab}
    =
    \frac{p_{\varnothing}^2 + p_\varnothing p_{\NE} + p_\varnothing p_{\NW} - p_{\NE}p_{\NW}}
    {p_{\varnothing}^2 + p_\varnothing p_{\NE} + p_\varnothing p_{\NW} + p_{\NE}p_{\NW}}.
\ee
In this section we will only need to consider the case $a=b$. Let $V_0$ denote all the vertices of $\TT_{L,M}$ on the $x$ axis, that is, with vertical co-ordinate equal to 0. Let $E$ denote all edges incident to at least one vertex in $V_0$, and let $E_0$ denote vertical the edges in $E$ below the vertices $V_0$, while $E_1$ is the vertical edges in $E$ lying above the vertices $V_0$. For $\kappa_0$ an arrow configuration on $E_0$ and $\kappa_1$ an arrow configuration on $E_1$, define 
\bes
    (T_{L,a,b,c})_{\kappa_1,\kappa_2}
    =
    (T_{L})_{\kappa_1,\kappa_2}
    =
    \sum_{\substack{\kappa \\ \kappa|_{E_i}=\kappa_i}} 
    a^{n_1(\kappa)+n_2(\kappa)}b^{n_3(\kappa)+n_4(\kappa)}c^{n_5(\kappa)+n_6(\kappa)},
\ees
where the sum is over arrow configurations $\kappa$ on only the edges $E$, satisfying the ice rule at all vertices in $V_0$, such that $\kappa|_{E_i}=\kappa_i$ for $i=1,2$, and $n_j(\kappa)$ is the number of vertices of $V$ of type $j$ in $\kappa$.

Identifying an up arrow with the basis vector $|+\rangle$ of $\CC^2$ and a down arrow with the basis vector $|-\rangle$, one has that $T_L$ is a matrix acting on the Hilbert space $(\CC^2)^{\{-L+1,\dots,L\}}$. Similarly to the quantum models, define for a linear operator $A$ acting on $(\CC^2)^{\{-L+1,\dots,L\}}$,
\bes
    \langle A \rangle^{\sv}_{\TT_{L,M}}
    =
    \frac{\Tr[A T_L^M]}{\Tr[T_L^M]}
\ees
and define for a vector $\Psi\in(\CC^2)^{\{-L+1,\dots,L\}}$,
\bes
    \langle A \rangle^{\sv,\Psi}_{\TT_{L,M}}
    =
    \frac{\langle \Psi | T_L^{M/2} A T_L^{M/2} | \Psi \rangle}
    {\langle \Psi | T_L^M | \Psi \rangle}.
\ees

\subsection{Proof of convergence}
The following is the main proposition of this section. We write $\langle\cdot\rangle^{\mathrm{XXZ}}$ for states of the $XXZ$ model, for clarity.
\begin{proposition}\label{prop:xxz-6v-same}
    Consider the six-vertex model with parameters $a,b,c>0$ and the XXZ chain with parameter $\Delta$ satisfying \eqref{eq:delta}.
    \begin{enumerate}
    \item We have that for all linear operators $A$ on $(\CC^2)^{\{-L+1,\dots,L\}}$ and $*=\bullet,\circ$
        \bes
        \lim_{\beta\to\infty}
        \langle A \rangle^{\mathrm{XXZ}}_{\TT_L,\beta}
        =
        \lim_{M\to\infty}
        \langle A \rangle^{\sv}_{\TT_{L,M}}
        =
        \lim_{\beta\to\infty}
        \langle A \rangle^{\mathrm{XXZ},\Psi^*_L}_{\TT_L,\beta}
        =
        \lim_{M\to\infty}
        \langle A \rangle^{\sv,\Psi^*_L}_{\TT_{L,M}},
        \ees
    \item Now let $\Delta\in[-1,1/2]$. The above expressions converge as $L\to\infty$ for any fixed local operator $A$.
    \end{enumerate}
\end{proposition}

\begin{proof}[Proof of Part 1 of Proposition \ref{prop:xxz-6v-same}]

    It is well known that $T_L$ and $H_{\TT_L}$ commute (see for example \cite[Lemma 5.1]{dc-etal-bethe-ansatz-exposition}). They therefore have the same eigenvectors. It is also well known (see for example \cite[Theorem 2.3]{dc-etal-bethe-ansatz-exposition}) that the largest eigenvalue of $T_L$ has a unique eigenvector $\Psi^{\mathrm{grd}}_L$, and that this is the unique eigenvector of the smallest eigenvalue of $H_{\TT_L}$. 
    
    It follows that 
    \be\label{eq:xxz-6v-equal-1}
        \lim_{\beta\to\infty}
        \langle A \rangle^{\mathrm{XXZ}}_{\TT_L,\beta}
        =
        \lim_{M\to\infty}
        \langle A \rangle^{\sv}_{\TT_{L,M}}
        =
        \frac{\langle \Psi^{\mathrm{grd}}_L | A | \Psi^{\mathrm{grd}}_L\rangle}{\langle \Psi^{\mathrm{grd}}_L| \Psi^{\mathrm{grd}}_L\rangle}.
    \ee
    A similar equality holds for the states seeded by $\Psi^*_L$:
    \be\label{eq:xxz-6v-equal-2}
        \lim_{\beta\to\infty}
        \langle A \rangle^{\mathrm{XXZ},\Psi^*_L}_{\TT_L,\beta}
        =
        \lim_{M\to\infty}
        \langle A \rangle^{\sv,\Psi^*_L}_{\TT_{L,M}}
        =
        \frac{\sum_{i} \langle \Psi^*_L | \Psi^{i}_L \rangle \langle \Psi^{i}_L | A | \Psi^{i}_L\rangle}
        {\langle \Psi^*_L | \Psi^{i}_L \rangle \langle \Psi^{i}_L| \Psi^{i}_L\rangle},
    \ee
    where the sum is over eigenvectors $\Psi^i_L$ of $T_L$ with largest eigenvalue among those which satisfy $\langle \Psi^*_L | \Psi^{i}_L \rangle>0$. To prove that \eqref{eq:xxz-6v-equal-1} and \eqref{eq:xxz-6v-equal-2} are equal it therefore suffices to prove that $\langle \Psi^*_L | \Psi^{\mathrm{grd}}_L \rangle>0$. 

    We claim that $\mu^{\sv}_{\TT_{L,M}}$ and $\mu^{\sv,\bullet}_{\cyl_{L,M}}$ both converge to the same measure on $\{-L+1,\dots,L\}\times\ZZ$ as $M\to\infty$. Indeed, $\mu^{\sv}_{\TT_{L,M}}$ and $\mu^{\sv}_{\TT_{L,M}}[\ \cdot \ | \{\text{balanced}\}]$ converge to the same measure as $M\to\infty$ by \cite[Lemma 4.11 and Lemma 18.7]{dc-etal-GFF-convergence}, where $\{\text{balanced}\}$ is the condition that on each row of vertical edges, there are an equal number of up and down arrows. This balanced condition naturally holds under $\mu^{\sv,\bullet}_{\cyl_{L,M}}$. The convergence follows from considering the system as a recurrent Markov chain on the arrow configurations on rows of vertical edges. 
    
    From this claim, for operators $A$ for which $\langle A \rangle^{\sv}_{\TT_{L,M}} = \mu^{\sv}_{\TT_{L,M}}[A']$ for some function $A'$, we have that \eqref{eq:xxz-6v-equal-1} and \eqref{eq:xxz-6v-equal-2} are equal. We claim that the operator $|\Psi^\bullet_L\rangle \langle \Psi^\circ_L |$ is such an operator. Indeed, let $\gamma$ be the path of edges of $E^\bullet$ whose midpoints lie on the $x$ axis and consider the event $\gamma\subset\xi^\bullet$. We have that by the FKG inequality and comparison of boundary conditions,
    \bes
    \begin{split}
    \mu^{\sv,\bullet}_{\cyl_{L,M}}[\gamma\subset\xi^\bullet]
    \ge 
    \mu^{\sv,\bullet}_{\cyl_{L,M}}[\gamma\subset\xi^{\bplus}]
    \ge
    \mu^{\sv,\bminus}_{\cyl_{L,M}}[\gamma\subset\xi^{\bplus}]>0,
    \end{split}
    \ees
    and the last expression is increasing in $M$ by comparison of boundary conditions, so the limit is strictly positive. Finally,  
    \bes
        \lim_{M\to\infty}\mu^{\sv}_{\TT_{L,M}}[\gamma\subset\xi^\bullet]
        =
        \frac{\Tr \left[ T_L^{M/2-1} | \Psi^\bullet_L\rangle \langle \Psi^\circ_L | T_L^{M/2} \right]}
            {\Tr[ T_L^M ]}
        =
        \frac{\langle \Psi^{\mathrm{grd}}_L | \Psi^\bullet_L\rangle \langle \Psi^\circ_L  | \Psi^{\mathrm{grd}}_L\rangle}
        {\langle \Psi^{\mathrm{grd}}_L| \Psi^{\mathrm{grd}}_L\rangle}>0,
    \ees
    Here we use that $\gamma\subset\xi^\bullet$ induces $\bullet$ boundary conditions below $\gamma$ and above $\gamma$, meaning: for every white face in the row of faces with centres having vertical coordinate $-\tfrac{1}{2}$ (resp. $+\tfrac{1}{2}$), the vertical edge on its left side is oriented upwards if and only if the vertical edge on its right side is oriented downwards. This is the same as the condition given by the vector $\Psi^\bullet_L$ for the row at height $-\tfrac{1}{2}$ (and the vector $\Psi^\circ_L$ for the row at height $+\tfrac{1}{2}$). 
    Hence the proof that \eqref{eq:xxz-6v-equal-1} and \eqref{eq:xxz-6v-equal-2} are equal is complete, and this finishes the proof of part 1 of Proposition \ref{prop:xxz-6v-same}.
\end{proof}    
    We turn to part 2 of the proposition, the convergence in infinite volume. It suffices to prove the proposition when $A=\prod_{k=1}^r E_{i_k^-,i_k^+}^{x_k}$, where $E_{i_k^-,i_k^+}^{x_k}$ is the $2\times2$ elementary matrix with its $(i_k^-,i_k^+)$ entry equal to 1 and other entries equal to 0, with support at the single vertex $x_k\in\ZZ$. Let $X=\{x_k\}_{k=1}^r$ and $\ul{i}=(i_1^-,i_1^+,\dots,i_r^-,i_r^+)$. Indeed, any elementary matrix with support $X$ can be written as such a product, and so any linear operator with support $X$ can be written as a linear combination of such products.

    
    We write $E_X$ for the set of vertical edges in $\TT_{L,M}$ with midpoints at $X\times\{\tfrac{1}{2}\}$; for $k=1,\dots,r$ write $e_k$ for the edge with midpoint $z_k:=(x_k,\tfrac{1}{2})$. Write $E'_X$ for the $2r$ half-edges of the edges in $E_X$, with $e_k^+$ the upper half-edge of $e_k$ and $e_k^-$ the lower half-edge. To the vector $\ul{i}$ we associate a configuration of arrows to the half-edges in $E'_X$: if $i_k^\pm=1$, we orient $e_k^\pm$ upwards, and otherwise downwards. We say such a half-edge which is oriented towards the midpoint $z_k$ is a sink, and is a source if if is oriented away from $z_k$. We call an edge a disorder if both of its half-edges are sinks or both are sources.
    
    
    Write $\Om^{\sv,\bullet,X,\ul{i}}_{\cyl_{L,M}}$ for the set of arrow configurations on the edges of $\cyl_{L,M}$ (all edges with at least one end in $V=\{-L+1,\dots,L\}\times\{-M+1,\dots,M\}$) with the boundary conditions $\bullet$ at the top and bottom of the cylinder, and with the half-edges in $E'_X$ fixed to be those given by $\ul{i}$ above, and which satisfy the ice rule. Let us be clear again what we mean by boundary conditions $\bullet$ in terms of arrow configurations: we mean that for every white face in $\partial_F^+$ (resp. $\partial_F^-$) (faces with centres having vertical coordinate $M+\tfrac{1}{2}$ (resp. $-M+\tfrac{1}{2}$), the vertical edge on its left side is oriented upwards if and only if the vertical edge on its right side is oriented downwards. Note that in terms of the height function, this condition is that the height function is constant on the black faces of $\partial_F^+$ (and also constant on the black faces of $\partial_F^-$). Write $\Om^{\sv,\bullet}_{\cyl_{L,M}}$ for the same with no disorders ($X=\es$). Write
    \bes
    	Z^{\sv,\bullet,X,\ul{i}}_{\cyl_{L,M}}
    	=
    	\sum_{\kappa\in \Om^{\sv,\bullet,X,\ul{i}}_{\cyl_{L,M}} }
    	a^{n_1+n_2}b^{n_3+n_4}c^{n_5+n_6},
        \qquad
        Z^{\sv,\bullet}_{\cyl_{L,M}}
    	=
    	\sum_{\kappa\in \Om^{\sv,\bullet}_{\cyl_{L,M}}}
    	a^{n_1+n_2}b^{n_3+n_4}c^{n_5+n_6},
    \ees
    the latter expression is the partition function of the measure $\mu^{\sv,\bullet}_{\cyl_{L,M}}$. Similarly, define $Z^{\sv,\bullet,X,\ul{i}}_{V}$ and $Z^{\sv,\bullet}_{V}$ for the same in the planar domain $V$. It is easy to prove that for $M$ odd,
    \be\label{eq:disorders}
    	\frac{Z^{\sv,\bullet,X,\ul{i}}_{\cyl_{L,M}}}
    	{Z^{\sv,\bullet}_{\cyl_{L,M}}}
    	=
    	\langle \prod_{k=1}^r E_{i_k^-,i_k^+}^{x_k} \rangle^{\sv, \Psi^\bullet_L}_{\TT_{L,M}}.
    \ee
    Indeed, the vector $\Psi^\bullet_L = 
    \frac{1}{\sqrt{n}} \sum_{\substack{i\in\Lambda_L \\ i \ \mathrm{even}}} \sum_{a=1}^2 | a,a \rangle_{i-1, i}$ implies that the vertical arrows in the top row of $\cyl_{L,M}$ on each pair of columns $i-1, i$ ($i$ even) are different (one up, one down); this is precisely the $\bullet$ boundary condition, providing $M$ is odd. If $M$ is even, we get $\circ$ in the superscripts on the left hand side; the proof below will be identical either way, so we assume $M$ is odd for the rest of this section. \\
        
    Note that \eqref{eq:disorders} is non-zero only if the number of sources equals the number of sinks. We will show that \eqref{eq:disorders} converges as $L,M\to\infty$ in any order (including the required $M\to\infty$ and then $L\to\infty$). \\

    Before we begin the proof, let us note that if we could prove that the mirror model $\mu^{\mir}_{\ul{p}}$ on $\ZZ^2$ had no infinite loops almost surely, then the convergence of \eqref{eq:disorders} would be simple. We describe briefly how this proof would work, as it contains a tool for our actual proof.

    The key observation is that \eqref{eq:disorders} can be written as an expectation in the mirror model. Indeed:
    \be\label{eq:disorders-mirrors}
    	\frac{Z^{\sv,\bullet,X,\ul{i}}_{\cyl_{L,M}}}
    	{Z^{\sv,\bullet}_{\cyl_{L,M}}}
    	=
    	\langle \prod_{k=1}^r E_{i_k^-,i_k^+}^{x_k}  \rangle^{\sv, \Psi^\bullet_L}_{\TT_{L,M}}
        =
        \mu^{\mir,\bullet}_{\cyl_{L,M}}[2^{-l_{E_X}(m)} \mathbbm{1}_{\{m\sim(X,\ul{i})\}}],
    \ee
    where $l_{E_X}(m)$ is the number of loops in the mirror model configuration which pass through the edges $E_X$, and $m\sim(X,\ul{i})$ if and only if each half-edge in $E'_X$ which is a sink according to $\ul{i}$ is connected to a source by the loops of $m$, and vice-versa.
    
    The proof is a switching bijection. Given $\kappa\in\Om^{\sv,\bullet,X,\ul{i}}_{\cyl_{L,M}}$, sample a mirror configuration $m$ according to the rules \eqref{eq:sample-mirror-from-6v}, with $v\in m^\bullet$ for all $v\notin\cyl_{L,M}$. Then in the joint configuration $(\kappa,m)$ there is a set of oriented paths starting at sources in $E'_X$ and ending at sinks. Reversing the orientation of the arrows along some of these paths turns the configuration into one in $\Om^{\sv,\bullet}_{\cyl_{L,M}}$ (no disorders), ie. where in each pair of half-edges $e_k^+$ and $e_k^-$, one is a source and one is a sink. This reversal also preserves the weight of the configuration. One can then take the marginal on $m$, and one only needs to remember that the loops passing through $E_X$ have a deterministic orientation, leading to the factor $2^{-l_{E_X}(m)}$. A similar proof was used in \cite[Lemma 2.4]{bjornberg-ryan-dimerization}. 

    Now, if one had that $\mu^{\mir}_{\ul{p}}$ had no infinite loops almost surely, it would follow that the right hand side of \eqref{eq:disorders-mirrors} converges as $L,M\to\infty$, and the XXZ convergence would be proved.
    Since we do not have this, we take a different route. The key idea is to use the percolation process $\xi^\bullet$ \eqref{eq:sampling-xi} and the fact that there are almost surely infinitely many circuits of $\xi^\bullet$ around the origin due to Theorem \ref{thm:6v-deloc}. Crucial will be the domain Markov property enjoyed by $\xi^\bullet$.

    Before we begin, we note a key tool we will use: uniformly in domains $V$,
    \be\label{eq:disorders-mirrors-2}
    	\frac{Z^{\sv,\bullet,X,\ul{i}}_{V}}
    	{Z^{\sv,\bullet}_{V}}
        =
        \mu^{\mir,\bullet}_{V}[2^{-l_{E_X}(m)} \mathbbm{1}_{\{m\sim(X,\ul{i})\}}]
        \le 1.
    \ee

\begin{proof}[Proof of Part 2 of Proposition \ref{prop:xxz-6v-same}] 
The proof will boil down to the following Lemma.
    \begin{lemma}\label{lem:disorders-estimate}
    	There exists a constant $C_X$ independent of $L,M$ such that for all $N\in\NN$ and all $L,M\ge N$,
    	\be\label{eq:disorders-estimate}
    		\left| 
    		\frac{Z^{\sv,\bullet,X,\ul{i}}_{\cyl_{L,M}}}
    		{Z^{\sv,\bullet}_{\cyl_{L,M}}}
    		-
    		\mu^{\sv,\bullet}_{\cyl_{L,M}}[f_N]
    		\right|
    		\le 
    		C_X
    		\mu^{\sv,\bullet}_{\cyl_{L,M}}[\eps_N],
    	\ee
    	where $f_N$ and $\eps_N$ are local functions of $(\kappa,\xi^\bullet)$, dependent only on the configuration in $\Lambda_N$, and 
    	\bes
    	\lim_{N\to\infty}\lim_{L,M\to\infty} \mu^{\sv,\bullet}_{\cyl_{L,M}}[\eps_N] =0.
    	\ees
    \end{lemma}
    
    We prove part 2 of Proposition \ref{prop:xxz-6v-same} from the lemma; in particular we prove that \eqref{eq:disorders} is Cauchy. Fix $\eps>0$. There exists an $N_0,N_1\in\NN$ such that for all $L,M\ge N_1$, 
    \be\label{eq:disorders-working-1}
    	\mu^{\sv,\bullet}_{\cyl_{L,M}}[\eps_{N_0}]<\eps/3C_X;
    \ee
    indeed, find $N_0$ such that 
    $\lim_{L,M\to\infty} \mu^{\sv,\bullet}_{\cyl_{L,M}}[\eps_{N_0}]<\eps/6C_X$, and then let $L,M$ large enough so that 
    \bes
    	\left|
    	\mu^{\sv,\bullet}_{\cyl_{L,M}}[\eps_{N_0}]
    	-
    	\lim_{L,M\to\infty} \mu^{\sv,\bullet}_{\cyl_{L,M}}[\eps_{N_0}]
    	\right|
    	<\eps/6C_X.
    \ees
    Now with $N_0$ fixed, there is an $N_2\in\NN$ such that for all $L,M,L',M'\ge N_2$,
    \be\label{eq:disorders-working-2}
    	\left| 
    	\mu^{\sv,\bullet}_{\cyl_{L,M}}[f_{N_0}]
    	-
    	\mu^{\sv,\bullet}_{\cyl_{L',M'}}[f_{N_0}]
    	\right|
    	\le \eps/3,
    \ee
    where we use that $\mu^{\sv,\bullet}_{\cyl_{L,M}}$ is convergent as $L,M\to\infty$. 
	Taking $L,M,L',M'>\max\{N_1,N_2\}$, we find that by combining \eqref{eq:disorders-working-1}, \eqref{eq:disorders-working-2} and \eqref{eq:disorders-estimate}, we have
	\bes
		\left| 
		\frac{Z^{\sv,\bullet,X,\ul{i}}_{\cyl_{L,M}}}
		{Z^{\sv,\bullet}_{\cyl_{L,M}}}
		-
		\frac{Z^{\sv,\bullet,X,\ul{i}}_{\cyl_{L',M'}}}
		{Z^{\sv,\bullet}_{\cyl_{L',M'}}}
		\right|
		\le \eps.
	\ees
	This concludes the proof of Proposition \ref{prop:xxz-6v-same}.
\end{proof}

\begin{proof}[Proof of Lemma \ref{lem:disorders-estimate}]
	We write $\kappa\sim(X,\ul{i})$ to mean $\kappa\in \Om^{\sv,\bullet,X,\ul{i}}_{\cyl_{L,M}} $ for brevity.

    Let $\cO_{n,N}$ be the event that there exists a circuit of $\xi^\bullet$ contained in $\Lambda_N\setminus\Lambda_n$, surrounding $\Lambda_n$ (n will depend only on $X$ and we will fix it later; for now we only require that $n$ be large enough so that $\Lambda_n$ contains $E_X$). Then
	\be\label{eq:disorders-good-bad}
	\begin{split}
	\frac{Z^{\sv,\bullet,X,\ul{i}}_{\cyl_{L,M}}}
	{Z^{\sv,\bullet}_{\cyl_{L,M}}}
	&=
	\frac{1}
	{Z^{\sv,\bullet}_{\cyl_{L,M}}}
	\sum_{\kappa\sim(X,\ul{i}) }
	\sum_{\xi^\bullet\sim\kappa}
	\mathbbm{1}(\cO_{n,N})
	w[\kappa,\xi^\bullet]\\
	&\qquad+
	\frac{1}
	{Z^{\sv,\bullet}_{\cyl_{L,M}}}
	\sum_{\kappa\sim(X,\ul{i}) }
	\sum_{\xi^\bullet\sim\kappa}
	\mathbbm{1}(\cO^c_{n,N})
	w[\kappa,\xi^\bullet],
	\end{split}
	\ee
    where $\xi^\bullet\sim\kappa$ if $\xi^\bullet$ is a valid configuration sampled via \eqref{eq:sampling-xi}, \eqref{eq:sample-xi-from-arrows} from $\kappa$. Here $w[\kappa,\xi^\bullet] =  \prod_{v\in\cyl_{L,M}}w_v[\kappa,\xi^\bullet]$, where $w_v[\kappa,\xi^\bullet]$ is given by \eqref{eq:product-weights-kappa-xi}.    

    We will show that 
    \be\label{eq:disorders-estimate-part-1}
	\begin{split}
	\frac{1}
	{Z^{\sv,\bullet}_{\cyl_{L,M}}}
	\sum_{\kappa\sim(X,\ul{i}) }
	\sum_{\xi^\bullet\sim\kappa}
	\mathbbm{1}(\cO_{n,N})
	w[\kappa,\xi^\bullet]
    =
    \mu^{\sv,\bullet}_{\cyl_{L,M}}[f_N]
	\end{split}
	\ee
    and that 
    \be\label{eq:disorders-estimate-part-2}
	\begin{split}
	\frac{1}
	{Z^{\sv,\bullet}_{\cyl_{L,M}}}
	\sum_{\kappa\sim(X,\ul{i}) }
	\sum_{\xi^\bullet\sim\kappa}
	\mathbbm{1}(\cO^c_{n,N})
	w[\kappa,\xi^\bullet]
    \le
    C_X
    \mu^{\sv,\bullet}_{\cyl_{L,M}}[\eps_N],
	\end{split}
    \ee
    where $f_N$, $\eps_N$ are local functions dependent only on the configurations in $\Lambda_N$, and we have $\lim_{L,M\to\infty} \mu^{\sv,\bullet}_{\cyl_{L,M}}[\eps_N] \to0$ as $N\to\infty$.\\

    We prove \eqref{eq:disorders-estimate-part-1} first. Let $\gamma^\bullet_N$ be the outermost circuit of $\xi^\bullet$ in $\Lambda_N$ surrounding $\Lambda_n$. Conditioning on $\gamma_N^\bullet=\gamma$ gives a domain Markov property across $\gamma$ (see \eqref{eq:6v-domain-markov-xi}). We observe that this domain Markov property extends to our case where there are disorders on $E_X$. Indeed, the condition that $e\in E^\bullet_{NE}$ lies in $\xi^\bullet$ can be written in terms of the arrow configuration $\kappa$: $v\in\kappa_a$ or $v\in\kappa_c$, where both options carry the same weight $a$ ($v\in\ZZ^2$ is the midpoint of $e$). This condition can be rewritten as a product of two conditions, one on the arrows on one side of $e$ and one on those on the other side. Indeed, let $e_{\mathrm{north}},e_{\mathrm{south}},e_{\mathrm{east}},e_{\mathrm{west}}$ be the four edges of $\ZZ^2$ incident to $v$, denoted by the direction they point from $v$. Then $e\in E^\bullet_{NE}$ if and only if both of the following hold:
    \be\label{eq:domain-markov-ref-pos-condition}
         \begin{split}
             e_{\mathrm{north}} \ \text{oriented towards} \ v  \ & \Leftrightarrow  \ e_{\mathrm{west}} \ \text{oriented away from} \ v\\
             e_{\mathrm{south}} \ \text{oriented towards} \ v  \ & \Leftrightarrow  \ e_{\mathrm{east}} \ \text{oriented away from} \ v.
         \end{split}
    \ee
    Similar holds for $e\in E^\bullet_{NW}$ with west and east exchanged in the condition. Using this, we find that
    
    \be\label{eq:disorders-estimate-part-1-working-1}
	\begin{split}
	&\frac{1}
	{Z^{\sv,\bullet}_{\cyl_{L,M}}}
	\sum_{\kappa\sim(X,\ul{i}) }
	\sum_{\xi^\bullet\sim\kappa}
	\mathbbm{1}_{\cO_{n,N}}
	w[\kappa,\xi^\bullet]\\
    &\qquad=
    \frac{1}
	{Z^{\sv,\bullet}_{\cyl_{L,M}}}
    \sum_\gamma
    a^{|\gamma\cap E_{\NE}|}b^{|\gamma\cap E_{\NW}|}
    \left[
	\sum_{\substack{\kappa_{\inn}\sim(X,\ul{i}) \\
    \kappa_{\inn}\sim\gamma} }
    w[\kappa_{\inn}]
    \right]
    \left[
    \sum_{\kappa_{\outt} \sim \gamma}
	\sum_{\xi_{\outt}^\bullet\sim\kappa_{\outt}}
	\mathbbm{1}_{\{\gamma=\gamma^\bullet_N(\xi_{\outt}^\bullet)\}}
	w[\kappa_{\outt},\xi_{\outt}^\bullet]
    \right],\\
    &\qquad=
    \frac{1}
	{Z^{\sv,\bullet}_{\cyl_{L,M}}}
    \sum_\gamma
    a^{|\gamma\cap E_{\NE}|}b^{|\gamma\cap E_{\NW}|}
    Z^{\sv,\bullet,X,\ul{i}}_{\gamma_{\inn}}
    \sum_{\kappa_{\outt} \sim \gamma}
	\sum_{\xi_{\outt}^\bullet\sim\kappa_{\outt}}
	\mathbbm{1}_{\{\gamma=\gamma^\bullet_N(\xi_{\outt}^\bullet)\}}
	w[\kappa_{\outt},\xi_{\outt}^\bullet],\\
    &\qquad=
    \frac{1}
	{Z^{\sv,\bullet}_{\cyl_{L,M}}}
    \sum_\gamma
    a^{|\gamma\cap E_{\NE}|}b^{|\gamma\cap E_{\NW}|}
    \frac{Z^{\sv,\bullet,X,\ul{i}}_{\gamma_{\inn}}}
    {Z^{\sv,\bullet}_{\gamma_{\inn}}}
    \left[
	\sum_{\substack{\kappa_{\inn}\sim\gamma} }
    w[\kappa_{\inn}]
    \right]
    \sum_{\kappa_{\outt} \sim \gamma}
	\sum_{\xi_{\outt}^\bullet\sim\kappa_{\outt}}
	\mathbbm{1}_{\{\gamma=\gamma^\bullet_N(\xi_{\outt}^\bullet)\}}
	w[\kappa_{\outt},\xi_{\outt}^\bullet],\\
    &\qquad=
    \mu^{\sv,\bullet}_{\cyl_{L,M}}\left[ 
    \mathbbm{1}_{\cO_{n,N}}
    \frac{Z^{\sv,\bullet,X,\ul{i}}_{(\gamma^\bullet_N)_{\inn}}}
    {Z^{\sv,\bullet}_{(\gamma^\bullet_N)_{\inn}}}
    \right].\\
	\end{split}
	\ee
    Here sum over $\gamma$ is the sum over realisations of $\gamma^\bullet_N$, $\gamma_{\inn}$ is the domain of vertices lying strictly in the interior of $\gamma$. Further $\kappa_{\inn}$ denotes arrow configurations on edges of $\cyl_{L,M}$ lying in the interior of $\gamma$ satisfying the ice rule; $\kappa_{\inn}\sim(X,\ul{i})$ means the half-edges of $E'_X$ are oriented according to $\ul{i}$; $\kappa_{\inn}\sim\gamma$ means at each $v\in\gamma$ the one condition from \eqref{eq:domain-markov-ref-pos-condition} pertaining to the pair of edges lying inside $\gamma$ holds. Further $\kappa_{\outt}$ is an arrow configuration on edges outside of $\gamma$ satisfying the ice rule,  $\kappa_{\outt}\sim\gamma$ is analogous to $\kappa_{\inn}\sim\gamma$, and $\xi^\bullet_{\outt}$ is a percolation configuration on edges of $E^\bullet$ lying on or outside of $\gamma$. Finally
    $w[\kappa_{\inn}] 
    = \prod_{v\in\gamma_{\inn}}a^{\mathbbm{1}_{\{v\in\kappa_a\}}} b^{\mathbbm{1}_{\{v\in\kappa_b\}}} c^{\mathbbm{1}_{\{v\in\kappa_c\}}}$. The fact that in the first equality we can write the weights as a product of the two factors is precisely the domain Markov property. In the last equality we use the domain Markov property again in reverse.
    This proves \eqref{eq:disorders-estimate-part-1} with 
        $f_N=
    \mathbbm{1}_{\cO_{n,N}} 
    Z^{\sv,\bullet,X,\ul{i}}_{(\gamma^\bullet_N)_{\inn}} / 
    Z^{\sv,\bullet}_{(\gamma^\bullet_N)_{\inn}}.$\\

    Now we prove \eqref{eq:disorders-estimate-part-2}. Let $\ell$ be a circuit of faces of $\cyl_{L,M}$ which surrounds $E_X$ with minimal length. Since the number of sources on $E_X$ equals the number of sinks on $E_X$, for any $\kappa\sim(X,\ul{i})$, $\kappa$ has a well-defined height function along $\ell$ (up to overall shift). Note though that this height function cannot necessarily be extended to a valid height function with no disorders on the interior of $\ell$. For $g$ a given height function along $\ell$ (up to overall shift), write $\mu^{\sv,\bullet,g}_{\cyl_{L,M}\setminus \ell_{\inn}}$ for the measure on arrow configurations on $\cyl_{L,M}\setminus \ell_{\inn}$, $\ell_{\inn}$ the faces on the interior of $\ell$, whose height function is $g$ on $\ell$ and has constant height function on the black faces of each of $\partial^+_{F^\bullet}$ and $\partial^-_{F^\bullet}$, the black faces in $\partial^+_F$ and $\partial^-_F$ respectively. By the domain Markov property of arrow configurations \eqref{eq:6v-domain-markov-arrows} (which also applies in this case with disorders),
    \be\label{eq:disorders-estimate-part-2-working-1}
	\begin{split}
	\frac{1}
	{Z^{\sv,\bullet}_{\cyl_{L,M}}}
	\sum_{\kappa\sim(X,\ul{i}) }
	\sum_{\xi^\bullet\sim\kappa}
	\mathbbm{1}(\cO^c_{n,N})
	w[\kappa,\xi^\bullet]
    &=
    \sum_g
    \sum_{\substack{\kappa_\inn\sim(X,\ul{i}) \\
    \kappa_\inn \sim g}}
    w[\kappa_{\inn}]
    \cdot
    \frac{Z^{\sv,\bullet,g}_{\cyl_{L,M}\setminus \ell_{\inn}}}
	{Z^{\sv,\bullet}_{\cyl_{L,M}}}
    \mu^{\sv,\bullet,g}_{\cyl_{L,M}\setminus \ell_{\inn}}
    [\cO^c_{n,N}],\\
	\end{split}
    \ee
    where $Z^{\sv,\bullet,g}_{\cyl_{L,M}\setminus \ell_{\inn}}$ is the partition function of the measure $\mu^{\sv,\bullet,g}_{\cyl_{L,M}\setminus \ell_{\inn}}$.
    If we define the height function to take the value 0 at the point of maximal height along $\ell$ according to $g$, we have a well-defined height function $h$.

    \begin{lemma}\label{lem:cyl-hf-fkg}
        The measure $\mu^{\sv,\bullet,g}_{\cyl_{L,M}\setminus \ell_{\inn}}$ satisfies the FKG inequality in the height function $h$.
    \end{lemma}
    \begin{proof}[Proof of Lemma \ref{lem:cyl-hf-fkg}]
    The proof of FKG for the measure $\mu^{\sv,\bullet,g}_{\cyl_{L,M}\setminus \ell_{\inn}}$ is very similar to that for the measure on a planar domain given in \cite[Proposition 2.2]{dckmo-6v-deloc}. One needs to modify it to create a Markov chain that deals with the boundary conditions on $\partial^+_{F^\bullet}$ and $\partial^-_{F^\bullet}$. 
    
    The measure $\mu^{\sv,\bullet,g}_{\cyl_{L,M}\setminus \ell_{\inn}}$ is irreducible with respect to the following two types of moves: type 1 being the usual changing of the value of $h$ at one square to $h\pm2$ (when the result is a valid configuration), and type 2 changing the value of $h$ on \textit{all} of $\partial^+_{F^\bullet}$ (resp. $\partial^-_{F^\bullet}$) to $h\pm2$ (again when the result is a valid configuration). Indeed, conditioning on $C^+, C^-$ being the values of $h$ on $\partial^+_{F^\bullet}$, $\partial^-_{F^\bullet}$, the measure is irreducible with respect to type 1 moves by the same proof as in the planar case. Then if $C^+, C^->0$, it is easy to show that a configuration exists from which one can use a type 2 move to reduce the larger of $C^+, C^-$ by 2. If both are negative one can similarly increase the lower by 2, and if they differ in sign, one can reduce the absolute value of either of them. By this process one can reduce to the case $C^+= C^-=0$, and irreducibility follows. 

    The FKG criterion (equation (A.2) of \cite{dckmo-6v-deloc}) then also holds for each of these moves; for type 1 moves by the same proof as the planar case, and for type 2 moves as they are actually a concatenation of several type 1 moves. The FKG property follows.
\end{proof}

    Let us now prove \eqref{eq:disorders-estimate-part-2} given this FKG property. We need two claims. The first claim is that that there exists a second circuit of faces $\ell'$ and a constant $C\le0$, depending on $g$ and $\ell$, such that $\ell'$ surrounds $\ell$ and such that the event $h\in\{C,C-1\}$ on all of $\ell'$ (ie. $C$ on even faces and $C-1$ on odd faces) has non-zero probability and is a decreasing event; that is, $h$ cannot take lower values than this on any of $\ell'$. 

    Indeed, let $h_0$ be a valid height function configuration. Using the moves of type 1 and 2 from the FKG proof, transform the configuration into one $h_1$ where there are no faces of height $\ge -1$, apart from on $\ell$ and on faces sharing an edge with those on $\ell$ where $g=0$. At this point, there is a circuit $\ell_1$ which is identical to $\ell$ except that it passes directly around all faces of $\ell$ where $g=0$. Define $g_1$ to be the height function $h_1$ restricted to $\ell_1$, and note that there is no configuration which takes lower values than $h_1$ on $\ell_1$. Now repeat this process; at each step a new $\ell_i$ is defined, on which $h_i$ takes values strictly less than the largest value $g_{i-1}$ takes, and the minimal value on all $\ell_i$ is the same. This process must terminate on a circuit $\ell':=\ell_k$ where $h_k$ takes values in $\{C,C-1\}$ for some $C\le0$, providing $L,M$ are large enough. Finally, note that at every step, the height function cannot possibly take lower values on $\ell_i$ than $h_i$ does, so the event is indeed decreasing.

    From now we take $n$ large enough so that $\Lambda_n$ contains $\ell'$. For $b\in\ZZ$, let $D_{b,N}$ be the event that there is a circuit of $h\ge b$ in $\Lambda_N$. The second claim we need is that for $n$ such that $\ell'\subset\Lambda_n$, there is a constant $b_n$ such that event $D_{b_n,N}$ occuring implies $\cO_{n,N}$ occurs. Indeed, there is a $b_n=b_n(g)$ such that if $D_{b_n,N}$ occurs, then there is a circuit $\gamma$ in $E^\bullet$ in $\Lambda_N\setminus\Lambda_n$ such that on the white faces adjacent to and outside $\gamma$, the height is $b_n$, and on the white faces adjacent to and inside $\gamma$, it is $b_n-2$. The circuit $\gamma$ then lives in $\om[\sigma^\circ]$ (see \eqref{eq:sample-xi-from-arrows}), so by definition of $\xi^\bullet$ \eqref{eq:sampling-xi}, $\gamma\subset\xi^\bullet$.
    
    Using both claims, we have:
    \be\label{eq:disorders-estimate-part-2-working-2}
	\begin{split}
    \mu^{\sv,\bullet,g}_{\cyl_{L,M}\setminus \ell_{\inn}}
    [\cO^c_{n,N}]
    &\le
    \mu^{\sv,\bullet,g}_{\cyl_{L,M}\setminus \ell_{\inn}}
    [D^c_{b_n,N}]\\
    &\le
    \mu^{\sv,\bullet,g}_{\cyl_{L,M}\setminus \ell_{\inn}}
    [D^c_{b_n,N} \ | \ h|_{\ell'}\in\{C,C-1\}]\\
    &=
    \mu^{\sv,\bullet}_{\cyl_{L,M}}
    [D^c_{b_n,N} \ | \ h|_{\ell'}\in\{C,C-1\}]\\
    &\le 
    C_g\cdot
    \mu^{\sv,\bullet}_{\cyl_{L,M}}
    [D^c_{b_n,N}],
	\end{split}
    \ee
    where the first inequality is the second claim above, the second is due to the FKG inequality and the first claim above, and the equality is the domain Markov property for arrow configurations. In the final expression we fix the height function to be $C$ at some point inside $\ell'$ so that the height function and the event $D^c_{b_n,N}$ are well-defined, and $C_g = \sup_{L,M}(\mu^{\sv,\bullet}_{\cyl_{L,M}}[h|_{\ell'}\in\{C,C-1\}])^{-1}$, which is bounded, since the probability in question converges as $L,M\to\infty$ due to Lemma \ref{lem:cylinder-6v-converges}, as it is a local function of the gradient, and the limit is bounded away from 0. Indeed there exist $n,N$ with $\ell'\subset\Lambda_n$ and $\mu^{\sv,\bullet}[\cO_{n,N}]>\tfrac{1}{2}$, and by conditioning on and restricting to the interior of $\gamma_N^\bullet$, one has $h|_{\ell'}\in\{C,C-1\}]$ with strictly positive probability. \\
    
    From the delocalisation of the height function of the six-vertex model, Theorem \ref{thm:6v-deloc}, the probability in the final expression tends to 0 as $L,M\to\infty$ and then $N\to \infty$, as required; the event $\eps_N$ in \eqref{eq:disorders-estimate-part-2} is the event $D^c_{b_n(g),N}$, which has largest probability as $g$ runs over its possible values. 
    
    Putting this back into \eqref{eq:disorders-estimate-part-2-working-1}, we have \be\label{eq:disorders-estimate-part-2-working-3}
	\begin{split}
	\frac{1}
	{Z^{\sv,\bullet}_{\cyl_{L,M}}}
	\sum_{\kappa\sim(X,\ul{i}) }
	\sum_{\xi^\bullet\sim\kappa}
	\mathbbm{1}(\cO^c_{n,N})
	w[\kappa,\xi^\bullet]
    &\le
    \mu^{\sv,\bullet}_{\cyl_{L,M}}
    [\eps_N]
    \max_g\{C_g\}
    \frac{Z^{\sv,\bullet,X,\ul{i}}_{\cyl_{L,M}}}
	{Z^{\sv,\bullet}_{\cyl_{L,M}}}.\\
	\end{split}
    \ee
    The ratio of partition functions is at most 1 by \eqref{eq:disorders-mirrors-2}, uniformly in $L,M,X,\ul{i}$.
    This finishes the proof of \eqref{eq:disorders-estimate-part-2}, and so finishes the proof of Lemma \ref{lem:disorders-estimate}. 
\end{proof}

    We have proved the following convergence of the XXZ model:
    \bes
        \langle\cdot\rangle 
            = \lim_{L\to\infty}\lim_{\beta\to\infty} \langle \cdot \rangle_{\beta,\TT_L}
            = \lim_{L\to\infty}\lim_{\beta\to\infty} \langle \cdot \rangle_{\beta,\TT_L}^{\Psi^*_L},
    \ees
    for $*=\bullet,\circ$. Note also that $\langle\cdot\rangle$ is translation invariant, since the six-vertex limit is also translation-invariant.

\subsection{Proof of extremality}\label{sec:extremality}
    
    We prove next that $\langle\cdot\rangle$ is extremal. We prove the equivalent condition (for translation-invariant states) that $\langle\cdot\rangle$ is mixing, that is, for all local observables $A$, $B$ and for $\tau_z$ the shift in $\ZZ$ by $z$, we have
    \be\label{eq:xxz-6v:mixing}
        \lim_{|z|\to\infty} \langle A \cdot (\tau_z B) \rangle = \langle A \rangle  \langle B \rangle.
    \ee
    See for example \cite[Corollary IV.2.3]{israel-book} or one can adapt the ideas of \cite{friedli-velenik}. Again it suffices to prove this for $A=\prod_{k=1}^r E_{i_k^-,i_k^+}^{x_k}$, $B=\prod_{k=1}^s E_{j_k^-,j_k^+}^{y_k}$. Write $X=\{x_k\}_{k=1}^r$ and $\ul{i}=(i_1^-,i_1^+,\dots,i_r^-,i_r^+)$, $Y=\{y_k\}_{k=1}^s$ and $\ul{j}=(j_1^-,j_1^+,\dots,j_s^-,j_s^+)$, and then the union of the two lists $\ul{ij}=(i_1^-,i_1^+,\dots,i_r^-,i_r^+,j_1^-,j_1^+,\dots,j_s^-,j_s^+)$.
    By \eqref{eq:disorders}, it suffices to prove that 
    \be\label{eq:xxz-mixing-disorders}
    \lim_{L,M\to\infty}
        \frac{Z^{\sv,\bullet,X\cup\tau_z Y, \ul{ij}}_{\cyl_{L,M}}}
	{Z^{\sv,\bullet}_{\cyl_{L,M}}}
    	\to
        \lim_{L,M\to\infty}
        \frac{Z^{\sv,\bullet,X,\ul{i}}_{\cyl_{L,M}}}
	       {Z^{\sv,\bullet}_{\cyl_{L,M}}}
            \cdot 
        \lim_{L,M\to\infty}
        \frac{Z^{\sv,\bullet,Y,\ul{j}}_{\cyl_{L,M}}}
	{Z^{\sv,\bullet}_{\cyl_{L,M}}}
    \ee
    as $|z|\to\infty$.

    We know from our proof of convergence that for all $\eps>0$, there exists an $N\in\NN$ such that
    \be\label{eq:extremal-work-1}
        \left| 
        \lim_{L,M\to\infty}
        \frac{Z^{\sv,\bullet,X,\ul{i}}_{\cyl_{L,M}}}
	    {Z^{\sv,\bullet}_{\cyl_{L,M}}}
        -
        \mu^{\sv}[f_N^{X,\ul{i}}]
        \right|
        \le \eps,
        \qquad
        \left| 
        \lim_{L,M\to\infty}
        \frac{Z^{\sv,\bullet,Y,\ul{j}}_{\cyl_{L,M}}}
	    {Z^{\sv,\bullet}_{\cyl_{L,M}}}
        -
        \mu^{\sv}[f_N^{Y,\ul{j}}]
        \right|
        \le \eps,
    \ee
    where we write $f_N^{X,\ul{i}}=\mathbbm{1}_{\cO_{n,N}} 
    Z^{\sv,\bullet,X,\ul{i}}_{(\gamma^\bullet_N)_{\inn}} / 
    Z^{\sv,\bullet}_{(\gamma^\bullet_N)_{\inn}}$ 
    and 
    $f_N^{Y,\ul{j}}=\mathbbm{1}_{\cO_{n,N}} 
    Z^{\sv,\bullet,Y,\ul{j}}_{(\gamma^\bullet_N)_{\inn}} / 
    Z^{\sv,\bullet}_{(\gamma^\bullet_N)_{\inn}}$.
    for $f_N$ exactly as written in \eqref{eq:disorders-f_N} and $f_N^B$ for the same with $\ul{i^-},\ul{i^+}$ replaced by $\ul{j^-},\ul{j^+}$, and we write $\mu^{\sv}$ for the infinite volume measure on $(\kappa,\xi^\bullet)$. 

    We will show that for $|z|$ large enough,
    \be\label{eq:extremal-work-2}
    \left|
    \lim_{L,M\to\infty}
        \frac{Z^{\sv,\bullet,X\cup\tau_z Y, \ul{ij}}_{\cyl_{L,M}}}
	{Z^{\sv,\bullet}_{\cyl_{L,M}}}
    -
    \mu^{\sv}\left[f_N^{X,\ul{i}}\right]
    \mu^{\sv}\left[f_N^{Y,\ul{j}}\right]
    \right|
    \le \eps,
    \ee
    which will suffice to prove \eqref{eq:xxz-mixing-disorders}. \\

    Let $\cO^{AB}_{n,N}=\cO_{n,N}\cap\tau_z\cO_{n,N}$ be the event that there is a circuit of $\xi^\bullet$ in $\Lambda_N$ surrounding $\Lambda_n$, and another circuit in $\tau_z\Lambda_N$ surrounding $\tau_z\Lambda_n$, where we take $n$ large enough so that $E_X,E_Y\subset\Lambda_n$. Similar to \eqref{eq:disorders-good-bad}, we have
    \be\label{eq:extremal-good-bad}
	\begin{split}
	\frac{Z^{\sv,\bullet,X\cup\tau_z Y, \ul{ij}}_{\cyl_{L,M}}}
	{Z^{\sv,\bullet}_{\cyl_{L,M}}}
	&=
	\frac{1}
	{Z^{\sv,\bullet}_{\cyl_{L,M}}}
	\sum_{\kappa\sim (X\cup\tau_z Y, \ul{ij}) }
	\sum_{\xi^\bullet\sim\kappa}
	\mathbbm{1}(\cO^{AB}_{n,N})
	w[\kappa,\xi^\bullet]\\
	&\qquad+
	\frac{1}
	{Z^{\sv,\bullet}_{\cyl_{L,M}}}
	\sum_{\kappa\sim (X\cup\tau_z Y, \ul{ij}) }
	\sum_{\xi^\bullet\sim\kappa}
	\mathbbm{1}((\cO^{AB}_{n,N})^c)
	w[\kappa,\xi^\bullet].
	\end{split}
	\ee

    An argument identical to the previous section shows that for $|z|$ large enough so that $\Lambda_N\cap\tau_z\Lambda_N=\es$,
    \bes
	\begin{split}
	\frac{1}
	{Z^{\sv,\bullet}_{\cyl_{L,M}}}
	\sum_{\kappa\sim (X\cup\tau_z Y, \ul{ij}) }
	\sum_{\xi^\bullet\sim\kappa}
	\mathbbm{1}(\cO^{AB}_{n,N})
	w[\kappa,\xi^\bullet]
    =
    \mu^{\sv,\bullet}_{\cyl_{L,M}}[f_N^{X,\ul{i}} \cdot \tau_x f_N^{Y,\ul{j}}].
	\end{split}
	\ees
    This converges to $\mu^{\sv}[f_N^{X,\ul{i}} \cdot \tau_x f_N^{Y,\ul{j}}]$ as $L,M\to\infty$ as the function is local, and then since the limiting measure $\mu^{\sv}$ on $(\kappa,\xi^\bullet)$ has the mixing property, this in turn tends to 
    $\mu^{\sv}[f_N^{X,\ul{i}}]
    \mu^{\sv}[f_N^{Y,\ul{j}}]$
    as $|z|\to\infty$. 
    
    Hence it suffices now to show that for some $N$ large enough, we have that the second term in \eqref{eq:extremal-good-bad} is smaller than $\eps$ for all $|z|$ large enough. We have 
    \bes
    \begin{split}
    \frac{1}
	{Z^{\sv,\bullet}_{\cyl_{L,M}}}
	\sum_{\kappa\sim (X\cup\tau_z Y, \ul{ij})}
	\sum_{\xi^\bullet\sim\kappa}
	\mathbbm{1}((\cO^{AB}_{n,N})^c)
	w[\kappa,\xi^\bullet]
    &\le 
    \frac{1}
	{Z^{\sv,\bullet}_{\cyl_{L,M}}}
	\sum_{\kappa\sim (X\cup\tau_z Y, \ul{ij}) }
	\sum_{\xi^\bullet\sim\kappa}
	\mathbbm{1}((\cO_{n,N})^c)
	w[\kappa,\xi^\bullet]\\
    & \qquad 
    +
    \frac{1}
	{Z^{\sv,\bullet}_{\cyl_{L,M}}}
	\sum_{\kappa\sim (X\cup\tau_z Y, \ul{ij})}
	\sum_{\xi^\bullet\sim\kappa}
	\mathbbm{1}((\tau_z\cO_{n,N})^c)
	w[\kappa,\xi^\bullet].
    \end{split}
    \ees
    We consider the first of these terms; the second is dealt with in an identical way. Similarly to \eqref{eq:disorders-estimate-part-2-working-1}, we take some circuits of faces $\ell_X$, $\ell_Y$ of faces surrounding $E_X$ and $\tau_z E_Y$ and condition on the relative heights $g_X$ and $g_Y$ along them, giving:
    \be\label{eq:extremal-working-3}
	\begin{split}
	\frac{1}
	{Z^{\sv,\bullet}_{\cyl_{L,M}}}
	\sum_{\kappa\sim  (X\cup\tau_z Y, \ul{ij})}
	\sum_{\xi^\bullet\sim\kappa}
	\mathbbm{1}((\cO_{n,N})^c)
	w[\kappa,\xi^\bullet]
    &=
    \sum_{g_X,g_Y}
    \sum_{\substack{\kappa_\inn\sim (X\cup\tau_z Y, \ul{ij})\\
    \kappa_\inn \sim g_X,g_Y}}
    w[\kappa_{\inn}]
    \cdot
    \frac{Z^{\sv,\bullet,g}_{\cyl_{L,M}\setminus \ell_{\inn}}}
	{Z^{\sv,\bullet}_{\cyl_{L,M}}}
    \mu^{\sv,\bullet,g}_{\cyl_{L,M}\setminus \ell_{\inn}}
    [(\cO_{n,N})^c],\\
	\end{split}
    \ee
    where now the $\mu^{\sv,\bullet,g}_{\cyl_{L,M}\setminus \ell_{\inn}}$ is the measure on height functions defined outside of both $\ell_X$ and $\ell_Y$ such that the height gradient on $\ell_X$ is $g_X$ and on $\ell_Y$ is $g_Y$, and has the usual flat boundary conditions on the boundaries $\partial^+_F$ and $\partial^-_F$ of the cylinder, and we fix the height to be 0 at its highest value on $\ell_X$. The factor $Z^{\sv,\bullet,g}_{\cyl_{L,M}\setminus \ell_{\inn}}$ is its partition function. 

    The same proof as Lemma \ref{lem:cyl-hf-fkg} shows that the height function under $\mu^{\sv,\bullet,g}_{\cyl_{L,M}\setminus \ell_{\inn}}$ satisfies the FKG inequality. Then the same argument as in the previous section tells us that 
    \bes
    \mu^{\sv,\bullet,g}_{\cyl_{L,M}\setminus \ell_{\inn}}
    [(\cO_{n,N})^c]
    \le 
    C_{X,Y} \cdot
    \mu^{\sv,\bullet,g_Y}_{\cyl_{L,M}\setminus (\ell_Y)_{\inn}}
    [(D_{b_n,N})^c],
    \ees
    where $\mu^{\sv,\bullet,g_Y}_{\cyl_{L,M}\setminus (\ell_Y)_{\inn}}$ is the measure on height functions now only on the outside of $\ell_Y$, with the height fixed to be 0 at some given point in the interior of $\ell_X$, with the same boundary conditions on $\ell_Y$, $\partial^+_F,\partial^-_F$ as before, and $D_{b_n,N}$ is the existence of a circuit of height $b_n$ in $\Lambda_N$, surrounding $\Lambda_n$, and $b_n=b_n(g_X)>0$ depends only on $n$ and $g_X$. 

    Finally, the proof of convergence in the previous section showed that 
    for $N$ large enough, 
    \bes
    \mu^{\sv,\bullet,g_Y}_{\cyl_{L,M}\setminus (\ell_Y)_{\inn}}
    [\tau_z(\cO_{n,N})^c] <\eps,
    \ees
    and further from mixing and the domain Markov property on arrow configurations, we have
    \bes
    \mu^{\sv,\bullet,g_Y}_{\cyl_{L,M}\setminus (\ell_Y)_{\inn}}
    [(D_{b_n,N})^c \ | \ \tau_z\cO_{n,N}]
    \le C'_{X,Y}
    \mu^{\sv,\bullet}_{\cyl_{L,M}}
    [(D_{b_n,N})^c \ | \ \tau_z\cO_{n,N}]
    \to 
    \mu^{\sv,\bullet}_{\cyl_{L,M}}
    [(D_{b_n,N})^c]
    \ees
    as $|z|\to\infty$. The final expression is less than $\eps$ for $N$ large enough, and so combining the above two inequalities and changing the value of $\eps$ above, we have $\mu^{\sv,\bullet,g}_{\cyl_{L,M}\setminus \ell_{\inn}}
    [(\cO_{n,N})^c]<\eps$. This finishes the proof of \eqref{eq:extremal-work-2}, which proves \eqref{eq:xxz-mixing-disorders}, which proves \eqref{eq:xxz-6v:mixing}, and finishes the proof that the ground state $\langle\cdot\rangle$ is extremal.


\subsection{Spin-spin correlations and connection probabilities in the mirror model}

We finish this section with a lemma relating the spin-spin correlations in the XXZ chain and connection probabilities in the mirror model. It is a simple corollary of \eqref{eq:disorders-mirrors}. An analogous statement for the continuous loop model was proved by Ueltschi \cite[Theorem 3.3]{ueltschi}.
\begin{lemma}\label{lem:spin-spins-are-mirror-connections}
	Set $\ul{p}\in\cP$, noting that $\Delta=\Delta(\ul{p})$ given by \eqref{eq:6v-weights} and \eqref{eq:delta} lies in $[-1,1)$. 
    \begin{enumerate}
        \item We have that 
	\bes
		\lim_{\beta\to\infty} \langle S_0^{(1)} S_x^{(1)} \rangle^{\mathrm{XXZ}}_{\TT_L,\beta} 
		=
		\frac{1}{4} 
		\lim_{M\to\infty} \mu^{\mir}_{\TT_{L,M},\ul{p}}[e_0 \leftrightarrow e_x],
	\ees
	where $e_x$ is the vertical edge in $\TT_{L,M}$ with midpoint at $(x,\tfrac{1}{2})$. Moreover for $\ul{p}$ satisfying $p_{\NW},p_{\NE}\ge p_{\es}$, 
    \be\label{eq:mir-conn-prob-decay-to-0}
		\mu^{\mir}_{\ul{p}}[e \leftrightarrow e'] \le \lim_{V\nearrow\ZZ^2} \mu^{\mir}_{V,\ul{p}}[e \leftrightarrow e']
        \to 0
	\ee
    as $||e-e'||\to\infty$ for any pair of edges $e,e'\in \EE(\ZZ^2)$.
    \item We have 
	\bes
	\begin{split}
	\lim_{\beta\to\infty} \langle S_0^{(3)} S_x^{(3)} \rangle^{\mathrm{XXZ}}_{\TT_L,\beta} 
	&=
	\frac{1}{4} (-1)^{x} \lim_{M\to\infty} \mu^{\sv}_{\TT_{L,M},a,b,c}[\sigma^\bullet_{y_0} \sigma^{\circ}_{z_0}\sigma^\bullet_{y_x} \sigma^{\circ}_{z_x}]\\
	&=
	\frac{1}{4}  \left[
	\lim_{M\to\infty} \mu^{\mir}_{\TT_{L,M},\ul{p}}[e_0 \xleftrightarrow{+} e_x]
	-
	\mu^{\mir}_{\TT_{L,M},\ul{p}}[e_0 \xleftrightarrow{-} e_x] 
	\right]\\
    &=
	\frac{1}{4} 
	\lim_{M\to\infty} \mu^{\sv}_{\TT_{L,M},a,b,c}[(h(u_0)-h(u_0'))(h(u_x)-h(u_x')],
	\end{split}
	\ees
    where $y_i$ (resp. $z_i$) is the black (resp. white) face adjacent to the edge $e_i$, while $u_i$ (resp. $u_i'$) is the face adjacent to and to the left of (resp. right of) $e_i$; where $e_0 \xleftrightarrow{+} e_x$ if $e_0 \leftrightarrow e_x$ and, upon orienting the loop through $e_0$ and $e_x$, the edges $e_0$ and $e_x$ are both oriented upwards or both downwards, and $e_0 \xleftrightarrow{-} e_x$ if exactly one is oriented upwards. Note the right hand side converges to $0$ as $|x|\to\infty$.
	\end{enumerate}
\end{lemma}

\begin{proof}
    By the definition of the spin operator $S^{(1)}$ and Proposition \ref{prop:xxz-6v-same}, $\lim_{\beta\to\infty} \langle S_0^{(1)} S_x^{(1)} \rangle^{\mathrm{XXZ}}_{\TT_L,\beta}$ is equal to $1/4$ times the sum of four terms $\lim_{M\to\infty} \langle E_{(\ul{i^-},\ul{i^+})} \rangle^{\sv}_{\TT_{L,M}}$ corresponding to there being disorders on the edges $e_0$ and $e_x$ which are both either two sinks or two sources. The cases where both are two sinks or both are two sources give 0, as there is no valid oriented loop configuration. In both other cases, the formula on the right hand side of \eqref{eq:disorders-mirrors} is $\tfrac{1}{2}\mu^{\mir}_{\TT_{L,M},\ul{p}}[e_0 \leftrightarrow e_x]$, giving the result. 

    Similarly, $\lim_{\beta\to\infty} \langle S_0^{(3)} S_x^{(3)} \rangle^{\mathrm{XXZ}}_{\TT_L,\beta}$ is $1/4$ times the sum of four terms $\lim_{M\to\infty} \langle E_{(\ul{i^-},\ul{i^+})} \rangle^{\sv}_{\TT_{L,M}}$, corresponding to there being disorders on the edges $e_0$ and $e_x$, which are both two upward arrows or two downward arrows, and each term is multiplied by $(-1)$ for each edge which are two downward arrows. The result follows from analysing \eqref{eq:disorders-mirrors}. The equalities giving the expressions in terms of the spins and height functions follow from the definition of $\sigma$ and \eqref{eq:sv-spin-hf}.

    The mixing property we proved in Section \ref{sec:extremality} proves the decay to 0 of both spin-spin correlations. Further, the proof of extremality did not require that the disorders in the six-vertex model lie on edges in the same horizontal row; hence we have $\lim_{L,M\to\infty} \mu^{\mir}_{\TT_{L,M},\ul{p}}[e \leftrightarrow e']\to 0$ as $||e-e'||\to\infty$ for any $e,e'$. Finally, 
    for all $\eps>0$ there exists an $n\in\NN$ such that $\mu^{\mir}_{\ul{p}}[e \xleftrightarrow{\Lambda_n} e'] \ge (1-\eps) \mu^{\mir}_{\ul{p}}[e \leftrightarrow e']$, where $\{e \xleftrightarrow{\Lambda_n} e'\}$ is the event that $e,e'$ are connected by a path which stays in $\Lambda_n$. Then we have $\mu^{\mir}_{\ul{p}}[e \xleftrightarrow{\Lambda_n} e'] = \lim_{L,M\to\infty} \mu^{\mir}_{\TT_{L,M},\ul{p}}[e \xleftrightarrow{\Lambda_n} e']$ as the event is local, and the inequality in \eqref{eq:mir-conn-prob-decay-to-0} follows. 
\end{proof}

The polynomial decay of the correlator $\langle S^{(1)}_0 S^{(1)}_x\rangle$ \eqref{eq:xxz-poly-decay} is, by the lemma above, a consequence of the polynomial decay of connection probability in the mirror model, part 3 of Theorem \ref{thm:main-mirror}. This completes the proof of part 1 of Theorem \ref{thm:main-xxz}.
\end{proof}

We also remark that the convergence we've proved in this section implies that for all $e_1,e_2$ edges in $\ZZ^2$, the limit $\lim_{V\nearrow\ZZ^2}\mu^{\mir}_{V,\ul{p}}[e_1\leftrightarrow e_2]$ exists, so the claims in Theorem \ref{thm:main-mirror} using it are well-defined.

\part{New coupling through eight-vertex and Ashkin-Teller}\label{part:new-coupling}

In this part we prove the rest of our result on the mirror model Theorem \ref{thm:main-mirror}. The main two things still to be proved are that for $p_\NW\cdot p_\NE\le p_\es$ the limiting measure is the unique Gibbs measure, and that connection probabilities decay to 0 polynomially fast in the symmetric case. This in turn gives that the spin-spin correlations in the XXZ chain decay polynomially fast.

In this section we prove more of Theorem \ref{thm:main-mirror}, which we formulate in the following proposition.
\begin{proposition}\label{prop:mir-main-part-2}
    Let $\ul{p}\in\cP$ with $p_\NW,p_\NE\ge p_\es$,
    with $p_\NW\cdot p_\NE\le p_\es$, the measure $\mu_{\ul{p}}^\mir$ is Gibbs and is the unique Gibbs measure and the unique thermodynamic limit, that is, for all $m'\in\Om^\mir_{\ZZ^2}$, we have the weak convergence
        \be
        	\lim_{V \nearrow \ZZ^2} \mu^{\mir;m'}_{V,\ul{p}} = \mu_{\ul{p}}^\mir.
        \ee
\end{proposition}

This proposition, as well as the study of the polynomial decay of the two-point function in Section \ref{sec:symmetric-case}, is facilitated by a new coupling between the Mirror model and a certain eight-vertex model. This eight-vertex model can be in turn coupled to the Ashkin-Teller model, which in turn can be coupled with the six-vertex model; indeed recall Figure \ref{fig:couplings-diagram}.


\section{The mirror model and the eight-vertex and Ashkin-Teller models}\label{sec:new-coupling}

\subsection{The (self-dual) Ashkin-Teller and eight-vertex models}
The Ashkin-Teller model on the graph $G^\bullet$ is defined as follows. Let $\Om^{\AT}_{G^\bullet}$ be the set of configurations $\tau,\tau':F^\bullet\to\{\pm1\}$.
Let $J,U:E^\bullet\to\RR$, $J>0$. We will consider only coupling constants $J,U$ constant ($=J_{\NW}, U_\NW$) on $E^\bullet_{\NW}$, the edges of $E^\bullet$ oriented north-west, and constant ($=J_{\NE}, U_\NE$) on $E^\bullet_{\NE}$, defined similarly. Then the (symmetric in $\tau$ and $\tau'$) Ashkin-Teller measure on $\G^\bullet$ with free boundary conditions is given by
\be\label{eq:AT-measure}
\begin{split}
    \mu^{\AT}_{G^\bullet,J,U}[(\tau,\tau')] 
    &\propto
    \exp\left[
    \sum_{xz\in E_\NE^\bullet} J_{\NE}\tau_x\tau_{z} + J_{\NE}\tau'_x\tau'_{z} + U_{\NE}\tau_x\tau_{z}\tau'_x\tau'_{z}
    \right]\\
    & \qquad \qquad \cdot
    \exp\left[
    \sum_{xz\in E_\NW^\bullet} J_{\NW}\tau_x\tau_{z} + J_{\NW}\tau'_x\tau'_{z} + U_{\NW}\tau_x\tau_{z}\tau'_x\tau'_{z}
    \right].
\end{split}
\ee
In particular, we work on the self-dual curve, parametrised as:
    \be\label{eq:AT-coupling-constants}
        \begin{split}
            e^{-2J_\NW-2U_\NW}=\frac{1-p_\NW}{1+p_\NW}&; \ \ e^{-4J_\NW}=\frac{p_\NE-p_\es}{1+p_\NW};\\
            e^{-2J_\NE-2U_\NE}=\frac{1-p_\NE}{1+p_\NE}&; \ \ e^{-4J_\NE}=\frac{p_\NW-p_\es}{1+p_\NE}.
        \end{split}
    \ee
Recalling the condition \eqref{eq:p-sum-to-1}, that the p variables sum to 1, one can check that this model is self-dual in the sense that its dual is in fact the model with $\NE$ and $\NW$ exchanged. Note that when $p_\NW=p_\NE$, these equations reduce to the well-known symmetric version $\sinh2J=e^{-2U}$. See \cite{aoun-dober-glazman} for an overview of the Ashkin-Teller model.

We define the measure $\mu^{\AT;+,\f}$ to be the free measure, conditioned on $\tau_x=1$ for all $x\in\partial F^\bullet$. Similarly we define $\mu^{\AT;\f,+}, \mu^{\AT;+,+}$. \\

The eight-vertex model we will consider is defined exactly as the Ashkin Teller model is, only that one set of spins lies on the dual graph, that is, $G^\circ$. One should think that one fixes $\tau$ in the Ashkin-Teller model and applies duality to $\tau'$ - see the coupling in Lemma \ref{lem:8v-AT-coupling}. Let $\Om^{\eiv}_{V}$ be the set of configurations $(\pi^\bullet,\pi^\circ)$, where $\pi^\bullet:F^\bullet\to\{\pm1\}$, and $\pi^\circ:F^\circ\to\{\pm1\}$. Let $K:E^\bullet\cup E^\circ\to\RR_{>0}$, $W:V\to\RR$. We will work with coupling constants $W$ constant everywhere and $K$ which are constant ($=K_\NW$) on $E_\NW$, the set of edges in  $E^\bullet\cup E^\circ$ oriented north-west, and constant ($=K_\NE$) on $E_\NE$, the edges in  $E^\bullet\cup E^\circ$ oriented north-east. We define an eight-vertex model measure on $G$ with free boundary conditions as
\be\label{eq:mixed-AT-measure}
\begin{split}
    \mu^{\eiv}_{V,K,W}[(\pi^\bullet,\pi^\circ)] 
    &\propto
    \exp\left[
    \sum_{xx'\in E_\NE^\bullet} 
        K_{\NE}\pi^\bullet_x\pi^\bullet_{x'} 
        + 
        K_{\NW}\pi^\circ_y\pi^\circ_{y'} 
        + 
        W\pi^\bullet_x\pi^\bullet_{x'}\pi^\circ_y\pi^\circ_{y'}
    \right]\\
    & \qquad \qquad \cdot
    \exp\left[
    \sum_{xx'\in E_\NW^\bullet} 
        K_{\NW}\pi^\bullet_x\pi^\bullet_{x'} 
        + 
        K_{\NE}\pi^\circ_y\pi^\circ_{y'} 
        + 
        W\pi^\bullet_x\pi^\bullet_{x'}\pi^\circ_y\pi^\circ_{y'}
    \right],
\end{split}
\ee
where for $e=xx'\in E^\bullet$, $e^*=yy'\in E^\circ$ is its dual edge. We define the measures $\mu^{\eiv;\bplus}_{V,K,W}$, $\mu^{\eiv;\wplus}_{V,K,W}$, $\mu^{\eiv;\bplus,\wplus}_{V,K,W}$ similarly to the six-vertex model. We work with constants $K,W$ on the self-dual curve:
\be\label{eq:8v-coupling-constants}
\begin{split}
    e^{-2K_\NW-2W} &= p_\NW; \\
    e^{-2K_\NE-2W} &= p_\NE; \\
    e^{-2K_\NW-2K_\NE} &= p_\es.
\end{split}
\ee

\begin{proposition}\label{prop:8v-AT}
    Let $\ul{p}\in\cP$ such that $p_\NE,p_\NW\ge p_\es$ and let $K,W$ satisfy \eqref{eq:8v-coupling-constants} and J,U satisfy \eqref{eq:AT-coupling-constants}.
    \begin{enumerate}
        \item In the self-dual eight-vertex model with parameters $K,W$, the measures $\mu^{\eiv}_{V,K,W}$,   $\mu^{\eiv;\bplus}_{V,K,W}$, $\mu^{\eiv;\wplus}_{V,K,W}$, $\mu^{\eiv;\bplus,\wplus}_{V,K,W}$ converge weakly as $V\nearrow\ZZ^2$ to a common limit $\mu^{\eiv}_{K,W}$, which is Gibbs, $\ZZ^2$ translation-invariant and ergodic, and extremal, and satisfies 
        \be\label{eq:8v-decay-correlations}
        \begin{split}
            \lim_{||x_1-x_2||\to\infty}& \mu^\eiv_{K,W}[\pi^\bullet_{x_1}\pi^\bullet_{x_2}]=0;\\
            \lim_{\substack{||x_1-x_2||\to\infty \\ ||y_1-y_2||\to\infty}}&
            \mu^\eiv_{K,W}[\pi^\bullet_{x_1}\pi^\bullet_{x_2}
            \pi^\circ_{y_1}\pi^\circ_{y_2}]=0.
        \end{split}
        \ee
        Further if $\ul{p}$ satisfies $p_\NW\cdot p_\NE\ge p_\es$ (equivalently $W\le0$), then $\mu^{\eiv}_{K,W}$ is the unique Gibbs measure and the unique thermodynamic limit, that is, for all $\pi'\in\Om^\eiv_{\ZZ^2}$, we have the weak convergence
        \be
        	\lim_{V \nearrow \ZZ^2} \mu^{\eiv;\pi'}_{V,K,W} = \mu_{K,W}^\eiv.
        \ee
        \item In the self-dual Ashkin-Teller model with parameters $J,U$, the measures $\mu^{\AT}_{G^\bullet,J,U}$,   $\mu^{\AT;+,f}_{G^\bullet,J,U}$, $\mu^{\AT;f,+}_{G^\bullet,J,U}$, $\mu^{\AT;+,+}_{G^\bullet,J,U}$ converge weakly as $G^\bullet\nearrow\GG^\bullet$ to a common limit $\mu^{\AT}_{J,U}$, which is Gibbs, $\GG^\bullet$ translation-invariant and ergodic, and extremal, and satisfies 
        \be\label{eq:AT-decay-correlations}
        \begin{split}
            \lim_{||x_1-x_2||\to\infty}& \mu^\AT_{J,U}[\tau_{x_1}\tau_{x_2}]=0;\\
        \end{split}
        \ee
        Further if $\ul{p}$ satisfies $p_\NW\cdot p_\NE\le p_\es$, then $\mu^{\sv}_{J,U}$ is the unique Gibbs measure and the unique thermodynamic limit, that is, for all $(\tau_0,\tau'_0)\in\Om^\sv_{\GG^\bullet}$, we have the weak convergence
        \be
        	\lim_{G^\bullet \nearrow \GG^\bullet} \mu^{\AT;(\tau_0,\tau'_0)}_{V,J,U} = \mu_{J,U}^\AT.
        \ee
    \end{enumerate}
\end{proposition}

Note that in the symmetric case, this proposition covers the Ashkin-Teller range $J\ge U$ on the self-dual line, and the eight-vertex range $K\ge W$ on the self-dual line.


\subsection{The new coupling}
We prove Proposition \ref{prop:mir-main-part-2} from Proposition \ref{prop:8v-AT}.

Let us first describe heuristically the coupling between the mirror model and the eight-vertex model. One can replace the factor $2^{l(m)}$ in the measure \eqref{eq:mir-measure} by colouring the loops either red or blue. Given the red loops and the locations of the mirrors, one can recover the blue loops, since every edge hosts exactly one loop, and the mirrors tell the loops where to go. Now taking the marginal on the set of red-coloured edges gives a configuration of even subgraphs of $G$. This can be interpreted as a spin model on the faces of $G$, with spins $\sigma$ taking values in $\{\pm 1\}$, and changing value exactly when one crosses a red edge. 

This spin model is the eight-vertex model described above. It is also studied under the name ``mixed Ashkin Teller model'' \cite{HDJS-infAT}. Fixing the spins $\sigma$ on the black faces $G^\bullet$, and applying a duality to those on $\G^\circ$, one obtains a spin model on $G^\circ$ with two spins $\sigma,\sigma'\in\{\pm\}$; this is the normal Ashkin teller model (the term ``mixed'' Ashkin Teller comes from this duality in half of the spins). 

To go the other way in the coupling, one takes a spin configuration, colours the domain walls red and other edges blue, and then samples mirrors at each vertex independently, with weights proportional to $p_{\NE},p_{\NW},p_\varnothing$, conditional on loops being consistently coloured. Note that at some vertices the choice of the mirror configuration is deterministic.

\be\label{eq:sample-mirror-from-8v}
    \begin{cases}
        m_v \in m_\es \ \mathrm{deterministically}  
        & \mathrm{if} \ \pi^\bullet_{x_1}\neq \pi^\bullet_{x_2}, \ 
        \pi^\circ_{x_1}\neq \pi^\circ_{x_2}\\
        m_v \in m^\bullet \ \mathrm{deterministically}    
        & \mathrm{if} \ \pi^\bullet_{x_1} = \pi^\bullet_{x_2}, \ 
        \pi^\circ_{x_1}\neq \pi^\circ_{x_2}\\
        m_v \in m^\circ   \ \mathrm{deterministically}  
        & \mathrm{if} \ \pi^\bullet_{x_1} \neq \pi^\bullet_{x_2}, \ 
        \pi^\circ_{x_1} = \pi^\circ_{x_2}\\
        m_v \in m_* \ \mathrm{w.p.} \ p_{*} \ \mathrm{for \ all} \ *\in\{\NW,\NE,\es\},
        & \mathrm{if} \ \pi^\bullet_{x_1} = \pi^\bullet_{x_2}, \ 
        \pi^\circ_{x_1} = \pi^\circ_{x_2},\\
    \end{cases}
\ee

\begin{proposition}[New coupling of mirror and eight-vertex models]\label{prop:mir-8v-coupling}
    Let $\ul{p}\in\cP$. Let $(K_\NW, K_\NE, W)$ satisfy \eqref{eq:8v-coupling-constants}. 
    \begin{enumerate}
        \item Take $m\sim \mu^{\mir}_{V,\ul{p}}$ and independently colour the loops red with probability $\tfrac{1}{2}$ each. Let $\pi=(\pi^\bullet,\pi^\circ)\in \Om^{\eiv}_{V}$ be defined by setting $\pi_{x_0}=\pm 1$ independently and uniformly at some fixed face $x_0$, and then for every pair of adjacent faces $x\in F^\bullet$, $y\in F^\circ$ separated by an edge $e$, let $\pi^\bullet_x\neq\pi^\circ_y$ if and only if $e$ lies in a red loop. Then $\pi\sim \mu^{\eiv}_{V}$. 
        \item Take $\pi\sim \mu^{\eiv}_{V}$ and sample a mirror configuration $m$ according to the rules \eqref{eq:sample-mirror-from-8v}. Then $m\sim\mu^{\mir}_{V,\ul{p}}$. 
        \item There exists an analogous coupling between $\mu^{\eiv;\bplus,\wplus}_{V,K,W}$ and $\mu^{\mir}_{V,p}$, the only difference being that one sets $\pi^\bullet\equiv \pi^\circ\equiv +1$ on $\partial F$ instead of the random $\pi_{x_0}$. Similarly, there exists an analogous coupling between $\mu^{\eiv;\bplus}_{V,K,W}$ and $\mu^{\mir;\bullet}_{V,\ul{p}}$, where one fixes only $\pi^\bullet\equiv1$ on $\partial F^\bullet$ and in sampling $m$ one fixes $\partial V\subset m^\bullet$; similar for $\mu^{\eiv;\wplus}_{V,K,W}$ and $\mu^{\mir;\circ}_{V,\ul{p}}$.
    \end{enumerate}

        
\end{proposition}

\begin{proof}
    Colour each of the loops or paths in a mirror configuration red, independently with probability $\tfrac{1}{2}$. This produces a joint measure on mirrors and configurations $L$, colourings of edges either red or uncoloured, with density proportional to:
    \bes
        p_\NW^{|m_\NW|} p_\NE^{|m_\NE|} p_\es^{|m_\es|}
            \mathbbm{1}\{m\sim L\},
    \ees
    where $m\sim L$ if and only if $m$ and $L$ are compatible, that is, the colouring $L$ is constant on each of the loops of $m$. The red loops are essentially equivalent to a configuration $(\pi^\bullet,\pi^\circ)$ of spins. Indeed, fix some face $x\in F^\bullet$ and let $\pi^\bullet_x=\pm1$ independently with probability 1/2. Let the value of the spin change sign between adjacent faces of $\ZZ^2$ if and only if one crosses a red loop in $L$. Hence we can rewrite the joint measure as proportional to
    \bes
        p_\NW^{|m_\NW|} p_\NE^{|m_\NE|} p_\es^{|m_\es|}
            \mathbbm{1}\{m\sim (\pi^\bullet,\pi^\circ)\},
    \ees
    where $m\sim (\tau,\sigma)$ if and only if the disagreements in $(\pi^\bullet,\pi^\circ)$ produce a subset of loops of $m$. 
    
    It remains to show that the marginal on spins is $\mu^{\eiv;\f,\f}_{V,K,W}$. Consider a compatible configuration $(m;\pi^\bullet,\pi^\circ)$. A vertex of $G$ is associated to an edge of $E^\bullet$ and an edge of $E^\circ$. Let $\om[\pi^\bullet]$ be the edges of $E^\circ$ either side of which the spins of $\pi^\bullet$ differ, and $\om[\pi^\circ]$ similar. Note both of these sets can be considered subsets of $V$, as $V,E^\bullet,E^\circ$ are in bijection. Now, each vertex $v\in V$ lies in one of $\om[\pi^\bullet]\triangle \om[\pi^\circ]$, $\om[\pi^\bullet]\cap \om[\pi^\circ]$, or $V\setminus(\om[\pi^\bullet]\cup \om[\pi^\circ])$. 

    In the first two cases, the value of $m$ at the vertex $v$ is determined. Indeed, in the first case, the spins produce one loop at the vertex which makes a turn, so a mirror must be present and must be oriented so the loop reflects off it. In the second case, the spins produce one loop which passes straight through the vertex, so there must me no mirror. In the third case, either every edge of $G$ emanating from $v$ is blue, or they are all red. In either case, any mirror configuration is permitted: NE, NW, or no mirror, whose probabilities sum to 1. The marginal on $\pi$ is then proportional to
    \bes
    \begin{split}
        (p_\NE)^{|\om_\NE|}
        (p_\NW)^{|\om_\NW|}
        (p_\es)^{|\om[\pi^\bullet]\cap\om[\pi^\circ]|},
    \end{split}
    \ees
    where here we write $\om_\NE$ for those edges of $\om[\pi^\bullet]\triangle\om[\pi^\circ]$ oriented north-east, and $\om_\NW$ for those oriented north-west. It is straightforward to check that this is $\mu^{\eiv}_{V,K,W}$.\\

    For the coupling between $\mu^{\eiv;\bplus,\wplus}_{V,K,W}$ and $\mu^{\mir}_{V,\ul{p}}$, observe that since the number of loops which intersect $\D_E V$ is exactly $\tfrac{1}{2}|\D_E V|$, one can equivalently write the measure \eqref{eq:mir-measure} with $l(m)$ meaning the number of loops which do not intersect $\D_E V$. From here, the proof only has two differences from the above: paths that do intersect $\D_E V$ are deterministically uncoloured, and instead of setting $\tau_x=1$ independently with probability 1/2, one sets $\tau_x=1$ deterministically for every $x\in\partial F$. 

    Similarly, for the coupling between $\mu^{\eiv;\bplus}_{V,K,W}$ and $\mu^{\mir;\bullet}_{V,\ul{p}}$, everything is identical to the original proof, except one sets $\tau_x=1$ deterministically for every $x\in\partial F^\bullet$; similar for $\mu^{\eiv;\wplus}_{V,K,W}$ and $\mu^{\mir;\circ}_{V,\ul{p}}$.
\end{proof}



The strength of this coupling is that the probability that two edges are connected in the mirror model is equal to a four-point function of spins in the eight-vertex model. For edges $e_i\in E$, $i=1,2$, let $x_i,y_i$ be the (resp. black, white) faces of $G$ either side of the edge $e_i$.
	
\begin{lemma}\label{lem:8v-connection-probs}
    Let $\ul{p}\in\cP$ and $K,W$ satisfy \eqref{eq:8v-coupling-constants}. Denote by $e_1\leftrightarrow e_2$ the event that $e_1$ and $e_2$ are connected by a loop in the mirror model. Then
    \bes
        \begin{split}
        \mu^{\mir}_{V,\ul{p}}[e_1\leftrightarrow e_2]
        =
        \mu^{\eiv}_{V,K,W}
        [\pi^\bullet_{x_1}\pi^\circ_{y_1}\pi^\bullet_{x_2}\pi^\circ_{y_2}].
        \end{split}
    \ees
\end{lemma} 
\begin{proof}
        The proof is similar to many others which use a colour-switching argument. 
        We have:
        \bes
        \begin{split}
            \mu^{\eiv}_{V,K,W}
            [\pi^\bullet_{x_1}\pi^\circ_{y_1}\pi^\bullet_{x_2}\pi^\circ_{y_2}]
            &=
            \mu^{\mir-\eiv}_{V,K,W}
            [\pi^\bullet_{x_1}\pi^\circ_{y_1}\pi^\bullet_{x_2}\pi^\circ_{y_2}\mathbbm{1}\{e_1\leftrightarrow e_2\}]\\
            &\ \ \ +
            \mu^{\mir-\eiv}_{V,K,W}
            [\pi^\bullet_{x_1}\pi^\circ_{y_1}\pi^\bullet_{x_2}\pi^\circ_{y_2}\mathbbm{1}\{e_1\not\leftrightarrow e_2\}]\\
            &=
            \mu^{\mir;\f}_{V,\ul{p}}[e_1\leftrightarrow e_2].
        \end{split}
        \ees
        Here the second equality comes from the fact that if $e_1$ and $e_2$ are not connected by a loop in the coupling, then one can change the colour of the loop passing through $e_1$, which does not change the weight of the contributing mirrors and spins, but does multiply the product $\pi^\bullet_{x_1}\pi^\circ_{y_1}\pi^\bullet_{x_2}\pi^\circ_{y_2}$ by $-1$. Hence all such configurations cancel. Finally, on the event $\{e_1\leftrightarrow e_2\}$, the product $\pi^\bullet_{x_1}\pi^\circ_{y_1}\pi^\bullet_{x_2}\pi^\circ_{y_2}$ is always 1. 

\end{proof}

\begin{proof}[Proof of Proposition \ref{prop:mir-main-part-2}]

    Take the infinite volume eight-vertex measure $\mu^\eiv_{K,W}$ described in Proposition \ref{prop:8v-AT}. Sample mirrors at each vertex according to the procedure \eqref{eq:sample-mirror-from-8v}. One can prove that the marginal $\mu$ on mirrors must be $\mu^\mir_{\ul{p}}$. Indeed, use the coupling above and the local sampling of mirrors given spins and apply the same proof as in the proof of Proposition \ref{prop:mir-conv-from-6v} to show that $\mu^\mir_{V,\ul{p}}$, $\mu^{\mir;\bullet}_{V,\ul{p}}$, and $\mu^{\mir;\circ}_{V,\ul{p}}$ all converge to $\mu$. By Proposition \ref{prop:mir-conv-from-6v}, these measures already converge to $\mu^\mir_{\ul{p}}$.\\

    We prove that if $\ul{p}$ satisfies $p_\NW\cdot p_\NE\ge p_\es$, then $\mu_{\ul{p}}^\mir$ is the unique thermodynamic limit (this implies it is the unique Gibbs measure, as every extremal Gibbs measure can be expressed as a thermodynamic limit; see \cite[Theorem 6.63]{friedli-velenik})). 
    
    Let $V\subset\ZZ^2$ be a domain, let $m_0\in\Om^\mir_{\ZZ^2}$ (possibly with infinite paths) and consider the measure $\mu^{\mir;m_0}_{V,\ul{p}}$. We apply the coupling from Proposition \ref{prop:mir-8v-coupling} to obtain a eight-vertex measure on $V$ - let us describe this measure precisely. One obtains a (partial) matching $M$ of the edges in $\partial_EV$ (the set of edges with exactly one endpoint in $V$) given by connectivity via the segments of loops and paths of $m_0$ lying outside of $V$. In constructing the coupling, we colour every loop or path in a configuration $m|_V\cup m_0|_{\ZZ^2\setminus V}$ independently uniformly red or blue, and apply the remaining steps as usual to obtain a spin configuration on $F$. 
    
    The boundary conditions one obtains on the spins $(\pi^\bullet,\pi^\circ)$ are as follows: for each pair of edges $e_1,e_2$ in the pairing $M$, the spins either side of $e_1$ agree if and only if those either side of $e_2$ agree. Write $A_{m_0}$ for this event, and write $\mu^{\eiv;m_0}_{V,K,W}$ for the free boundary conditions measure conditioned on $A_{m_0}$. Note that this measure specialises to the free measure and the measures $\mu^{\eiv;\bplus}_{V,K,W}$, $\mu^{\eiv;\wplus}_{V,K,W}$ in the cases where $M$ is the empty pairing, $m_0$ has mirrors everywhere connecting black faces, and $m_0$ has mirrors everywhere connecting white faces, respectively. 
    
    One can write the measure $\mu^{\eiv;m_0}_{V,K,W}$ as a convex combination: for some cylinder event $B$,
    \bes
    \begin{split}
        \mu^{\eiv;m_0}_{V,K,W}[B]
        &=
        \sum_{\substack{\pi'\in\{\pm1\}^{\partial F} \\ \pi'\in A_{m_0}}}
        \mu^{\eiv;m_0}_{V,K,W}[\pi|_{\partial F} = \pi']
        \mu^{\eiv;m_0}_{V,K,W}[B \ | \ \pi|_{\partial F} = \pi']\\
        &=
        \sum_{\substack{\pi'\in\{\pm1\}^{\partial F} \\ \pi'\in A_{m_0}}}
        \mu^{\eiv;m_0}_{V,K,W}[\pi|_{\partial F} = \pi']
        \mu^{\eiv;\pi'}_{V,K,W}[B]
    \end{split}
    \ees
    Let $V_k$ be a sequence of domains in $\ZZ^2$ such that $V_k\nearrow\ZZ^2$ and the sets of boundary faces $\partial F(V_k)$ are all disjoint. Let $\pi^{\max}\in\Om^\eiv_{\ZZ^2}$ be such that $\pi^{\max}|_{\partial F(V_k)}$ maximises $\mu^{\eiv;\pi'}_{V,K,W}[B]$ among all $\pi'\in A_{m_0}$, for all $k$. Here we use the sets $\partial F(V_k)$ being all disjoint so that we may choose $\pi^{\max}$ freely on each $\partial F(V_k)$. Let $\pi^{\min}$ be defined similarly; note both configurations are dependent on the event $B$. We have
    \bes
        \mu^{\eiv;\pi^{\min}}_{V_k,K,W}[B] \le \mu^{\eiv;m_0}_{V_k,K,W}[B] \le \mu^{\eiv;\pi^{\max}}_{V_k,K,W}[B].
    \ees
    By Proposition \ref{prop:8v-AT}, for $p_\NW\cdot p_\NE\ge p_\es$, the left and right hand sides converge to $\mu^{\eiv}_{K,W}[B]$, and so so does $\mu^{\eiv;m_0}_{V_k,K,W}[B]$. We can remove the condition that the sets $\partial F(V_k)$ are all disjoint: indeed, for an arbitrary sequence $V_k\nearrow\ZZ^2$, any subsequence has a further subsequence $V_{k_l}$ such that the sets $\partial F(V_{k_l})$ are all disjoint - this is enough to conclude the convergence for the original sequence. Since $B$ was an arbitrary cylinder event, we have that $\mu^{\eiv;m_0}_{V,K,W}$ converges to $\mu^{\eiv}_{K,W}$ weakly as $V\to\infty$. By the same arguments as at the start of this proof where we showed $\mu=\mu^\mir_{\ul{p}}$, we find that $\mu^{\mir;m_0}_{V,\ul{p}}$ converges to $\mu^\mir_{\ul{p}}$ weakly.    
\end{proof}



\section{Proofs of results for eight-vertex and Ashkin-Teller models}\label{sec:8v-AT}

In this section we prove Proposition \ref{prop:8v-AT}, as well as a corresponding theorem for the randomm cluster represetation of the eight-vertex and Ashkin-Teller models, Proposition \ref{prop:8v-AT-RC} below. The main input is Theorem \ref{thm:6v-deloc}, the delocalisation of the six-vertex height function. The proof uses standard arguments and existing couplings between the eight-vertex and Ashkin-Teller models, and between the Ashkin-Teller and six-vertex models.

\subsection{The Ashkin-Teller random cluster model}

A key tool in the study of Ashkin-Teller models is an Edwards-Sokal type percolation representation introduced in \cite{pfister-velenik-AT}. Taking some pair $s,s'$ from the three spins $\tau,\tau'$ and the product $\tau\tau'$, one can fix $s'$ and apply an Edwards-Sokal type expansion for $s$ to obtain a percolation configuration $\om$. Spin-spin correlations in $s$ are then given by connection probabilities in $\om$. For a thorough overview of these representations, see \cite{dober}.

For the model symmetric in $\tau$ and $\tau'$ with $J>|U|$, one finds two distinct representations by taking $(s,s')=(\tau,\tau')$, and $(s,s')=(\tau,\tau\tau')$. In both cases, one can couple the representations $\om$ for $s$ and $\om'$ for $s'$ to form what is termed in some of the literature the Ashkin-Teller random cluster (ATRC) model. In the case $(s,s')=(\tau,\tau')$, this leads to a coupling between the Ashkin-Teller and eight-vertex models, while in the case $(s,s')=(\tau,\tau\tau')$, a coupling between the Ashkin-Teller and six-vertex models. 

The same representations hold for the eight-vertex model, with the exception that since the spins $\pi^\bullet$ and $\pi^\circ$ are defined on different graphs, their product is not well-defined, and so one is forced to choose $(s,s')=(\pi^\bullet,\pi^\circ)$. If $(\om,\om')$ is distributed according to the ATRC measure (where one chooses the ATRC model in the case $(s,s')=(\tau,\tau')$), we will find that $(\om,(\om')^*)=(\eta^\bullet,\eta^\circ)$ distributed according to the the eight-vertex version of the same. This is a measure on pairs of percolation configurations $\eta^\bullet,\eta^\circ$ on $G^\bullet$, $G^\circ$ respectively. 

To avoid confusion with the two possible choices of the ATRC model, we will mainly use the term eight-vertex random cluster model (8vRC) for this model, but one should keep in mind (particularly readers familiar with the Ashkin-Teller model) that it is essentially identical to talking about the ATRC model with $(s,s')=(\tau,\tau')$.

To summarise, the Ashkin-Teller and eight-vertex models are related as follows: take Ashkin-Teller spins $\tau,\tau'$, fix $\tau$, and apply a duality in $\tau'$ - one obtains a pair of spins $(\tau=\pi^\bullet, \pi^\circ)$ distributed as a eight-vertex model. This duality can be made into an explicit coupling via the ATRC/8vRC model, where the duality is simply taking the dual configuration of one component. \\

Let us define the 8vRC model. Fix $V\subset\ZZ^2$. Let $\Om^{\eivRC}_V$ be the set of pairs $(\eta^\bullet,\eta^\circ)\in\{0,1\}^{E^\bullet}\times\{0,1\}^{E^\circ}$. For such a configuration, write $\eta_{\lambda_1,\lambda_2}$ for the set of edges $e\in E^\bullet$ such that $\eta_{\pi^\bullet}(e)=\lam_1$, and $\eta_{\pi^\circ}(e^*)=\lam_2$. Let $\eta_0=(\eta_0^\bullet,\eta_0^\circ)\in\Om^{\eivRC}_{\ZZ^2}$. On $\Om^{\eivRC}_V$, we have the measure with boundary conditions $\eta_0$:
\bes
\begin{split}
    \mu^{\eivRC;\eta_0}_{V,\ul{p}}[(\eta^\bullet,\eta^\circ)]
    &\propto
    p_\es^{|\eta_{0,0}|}
    (2p_\es)^{|\eta_{1,1}|}\\
    & \ \cdot
    (p_\NW-p_\es)^{
    |\eta_{1,0}\cap E^\bullet_\NW|
    +
    |\eta_{0,1}\cap E^\bullet_\NE|}
    (p_\NE-p_\es)^{
    |\eta_{0,1}\cap E^\bullet_\NW|
    +
    |\eta_{1,0}\cap E^\bullet_\NE|}\\
    & \ \cdot
    2^{k(\eta^\bullet;\eta_0) + k(\eta^\circ;\eta_0)},
\end{split}
\ees
where $k(\eta^\bullet;\eta_0)$ is the number of clusters of the configuration $\eta^\bullet|_{E^\bullet} \cup \eta_0^\bullet|_{\EE\setminus E^\bullet}$ intersecting edges in $E^\bullet$, and $k(\eta^\circ;\eta_0)$ is defined similarly. Write $\mu^{\eivRC;1,0}_{V,\ul{p}}$ for the measure where $\eta_0^\bullet\equiv1$ and $\eta_0^\circ\equiv0$, and define the measures, $\mu^{\eivRC;0,0}_{V,\ul{p}}, \mu^{\eivRC;0,1}_{V,\ul{p}}$ and $\mu^{\eivRC;1,1}_{V,\ul{p}}$ similarly.

Say that a measure $\mu$ on $\Om^{\eivRC}_{\ZZ^2}$ is a Gibbs measure for the 8vRC model if for all $\eta_0\in\Om^{\eivRC}_{\ZZ^2}$, 
\bes
    \mu[\eta|_V = \om \ | \ \eta|_{\ZZ^2\setminus V} = \eta_0|_{\ZZ^2\setminus V}]
    =
    \mu^{\eivRC;\eta_0}_{V,\ul{p}}[\om].
\ees

\begin{proposition}\label{prop:8v-AT-RC}
    Let $\ul{p}\in\cP$ such that $p_\NE,p_\NW\ge p_\es$ and let $K,W$ satisfy \eqref{eq:8v-coupling-constants} and J,U satisfy \eqref{eq:AT-coupling-constants}.
    The eight-vertex Random Cluster model with parameters $\ul{p}$ has a unique Gibbs measure $\mu^{\eivRC}$ on $\ZZ^2$, which is also $\ZZ^2$ translation-invariant and ergodic, and extremal. It is the weak limit of the measures $\mu^{\eivRC;\clubsuit}_{V,K,W}$, for $\clubsuit\in\{0,1\}^2$ as $V\nearrow\ZZ^2$, and is supported on configurations where $\om^\bullet$, $\om^\circ$, $(\om^\bullet)^*$, and $(\om^\circ)^*$ have no infinite clusters. A similar statement holds for the ordinary Ashkin-Teller Random Cluster model.
\end{proposition}

\begin{remark}
    Note that the convergence of some of the measures described above to some infinite volume limits is known, due to the FKG inequalities in the ATRC model. Their limits being the same is new.
\end{remark}

The basic plan of the proofs in this section is that Theorem \ref{thm:6v-deloc} (six-vertex delocalisation) implies Proposition \ref{prop:8v-AT-RC}, which in turn implies the eight-vertex and Ashkin Teller results, Proposition \ref{prop:8v-AT}.

\subsection{FKG and stochastic domination in 8vRC}

It is known that the Ashkin-Teller random cluster model (and therefore the $\eivRC$ model) satisfies an FKG inequality in the range that we are working in \cite{pfister-velenik-AT}. In particular, this inequality holds for two different orderings of the configurations, depending on the range of parameters. Let us make this precise. 

Consider $\eta_1,\eta_2\in\Om^{\eivRC}_V$. We say that 
\be\label{eq:orderings-8vRC}
\begin{split}
	\eta_1\le\eta_2 \qquad 
	&\mathrm{if} \qquad 
	\eta_1^\bullet(e)\le \eta_2^\bullet(e), \ 
	\eta_1^\circ(e)\le \eta_2^\circ(e);\\
	\eta_1\tilde{\le}\eta_2 \qquad 
	&\mathrm{if} \qquad 
	\eta_1^\bullet(e)\le \eta_2^\bullet(e), \ 
	\eta_1^\circ(e)\ge \eta_2^\circ(e),\\
\end{split}
\ee
for all $e\in E^\bullet$. We say an event $A$ is $\bar{\le}$-increasing with respect to an ordering $\bar{\le}$ if $\eta_1\in A$ and $\eta_1\bar{\le}\eta_2$ imply that $\eta_2\in A$. A measure $\mu$ on $\Om^\eivRC_V$ satisfies the FKG inequality for the ordering $\bar{\le}$ if for all $\bar{\le}$-increasing events $A,B$, one has:
\bes
	\mu[A\cap B] \ge \mu[A]\mu[B].
\ees

\begin{lemma}[\cite{pfister-velenik-AT}]\label{lem:8vRC-FKG}
	Let $\ul{p}\in\cP$ with $p_\NE,p_\NW \ge p_\es$, and let $K,W$ satisfy \eqref{eq:8v-coupling-constants}. The measures $\mu^{\eivRC;i,j}_{G,K,W}$ for $i,j\in\{0,1\}$ satisfy the FKG inequality:
	\begin{enumerate}
		\item for the ordering $\le$ if $p_\NE\cdot p_\NW \ge p_\es$;
		\item for the ordering $\tilde{\le}$ if $p_\NE\cdot p_\NW \le p_\es$.
	\end{enumerate}
\end{lemma}

Let $\mu,\nu$ be two measures on $\Om^{\eivRC}_V$. We say that $\mu$ stochastically dominates $\nu$ with respect to the ordering $\le$ if for all $\le$-increasing events $A$, we have $\mu[A]\ge\nu[A]$. We write $\nu \prec\mu$ for this condition, and $\nu \ \tilde{\prec} \ \mu$ for the same condition with respect to the ordering $\tilde{\le}$. It follows straightforwardly that for all $i,j,i',j'\in\{0,1\}$,
\be\label{eq:8vRC-stoch-dom}
\begin{split}
	\mu^{\eivRC;i,j}_{V,K,W} &\prec \mu^{\eivRC;i',j'}_{V,K,W} \ \mathrm{for} \ i\le i', \ j\le j' 
	\qquad \mathrm{if} \ p_\NE\cdot p_\NW \ge p_\es;\\
	\mu^{\eivRC;i,j}_{V,K,W} \ &\tilde{\prec} \ \mu^{\eivRC;i',j'}_{V,K,W} \ \mathrm{for} \ i\le i', \ j\ge j' 
	\qquad \mathrm{if} \ p_\NE\cdot p_\NW \le p_\es.
\end{split}
\ee
Standard arguments then show that for $p_\NE\cdot p_\NW \ge p_\es$, the measures $\mu^{\eivRC;1,1}_{V,K,W}$ and $\mu^{\eivRC;0,0}_{V,K,W}$ have weak infinite volume limits $\mu^{\eivRC;1,1}_{K,W}$ and $\mu^{\eivRC;0,0}_{K,W}$ respectively, while for $p_\NE\cdot p_\NW \le p_\es$, the measures $\mu^{\eivRC;1,0}_{V,K,W}$ and $\mu^{\eivRC;0,1}_{V,K,W}$ have weak infinite volume limits $\mu^{\eivRC;1,0}_{K,W}$ and $\mu^{\eivRC;0,1}_{K,W}$ respectively. It is straightforward (see eg. \cite[Section 4]{grimmett-RC}) to show that these limits are Gibbs, $\GG^\bullet$-invariant and ergodic, and extremal.

\subsection{Unique Gibbs measure for 8vRC}

Further standard arguments (see eg. \cite[Section 4]{grimmett-RC}) show that any Gibbs measure $\mu$ for the 8vRC model satisfies:
\bes
\begin{split}
	\mu^{\eivRC;0,0}_{K,W} &\prec \mu \prec \mu^{\eivRC;1,1}_{V,K,W}  
	\qquad \mathrm{if} \ p_\NE\cdot p_\NW \ge p_\es;\\
	\mu^{\eivRC;0,1}_{K,W} \ &\tilde{\prec} \ \mu \ \tilde{\prec} \ \mu^{\eivRC;1,0}_{K,W}
	\qquad \mathrm{if} \ p_\NE\cdot p_\NW \le p_\es.
\end{split}
\ees
In particular if the right and left hand sides of these stochastic inequalities are equal, then there exists a unique Gibbs measure. Further again, standard arguments show that in the first case, if $\mu^{\eivRC;1,1}$ exhibits no infinite cluster in both $\eta^\bullet$ and $\eta^\circ$ almost surely, then $\mu^{\eivRC;1,1}=\mu^{\eivRC;0,0}$. Note that shifting $\eta^\bullet$ by $(1,0)$ under $\mu^{\eivRC;1,1}$ has the same law as $\eta^\circ$, so it suffices here to show that there is almost surely no infinite cluster in $\eta^\bullet$. 

In the second case, it suffices to show that $\mu^{\eivRC;1,0}$ exhibits no infinite cluster in both $\eta^\bullet$ and $(\eta^\circ)^*$ almost surely. Under the measure $\mu^{\eivRC;1,0}$, $\eta^\bullet$ and $(\eta^\circ)^*$ have the same marginal; indeed, they are jointly distributed as the usual ATRC measure under $++$ boundary conditions, which is symmetric in its two components. Hence here also it suffices to show that $\eta^\bullet$ has no infinite cluster almost surely.\\

Let us briefly describe the rest of the proof. We will first show how the eight-vertex model, the Ashkin-Teller model and the 8vRC model can be coupled. From this coupling we'll show that
\be\label{eq:8v-8vRC-connections}
	\mu^{\eivRC;1,1}_{\ul{p}}[x\xleftrightarrow{\eta^\bullet} \infty] 
	=
	\mu^{\eiv;\bplus,\wplus}_{K,W}[\pi^\bullet_x]
	=
	\mu^{\AT;+,\f}_{J,U}[\tau_x],
\ee
where $\mu^{\eiv;\bplus,\wplus}_{K,W}, \mu^{\AT;+,\f}_{J,U}$ are the infinite volume weak limits of $\mu^{\eiv;\bplus,\wplus}_{V,K,W}$ and $\mu^{\AT;+,\f}_{V,J,U}$, respectively, and are shown to exist in the relevant regime via the coupling. Analogous equalities hold for the measures $\mu^{\eivRC;1,0}_{\ul{p}}, \mu^{\eiv;\bplus}_{K,W}$, and $\mu^{\AT;+,+}_{J,U}[\tau_x]$. To show that the 8vRC model has a unique Gibbs measure, it therefore suffices to show that $\mu^{\AT;+,+}_{J,U}[\tau_x] = \mu^{\AT;+,\f}_{J,U}[\tau_x] =0$. This will follow from a further coupling between the Ashkin-Teller model and the six-vertex model and Theorem \ref{thm:6v-deloc}.  \\

\subsection{Coupling eight-vertex, 8vRC and Ashkin-Teller}
We describe how the three models are coupled. First we describe how to sample the random cluster representation $\eta$ from the eight-vertex model, and $\om$ (satisfying $(\om,(\om')^*)=(\eta^\bullet,\eta^\circ)$) from the Ashkin-Teller model. 

Recall there are bijections between the sets $V,E^\bullet,E^\circ$. In the following we identify them via these bijections. For $\pi\in\Om^{\eiv}_V$, say $v\in\om[\pi^\bullet]$ if $\pi^\bullet_x\neq\pi^\bullet_{x'}$ where $x,x'\in F^\bullet$ are the black faces touching $v$. Say $v\in\om[\pi^\circ]$ similarly. Sample $\eta\in\Om^{\eivRC}_{V}$ at $v$ (corresponding to $e\in E^\bullet_\NW$ and $e^*\in E^\circ_\NE$) as:
\be\label{eq:sample-8vRC}
\begin{cases}
        \eta^\bullet_e=\eta^\circ_{e^*}=0 \ \mathrm{deterministically}  
        & \mathrm{if} \ v \in \om[\pi^\bullet]\cap\om[\pi^\circ]\\
        \eta^\bullet_e=0 \ \mathrm{deterministically}, 
        \ \eta^\circ_{e^*}=1 \ \mathrm{w.p.} \ \tfrac{p_{\NE}-p_\es}{p_{\NE}}     
        & \mathrm{if} \ v \in \om[\pi^\bullet]\setminus\om[\pi^\circ]\\
        \eta^\circ_{e^*}=0 \ \mathrm{deterministically}, 
        \ \eta^\bullet_e=1 \ \mathrm{w.p.} \ \tfrac{p_{\NW}-p_\es}{p_{\NW}} 
        & \mathrm{if} \ v \in \om[\pi^\circ]\setminus\om[\pi^\bullet]\\
        (\eta^\bullet_e,\eta^\circ_{e^*})=(0,0) \ \text{w.p.} \ p_\es, \
        (\eta^\bullet_e,\eta^\circ_{e^*})=(1,0) \ \text{w.p.} \ p_\NW-p_\es, \\
        \qquad (\eta^\bullet_e,\eta^\circ_{e^*})=(0,1) \ \text{w.p.} \ p_\NE-p_\es, \
        (\eta^\bullet_e,\eta^\circ_{e^*})=(1,1) \ \text{w.p.} \ 2p_\es 
        & \mathrm{if} \ v \notin \om[\pi^\bullet]\cup\om[\pi^\circ]\\
        \text{The same with $\NW$ and $\NE$ exchanged} 
        & \mathrm{if} \ e\in E^\bullet_\NE, \ e^*\in E^\circ_\NW.
\end{cases}
\ee

Now we describe how to sample $(\om,\om')$ from an Ashkin-Teller configuration $(\tau,\tau')$. For $e\in E^\bullet_\NW$, 
\be\label{eq:sample-ATRC}
\begin{cases}
        \om_e=\om'_e=0 \ \mathrm{deterministically}  
        & \mathrm{if} \ e^* \in \om[\tau]\cap\om[\tau']\\
        \om_e=0 \ \mathrm{deterministically}, 
        \ \om'_{e}=1 \ \mathrm{w.p.} \ \tfrac{2p_\es}{p_{\NE}+p_\es}     
        & \mathrm{if} \ e^* \in \om[\tau]\setminus\om[\tau']\\
        \om'_e=0 \ \mathrm{deterministically}, 
        \ \om_{e}=1 \ \mathrm{w.p.} \ \tfrac{2p_\es}{p_{\NE}+p_\es}     
        & \mathrm{if} \ e^*v \in \om[\tau']\setminus\om[\tau]\\
        (\om_e,\om'_{e})=(0,0) \ \text{w.p.} \ \tfrac{p_\NE-p_\es}{1+p_\NW}, \
        (\om_e,\om'_{e})=(1,0) \ \text{w.p.} \ \tfrac{2p_\es}{1+p_\NW}, \\
        \qquad (\om_e,\om'_{e})=(0,1) \ \text{w.p.} \ \tfrac{2p_\es}{1+p_\NW}, \
        (\om_e,\om'_{e})=(1,1) \ \text{w.p.} \ \tfrac{2(p_{\NW}-p_\es)}{1+p_\NW}
        & \mathrm{if} \ e^* \notin \om[\tau]\cup\om[\tau']\\
        \text{The same with $\NW$ and $\NE$ exchanged} 
        & \mathrm{if} \ e\in E^\bullet_\NE.
\end{cases}
\ee

\begin{lemma}\label{lem:8v-AT-coupling}
    Let $\ul{p}\in\cP$ with $p_{\NW},p_{\NE}\ge p_\es$, let $K,W$ satisfy \eqref{eq:8v-coupling-constants}, and let $J,U$ satisfy \eqref{eq:AT-coupling-constants}.
\begin{enumerate}
    \item If one takes $\eta\sim\mu^{\eivRC,0,0}_{V,\ul{p}}$ and samples $\pi\in\Om^{\eiv}_{V}$ by independently and uniformly assigning spins $\pm1$ to the clusters of $\eta^\bullet$ and of $\eta^\circ$, then $\pi\sim\mu^{\eiv}_{V,K,W}$.
    \item If one takes $\eta\sim\mu^{\eivRC,0,0}_{V,\ul{p}}$ and samples $(\tau,\tau')\in\Om^{\AT}_{G^\bullet}$ by assigning spins $\pm 1$ to each cluster of $(\eta^\bullet)$ and each cluster of $(\eta^\circ)^*$ independently and uniformly, with the exception that all clusters of $(\eta^\circ)^*$ touching $\partial V^\bullet$ are given the spin $+1$ deterministically, then $(\tau,\tau')\sim\mu^{\AT;\f,+}_{G^\bullet,J,U}$. 
    \item If one takes $\pi\sim\mu^{\eiv}_{V,K,W}$ and samples a percolation configuration $\eta\in\Om^{\eivRC}_{V}$ according to the rules \eqref{eq:sample-8vRC}, or if one takes  $(\tau,\tau')\sim\mu^{\AT;\f,+}_{G^\bullet,J,U}$ and samples a percolation configuration $\om\in\Om^{\ATRC}_{V}$ according to the rules \eqref{eq:sample-ATRC} (and sets $(\eta^\bullet,\eta^\circ)=(\om,(\om')^*)$), then $\eta\sim\mu^{\eivRC,0,0}_{V,\ul{p}}$.
    
    \item Similar couplings hold between the measures $\pi\sim\mu^{\eiv;\spadesuit}_{V,K,W}$, $\eta\sim\mu^{\eivRC;\clubsuit}_{V,\ul{p}}$ and $(\tau,\tau')\sim\mu^{\AT;\diamond}_{G^\bullet,J,U}$, where $\spadesuit=\bplus\wplus, \diamond=+,\f$ if $\clubsuit=1,1$, and $\spadesuit=\bplus, \diamond=+,+$ if $\clubsuit=1,0$, and $\spadesuit=\wplus, \diamond=\f,\f$ if $\clubsuit=0,1$. 
    \item Further, if the following infinite volume limits exist,
    \be\label{eq:8v-8vRC-connections-full}
    \begin{split}
	\mu^{\eivRC,1,1}_{\ul{p}}[x\xleftrightarrow{\eta^\bullet} \infty] 
	&=
	\mu^{\eiv,\bplus,\wplus}_{K,W}[\pi^\bullet_x]
	=
	\mu^{\AT,+,\f}_{J,U}[\tau_x]; \\
    \mu^{\eivRC,1,0}_{\ul{p}}[x\xleftrightarrow{\eta^\bullet} \infty] 
	&=
	\mu^{\eiv,\bplus}_{K,W}[\pi^\bullet_x]
	=
	\mu^{\AT,+,+}_{J,U}[\tau_x].
    \end{split}
    \ee
\end{enumerate}
\end{lemma}

The proof of the first three parts are either contained in or are straightforward adaptations of the proofs in \cite{pfister-velenik-AT}, so we leave them to the reader. We note that when the infinite volume limit of the 8vRC model with boundary conditions $\clubsuit\in\{0,1\}^2$ exists, it follows that the limit for the two spin measures exist; this is a standard application of a standard Burton-Keane argument \cite{burton-keane}. The converse also holds as the sampling of the random cluster model is local from each of the spin models. By the FKG inequalities of Lemma \ref{lem:8vRC-FKG}, the limit for $\clubsuit=(1,1)$ (resp. $\clubsuit=(1,0)$) exists for $p_\NE \cdot p_\NW\ge p_\es$ (resp. $p_\NE \cdot p_\NW\le p_\es$).
Finally, a standard colour-switching argument shows that \eqref{eq:8v-8vRC-connections-full} holds. \\

\subsection{Coupling Ashkin-Teller and six-vertex}
As noted above, the Ashkin-Teller model and the six-vertex model have a very similar coupling to that of the Ashkin-Teller model and the eight-vertex model. This was given in \cite[Proposition 8.1]{glaz-peled}. The coupling is facilitated by the second ATRC-type percolation configuration, this time where $(s,s')$ = $(\tau,\tau\tau')$.

Meanwhile one can also sample $\xi^\bullet\in\Om^{\perc}_{G^\bullet}$ from an Ashkin-Teller configuration $(\tau,\tau')\in\Om^{\AT}_{G^\bullet}$.
\be\label{eq:sampling-xi-from-AT}
\begin{cases}
	\xi^\bullet(e)=1 \ \mathrm{w.p.} \ \tfrac{2(p_\NE+p_\es)}{1+p_{\NE}}
	& \mathrm{if} \ e\in E^\bullet_\NE, \ e^*\in\om[\tau]\cap\om[\tau']\\
	\xi^\bullet(e)=1 \ \mathrm{w.p.} \ \tfrac{2(p_\NW+p_\es)}{1+p_{\NW}}
	& \mathrm{if} \ e\in E^\bullet_\NW, \ e^*\in\om[\tau]\cap\om[\tau']\\
	\xi^\bullet(e)=0 \ \mathrm{deterministically} 
	& \text{otherwise}\\
\end{cases}
\ee

\begin{proposition}\label{prop:6v-AT-coupling}
Let $\ul{p}\in\cP$ with $p_\NW,p_\NE\ge p_\es$, and let $a,b,c$ satisfy \eqref{eq:6v-weights} and $J=(J_\NW,J_\NE),U=(U_\NW,U_\NE)$ satisfy \eqref{eq:AT-coupling-constants}. For all $V$ finite and also for $V=\ZZ^2$, we have
\begin{enumerate}
    \item If one takes $\sigma\sim\mu^{\sv}_{V,a,b,c}$ and samples $\xi^\bullet$ according to the rule \eqref{eq:sampling-xi}, and then samples a spin configuration $\tau\in\{\pm1\}^{F^\bullet}$ by assigning $+$ or $-$ to the clusters of $\xi^\bullet$ independently and uniformly, then $(\tau, \tau\sigma^\bullet)\sim\mu_{G^\bullet,J,U}^{\AT,+,\f}$. That is, $\tau$ is distributed as a simple Ashkin-Teller spin, and $\sigma^\bullet$ is distributed as the product of Ashkin-Teller spins.
    \item If one takes $(\tau,\tau')\sim\mu_{G^\bullet,J,U}^{\AT,+,\f}$ and samples $\xi^\bullet$ according to the rule \eqref{eq:sampling-xi-from-AT}, then samples $\sigma^\circ\in\{\pm1\}^{F^\circ}$ by assigning $+$ or $-$ independently and uniformly to clusters of $(\xi^\bullet)^*$, then $(\tau\tau',\sigma^\circ)\sim\mu_{V,a,b,c}^{\sv}$. 
    \item One has the same relationship between the measures $\mu^{\sv,\bplus}_{V,a,b,c}$ and $\mu^{\AT;+,+}_{V,J,U}$.
    \item One has that in infinite volume,
    \bes
        \mu^{\AT;+,\f}_{J,U}[\tau_x] 
	=
	\mu^{\AT;++}_{J,U}[\tau_x] 
	=
	\mu^{\sv}_{a,b,c}[x \xleftrightarrow{\xi^\bullet} \infty]
	=
	0.
    \ees
\end{enumerate}
\end{proposition}

The proof is identical to that of \cite[Proposition 8.1]{glaz-peled}, so we omit it. To extend to infinite volume, one needs to use that $\xi^\bullet$ (and its dual) exhibits no infinite clusters almost surely, which is proved in the proofs of Theorem 3 and Proposition 4.11 of \cite{glazman-lammers} (the delocalisation in the six-vertex model). The equality in part 4 above follows by a standard colour-switching argument. 

Combining this with \eqref{eq:8v-8vRC-connections-full}, one concludes that the 8vRC model has a unique Gibbs measure, which we call $\mu^{\eivRC}_{\ul{p}}$. The stochastic domination \eqref{eq:8vRC-stoch-dom} shows that the remaining measures from $\mu^{\eivRC;i,j}_{V,\ul{p}}$ for $i,j\in\{0,1\}$ converge to $\mu^{\eivRC}_{\ul{p}}$. The symmetry between $\eta^\bullet$ and $\eta^\circ$ under $\mu^{\eivRC;1,1}_{V,\ul{p}}$ upgrades the $G^\bullet$ translation-invariance and ergodicity to $\ZZ^2$. Finally, by self-duality, if $\eta\sim\mu^{\eivRC}_{\ul{p}}$ then $\eta^*\sim\mu^{\eivRC}_{\ul{p}}$ too, so $\eta^*$ also has no infinite cluster in either component almost surely. This completes the proof of part 2 of Proposition \ref{prop:8v-AT}.\\

\subsection{Eight-vertex and Ashkin-Teller results from 8vRC results}

\begin{proof}[Proof of Proposition \ref{prop:8v-AT}]

We prove part 1, the eight-vertex model part (part 2 is essentially identical). We know that the weak limits of all four 8vRC measures exist and are equal by Proposition \ref{prop:8v-AT-RC} (write $\mu^{\eivRC}_{\ul{p}}$ for this common measure). Let $\mu^{\eiv}_{K,W}$ be the measure on eight-vertex spins given by assigning $\pm1$ uniformly and independently to every cluster of $\eta^\bullet$ and every cluster of $\eta^\circ$, where $\eta\sim\mu^{\eivRC}_{\ul{p}}$. By standard arguments (see \cite[Theorem 4.91]{grimmett-RC}), it follows that the weak limits of $\mu^{\eiv}_{V,K,W}$, $\mu^{\eiv;\bplus}_{V,K,W}$, $\mu^{\eiv;\wplus}_{V,K,W}$, $\mu^{\eiv;\bplus,\wplus}_{V,K,W}$ all exist and are equal to $\mu^{\eiv}_{K,W}$. The limit is straightforwardly shown to be Gibbs (since the interaction is local), $\ZZ^2$ translation-invariant (since $\mu^{\eivRC}_{\ul{p}}$ is), and since $\eta\sim\mu^{\eivRC}_{K,W}$ has at most one infinite cluster almost surely (indeed it has zero), $\mu^{\eiv}_{K,W}$ is ergodic; together these properties imply extremality (see e.g. \cite[Section 14]{georgii}). By the coupling and another standard colour-switching argument, we have
\be\label{eq:8v-8vRC-2pt-4pt}
\begin{split}
	\mu^\eiv_{K,W}[\pi^\bullet_{x_1}\pi^\bullet_{x_2}]
	&=
	\mu^{\eivRC}_{\ul{p}}[x_1\xleftrightarrow{\eta^\bullet} x_2];\\
	\mu^\eiv_{K,W}[\pi^\bullet_{x_1}\pi^\bullet_{x_2}
	\pi^\circ_{y_1}\pi^\circ_{y_2}]
	&=
	\mu^{\eivRC}_{\ul{p}}[x_1\xleftrightarrow{\eta^\bullet} x_2,
	y_1\xleftrightarrow{\eta^\circ} y_2],
\end{split}
\ee
and the right hand side of both of these equations tends to 0 as $||x_1-x_2||_1,||y_1-y_2||_1\to\infty$.\\

It remains to prove that for $\ul{p}$ satisfies $p_\NW\cdot p_\NE\ge p_\es$, then $\mu^{\eiv}_{K,W}$ is the unique Gibbs measure and the unique thermodynamic limit. This proof is based on that of Theorem 11.3 of \cite{grimmett-RC}, although requires us to be a little more technically careful. It relies crucially on the joint FKG with respect to $\le$ of $\eta^\bullet,\eta^\circ$.

Let $\pi'\in\Om^{\eiv}_{\ZZ^2}$ and let $\pi\sim\mu^{\eiv;\pi'}_{V,K,W}$. Let $B^{\bplus}_{\pi'}$ be the set of faces in $\partial F^\bullet$ such that $(\pi')^\bullet=1$, and similarly define $B^{\wplus}_{\pi'}, B^{\bminus}_{\pi'}$ and $B^{\wminus}_{\pi'}$. Sampling an edge configuration $\eta$ from $\pi$ on the edges $E$ in the usual way \eqref{eq:sample-8vRC}, one finds that $\eta$ is distributed according to the measure $\mu^{\eivRC;\eta'}_{V,\ul{p}}[\cdot \ | \ \c B]$, where $\eta'$ is any configuration such that outside $V$, all vertices of $B^{\bplus}_{\pi'}\cup B^{\bminus}_{\pi'}$ are connected in $\eta^\bullet$ and all vertices of $B^{\wplus}_{\pi'}\cup B^{\wminus}_{\pi'}$ are connected in $\eta^\circ$, and $\c B$ is the event that $B^{\bplus}_{\pi'}$ and $B^{\bminus}_{\pi'}$ are not connected in $\eta^\bullet$ inside $V$, and $B^{\wplus}_{\pi'}$ and $B^{\wminus}_{\pi'}$ are not connected in $\eta^\circ$ inside $V$.

Crucially, the event $\c B$ is a decreasing event with respect to the ordering $\le$, and so in the range $p_\NW\cdot p_\NE\ge p_\es$ we can use the joint FKG property of $\eta^\bullet$ and $\eta^\circ$ with respect to $\le$ to prove that $\mu^{\eivRC;\eta'}_{V,\ul{p}}[\cdot \ | \ \c B]$ is stochastically dominated by $\mu^{\eivRC;1,1}_{V,\ul{p}}$. Let $A$ be some cylinder event for the eight-vertex model, and $\Lambda\subset\ZZ^2$ some set containing the support of $A$ (think of $\Lambda$ as very large. We proved earlier that the limit of of the measure $\mu^{\eivRC;1,1}_{V,\ul{p}}$ displays no infinite clusters almost surely. It follows that as $k\to\infty$, with probability tending to 1, the measure $\mu^{\eivRC;\eta'}_{\Lambda_k,\ul{p}}[\cdot \ | \ \c B]$ has no connection between $\Lambda$ and $\partial\Lambda_k$ in either $\eta^\bullet$ or $\eta^\circ$.

Let $\Gamma^\bullet\subset \Lambda_k$ be the largest set containing $\Lambda$ such that $\eta^\bullet\equiv0$ on $\Delta_{\EE^\bullet}(\ZZ^2\setminus\Gamma^\bullet)$, which is the edges of $\EE^\bullet$ outside of $\Gamma^\bullet$ incident to an edge of $\Gamma^\bullet$. The event $\Gamma^\bullet=\gamma^\bullet$ is measurable with respect to $\eta^\bullet|_{\Lambda_k\setminus\gamma^\bullet}$. Let $\om^\circ$ be a configuration on $\EE^\circ\setminus\gamma^\bullet$ which agrees with $\eta'$ on $\EE^\circ\setminus\Lambda_k$. Let $B^{\wplus}_{\om^\circ}$ be the set of vertices on the boundary of $\gamma^\bullet$ which are connected to a vertex of $B^{\wplus}_{\pi'}$ by $\om^\circ$, and $B^{\wminus}_{\om^\circ}$ similar. 

One can show that conditioned on the event $\Gamma^\bullet=\gamma^\bullet$ and the event $\eta^\circ|_{\Lambda_k\setminus\gamma^\bullet}=\om^\circ|_{\Lambda_k\setminus\gamma^\bullet}$, the measure 
$\mu^{\eivRC;\eta'}_{\Lambda_k,\ul{p}}[\cdot \ | \ \c B]$
restricted to $\gamma^\bullet$ is 
$\mu^{\eivRC;\eta''}_{\gamma^\bullet,\ul{p}}[\cdot \ | \ \c B'']$,
where $(\eta'')^\bullet=0$, $(\eta'')^\circ|_{\EE\setminus\gamma^\bullet} = \om^\circ$, and where $\c B''$ is the event that $B^{\wplus}_{\om^\circ}$ and $B^{\wminus}_{\om^\circ}$ are not connected by $\eta^\circ$ inside $\gamma^\bullet$. \\

Consider now this new measure $\mu^{\eivRC;\eta''}_{\gamma^\bullet,\ul{p}}[\cdot \ | \ \c B'']$. We've managed to remove the conditioning on $\eta^\bullet$; we now repeat the process with $\eta^\circ$ to remove the conditioning on the event $\c B''$. Let $\Gamma^\circ\subset\gamma^\circ$ be the largest set containing $\Lambda$ such that $\eta^\circ\equiv0$ on the edges of $\EE^\bullet$ outside of $\Gamma^\circ$ which are incident to an edge of $\Gamma^\circ$. The event $\Gamma^\circ=\gamma^\circ$ is measurable with respect to $\eta^\circ|_{\gamma^\bullet\setminus\gamma^\circ}$. Let $\om^\bullet$ be some configuration on $\gamma^\bullet\setminus\gamma^\circ$ which is zero outside of $\gamma^\bullet$. 

One can show that conditioned on the event $\Gamma^\circ=\gamma^\circ$ and the event $\eta^\bullet|_{\gamma^\bullet\setminus\gamma^\circ}=\om^\bullet$, the measure $\mu^{\eivRC;\eta''}_{\gamma^\bullet,\ul{p}}[\cdot \ | \ \c B'']$ restricted to $\gamma^\circ$ is the measure $\mu^{\eivRC;\eta'''}_{\gamma^\circ,\ul{p}}$, where $(\eta''')^\bullet=\om^\bullet$ and $(\eta''')^\circ=0$. \\

Now since $\Lambda$ was arbitrary, we have that conditioned on the event $\Lambda\not\leftrightarrow\partial\Lambda_k$,
$\mu^{\eivRC;\eta'}_{\Lambda_k,\ul{p}}[A \ | \ \c B]$,  is some convex combination of $\mu^{\eivRC;\eta'''}_{\gamma^\circ,\ul{p}}[A]$ over configurations $\eta'''$ and domains $\gamma^\circ\subset\Lambda_k$ containing $\Lambda$. By comparison of boundary conditions, this measure restricted to $\Lambda$ stochastically dominates $\mu^{\eivRC;0,0}_{\Lambda,\ul{p}}$. We've already noted that $\mu^{\eivRC;\eta'}_{\Lambda_k,\ul{p}}[A \ | \ \c B]$ is stochastically dominated by $\mu^{\eivRC;1,1}_{\Lambda_k,\ul{p}}$, whose restriction to $\Lambda$ is dominated by $\mu^{\eivRC;1,1}_{\Lambda,\ul{p}}$. Now by part 2 of Proposition \ref{prop:8v-AT}, and since $\Lambda\not\leftrightarrow\partial\Lambda_k$ has probability arbitrarily close to 1 for large enough $k$, one has that $\mu^{\eivRC;\eta'}_{\Lambda_k,\ul{p}}[A \ | \ \c B]$ converges to $\mu^{\eivRC}_{\ul{p}}$. By the same arguments as used earlier in the paper, it follows that our original measure on spins, $\mu^{\eiv;\pi'}_{V,K,W}$, converges to $\mu^{\eiv}_{K,W}$, which completes the proof of part 1 of Proposition \ref{prop:8v-AT}.\\

The proof of part 2 of Proposition \ref{prop:8v-AT} (the Ashkin-Teller part) follows the blueprint of part 1; we leave the details to the reader.
    
\end{proof}



\section{Polynomial decay of correlations}\label{sec:symmetric-case}
In this section we finish the proof of Theorem \ref{thm:main-mirror}.

\begin{proposition}\label{prop:poly-decay}We have the following. 
\begin{enumerate}
    \item 
        If $p_\NW=p_\NE\ge p_\es$, there exist $\alpha_1,\alpha_2,C_1,C_2>0$ such that for all edges $e_1,e_2$ of $\ZZ^2$ with $||e_1-e_2||=n$, we have 
        \be\label{eq:mir-poly-decay-repeat}
        \begin{split}
        	\mu_{\ul{p}}^\mir[e_1\leftrightarrow e_2]
        	&\le
        	C_2n^{-\alpha_2},\\
        	C_1n^{-\alpha_1}
        	\le
            \lim_{V\nearrow\ZZ^2}
            \mu_{V,\ul{p}}^{\mir}
            [e_1\leftrightarrow e_2]
            &\le
        	C_2n^{-\alpha_2}.
        \end{split}
        \ee
    \item 
        In the case $p_\NW\cdot p_\NE= p_\es$, we also have \eqref{eq:mir-poly-decay-repeat} for all pairs of edges, and moreover we have $\alpha_1=\alpha_2=\tfrac{1}{2}$. 
    \item 
        If $p_\NW\cdot p_\NE\le p_\es$, then the lower bound $C_1n^{-\alpha_1} \le 
            \lim_{V\nearrow\ZZ^2}
            \mu_{V,\ul{p}}^{\mir}
            [e_1\leftrightarrow e_2]$ holds for all pairs of edges.
    \item If $\Delta(a,b,c)=\Delta(\ul{p})\in[-1,-\tfrac{1}{2}]$, then for $e_1,e_2$ vertical edges on the same horizontal row at distance $n$ from each other,
    \be\label{eq:mir-exact-decay-hugo}
        \begin{split}
            \lim_{V\nearrow\ZZ^2}
            \left|\mu^{\mir}_{V,\ul{p}}[e_1 \xleftrightarrow{+} e_2]
	           -
	           \mu^{\mir}_{V,\ul{p}}[e_1 \xleftrightarrow{-} e_2] 
               \right|
            &=
            \mu^{\sv}_{a,b,c}[(h(u_1)-h(u_1'))(h(u_2)-h(u_2')]\\
            &=
        	-\frac{1}{\pi\cdot\arccos(\Delta)}
            n^{-2}
            + 
            o(n^{-2}).
        \end{split}
        \ee
\end{enumerate}
\end{proposition}
Recall $u_i$ (resp. $u_i'$) is the face adjacent to and to the left of (resp. right of) $e_i$.
Combining this proposition with Lemma \ref{lem:spin-spins-are-mirror-connections}, and noting that $\ul{p}\in\cP$ and $p_\NW=p_\NE\ge p_\es$ implies $\Delta=\Delta(\ul{p})\in[-1,\tfrac{1}{2}]$, one proves the polynomial decay of the spin-spin correlations in the XXZ chain: for $\Delta\in[-1,\frac{1}{2}]$,
\be\label{eq:xxz-poly-decay-repeat}
    C_1|x|^{-\alpha_1}
    \le
    \langle S^{(1)}_0 S^{(1)}_x\rangle 
    \le
    C_2|x|^{-\alpha_2},
\ee
as well as proving Part 1 of Theorem \ref{thm:xxz-decay-rates} as well as Part 2 of Theorem \ref{thm:xxz-decay-rates} in the case $t=0$.\\

We prove the upper bound of \eqref{eq:mir-poly-decay-repeat} by proving a strong RSW property in the 8vRC model, and we prove the lower bound of \eqref{eq:mir-poly-decay-repeat} by using the existing strong RSW property in the six-vertex model.

\subsection{Strong RSW in 8vRC}
In this subsection, every square $\Lambda_n$ and every rectangle is in the lattice $\GG^\bullet$, identifying it with $\ZZ^2$.
 
\begin{proposition}\label{prop:strong-RSW}
    Let $\ul{p}\in\c P$ with $p_\NE=p_\NW\ge p_\es$ and let $K,W$ satisfy \eqref{eq:8v-coupling-constants}. Let $\rho>0$. There exists $c=c(\rho)$ such that for all $n\in\NN$, and all $\eta'\in\Om^{\eivRC}_{\ZZ^2}$,
    \bes
        c< \mu^{\eivRC;\eta'}_{R,\ul{p}}[H(n,\rho n)] < 1-c,
    \ees
    where $H(n,\rho n)$ is the event that $\eta^\bullet$ connects the left side of the rectangle $[0,\rho n]\times[0,n]$ to the right side via a crossing inside that rectangle, and $R$ is the larger rectangle $[-n,(\rho+1)n]\times[-n,2n]$.
\end{proposition}

The proof follows the dichotomy argument of \cite{dc-sid-tass-cont} and \cite{hugo-notes} and is very similar to its counterpart for the critical random cluster model on $\ZZ^2$ with $q\in[1,4]$. We sketch it here. 

For shorthand we will write $\ul{0}$ (resp. $\ul{1}$) for the boundary conditions $0,0$ (resp. $1,1$) in the regime $p_\NE\cdotp_\NW\ge p_\es$, and $0,1$ (resp. $1,0$). These are the minimal (resp. maximal) boundary conditions with respect to the ordering $\le$ or $\tilde{\le}$ for the relevant regime. We follow as much as possible the notation of \cite{hugo-notes}.

To help with following the proof, we break it into steps. We start with the fact that the measures $\mu^{\eivRC;i,j}_{\ul{p}}$ are all equal for $i,j\in\{0,1\}$. 
\begin{enumerate}
	\item[\bf{S1}] The escape probability in $\mu^{\eivRC}_{\ul{p}}$ decays slowly: for any black face $x$,
	\bes
		\lim_{k\to\infty}\tfrac{1}{k^{1/3}} \log 
		\mu^{\eivRC}_{\ul{p}}[x\xleftrightarrow{\eta^\bullet} \partial\Lambda_k(x)]=0.
	\ees
	\item[\bf{S2}] We have the box-crossing property in the infinite volume measure: for any $\rho>0$, there exists $c=c(\rho)$ such that for all $n\in\NN$,
	\bes
		c< \mu^{\eivRC}_{\ul{p}}[H(n,\rho n)] < 1-c.
	\ees
	\item[\bf{S3}] We have the box-crossing property in a strip: let $\SS=\ZZ\times[-n,2n]$ and let $\eta'$ denote boundary conditions $\ul{0}$ on the bottom side $\ZZ\times\{-n\}$ and $\ul{1}$ on the top side $\ZZ\times\{2n\}$. Then
	\bes
	\mu^{\eivRC;\eta'}_{\SS,\ul{p}}[H(n,\rho n)] >c.
	\ees
	\item[\bf{S4}] We have the box-crossing property in a closer domain (this is also known as the Pushing Lemma): Let $D_{\rho,l}=[0,\rho n]\times[-n,ln]$ and let $\eta'$ denote boundary conditions $\ul{0}$ on the bottom side of $D_{\rho,l}$ and $\ul{1}$ on the other three sides. For all $\rho>0$, $l\ge2$, there exists $c=c(\rho,l)$ such that for all $n\in\NN$,
	such that 
	\bes
	\mu^{\eivRC;\eta'}_{D_{\rho,l},\ul{p}}[H(n,\rho n)] >c.
	\ees
	\item[\bf{S5}] We have a renormalisation inequality. Let $\c A_n$ be the event that there exists a circuit of open edges in $\eta^\bullet$ in $\Lambda_{2n}\setminus\Lambda_n$ surrounding the origin, and set $u_n = \mu^{\eivRC;\ul{0}}_{\Lambda_{8n},\ul{p}}$. Then there exists $0<C<\infty$ such that 
	\bes
		u_{7n} \le C u_n^2.
	\ees
	\item[\bf{S6}] The quantity $u_n$ is either bounded away from 0 uniformly in $n$, or $u_{7^kn}$ decays stretched-exponentially fast in $k$. If the former holds then Proposition \ref{prop:strong-RSW} holds, and if the latter holds, then $\mu^{\eivRC}_{\ul{p}}[x_0\leftrightarrow \partial\Lambda_n]$ decays stretched-exponentially fast in $n$.
\end{enumerate}

	The statement $\bf{S1}$ is straightforward to prove: by self-duality (and crucially the symmetry $p_\NE=p_\NW$), $\mu^{\eivRC;\ul{0}}_{\ul{p}}[H(n,n+1)] + \mu^{\eivRC;\ul{1}}_{\ul{p}}[H(n,n+1)]=1$, and then by stochastic domination, 
	\bes
		\mu^{\eivRC;\ul{1}}_{\ul{p}}[H(n,n)]
		\ge 
		\mu^{\eivRC;\ul{1}}_{\ul{p}}[H(n,n+1)]
		\ge 
		\frac{1}{2}
		\ge 
		\mu^{\eivRC;\ul{0}}_{\ul{p}}[H(n,n+1)],
	\ees
	and further
	\bes
		n \mu^{\eivRC;\ul{0}}_{\ul{p}}[0\leftrightarrow \partial\Lambda_n]
		\ge 
		\mu^{\eivRC;\ul{0}}_{\ul{p}}[H(n,n+1)]
		=
		\mu^{\eivRC;\ul{1}}_{\ul{p}}[H(n,n+1)]
		=\frac{1}{2},
	\ees
	where in the first equality we use $\mu^{\eivRC;\ul{0}}_{\ul{p}} =\mu^{\eivRC;\ul{1}}_{\ul{p}}$ from Proposition \ref{prop:8v-AT}. The statement $\bf{S1}$ follows readily.\\
	
	The statement $\bf{S2}$ follows from the general proof of Köhler-Schindler and Tassion \cite{koehler-schindler-tassion}, using the FKG property of $\eta^\bullet$ and, crucially, the symmetry $p_\NE=p_\NW$.\\
	
	The proofs of $\bf{S2}\implies\bf{S3}\implies\bf{S4}\implies\bf{S5}\implies\bf{S6}$ are identical to those for critical FK percolation \cite{dc-sid-tass-cont}, \cite{hugo-notes}; they rely on the FKG inequality, the comparison of boundary conditions, the domain Markov property, and the symmetry $p_\NE=p_\NW$. \\

\subsection{Polynomial decay}
We prove the upper bound in \eqref{eq:mir-poly-decay-repeat}. Using Proposition \ref{prop:strong-RSW}, and in the same way as with FK percolation, one proves that
\bes
	\mu^{\eivRC}_{\ul{p}}[0\xleftrightarrow{\eta^\bullet} \partial\Lambda_n] \le n^{-\alpha}.
\ees
Indeed, Proposition \ref{prop:strong-RSW} and shows that $\mu^{\eivRC;\ul{1}}_{\Lambda_{4n}\setminus\Lambda_n}[\c H([2n,3n]\times[-3n,3n])]>c$; applying the FKG inequality to the 4 rotated (by multiples of $\tfrac{\pi}{2}$) versions of this event, one finds that 
\bes
\mu^{\eivRC;\ul{1}}_{\Lambda_{4n}\setminus\Lambda_n}
[\partial\Lambda_n\xleftrightarrow{\eta^\bullet}\partial\Lambda_{4n}]
\le 
1-
\mu^{\eivRC;\ul{1}}_{\Lambda_{4n}\setminus\Lambda_n}[\c H([2n,3n]\times[-3n,3n])]^4
\le 
1-c^4.
\ees
Then many applications of the domain markov property and comparison of boundary conditions gives:
\bes
\begin{split}
	\mu^{\eivRC}_{\ul{p}}[0\xleftrightarrow{\eta^\bullet} \partial\Lambda_n]
	\le 
	\mu^{\eivRC;\ul{1}}_{\Lambda_n,\ul{p}}[0\xleftrightarrow{\eta^\bullet} \partial\Lambda_n]
	&\le 
	\prod_{4^k\le n} \mu^{\eivRC;\ul{1}}_{\Lambda_{4^k}\setminus\Lambda_{4^{k-1}},\ul{p}}
	[\partial\Lambda_{4^{k-1}}\xleftrightarrow{\eta^\bullet} \partial\Lambda_{4^k}]\\
	&\le 
	(1-c^4)^{\lfloor \log_4 n\rfloor} \le n^{-\alpha}.
\end{split}
\ees
From this it follows that
\bes
\begin{split}
	\lim_{V\nearrow\ZZ^2}\mu^{\mir}_{V,\ul{p}}[e_1\leftrightarrow e_2]
	=
	\mu^\eiv_{K,W}[\pi^\bullet_{x_1}\pi^\bullet_{x_2}
	\pi^\circ_{y_1}\pi^\circ_{y_2}]
	&=
	\mu^{\eivRC}_{\ul{p}}[x_1\xleftrightarrow{\eta^\bullet} x_2,
	y_1\xleftrightarrow{\eta^\circ} y_2]\\
	&\le
	\mu^{\eivRC}_{\ul{p}}[x_1\xleftrightarrow{\eta^\bullet} x_2]\\
	&\le
	\mu^{\eivRC}_{\ul{p}}[0\leftrightarrow \partial\Lambda_{|x_1-x_2|}],
\end{split}
\ees
and we have the upper bound of \eqref{eq:mir-poly-decay-repeat}.\\

Obtaining lower bound in the regime $p_\NE\cdot p_\NW\le p_\es$ is also fairly straightforward after one notices that from the FKG inequality one has
\bes
\mu^{\eivRC}_{\ul{p}}[x_1\xleftrightarrow{\eta^\bullet} x_2,
y_1\xleftrightarrow{\eta^\circ} y_2]\\
\ge
\mu^{\eivRC}_{\ul{p}}[x_1\xleftrightarrow{\eta^\bullet} x_2]
\mu^{\eivRC}_{\ul{p}}[y_1\xleftrightarrow{\eta^\circ} y_2].
\ees
However since this inequality does not hold for the range $p_\NE\cdot p_\NW\ge p_\es$, we turn to the following proof which makes full use of the strong RSW property of Proposition \ref{prop:strong-RSW}, and works for the whole range $p_\NE,p_\NW\ge p_\es$. The idea is to estimate the probability of a crossing in $\eta^\bullet$ from $x_1$ to $x_2$ and a crossing in $\eta^\circ$ from $y_1$ to $y_2$ which stay away from one another. Let's make this precise.

For ease of notation, we work with the case where $e_1$ is the horizontal edge with left endpoint the origin (with $x_1$ above and $y_1$ below), and $e_2$ is the shift of $e_2$ by $(2^K,0)$. For other cases an analogous proof holds. For the rest of this section, the sets $\Lambda_n$ and all rectangles are in the original lattice $\ZZ^2$. 

Let $T_k=\Lambda_{2^k}\setminus\Lambda_{2^{k-1}}\cap\{[0,\infty)\times[2^{k-2},\infty)\}$. Let $A_k$ be the event that $2^{k-1}\times[2^{k-2},2^{k-1}]$ is connected in $\eta^\bullet$ within $A_k\cup A_{k+1}$ to $2^{k+1}\times[2^{k},2^{k+1}]$. Let $B_k$ be the event that the top and bottom of the rectangle $[2^{k-1},2^k]\times[2^{k-2},2^{k+2}]$ are connected in $\eta^\bullet$. 

By Proposition \ref{prop:strong-RSW}, we have that there exists $c>0$ such that for all $k\in\NN$ and all $\eta'\in\Om^{\eivRC}_{\ZZ^2}$, 
\bes
	\mu^{\eivRC}_{\ul{p}}[A_k] \ge c; \qquad
	\mu^{\eivRC}_{\ul{p}}[B_k] \ge c.
\ees
Let $A_k'$ (resp. $B_k'$) be $A_k$ (resp. $B_k$) reflected in the line $x=2^{K-1}$. Write $\c A =\bigcap_{k\le K-1} A_k$, $\cB = \bigcap_{k\le K} B_k$, and $\c A'$, $\c B'$ their reflections in the line $x=2^{K-1}$. Let $C$ be the event $H([2^{K-2},3\cdot 2^{K-2}]\times[2^{K-3},2^{K-2}])$. Observe that on the event $\cA \cap \c A' \cap \c B \cap \c B' \cap C$, the faces $x_1$ and $x_2$ are connected in $\eta^\bullet$. Further, by the FKG inequality and Proposition \ref{prop:strong-RSW}, there exists $c>0$ such that uniformly in $K$,
\bes
	\mu^{\eivRC}_{\ul{p}}[\cA \cap \c A' \cap \c B \cap \c B' \cap C] \ge c^K.
\ees

For an event $D$, write $D^*$ for its reflection in the $x$ axis. By Proposition \ref{prop:strong-RSW}, we have that conditional on $\cA \cap \c A' \cap \c B \cap \c B' \cap C$, each of the events $\cA^*,  (\c A')^*,\c B^*, (\c B')^*, C^*$ all still have probability at least $c$. One therefore finds that 
\bes
\mu^{\eivRC}_{\ul{p}}[\cA \cap \c A' \cap \c B \cap \c B' \cap C \cap (\cA \cap \c A' \cap \c B \cap \c B' \cap C)^*] \ge (c^2)^K.
\ees
On this event $x_1$ and $x_2$ are connected in $\eta^\bullet$ and $y_1$ and $y_2$ are connected in $\eta^\circ$, and so
\bes
	\lim_{V\nearrow\ZZ^2}\mu^{\mir}_{V,\ul{p}}[e_1\leftrightarrow e_2]
	=
	\mu^\eiv_{K,W}[\pi^\bullet_{x_1}\pi^\bullet_{x_2}
	\pi^\circ_{y_1}\pi^\circ_{y_2}]
	=
	\mu^{\eivRC}_{\ul{p}}[x_1\xleftrightarrow{\eta^\bullet} x_2,
	y_1\xleftrightarrow{\eta^\circ} y_2] \ge (c^2)^K \ge n^{-\alpha_1}.
\ees
We can replace $\mu^{\mir}_{V,\ul{p}}$ in the first expression by the mirror measure with any boundary conditions for which the convergence of part 1 of Theorem \ref{thm:main-mirror} holds.
This completes the proof of \eqref{eq:mir-poly-decay-repeat}. \\

\subsection{Exact decay rates}
In this section we prove parts $2,3$, and 4 of Proposition \ref{prop:poly-decay}.

\begin{proof}[Proof of Part 2 of Proposition \ref{prop:poly-decay}]
    In the case $p_\NW\cdot p_\NE= p_\es$ observe that if $\pi=(\pi^\bullet,\pi^\circ)\sim\mu^{\eiv}_{K,W}$ then $\pi^\bullet$ and $\pi^\circ$ are independent and both distributed as critical Ising models on the square lattice. Then
\bes
    \lim_{V\nearrow\ZZ^2}\mu^{\mir}_{V,\ul{p}}[e_1\leftrightarrow e_2]
	=
	\mu^\eiv_{K,W}[\pi^\bullet_{x_1}\pi^\bullet_{x_2}
	\pi^\circ_{y_1}\pi^\circ_{y_2}]
    =
    \mu^{\mathrm{Ising}}_{\ZZ^2,\beta_c}[\pi^\bullet_{x_1}\pi^\bullet_{x_2}]\cdot
	\mu^{\mathrm{Ising}}_{\ZZ^2,\beta_c}[\pi^\circ_{y_1}\pi^\circ_{y_2}]
    \sim
    (||e_1-e_2||^{-1/4})^2.
\ees
The decay here classical, see \cite{mccoy-wu} or the more recent \cite{chel-hong-izy}.
\end{proof}

\begin{proof}[Proof of Part 3 of Proposition \ref{prop:poly-decay}]
    Our work in Lemma \ref{lem:spin-spins-are-mirror-connections} shows that
    for $e_1$, $e_2$ vertical edges on the same horizontal row (or horizontal edges on the same vertical column), we have
    \bes
        \lim_{V\nearrow\ZZ^2}\mu^{\mir}_{V,\ul{p}}[e_1\leftrightarrow e_2]
    	=
    	\mu^{\eivRC}_{\ul{p}}[y_1\xleftrightarrow{\eta^\bullet} y_2,
	       z_1\xleftrightarrow{\eta^\circ} z_2]
        =
        \mu^{\eivRC}_{\ul{p'}}[y_1\xleftrightarrow{\eta^\bullet} y_2,
	       z_1\xleftrightarrow{\eta^\circ} z_2]
    	=
    	\mu^\eiv_{K',W'}[\pi^\bullet_{x_1}\pi^\bullet_{x_2}
    	\pi^\circ_{y_1}\pi^\circ_{y_2}],
    \ees
    for all $\ul{p}$, $\ul{p'}\in\cP$ such that $\Delta(\ul{p})=\Delta(\ul{p'})$. By Part 1 of Proposition \ref{prop:poly-decay}, we have the lower bound $C_1n^{-\alpha_1} \le 
            \lim_{V\nearrow\ZZ^2}
            \mu_{V,\ul{p}}^{\mir}
            [e_1\leftrightarrow e_2]$ in the case $e_1$, $e_2$ vertical edges on the same horizontal row (or horizontal edges on the same vertical column), for $p_\NW=p_\NE\ge p_\es$, so we have it for all $\ul{p}\in\cP$ with $\Delta\in[0,\tfrac{1}{2}]$.

    By using finite energy, the lower bound for all pairs $e_1$, $e_2$ follows from
    \bes
    \begin{split}
        \lim_{V\nearrow\ZZ^2}\mu^{\mir}_{V,\ul{p}}[e_1\leftrightarrow e_3]
    	&=
    	\mu^{\eivRC}_{\ul{p}}[y_1\xleftrightarrow{\eta^\bullet} y_3,
	       z_1\xleftrightarrow{\eta^\circ} z_3]\\
        &\ge 
        \mu^{\eivRC}_{\ul{p}}[y_1\xleftrightarrow{\eta^\bullet} y_2 \xleftrightarrow{\eta^\bullet} y_3,
	       z_1\xleftrightarrow{\eta^\circ} z_2 \xleftrightarrow{\eta^\circ} z_3]\\
        &\ge 
        \mu^{\eivRC}_{\ul{p}}[y_1\xleftrightarrow{\eta^\bullet} y_2,
	       z_1\xleftrightarrow{\eta^\circ} z_2]
        \mu^{\eivRC}_{\ul{p}}[y_2\xleftrightarrow{\eta^\bullet} y_3,
	       z_2\xleftrightarrow{\eta^\circ} z_3]\\
    	&=
    	\lim_{V\nearrow\ZZ^2}\mu^{\mir}_{V,\ul{p}}[e_1\leftrightarrow e_2]
        \lim_{V\nearrow\ZZ^2}\mu^{\mir}_{V,\ul{p}}[e_2\leftrightarrow e_3],
    \end{split}
    \ees
    where in the second inequality we use the FKG inequality for $\eta$, Lemma \ref{lem:8vRC-FKG}, which holds for $p_\NW,p_\NE\ge p_\es$ and $p_\NW\cdot p_\NE\le p_\es$. 
\end{proof}

\begin{proof}[Proof of Part 4 of Proposition \ref{prop:poly-decay}]
    We claim that the decay is a consequence of the convergence of the height function of the six-vertex model with $\Delta\in[-1,-\tfrac{1}{2}]$ to the Gaussian free field, recently proved in \cite{dc-etal-GFF-convergence}. 

    Recall $e_1,e_2$ are vertical edges on the same horizontal row at distance $n$ from each other, and $u_i$ is (resp. $u_i'$) is the face adjacent to and to the left of (resp. right of) $e_i$.
    Theorem 4.12 and Lemma 6.6 of \cite{dc-etal-GFF-convergence} together prove that 
    \bes
        \mu^{\sv}_{a,b,c}[(h(u_0)-h(u_0'))(h(u_x)-h(u_x')]
        =
        \int -a^2 (1-a)^n \d \mu_\infty(a,b),
    \ees
    where $\mu_\infty$ is a specific finite positive measure on $\RR_{>0}\times\RR$, supported on $(0,2]\times[-\pi,\pi]$. Writing $\delta=1/n$, one has that the above is equal to
    \bes
        \int -(a/\delta)^2 (1-(a/\delta))^{1/\delta} \d \mu_\infty^{(\delta)}(a,b),
    \ees
    where $\mu_\infty^{(\delta)}(U):=\mu_\infty(\delta U)$ for all $U\subset\RR_{>0}\times\RR$ measurable. Then Theorem 9.1 of \cite{dc-etal-GFF-convergence} shows that $\mu_\infty^{(\delta)}\to\mu$ in the space of measures $\cM$ (see Definition 6.3 of \cite{dc-etal-GFF-convergence}), where $\mu=\tfrac{1}{2}(\delta_{b=a}+\delta_{b=-a})\frac{\sigma^2}{2\pi a} \d a$. Here $\sigma^2 = \frac{2}{\arccos(\Delta)}$. This, together with an identical proof to that of Lemma 6.8 of \cite{dc-etal-GFF-convergence} gives (taking out a factor of $\delta^{-2}=n^2$):
    \bes
        \int -a^2 (1-(a/\delta))^{1/\delta} \d \mu_\infty^{(\delta)}(a,b)
        \to
        \int -a^2 e^{-a} \d \mu 
        =
        -\frac{\sigma^2}{2\pi} \int_0^\infty ae^{-a} \d a
        =
        -\frac{\sigma^2}{2\pi}
    \ees
    as $\delta\to0$.
    This completes the proof of Proposition \ref{prop:poly-decay} and also the proof of Part 1 of Theorem \ref{thm:xxz-decay-rates}.
\end{proof}


\part{Space-time models}\label{part:space-time}
\section{The space-time six-vertex height function delocalises}\label{sec:delocalisation}

\subsection{Definitions of space-time six-vertex models}
In this section we define the further representations of the space-time six vertex model, in terms of arrow configurations and Ising spins, as well as the percolation configurations $\xi^\bullet$, $\xi^\circ$. To keep track of notation, all of the sets of configurations for space-time models will be denoted $\ul{\Om}$ instead of $\Om$, and the measures denoted $\nu$ instead of $\mu$.\\

We start with arrow configurations. Let $F=(V,F)$ be a domain. 
Define $E=(\overline{V}+(\tfrac{1}{2},0))\cup (\overline{V}-(\tfrac{1}{2},0))$. Define $\partial_\horiz E$ to be those points $x$ of $\partial E$ with one of $x\pm(\tfrac{1}{2},0)$ not in $V$ (note $\partial_\horiz E$ is a union of intervals). Define $\partial_\vert E$ to be $\partial E \setminus (V+(\tfrac{1}{2},0))\cup (V-(\tfrac{1}{2},0))$ (note $\partial_\vert E$ is a finite set of discrete points). Let $\partial E = \partial_\horiz E \cup \partial_\vert E$ (this is the continuous analogue of the set of edges with exactly one endpoint in $V$).

Let $\ul{\Om}^{\sv,\rightarrow}_{F}$ be the set of functions $\kappa:E\to\{\uparrow,\downarrow\}$ such that the restriction to any interval is piecewise constant and right continuous, and the points of discontinuity of $\kappa$ can be partitioned into pairs at adjacent points $v\pm(\tfrac{1}{2},0)$, $v\in V$, and $\kappa$ always takes different values at the points $v\pm(\tfrac{1}{2},0)$. We also assign an orientation to each point of discontinuity of $\kappa$; namely we say $v\in\kappa_{\rightarrow}$ if $\kappa_{v-(\tfrac{1}{2},0)}=\downarrow$ and $\kappa_{v+(\tfrac{1}{2},0)}=\uparrow$, and vice-versa for $v\in\kappa_{\leftarrow}$.

Note that the condition that $\kappa$ always takes different values at the points $v\pm(\tfrac{1}{2},0)$, together with right continuity, forms the analogue of the ice rule in this setting. Define the associated sigma-algebra $\FF^{\sv,\rightarrow}_{F}$ analogously to the height function case, using the $w^\#$ topology \cite[Section A2.6]{daley-verejones-1}. 

For $\kappa\in\ul{\Om}^{\sv,\rightarrow}_{F}$, let $\om^{\uparrow\downarrow}[\kappa]$ be the set of $v\in V$ such that $\kappa$ takes differing values on $v\pm(\tfrac{1}{2},0)$, and let $\om^{\uparrow\uparrow}[\kappa]$ be those $v\in V$ such that it takes the same value. Define the six-vertex orientation measure on $\ul{\Om}^{\sv,\rightarrow}_{F}$ as $\nu^{\sv,\rightarrow}_{F,\Delta}$ defined by
\bes
    \d \nu^{\sv,\rightarrow}_{F,\Delta}
    \propto
    \d \nu^{\poi}_{V,1}(\chi)
    \mathbbm{1}_{\{\kappa\sim\chi\}}
    \exp\left[u\left|\om^{\uparrow\uparrow}[\kappa]\right| + (1-u)\left|\om^{\uparrow\downarrow}[\kappa]\right|\right]
    ,
\ees
where $\chi$ is a Poisson point process of rate 1 and where $\kappa\sim\chi$ if and only if the points where $\kappa$ changes value is exactly the points $v\pm(\tfrac{1}{2},0)$ for $v\in\chi$.

The points of $\chi$ here are analogous to the $a$-type vertices in the discrete six-vertex model. The points $\om^{\uparrow\uparrow}[\kappa]$ are the $b$-type vertices and the points $\om^{\uparrow\downarrow}[\kappa]$ are the $c$-type vertices. 

Finally for $\kappa'\in\ul{\Om}^{\sv,\rightarrow}_{\cF}$, we define 
\bes
    \d \nu^{\sv,\rightarrow}_{F,\Delta}
    \propto
    \d \nu^{\poi}_{V,1}(\chi)
    \mathbbm{1}_{\{\kappa\sim\chi\}}
    \exp\left[u\left|\om^{\uparrow\uparrow}[\kappa]\right| + (1-u)\left|\om^{\uparrow\downarrow}[\kappa]\right|\right]
    \mathbbm{1}_{\{\kappa = \kappa'\ \text{on} \ \partial E\}}
\ees
to be the measure with boundary conditions $\kappa'$.\\

We turn to the spin representation. Recall that the discrete six-vertex model can be written as two Ising models, one on the black faces of $\ZZ^2$ and one on the white faces, such that their domain walls do not intersect. We do the same here, replacing classical Ising models with (the space-time representation of) quantum Ising models. 
The quantum Ising model is a quantum spin system, and when defined on a given graph $G$, has a probabilistic representation in terms of space-time configurations of classical spins, which can be seen as a scaling limit in one direction of a classical Ising model. We call this the space-time Ising model.

Let $F=(V,F)$ be a domain. Recall $\ul{\Om}^{\ising}_{F}$ is the set of functions  $\sigma:F\to\{\pm1\}$ whose restriction to any interval in $\{x\}\times\RR$ is right continuous and has finitely many points of discontinuity in any compact interval. 
We endow this space with the $w^\#$ topology \cite[Section A2.6]{daley-verejones-1} and write $\FF^{\ising}_{F}$ for the Borel sigma algebra generated by this topology.


For a configuration $\sigma\in\ul{\Om}^{\ising}_{F}$ recall $\om[\sigma]=\om_{\vert}[\sigma]\cup\om_{\horiz}[\sigma]$ is the set of domain walls of $\sigma$. That is, $\om_{\vert}[\sigma]$ is the set of points $(x,t)\in\ZZ\times\RR$ such that $\sigma_{(x-\tfrac{1}{2},t)} \neq \sigma_{(x+\tfrac{1}{2},t)}$, and $\om_{\horiz}[\sigma]$ the set of points of discontinuity of $\sigma$. 

Let $\lambda,J\in\RR$, $\lambda>0$. Define the space-time Ising model (with free boundary conditions) as the measure $\nu^{\ising}_{F,\Lambda,J}$ given by
\bes
	\d\nu^{\ising}_{F,\lambda,J}(\sigma)
	\propto
    \d\nu^\poi_{F,\lambda}(\chi)
	\mathbbm{1}_{\{\om_{\horiz}[\sigma]=\chi\}}
    e^{-J|\om_{\vert}[\sigma]|}
	,
\ees
where the sum is over all $\sigma\in\ul{\Om}^\ising_F$ satisfying the given condition. One should think of $J$ as the coupling constant in the horizontal direction (higher $J$ induces higher cost of a horizontal disagreement) and $\lambda$ the coupling constant in the vertical direction (higher $\lambda$ gives more disagreements in the vertical direction). 

Define the model with plus boundary conditions $\nu^{\ising,+}_{F,\lambda,J}$ (resp. minus $\nu^{\ising,-}_{F,\lambda,J}$) as the same model conditioned on $\sigma=1$ (resp. $\sigma=-1$) everywhere on $\partial F$.\\

We can now define the six-vertex spin representation. Let $F$ be a domain. Notice that $F^\bullet$ and $F^\circ$ are also domains, just rescaled horizontally by a factor of 2. 

Let $\Delta\in[-1,1)$, so $u\in[0,1)$. 
\bes
\begin{split}
	\nu^{\sv}_{F,\Delta}(\sigma^\bullet,\sigma^\circ)[\ \cdot \ ]
    =
    \nu^{\ising}_{F^\bullet,1,-\Delta}(\sigma^\bullet)
    \otimes
    \nu^{\ising}_{F^\circ,1,-\Delta}(\sigma^\circ)
    [ \ \cdot \ | \ \om[\sigma^\bullet]\cap\om[\sigma^\circ]=\varnothing].
\end{split}
\ees

One can write this explicitly as

\be\label{eq:q:6v-measure}
\begin{split}
	\d\nu^{\sv}_{F,\Delta}(\sigma^\bullet,\sigma^\circ)
	&\propto
	\d\nu^\poi_{F^\bullet,1}(\chi^\bullet)\d\nu^\poi_{F^\circ,1}(\chi^\circ)
	\mathbbm{1}_{\{\om_{\horiz}[\sigma^\bullet]=\chi^\bullet\}}
	\mathbbm{1}_{\{\om_{\horiz}[\sigma^\circ]=\chi^\circ\}}
	\mathbbm{1}_{\{\om[\sigma^\bullet]\cap\om[\sigma^\circ]=\es\}}
	e^{\Delta|\om_\vert[\sigma]|}.
\end{split}
\ee



We define $\nu^{\sv,\bplus,\wplus}_{F,\Delta}$ to be the measure $\nu^{\sv}_{F,\Delta}$ conditioned on the event $\sigma^\bullet,\sigma^\circ\equiv1$ on $\partial F$. Similarly one can define $\nu^{\sv,\bplus}_{F,\Delta}$, where only $\sigma^\bullet\equiv1$ on $\partial F^\bullet$, and $\nu^{\sv,\wplus}_{F,\Delta}$ similarly. The measures $\nu^{\sv,\bminus,\wminus}_{F,\Delta}$, $\nu^{\sv,\bminus}_{F,\Delta}$, $\nu^{\sv,\wminus}_{F,\Delta}$ are defined similarly. 

Finally, we also need the model with periodic boundary conditions in either the vertical or horizontal direction. Define $\nu^{\sv,\bullet}_{\cyl_{L,\beta}^h,
\Delta}$ to be the measure \eqref{eq:q:6v-measure} when $F$ is the cylinder $\cF\cap[-L+1,L]\times[-\tfrac{\beta}{2},\tfrac{\beta}{2}]$ where we identify the top and bottom points, $L$ is odd, and we condition that $\sigma^\bullet$ is constant on $\{L-\tfrac{1}{2}\}\times[-\tfrac{\beta}{2},\tfrac{\beta}{2}]$ and also constant on $\{-L-\tfrac{1}{2}\}\times[-\tfrac{\beta}{2},\tfrac{\beta}{2}]$ (note these two values can differ). Let $\nu^{\sv,\circ}_{\cyl_{L,\beta}^h,\Delta}$ be defined analogously with $\sigma^\circ$ and $L$ even. 

Let $\nu^{\sv,\bullet}_{\cyl_{L,\beta},\Delta}$ be the measure \eqref{eq:q:6v-measure} where again we let $F=\cF\cap[-L+1,L]\times[-\tfrac{\beta}{2},\tfrac{\beta}{2}]$, but this time we let the spins on $\{L-\tfrac{1}{2}\}\times[-\tfrac{\beta}{2},\tfrac{\beta}{2}]$ interact with those on $\{-L-\tfrac{1}{2}\}\times[-\tfrac{\beta}{2},\tfrac{\beta}{2}]$ (giving periodicity in the horizontal direction) and we condition on $\sigma^\bullet$ to be constant on $\partial^+ F^\bullet=\cF^\bullet\cap(\ZZ\times\{\tfrac{\beta}{2}\})$ and also constant on $\partial^- F^\bullet=\cF^\bullet\cap(\ZZ\times\{-\tfrac{\beta}{2}\})$ (these two values can differ). Let $\nu^{\sv,\circ}_{\cyl_{L,\beta},\Delta}$ be defined analogously using $\sigma^\circ$.

 Note that one can obtain the measure on height functions $\nu^{\sv-\hom,0,1}_{F,\Delta}$ from the measure on spins $\nu^{\sv,\bplus,\wplus}_{F,\Delta}$ by setting $h$ to take the values $0$ and $1$ on $\partial F$, and
\be\label{eq:q:h-from-sigma}
\begin{split}
    &1. \ \ \ \sigma^\bullet_x\sigma^\circ_y=h(x)-h(y) \qquad \mathrm{for} \  x\in F^\bullet, \  y=x\pm(1,0)\\
    &2. \ \ \  h \ \mathrm{constant \ on \ intervals \ where \ } \sigma \ \mathrm{is \ constant}.
\end{split}
\ee
The sigma algebra $\FF^{\sv,\#}_{F}$ we use is that generated by the $w^\#$ topology \cite[Section A2.6]{daley-verejones-1}.\\

We now define the bond percolations $\xi^\bullet$ and $\xi^\circ$ which are central to the argument of \cite{glazman-lammers}. 
Recall that $\om[\sigma^\bullet]$ can be thought of as a bond percolation configuration on $V^\circ$, and similar for $\om[\sigma^\circ]$ on $V^\bullet$. The configuration $\xi^\bullet$ will be a bond percolation configuration on $V^\bullet$, and $\xi^\circ$ is a bond percolation configuration on $V^\circ$. We define them as follows, conditional on a spin configuration $\sigma=(\sigma^\bullet,\sigma^\circ)$:
\be\label{eq:q:sample-both-xis-from-6v}
\begin{split}
    \text{On $\om[\sigma^\bullet]$}: &
    \begin{cases}
        \text{Set $\xi^\circ=1$, $\xi^\bullet=0$ deterministically}
    \end{cases}\\
    \text{On $\om[\sigma^\circ]$}: &
    \begin{cases}
        \text{Set $\xi^\bullet=1$, $\xi^\circ=0$ deterministically}
    \end{cases}\\
    \text{On $E^\bullet_{\vert}\setminus(\om[\sigma^\bullet]\cup\om[\sigma^\circ])$}: &
    \begin{cases}
        \text{Sample $X$ a PPP rate 1} \\
        \text{Set each $x\in X$ to be in $Y$ w.p. $1-2u$ independently.} \\
        \text{Set $\xi^\circ=1$ on $X$ and $\xi^\circ=0$ outside $X$.} \\ 
        \text{Set $\xi^\bullet=0$ on $Y$ and $\xi^\bullet=1$ outside $Y$.}
    \end{cases}\\
    \text{On $E^\circ_\vert\setminus(\om[\sigma^\bullet]\cup\om[\sigma^\circ])$}: &
    \begin{cases}
        \text{Do the same with black and white exchanged}
    \end{cases}
\end{split}
\ee

It is straightforward to check that for any configuration of $(\xi^\bullet,\xi^\circ)$ sampled as above, we have that if $\Delta\in[-1,0]$ (so $u\in[0,\tfrac{1}{2}]$), then
\be\label{eq:q:super-duality}
    (\xi^\bullet)^*\subset \xi^\circ \qquad \mathrm{and} \qquad 
    (\xi^\circ)^*\subset \xi^\bullet.
\ee
This property is sometimes called the ``super-duality'' of $\xi^\bullet$ and $\xi^\circ$. We note for clarity that if one wants to sample $\xi^\bullet$ and not $\xi^\circ$, one does the following:
\be\label{eq:q:sample-xi-black-from-6v}
\begin{split}
    \text{On $\om[\sigma^\bullet]$}: &
    \begin{cases}
        \text{Set $\xi^\bullet=0$ deterministically}
    \end{cases}\\
    \text{On $\om[\sigma^\circ]$}: &
    \begin{cases}
        \text{Set $\xi^\bullet=1$ deterministically}
    \end{cases}\\
    \text{On $E^\bullet_{\vert}\setminus(\om[\sigma^\bullet]\cup\om[\sigma^\circ])$}: &
    \begin{cases}
        \text{Sample $Y$ a PPP rate $1-2u$} \\
        \text{Set $\xi^\bullet=0$ on $Y$ and $\xi^\bullet=1$ outside $Y$.}
    \end{cases}\\
    \text{On $E^\circ_\vert\setminus(\om[\sigma^\bullet]\cup\om[\sigma^\circ])$}: &
    \begin{cases}
        \text{Sample $X$ a PPP rate 1} \\
        \text{Set $\xi^\bullet=1$ on $X$ and $\xi^\bullet=0$ outside $X$.} 
    \end{cases}
\end{split}
\ee

We will often write $\nu^{\sv}_{F,\Delta}$ for the measure on $(\sigma^\bullet,\sigma^\circ,\xi^\bullet,\xi^\circ)$ or $(\sigma^\bullet,\sigma^\circ,\xi^\bullet)$, depending on what's needed.

\begin{theorem}[Delocalisation]\label{thm:Q:delocalisation}
    Let $\Delta\in[-1,0]$, so $u\in[0,\frac{1}{2}]$. For any fixed $x\in\cF$,
    \be\label{eq:Q:delocalisation-repeat}
        \lim_{F\nearrow(\cF)}\Var_{\nu^{\sv-\hom;0,1}_{F,\Delta}}
        \left[h(x)\right]=\infty.
    \ee    
    The measure $\nu^{\sv,\rightarrow}_{F,\Delta}\to\nu^{\sv,\rightarrow}_{\Delta}$ 
    weakly in $(\ul{\Om}^{\sv,\rightarrow}_{\cF}, \FF^{\sv,\rightarrow,\#}_{\cF})$. Moreover, on the quadruple  $(\sigma^\bullet,\sigma^\circ,\xi^\bullet,\xi^\circ)$,
    the measures $\nu^{\sv}_{F,\Delta}, \nu^{\sv,\bplus\wplus}_{F,\Delta}, \nu^{\sv,\bplus}_{F,\Delta}$, $\nu^{\sv,\wplus}_{F,\Delta}$, $\nu^{\sv,*}_{\cyl_{L,\beta},\Delta}$, $\nu^{\sv,*}_{\cyl_{L,\beta}^h,\Delta}$ all converge to a common limit 
    $\nu^{\sv}_{\Delta}$ weakly in $(\ul{\Om}^{\sv}_{\cF}, \FF^{\sv,\#}_{\cF})$ as $F\nearrow(\cF)$ and as $L,\beta\to\infty$, for $*=\bullet,\circ$. The limits $L,\beta$ can be taken in either order or simultaneously, and $F$ can be taken to converge first to $\cF\cap([L_1,L_2]\times\RR)$ or $\cF\cap(\ZZ\times[\alpha_1,\alpha_2])$ and then to $\cF$. 
    
    The limit $\nu^{\sv}_{\Delta}$ is Gibbs, $\ZZ\times\RR$ translation-invariant and ergodic, and tail-trivial, and under this measure, the percolation processes $\xi^\dagger$ and $(\xi^\dagger)^*$ ($\dagger=\bullet,\circ)$ both do not percolate almost surely. As a consequence,
    \be\label{eq:Q:6v-decay-correlations}
    \begin{split}
        \lim_{||x_1-x_2||\to\infty} \nu^\sv_{\Delta}[\sigma^\bullet_{x_1}\sigma^\bullet_{x_2}]=0 
        \qquad
        \mathrm{and} 
        \qquad
        \lim_{\substack{||x_1-x_2||\to\infty \\ ||y_1-y_2||\to\infty}}
        \nu^\sv_{\Delta}[\sigma^\bullet_{x_1}\sigma^\bullet_{x_2}
        \sigma^\circ_{y_1}\sigma^\circ_{y_2}]=0.
    \end{split}
    \ee
\end{theorem}

The rest of this section is devoted to proving this theorem.

\subsection{Burton-Keane and Sheffield non-coexistence}
We need a couple of established results for ordinary bond and site percolation models in the space-time setting. Let's define what we mean by bond and site percolation here.

Let $G=(V,E)$ be a connected graph. Define a metric on $(V\times\RR)\cup(E\times\RR)$ given by the product of the graph distance on $G$ (edges are distance $\tfrac{1}{2}$ from each of their vertices) and the usual distance metric on $\RR$. 

Let $\mathfrak{B}$ be the set of (space-time) bond percolation configurations on $G\times\RR$, that is, the set of functions $\om:(V\times\RR)\cup(E\times\RR)\to\{0,1\}$ such that the points of $V\times\RR$ sent to 0 are locally finite in each copy of $\RR$, and the same for those points of $E\times\RR$ sent to 1. We say a measure $\mu$ on $\mathfrak{B}$ is a (space-time) bond percolation model on $G\times\RR$. We often identify a configuration $\om$ with the preimage of 1 under the map $\om$. We equip $\mathfrak{B}$ with the $w^\#$ topology generated by the metric $d^\#$ \cite[Section A2.6]{daley-verejones-1}. Here we identify $\om\in\mathfrak{B}$ with a simple point measure with point masses at each point of $V\times\RR$ sent to 0 and each point of $E\times\RR$ sent to 1 so that $d^\#$ is well defined. We write $\FF^\#(\mathfrak{B})$ for the Borel	$\sigma$-algebra generated by this topology.

Let $\mathfrak{S}$ be the set of (space-time) site percolation configurations, that is, the set of functions $\om:V\times\RR\to{0,1}$ which are right continuous in each copy of $\RR$ and have that the set points of discontinuity in each copy of $\RR$ is locally finite. We say a measure $\mu$ on $\mathfrak{S}$ is a (space-time) site percolation model. We often identify a configuration $\om$ with the preimage of 1 under the map $\om$. We equip $\mathfrak{S}$ with the $w^\#$ topology as before, setting the points of discontinuity of $\om$ to be the point masses. We write $\FF^\#(\mathfrak{S})$ for the Borel $\sigma$-algebra generated by this topology. 

\begin{remark}
	Note that if one takes Bernoulli site percolation on a discrete lattice and takes a scaling limit in one direction, one cannot obtain a space-time site percolation model as defined here. Indeed, one needs that in the vertical direction, the value of the percolation model is almost always the same, a behaviour not obtainable when the sites are independent. One should instead think of the space-time Ising model as the prototypical example of a space-time site percolation model.
\end{remark}

For a set $I\subset V\times\RR$ or $E\times\RR$, write $\tau_I$ for the sigma algebra generated by events dependent on the configuration outside of $I$. We say a measure $\nu$ on $\mathfrak{B}$ has the finite energy property if for all intervals $I\subset V\times\RR$ and $J\subset E\times\RR$, there exists a constants $c_1,c_2,c_3,c_4\in(0,1)$ such that
\bes
    c_1<\nu[I\subset\om]<c_2; \qquad c_3<\nu[J\cap\om=\es]<c_4.
\ees
Similarly we say that a measure $\nu$ on $\mathfrak{S}$ has the finite energy property if for all intervals $I\subset V\times\RR$, there exists a constants $c_1,c_2,c_3,c_4\in(0,1)$ such that
\bes
    c_1<\nu[I\subset\om]<c_2; \qquad c_3<\nu[I\subset(V\times\RR)\setminus\om]<c_4.\\
\ees    
\begin{proposition}[Burton-Keane]\label{prop:burton-keane}
    Let $G$ be a locally finite transitive amenable graph and let $\nu$ be a measure on $\mathfrak{B}$ or $\mathfrak{S}$ which is invariant under translations of both the graph $G$ and $\RR$, and satisfies the finite energy property. Then the number of infinite clusters of $\om$ under $\nu$ is almost surely 0 or 1.
\end{proposition}

A proof identical to that of Theorem 2.3.10 of \cite{bjornberg-thesis} applies here.\\

We need a version of Sheffield's non-coexistence theorem for space-time percolation models. Let $G=\ZZ$, and let $\G^*=\ZZ+\frac{1}{2}$. Identify the edges $G^*$ with the vertices of $G$ (write $v^*$ for the edge of $G^*$ associated to each $v\in V(G)$), and vice-versa ($e^*$ for the edge of $G^*$ associated to each $e\in E(G)$). For a bond percolation configuration $\om\in\mathfrak{B}$, $\om^*$ is the bond percolation configuration on $G^*\times\RR$ given by $\om(v^*,t)=1-\om(v,t)$ and $\om(e^*,t)=1-\om(e,t)$ for all $v\in V(G),e\in E(G)$ and all $t\in\RR$. For a site percolation configuration $\om\in\mathfrak{S}$ let $\om^*=1-\om$. 

There exists a partial order on $\mathfrak{B}$ and $\mathfrak{S}$ given by: $\om\le\om'$ if $\om(a,t)\le\om'(a,t)$ for all $a\in V,E$, $t\in\RR$. We say a function on $\mathfrak{B}$ and $\mathfrak{S}$ is increasing if it is $\om\le\om'$ implies $f(\om)\le f(\om')$. We say that a measure $\nu$ on $\mathfrak{B}$ or $\mathfrak{S}$ satisfies the FKG inequality if for all measurable increasing functions $f,g$, we have $\int f g \d\nu\ge\int f\d\nu \cdot \int g\d\nu$. 

We say $\om$ percolates as shorthand for the event that $\om$ contains an infinite cluster.

\begin{proposition}[Sheffield non-coexistence]\label{prop:non-coexistence}
    Let $G=\ZZ$ and let $\nu$ be a measure on $\mathfrak{B}$ or $\mathfrak{S}$ which is $\ZZ\times\RR$ translation invariant and satisfies the FKG inequality. Then
    \bes
        \nu[\om \ \mathrm{and} \ \om^* \ \mathrm{both \ percolate}]<1.
    \ees
\end{proposition}
An essentially identical proof to that of Theorem 1.5 of \cite{dc--raoufi-tassion} holds here.

\subsection{Basic properties of the space-time six vertex model: height function}
\begin{lemma}\label{lem:q:hf-fkg-markov} Let $\Delta\le0$, so $u\le\tfrac{1}{2}$.
\begin{enumerate}
    \item \normalfont{(FKG)} Let $h\sim\nu^{\sv-\hom,\tilde{h}}_{F,\Delta}$ \eqref{eq:q:6v-hom-measure}. Then $h$ is FKG: for all increasing, measurable functions $f,g$ of $h$, $\nu^{\sv-\hom,\tilde{h}}_{F,\Delta} [f\cdot g] \ge \nu^{\sv-\hom,\tilde{h}}_{F,\Delta}[f] \cdot \nu^{\sv-\hom,\tilde{h}}_{F,\Delta}[g]$.
    \item \normalfont{(Comparison of boundary conditions)} Let $\tilde{h}_0, \tilde{h}_1$ be (possibly random) height functions defined on $\partial F$ which are almost surely admissible (they can be extended to valid height functions on $F$), coupled such that $\tilde{h}_0\le\tilde{h}_1$ almost surely. Define $h_0\sim\nu^{\sv-\hom,\tilde{h}_0}_{F,\Delta}$ and $h_1\sim\nu^{\sv-\hom,\tilde{h}_1}_{F,\Delta}$. Then $h_0\preceq h_1$ ($h_0$ is stochastically dominated by $h_1$).
    \item \normalfont{(Domain Markov)} Let $h\sim\nu^{\sv-\hom,\tilde{h}}_{F,\Delta}$ and let $F'\subset F$ be a domain. Then
    \bes
        h|_{F'} \sim \nu^{\sv-\hom,h|_{\partial F'}}_{F',\Delta},
    \ees
    where $h|_{\partial F'}$ is the value of $h$ taken on $\partial F'$, which is a random variable measurable with respect to $\tau_{F\setminus F'}$, the sigma algebra generated by $h$ restricted to $F\setminus F'$.
    Written differently, for $A,B$ functions of the $h$ restricted to $F'$, $F\setminus F'$ respectively, we have 
    \bes
        \hat{\nu}^{\sv-\hom,\tilde{h}}_{F,\Delta}[A\cdot B]
        =
        \sum_{m\ge0}
        \int_{\partial F'} \d\mu^{\odot m}(\eta) \sum _{g \sim \eta} 
        \hat{\nu}^{\sv-\hom,g}_{F',\Delta}[A]
        \hat{\nu}^{\sv-\hom,g,\tilde{h}}_{F\setminus F',\Delta}[B],
    \ees
    where in each case $\hat{\nu}[f]=\cZ\nu[f]$ the unnormalised measure, and the second sum is over admissible height functions $g$ defined on $\partial F'$ whose points of discontinuity are exactly those of $\eta$.
    \item \normalfont{(Limits are Gibbs)} If $\nu^{\sv-\hom,\tilde{h}_k}_{F_k,\Delta}\to \nu$ weakly in $(\ul{\Om}^{\sv-\hom}_{\cF}, \FF^{\sv-\hom,\#}_{\cF})$ as $k\to \infty$, with $F_k\nearrow\cF$ as $k\to\infty$ and $\tilde{h}_k$ some sequence in $\ul{\Om}^{\sv-\hom}_{\cF}$, then $\nu$ is Gibbs.
\end{enumerate}
\end{lemma}
\begin{proof}
    Parts 1 and 2 are consequences of the same statements for the discrete model, due to Lemmas \ref{lem:appendix:FKG-proofs} and \ref{lem:appendix-convergence}. Part 4 is a simple consequence of part 3. We prove part 3. 
    To prove this we'll need Mecke's formula (see \cite[Theorem 4.4]{last-penrose}), which allows one to condition on the location of a finite number of Poisson points.

\begin{lemma}\label{lem:mecke}[Mecke's formula]
    For any $f$ a bounded local function of a Poisson point process $\chi$ on a set $V\subset\cF$, we have 
    \bes
    \nu^{\poi}_{V,1}\left[\sum_{\eta\subset\chi,  |\eta=k|} f(\chi,\eta)\right]
    =
    \int_V \d\mu^{\odot k}(\eta) \nu^{\poi}_{V,1}[f(\chi\cup\eta,\eta)].
    \ees
    where 
    \bes
     \d\mu^{\odot k}(\{x_1,\dotsc,x_k\})
    =\frac1{k!} \sum_{\pi\in S_k}
    \d\mu^{\otimes k}(x_{\pi(1)},\dotsc,x_{\pi(k)})
    \ees
    is the symmetrized $k$-fold product measure.
\end{lemma}

We have
\bes
\begin{split}
\hat{\nu}^{\sv-\hom;\tilde{h}}_{F,\Delta}[A\cdot B]
&=
\nu^{\poi}_{F,1}(\chi)
\left[
\sum_{h\sim\chi}
W_{F,\Delta}(h)
\mathbbm{1}_{h|_{\partial F} = \tilde{h}}
A\cdot B (h)
\right]\\
&=
\sum_{m\ge0}
\nu^{\poi}_{F,1}(\chi)
\left[
\sum_{h\sim\chi}
W_{F,\Delta}(h)
\mathbbm{1}_{h|_{\partial F} = \tilde{h}}
\mathbbm{1}_{\{m \ \mathrm{points \ of \ } \chi \ \mathrm{on} \ \partial F'\}}
A\cdot B (h)
\right]\\
&=
\sum_{m\ge0}
\int_{\partial F'} \d\mu^{\odot m}(\eta) \sum _{g \sim \eta} 
\nu^{\poi}_{F,1}(\chi)
\left[
\sum_{h\sim\chi}
W_{F,\Delta}(h)
\mathbbm{1}_{h|_{\partial F} = \tilde{h}, \ h|_{\partial F'} = g}
\mathbbm{1}_{\{\chi \cap \partial F'=\es\}}
A\cdot B (h)
\right]
\\
&=
\sum_{m\ge0}
\int_{\partial F'} \d\mu^{\odot m}(\eta) \sum _{g \sim \eta} 
\nu^{\poi}_{F',1}(\chi_1)
\left[
\sum_{h_1\sim\chi_1}
W_{F',\Delta}(h_1)
\mathbbm{1}_{h_1|_{\partial F'} = g}
\mathbbm{1}_{\{\chi_1 \cap \partial F'=\es\}}
A(h_1)
\right]\\
&\qquad\qquad\cdot
\nu^{\poi}_{F\setminus F',1}(\chi_2)
\left[
\sum_{h_2\sim\chi_2}
W_{F\setminus F',\Delta}(h_2)
\mathbbm{1}_{h_2|_{\partial F} = \tilde{h}, \ h_2|_{\partial F'} = g}
\mathbbm{1}_{\{\chi_2 \cap \partial F'=\es\}}
B(h_2)
\right]
\\
&=
\sum_{m\ge0}
\int_{\partial F'} \d\mu^{\odot m}(\eta) \sum _{g \sim \eta} 
\hat{\nu}^{\sv-\hom,g}_{F',\Delta}[A]
\hat{\nu}^{\sv-\hom,g,\tilde{h}}_{F\setminus F',\Delta}[B],
\end{split}
\ees
where the third equality is Mecke's formula and the fourth is the fact that everything factorises over the two domains $F'$ and $F\setminus F'$. Note that we can give a precise formula:
\bes
    \nu^{\sv-\hom,h|_{\partial F'}}_{F',\Delta}[A]
    =
    \sum_{m\ge0}
    \int_{\partial F'} \d\mu^{\odot m}(\eta) \sum _{g \sim \eta} 
    \frac{\cZ^{\sv-\hom,g,\tilde{h}}_{F\setminus F',\Delta}}
    {\cZ^{\sv-\hom,\tilde{h}}_{F,\Delta}}
    \hat{\nu}^{\sv-\hom,g}_{F',\Delta}[A].
\ees
\end{proof}

\subsection{Basic properties of the space-time six vertex model: spins and bond percolations}

For a spin configuration $\sigma\in\ul{\Om}^{\ising}_{F^\bullet}$ and a percolation configuration $\xi^\bullet\in\ul{\Om}^{\perc}_{F^\bullet}$, we say $\sigma^\bullet\perp\xi^\bullet$ if $\sigma^\bullet$ is constant on clusters of $\xi^\bullet$. The percolation $\xi^\bullet$ has a dual configuration $(\xi^\bullet)^*\in\ul{\Om}^{\perc}_{F^\circ}$, and we define $\sigma^\circ\perp(\xi^\bullet)^*$ similarly.

\begin{lemma}\label{lem:q:quantum-6v-triple}
    Let $\Delta\le0$, so $u\le\tfrac{1}{2}$. The marginal of $\nu^{\sv}_{F,\Delta}$ on the triple $(\sigma^\bullet,\sigma^\circ,\xi^\bullet)$ is
        \be\label{eq:quantum-6v-triple}
    \begin{split}
        \d\nu^\sv_{F,\Delta}(\sigma^\bullet,\xi^\bullet) \propto
        & \ 
        \d\nu^\poi_{E^\bullet_\vert,1}(\chi_1)
        \d\nu^\poi_{E^\bullet_\vert,1-2u}(\chi_2) 
        \d\nu^\poi_{E^\circ_\vert,1}(\chi_3)\\
        & \qquad\cdot
        \mathbbm{1}_{\{\sigma^\bullet\perp\xi^\bullet\}}
        \mathbbm{1}_{\{\chi_1=\om_{\horiz}[\sigma^\bullet]\}}
        \mathbbm{1}_{\{\chi_1\cup \chi_2=(\xi^\bullet)^*_\horiz\}}
        \mathbbm{1}_{\{\chi_3=\xi^\bullet_{\horiz}\}}\\
        & \qquad \cdot
        e^{u\left|\xi^\bullet_\vert\right|}
        e^{u\left|\om_\vert[\sigma^\bullet]\right|}
        e^{-u\left|E^\circ_\vert\setminus
        (\om_\vert[\sigma^\bullet]\cup\xi^\bullet_\horiz)\right|}
        \cdot
        2^{k((\xi^\bullet)^*)},
    \end{split}
    \ee
    where $k((\xi^\bullet)^*)$ is the number of clusters of $(\xi^\bullet)^*$. In particular, conditional on $\sigma^\bullet$ and $\xi^\bullet$, $\sigma^\circ$ is obtained by independently and uniformly assigning $\pm1$ to each cluster of $(\xi^\bullet)^*$. 
\end{lemma}

\begin{proof}[Proof of Lemma \ref{lem:q:quantum-6v-triple}]
    Upon renormalising, the measure $\nu^{\sv}_{F,\Delta}$ on $(\sigma^\bullet,\sigma^\circ)$ can we written as
    \be\label{eq:q:6v-measure-renormalised}
    \begin{split}
    	\d\nu^{\sv}_{F,\Delta}(\sigma^\bullet,\sigma^\circ)
    	&\propto
    	\d\nu^\poi_{E^\bullet_\vert,1}(\chi^\bullet)\d\nu^\poi_{E^\circ_\vert,1}(\chi^\circ)\\
    	&\qquad\qquad\cdot
        \mathbbm{1}_{\{\om_{\horiz}[\sigma^\bullet]=\chi^\bullet\}}
        \mathbbm{1}_{\{\om_{\horiz}[\sigma^\circ]=\chi^\circ\}}
    	\mathbbm{1}_{\{\om[\sigma^\bullet]\cap\om[\sigma^\circ]=\es\}}
    	e^{(1-u)(V\setminus|\om_\vert[\sigma]|)
        + u|\om_\vert[\sigma]| }.
    \end{split}
    \ee
    Using Mecke's formula and the sampling \eqref{eq:q:sample-xi-black-from-6v}, we have:
    \bes
    \begin{split}
        \nu^{\sv}_{F,\Delta}[f(\sigma^\bullet,\sigma^\circ,\xi^\bullet)]
        &=
        \sum_{k_\bullet,k_\circ\ge 0} 
        \int_{E^\bullet_\vert} \d\mu^{\odot k_\bullet}(\eta^\bullet)
        \int_{E^\circ_\vert} \d\mu^{\odot k_\circ}(\eta^\circ)
        \tfrac{1}{\cZ^{\sv}_{F,\Delta}}\\
        &\qquad\cdot
        \sum_{\sigma^\bullet\perp\sigma^\circ}
        \mathbbm{1}_{\{\eta^\bullet=\om_\horiz[\sigma^\bullet]\}}
        \mathbbm{1}_{\{\eta^\circ=\om_\horiz[\sigma^\circ]\}}
        e^{(1-u)(V\setminus|\om_\vert[\sigma]|)
        + u|\om_\vert[\sigma]| }\\
        &\qquad\cdot
        \nu^{\poi}_{E^\bullet_\vert,1}(Y)\otimes\nu^\poi_{E^\circ_\vert,1}(X)
        \left[ 
        \tfrac{1}{\cZ_Y}\tfrac{1}{\cZ_X}
        (1-2u)^{|Y|}
        \mathbbm{1}_{\{Y\cap\om[\sigma]=\es\}}
        \mathbbm{1}_{\{X\cap\om[\sigma]=\es\}}
        \right]\\
        &\qquad\cdot
        \sum_{\xi^\bullet}
        \mathbbm{1}_{\{ \xi^\bullet_\horiz = X \cup \om_\horiz[\om^\circ] \}}
        \mathbbm{1}_{\{ (\xi^\bullet)^*_\horiz = Y \cup \om_\horiz[\om^\bullet] \}}
        f(\sigma^\bullet,\sigma^\circ,\xi^\bullet),
    \end{split}
    \ees
    where $\cZ_X = e^{|E^\circ_\vert\setminus\om[\sigma]|}$, $\cZ_Y=e^{(1-2u)|E^\bullet_\vert\setminus\om[\sigma]|}$, and we use $\om[\sigma]=\om[\sigma^\bullet]\cup\om[\sigma^\circ]$. Now using Mecke's formula again, we have:
    \bes
    \begin{split}
        \nu^{\sv}_{F,\Delta}[f(\sigma^\bullet,\sigma^\circ,\xi^\bullet)]
        &=
        \sum_{k_\bullet,k_\circ\ge 0} 
        \int_{E^\bullet_\vert} \d\mu^{\odot k_\bullet}(\eta^\bullet)
        \int_{E^\circ_\vert} \d\mu^{\odot k_\circ}(\eta^\circ)
        \tfrac{1}{\cZ^{\sv}_{F,\Delta}}\\
        &\qquad\cdot
        \sum_{\sigma^\bullet\perp\sigma^\circ}
        \mathbbm{1}_{\{\eta^\bullet=\om_\horiz[\sigma^\bullet]\}}
        \mathbbm{1}_{\{\eta^\circ=\om_\horiz[\sigma^\circ]\}}
        e^{(1-u)(V\setminus|\om_\vert[\sigma]|)
        + u|\om_\vert[\sigma]| }\\
        &\qquad\cdot
        \sum_{k_X,k_Y\ge0}
        \int_{E^\bullet_\vert} \d\mu^{\odot k_X}(\eta_X)
        \int_{E^\circ_\vert} \d\mu^{\odot k_Y}(\eta_Y)\\
        &\qquad\cdot
        \left[ 
        \tfrac{1}{\cZ_Y}\tfrac{1}{\cZ_X}
        (1-2u)^{k_Y}
        \mathbbm{1}_{\{\eta_Y\cap\om[\sigma]=\es\}}
        \mathbbm{1}_{\{\eta_X\cap\om[\sigma]=\es\}}
        \right]\\
        &\qquad\cdot
        \sum_{\xi^\bullet}
        \mathbbm{1}_{\{ \xi^\bullet_\horiz = \eta_X \cup \eta^\circ \}}
        \mathbbm{1}_{\{ (\xi^\bullet)^*_\horiz = \eta_Y \cup \eta^\bullet \}}
        f(\sigma^\bullet,\sigma^\circ,\xi^\bullet)\\
            &=
        \sum_{k^\dagger,k_\circ,k_Y\ge 0} 
        \int_{E^\circ_\vert} \d\mu^{\odot k^\dagger}(\eta^\dagger)
        \int_{E^\bullet_\vert} \d\mu^{\odot k_\bullet}(\eta^\bullet)
        \int_{E^\bullet_\vert} \d\mu^{\odot k_Y}(\eta_Y)
        \tfrac{1}{\cZ^{\sv}_{F,\Delta}}
        (1-2u)^{k_Y}\\
        &\qquad\cdot
        \sum_{\substack{\sigma^\bullet\perp\xi^\bullet \\ \sigma^\circ\perp(\xi^\bullet)^*}}
        \mathbbm{1}_{\{\eta^\bullet=\om_\horiz[\sigma^\bullet]\}}
        \mathbbm{1}_{\{\om_\horiz[\sigma^\circ]\subset \eta^\dagger\}}
        \mathbbm{1}_{\{\eta_Y\cap\om_\vert[\sigma^\bullet]=\es\}}\\
        &\qquad\cdot
        \sum_{\xi^\bullet}
        \mathbbm{1}_{\{ \xi^\bullet_\horiz = \eta^\dagger \}}
        \mathbbm{1}_{\{ (\xi^\bullet)^*_\horiz = \eta_Y \cup \eta^\bullet \}}
        e^{u|\xi^\bullet_\vert|}
        e^{ u|\om_\vert[\sigma]| }
        e^{-u|E^\circ_\vert\setminus|\om_\vert[\sigma^\bullet]|}
        f(\sigma^\bullet,\sigma^\circ,\xi^\bullet)
    \end{split}
    \ees
    Here final equality is given by setting $k^\dagger=k_\circ+k_X$ and $\eta^\dagger=\eta^\circ\cup\eta_X$. Now using Mecke's formula one final time in reverse, one obtains \eqref{eq:quantum-6v-triple}, noticing that $\sigma^\circ\perp(\xi^\bullet)^*$ implies $\om_\horiz[\sigma^\circ]\subset\eta^\dagger=\xi^\bullet_\horiz$. The only condition on $\sigma^\circ$ that remains is $\sigma^\circ\perp(\xi^\bullet)^*$, which gives the property that $\sigma^\circ$ can be sampled by assigning independent uniform spins $\pm1$ to the clusters of $(\xi^\bullet)^*$.
\end{proof}

\begin{lemma}[Domain Markov]\label{lem:q:domain-markov}
    Let $\Delta\le0$, so $u\le\tfrac{1}{2}$. Let $F,F'$ be domains with $F\subset F'$. Let $A_n$ be a sequence of events dependent on the configuration in $F'$ such that $A_{n+1}\subset A_n$ for all $n$, $\nu^\sv_{F',u}[A_n]>0$ for all $n$ and $\bigcap_n A_n = \{\partial F \subset \xi^\bullet\}$. Then
    \bes
    	\nu^\sv_{F',\Delta}\left[\ \cdot \ | \ \{\partial F \subset \xi^\bullet\}\right]|_{F}
    	:=
    	\lim_{n\to\infty} \nu^\sv_{F',\Delta}[\ \cdot \ | \ A_n] |_F
    \ees
    exists, is independent of the configuration outside of $F$, and is distributed as $\nu^{\sv,\bplus}_{F,\Delta}$. The same holds when boundary conditions are imposed on $\nu^\sv_{F',u}$, and also when this measure is replaced by $\nu^{\sv,*}_{\cyl_{L,\beta},\Delta}$ or $\nu^{\sv,*}_{\cyl_{L,\beta}^h,\Delta}$, for $*=\bullet,\circ$ and $F\subset R_{L,\beta}$, where $R_{L,\beta}=\cF\cap[-L+1,L]\times[-\tfrac{\beta}{2},\tfrac{\beta}{2}]$.

    Also, let $\cO_{n,N}$ be the event that there exists a circuit of $\xi^\bullet$ in $\Lambda_N$ surrounding $\Lambda_n$, and let $\gamma_N^\bullet$ be the outermost such circuit. Then
    \bes
        \nu^{\sv}_{F,\Delta}[ \ \cdot \ | \cO_{n,N}]|_{(\gamma_N^\bullet)_{\inn}}
        =
        \nu^{\sv,\bplus}_{(\gamma_N^\bullet)_{\inn},\Delta}[ \ \cdot \ ].
    \ees
\end{lemma}

Recall that $\sigma^\bullet$ is constant on clusters of $\xi^\bullet$. Define $\xi^{\bplus}$ to be the part of $\xi^\bullet$ where $\sigma^\bullet=+1$ and $\xi^{\bminus}$ to be the part of $\xi^\bullet$ where $\sigma^\bullet=-1$.

\begin{lemma}[FKG]\label{lem:q:FKG} 
    Let $\Delta\le0$, so $u\le\tfrac{1}{2}$. The triple $(\sigma^\bullet,\xi^{\bplus},-\xi^{\bminus})$ satisfies the FKG inequality under the measure $\nu^\sv_{F,\Delta}$, as well as with the boundary conditions $\bplus$, $\bminus$ and/or $\wplus$, $\wminus$, and also under the measures $\nu^{\sv,*}_{\cyl_{L,\beta},\Delta}$, $\nu^{\sv,*}_{\cyl_{L,\beta}^h,\Delta}$, for $*=\bullet,\circ$. The same holds for the triple $(\sigma^\circ,\xi^{\wplus},-\xi^{\wminus})$. 
\end{lemma}

One can adapt the proof of FKG in \cite{glazman-lammers} to this setting, or one can view the measures here as limits of the discrete measures from Section \ref{sec:1st-coupling}. We choose the latter; see Lemmas \ref{lem:appendix:FKG-proofs} and \ref{lem:appendix-convergence}. A combination of the Domain Markov property Lemma \ref{lem:q:domain-markov} and the FKG property Lemma \ref{lem:q:FKG} gives the following.

\begin{lemma}[Comparison of boundary conditions]\label{lem:q:CBC}
    Let $F$, $F'$ be domains with $F\subset F'$ and let $\Delta\le0$, so $u\le\tfrac{1}{2}$. Then 
    \begin{enumerate}
    	\item $\nu^{\sv, \bminus\wplus}_{F,\Delta} \preceq_{\bullet} \nu^{\sv, \bplus\wplus}_{F,\Delta}$,
    	\item $\nu^{\sv, \bplus}_{F',u} |_{F} \preceq_{\bullet} \nu^{\sv, \bplus}_{F,\Delta}$,
    	\item $\nu^{\sv, \wplus}_{F',u} (\cdot \ | \ \sigma^\bullet|_A = \tau) |_F
    		\preceq_{\bullet} \nu^{\sv, \bplus}_{F,\Delta}$ for all $A\subset\partial F'$ and all $\tau:A\to\{\pm1\}$. In particular, we have $\nu^{\sv, \bminus}_{F,\Delta} \preceq_{\bullet} \nu^{\sv, \bplus \wplus}_{F,\Delta}
    		\preceq_{\bullet} \nu^{\sv, \bplus}_{F,\Delta}$,
        \item The same holds when $\nu^{\sv, \wplus}_{F',u}$ is replaced by $\nu^{\sv,*}_{\cyl_{L,\beta},\Delta}$ or $\nu^{\sv,*}_{\cyl_{L,\beta}^h,\Delta}$, for $*=\bullet,\circ$ and $F\subset R_{L,\beta}$. We therefore have 
        $\nu^{\sv, \bminus}_{R_{L,\beta},u} 
        \preceq_{\bullet} 
        \nu^{\sv,\bullet}_{\cyl_{L,\beta},\Delta}, \nu^{\sv,\bullet}_{\cyl^h_{L,\beta},\Delta}
    	\preceq_{\bullet} 
        \nu^{\sv, \bplus}_{R_{L,\beta},u}$.
    \end{enumerate}
\end{lemma}

\subsection{Convergence of measures and delocalisation}

We now adapt the proofs of Glazman and Lammers \cite{glazman-lammers} to prove Theorem \ref{thm:Q:delocalisation}. This comes in four steps: we show that $\nu^{\sv, \bplus}_{F_k,u}\to\nu^{\sv, \bplus}_{\Delta}$ as $F_k\nearrow(\ZZ+\tfrac{1}{2})\times\RR$, that $(\xi^\bullet)^*$ does not percolate under this measure, $\xi^\bullet$ also does not percolate under this measure, and finally that delocalisation occurs. While the strategy closely follows that of \cite[Sections 4.5 and 4.6]{glazman-lammers}, the proofs require technical modifications to the space-time setting. \\

Recall that we call the set of configurations for the six-vertex spin representation $\ul{\Om}^{\sv}_{\cF}$ and call the sigma algebra given by the $w^\#$ topology $\FF^{\sv,\#}_{\cF}$. We can appropriately enlarge these sets to include the $(\xi^\bullet,\xi^\circ)$ configurations. 


\begin{lemma}[Convergence of the measure with $\bplus$ boundary conditions]\label{lem:q:6v+convergence}
    Let $\Delta\le0$, so $u\le\tfrac{1}{2}$. For any sequence of domains $F_k=(V_k,F_k)\nearrow(\ZZ+\tfrac{1}{2})\times\RR$, as measures on configurations $(\sigma^\bullet, \sigma^\circ, \xi^\bullet,\xi^\circ)$, we have the weak convergence in $(\ul{\Om}^{\sv}_{\cF}, \FF^{\sv,\#}_{\cF})$:
    \bes
        \nu^{\sv, \bplus}_{F_k,u} \to \nu^{\sv, \bplus}_{\Delta}.
    \ees
    The limiting measure $\nu^{\sv, \bplus}_{\Delta}$ is a Gibbs measure for $(\sigma^\bullet,\sigma^\circ)$, and under this measure:
    \begin{enumerate}
        \item $(\sigma^\bullet,\xi^\bullet)$ is extremal and ergodic,
        \item $(\sigma^\bullet,\xi^{\bplus},-\xi^{\bminus})$ and $(\sigma^\bullet,\xi^{\wplus},-\xi^{\wminus})$ both satisfy the FKG inequality,
        \item The distribution of $\sigma^\circ$ given $(\sigma^\bullet,\xi^\bullet)$ is given by uniformly and independently assigning $\pm1$ to each cluster of $(\xi^\bullet)^*$,
        \item The distribution of $\sigma^\bullet$ given $(\sigma^\circ,\xi^\circ)$ is given by assigning $+1$ to the infinite cluster of $(\xi^\circ)^*$ (if it exists) and uniformly and independently assigning $\pm1$ to each finite cluster of $(\xi^\circ)^*$.
    \end{enumerate}
\end{lemma}
\begin{proof}
    Part 2 of Lemma \ref{lem:q:CBC} shows that the measures $\nu^{\sv,\bplus}_{F_k,u}$ restricted to a fixed domain $F$ are stochastically decreasing in $(\sigma^\bullet,\xi^{\bplus},-\xi^{\bminus})$ as $k$ increases, and \cite[Proposition 11.1.VI]{daley-verejones-2} shows that the measures are tight, so \cite[Corollary 6.4]{lindvall} shows that the marginal on the triple $(\sigma^\bullet,\xi^{\bplus},-\xi^{\bminus})$ on $F$ converges weakly in the $w^\#$-topology. Since $F$ was arbitrary, it follows that the weak convergence holds for the measures on all of $\cF =\ZZ+\tfrac{1}{2}\times\RR$ \cite[Theorem 11.1.VII]{daley-verejones-2}. Denote this infinite volume measure on $(\sigma^\bullet,\xi^{\bplus},-\xi^{\bminus})$ by $\nu$.
    
    The standard arguments available for discrete models \cite{grimmett-RC} apply to show that $\nu$ is $2\ZZ\times\RR$-ergodic and translation-invariant, tail-trivial and satisfies the FKG inequality. Indeed, an identical proof to that of \cite[Proposition 2.3.6]{bjornberg-thesis} shows that the limit is tail-trivial, $2\ZZ\times\RR$-invariance follows by interleaving the sequence $F_k$ with the sequence $\tau F_k$ for a given translation $\tau$. Ergodicity follows from an identical argument to \cite[Proposition 14.9]{georgii} and the FKG inequality follows because it holds for the finite volume measures and one can approximate any increasing event by bounded continuity events. 

    The bond percolation $(\xi^\bullet)^*$ under $\nu$ satisfies the conditions for the Burton-Keane argument \ref{prop:burton-keane} and therefore has at most one infinite cluster almost surely. This implies that the property that $\sigma^\circ$ is obtained by independently and uniformly assigning $\pm1$ to each cluster of $(\xi^\bullet)^*$ still holds in the infinite volume limit (see Lemma 6.5 of \cite{lis-spins} for a proof in the discrete). The proof of this fact requires modification in our setting, and we include the proof below. \\

    Let $\nu^{\sv,\bplus}_\Delta$ denote the measure on $(\sigma^\bullet,\sigma^\circ,\xi^\bullet)$ obtained from $\nu$ by assigning $\pm1$ independently and uniformly to each cluster of $(\xi^\bullet)^*$ to get $\sigma^\circ$. 
    For each variable $\eta$, write $\eta_k$ for the variable under the measure $\nu^{\sv, \bplus}_{F_k,u}$ and write $\eta_\infty$ for the variable under $\nu^{\sv, \bplus}_{\Delta}$. In the following we often think of a configuration $(\sigma^\bullet_k,\xi^\bullet_k)$ as a point measure on $(F_k\times\{\pm1\})\cup V_k$ (the points where $\sigma^\bullet_k$ is discontinuous (plus its values at those points), and the points of $\xi^\bullet_\horiz$ and $(\xi^\bullet)^*_\horiz$).

    We aim to show that for all bounded continuity sets $A_1,\dots,A_n\subset\cF$ of $\nu^{\sv,\bplus}_\Delta$, $(\sigma^\circ_k(A_1),\dots,\sigma^\circ_k(A_n))$ converges to $(\sigma^\circ_\infty(A_1),\dots,\sigma^\circ_\infty(A_n))$ weakly in $\RR^n$, where here we think of $\sigma^\circ_k$ as a simple random point measure on $((\ZZ+\tfrac{1}{2})\times\RR)\times\{\pm1\}$; this is equivalent to weak convergence \cite[Theorem 11.1.VII]{daley-verejones-2}. It will suffice to prove that $\sigma^\circ_k|_{F_K}\to\sigma^\circ_\infty|_{F_K}$ almost surely for all $K$. Note that a continuity set is defined by thinking of $\nu^{\sv,\bplus}_\Delta$ as a random point measure.
    
    Since the triple $(\sigma^\bullet,\xi^{\bplus},-\xi^{\bminus})$ converges weakly as $F_k\nearrow(\ZZ+\tfrac{1}{2})\times\RR$, by Skorokhod's representation theorem there exists a coupling $\PP$ of the random variables $(\sigma^\bullet,\xi^{\bplus},-\xi^{\bminus})_{k}$ for all $k\in\NN\cup\{\infty\}$ under which the convergence is almost sure. We can further require that under $\PP$, $(\sigma^\bullet,\xi^{\bplus},-\xi^{\bminus})_{k}\ge (\sigma^\bullet,\xi^{\bplus},-\xi^{\bminus})_{\infty}$ almost surely \cite[Theorem 6.1]{lindvall}.

    Fix $K\in\NN$ and $\eps>0$. Let $N\ge K$. For $\delta>0$ small enough (depending on $N$), by Lemma \ref{lem:q:pp-dom} and the ordering above, one can cover $(F_N\times\{\pm1\})\cup V_N$ with intervals $I_j$ of length at most $\delta$ such that under $\PP$ with probability at least $1-\eps$, each interval has at most 1 point of the point measure $(\sigma^\bullet,\xi^{\bplus},-\xi^{\bminus})_{k}$ and no two adjacent intervals have a point, for all $k$. Denote this event by $B_{N,\delta}$. By the pointwise almost sure convergence $(\sigma^\bullet,\xi^{\bplus},-\xi^{\bminus})_{k}\to(\sigma^\bullet,\xi^{\bplus},-\xi^{\bminus})_{\infty}$, there is some $k_0$ such that with probability 1 under $\PP$, for $k\ge k_0$, $(\sigma^\bullet,\xi^{\bplus},-\xi^{\bminus})_{k}[I_j] = (\sigma^\bullet,\xi^{\bplus},-\xi^{\bminus})_{\infty}[I_j]$ for all intervals $I_j$ (here $(\sigma^\bullet,\xi^{\bplus},-\xi^{\bminus})_{k}[I_j]$ is the number of Poisson points of each type appearing in $I_j$).

    Further, since $(\sigma^\bullet,\xi^{\bplus},-\xi^{\bminus})_{\infty}$ has at most one infinite cluster of $(\xi^\bullet)^*_\infty$ almost surely, we see that for $N$ large enough, with probability $1-\eps$ there is at most one cluster $\cC$ of $(\xi^\bullet)^*_{\infty}$ connecting $F_K$ with $\partial F_N$, and this cluster $\cC$ is part of the unique infinite cluster (if it exists). Call this event $C_{K,N}$. Since under $\PP$, $(\sigma^\bullet,\xi^{\bplus},-\xi^{\bminus})_{k}\ge (\sigma^\bullet,\xi^{\bplus},-\xi^{\bminus})_{\infty}$ almost surely, on the event $C_{K,N}$, for all $k$ larger than $N$, there is also at most one cluster of $(\xi^\bullet)_{k}$ connecting $F_K$ with $\partial F_N$, and this cluster also connects to the boundary of $F_K$. 

    Now on the event $C_{K,N}\cap B_{N,\delta}$, for $k\ge k_0$, the clusters of $(\xi^\bullet)^*_{k}$ inside $F_N$ are essentially fixed; to be precise, the number points of $((\xi^\bullet)^*_{k})_\horiz$ and of $(\xi^\bullet_k)_\horiz$ in each small interval $I_j$ is constant for all $k\ge k_0$. The locations of these points can still move around within their respective intervals as $k\to\infty$. In particular there is a bijection between the clusters of $(\xi^\bullet)^*_{k}$ inside $F_N$ and those of $(\xi^\bullet)^*_{\infty}$ inside $F_N$.

    Now sample $\sigma^\circ$ by independently and uniformly on each of these clusters, and using the bijection above, give each cluster of $(\xi^\bullet)^*_{k}$ the same spin $\sigma^\circ$ as that of the corresponding cluster in $(\xi^\bullet)^*_{\infty}$. Give the unique cluster connecting $F_K$ and $\partial F_N$ the spin $\sigma^\circ=1$ deterministically, if it exists. It is now straightforward to see that $\sigma^\circ_k|_{F_K}\to\sigma^\circ_\infty|_{F_K}$ almost surely, as a consequence of the above coupling and the convergence $(\xi^\bullet)^*_k\to (\xi^\bullet)^*_\infty$. The result of the claim follows. 
    
    The measure $\nu^{\sv,\bplus}_\Delta$ is therefore well-defined and the convergence extends to the triple $(\sigma^\bullet,\sigma^\circ,\xi^\bullet)$, and Part 3 of the Lemma is proved. Part 4 follows from an essentially identical argument to Part 3.
\end{proof}

\begin{lemma}
     Let $\Delta\in[-1,0]$, so $u\in[0,\tfrac{1}{2}]$. Under the measure $\nu^{\sv, \bplus}_{\Delta}$, the bond percolation $(\xi^\bullet)^*$ does not percolate almost surely, and the measure $\nu^{\sv, \bplus}_{\Delta}$ is tail-trivial and $2\ZZ\times\RR$-ergodic.
\end{lemma}
\begin{proof}
    Assume that $(\xi^\bullet)^*$ does percolate with positive probability; by the ergodicity of $(\sigma^\bullet,\xi^{\bplus},-\xi^{\bminus})$ proved above, it percolates with probability 1. By super-duality \eqref{eq:q:super-duality}, $(\xi^\bullet)^*\subset\xi^\circ$, so now $\xi^\circ$ percolates with probability 1 (this is where we use the condition $\Delta\ge-1$). Further, $\xi^\circ$ under the measure $\nu^{\sv, \bplus}_{\Delta}$ is dominated by $\xi^\circ$ under the measure $\nu^{\sv, \wplus}_{\Delta}$, so $\xi^\circ$ percolates with probability 1 under this measure too. By shifting by $(1,0)$, $\xi^\bullet$ percolates with probability 1 under the measure $\nu^{\sv, \bplus}_{\Delta}$. Recall that conditioned on $\xi^\bullet$, $\sigma^\bullet$ is obtained by assigning $+1$ to the infinite cluster (if it exists) and $\pm1$ uniformly and independently to each finite cluster. Hence if $\xi^\bullet$ percolates with probability 1 under the measure $\nu^{\sv, \bplus}_{\Delta}$, then so must $\xi^{\bplus}$. Now we've shown that both $\xi^{\bplus}$ and $(\xi^{\bplus})^*$ percolate with probability 1 and moreover $\xi^{\bplus}$ satisfies the FKG inequality and is $(2\ZZ\times\RR)$-invariant in law. This contradicts the non-coexistence theorem \ref{prop:non-coexistence}, and we conclude that $(\xi^\bullet)^*$ does not percolate almost surely.

    Finally, $\sigma^\circ$ is obtained by uniformly and independently assigning $\pm1$ to the clusters of $(\xi^\bullet)^*$. Tail triviality of $\sigma^\circ$ follows from the fact that the clusters of $(\xi^\bullet)^*$ are all finite almost surely, and ergodicity follows from tail-triviality as before. 
\end{proof}

\begin{lemma}\label{lem:q:xi-bullet-no-perc}
    Let $\Delta\in[-1,0]$, so $u\in[0,\tfrac{1}{2}]$. Under the measure $\nu^{\sv, \bplus}_{\Delta}$, the bond percolations $\xi^\bullet$ and $\xi^\circ$ do not percolate almost surely.
\end{lemma}
\begin{proof}
    The proof follows that of Theorem 3 of \cite{glazman-lammers}, using their ``$\TT$-circuit'' argument. In our setting, there is no need to introduce the $\TT$-connectivity, as the height function's sets of the form $\{h\ge a\}$ are space-time site percolations and Sheffield's non-coexistence theorem \ref{prop:non-coexistence} applies directly. (The $\TT$ connectivity and the usual $\ZZ^2$ connectivity collapse to the same thing in the limit $\ZZ^2\to(\ZZ\times\RR)$).
    
    The proof for $\xi^\circ$ is straightforward: if $\xi^\circ$ percolates with positive probability, then it does so with probability 1 by ergodicity, and it has a unique infinite cluster by the Burton-Keane argument \ref{prop:burton-keane}. This cluster is a cluster of $\xi^{\wplus}$ or $\xi^{\wminus}$ with probability $\tfrac{1}{2}$ each; but this contradicts the ergodicity of the triple $(\sigma^\circ,\xi^{\wplus},-\xi^{\wminus})$.\\

    Assume for contradiction that $\xi^\bullet$ does have an infinite cluster almost surely. By the Burton-Keane argument \ref{prop:burton-keane}, it must be unique, and by Lemma \ref{lem:q:6v+convergence} it must be a cluster of $\xi^{\bplus}$. For the $\bplus$ measure on a finite domain, $\nu^{\sv,\bplus}_{F,\Delta}$, let $\nu^{\sv-\hom;0}_{F,\Delta}$ be the height function defined by setting $h=0$ on $\partial F^\bullet$ and \eqref{eq:q:h-from-sigma}, and for the infinite volume measure, define $\nu^{\sv-\hom;0}_{\Delta}$ to be the height function defined by setting $h=0$ on the infinite cluster of $\xi^{\bplus}$ and \eqref{eq:q:h-from-sigma}. 
    
    We show next that $\nu^{\sv-\hom;0}_{F,\Delta}\to\nu^{\sv-\hom;0}_{\Delta}$ as $F\nearrow(\ZZ\times\RR)$. Fix $K\ge1$ and $\eps>0$. The proof is a very similar argument to that in the proof of Lemma \ref{lem:q:6v+convergence}, and we summarise it here. One can find a coupling $\PP$ of the measures $\nu^{\sv,\bplus}_{F_k,u}$, $k\in\NN$ and $\nu^{\sv,\bplus}_\Delta$ under which $\sigma_k$ converges to $\sigma_\infty$ almost surely. Further, for $N\ge K$ large enough and for $k\ge N$ large enough, one has the following. There is a partition of $F_N\times\{\pm1\}$ into intervals $I_j$ of length $\delta$ such that with probability at least $1-\eps$ under $\PP$, each $I_j$ contains at most one Poisson point of $\sigma_k$ and of $\sigma_\infty$ and no two adjacent intervals contain a point, and moreover, the data of which intervals $I_j$ contain a Poisson point is almost surely constant in $\sigma_k$ and equal to that of $\sigma_\infty$. Finally, there is a unique cluster of $\{\sigma_\infty=1\}$ connecting $F_K$ to $\partial F_N$, and this cluster is part of the infinite cluster of $\sigma_\infty$, and the corresponding cluster in $\sigma_k$ is also the unique cluster to connect $F_K$ to $\partial F_N$, and is itself connected to $\partial F_k$. 
    
    Now we sample $h$ inside $F_K$ according to the rules described above, and we see that $h_k$ and $h_\infty$ are identical up to the movement of the Poisson points of $\sigma_k$ inside their given intervals $I_j$. It's therefore clear that $\nu^{\sv-\hom;0}_{F_k,u}|_{F_K}\to\nu^{\sv-\hom;0}_{\Delta}|_{F_K}$ as $k\to\infty$, and the weak convergence on the whole space follows as $K$ was arbitrary. Finally, as the measures $\nu^{\sv-\hom;0}_{F_k,u}$ satisfy the domain Markov property, by Lemma \ref{lem:q:hf-fkg-markov} the limit $\nu^{\sv-\hom;0}_{\Delta}$ is Gibbs. \\

    The finite volume measures $\nu^{\sv-\hom,0}_{F_k,u}$ satisfy the FKG inequality in $h$ by Lemma \ref{lem:a:6v-hom--convergence}, and therefore so does the measure $\nu^{\sv-\hom,0}_{\Delta}$. One can similarly define $\nu^{\sv-\hom,1}_{\Delta}$ by taking the measure $\nu^{\sv,\wplus}_\Delta$ and assigning the height $1$ to the infinite cluster of $\{\sigma^\circ=1\}$ and \eqref{eq:q:h-from-sigma}. Write $\preceq$ for stochastic domination in $h$. To prove the Lemma, it suffices to show that 
    \be\label{eq:q:nu^1<nu^0}
        \nu^{\sv-\hom,1}_{\Delta} \preceq \nu^{\sv-\hom,0}_{\Delta}.
    \ee
    Indeed, iterating this gives $\nu^{\sv-\hom,4}_{\Delta} \preceq \nu^{\sv-\hom,0}_{\Delta}$ and the former measure is just $\nu^{\sv-\hom,0}_{\Delta}$ with all heights shifted up by 4, so clearly a contradiction.

    We now proceed with the $\TT$-circuit argument. The argument is essentially identical to that in \cite{glazman-lammers}. For $n\in\NN$, let $\Lambda_n=[-n,n]^2\cap((\ZZ+\tfrac{1}{2})\times\RR)$. The aim is to find, for all $n\in\NN$ and $\eps>0$, a coupling $\PP$ of $h^0\sim\nu^{\sv-\hom,0}_{\Delta}$ and $h^1\sim\nu^{\sv-\hom,1}_{\Delta}$ such that with probability at least $1-\eps$, 
    \be\label{eq:q:h^1<h^0}
        h^1|_{\Lambda_n} \le h^0|_{\Lambda_n}.
    \ee
    
    The random set $\{h=-1,-3,-5,\dots\}$ under $\nu^{\sv-\hom,0}_{\Delta}$ is a site percolation model on $\cF^\circ$, is $2\ZZ\times\RR$-invariant and satisfies the FKG inequality, and has the same law as its dual $\{h=1,3,5,\dots\}$. Hence neither percolate almost surely, by the non-coexistence theorem \ref{prop:non-coexistence}. Hence for fixed $n\in\NN$ and $\eps>0$, with probability at least $1-\eps$ under $\nu^{\sv-\hom,0}_{\Delta}$, there exists an $N$ such that $\Lambda_N$ contains a circuit of $\{h=1,3,5,\dots\}$ surrounding $\Lambda_{n+8}$. Let $\gamma_N$ be the outermost such circuit and write $F(\gamma_N)$ for the domain given by all points of $\cF$ on or inside of $\gamma_N$.
    
    Observe that if $h^0\sim\nu^{\sv-\hom,0}_{\Delta}$ then $h'$ defined by
    \be\label{eq:q:h^1}
        h'(x):=1-h^0(x-(1,0))
    \ee
    has law $\nu^{\sv-\hom,1}_{\Delta}$. We define $h^1$ as follows: 
    \begin{enumerate}
        \item Define $h^1$ on $\cF\setminus(F(\gamma_N)+(1,0))$ by \eqref{eq:q:h^1}.
        \item On $F(\gamma_N)+(1,0)$, define $h^1$ by sampling according to $\nu^{\sv-\hom,h^1}_{F(\gamma_N)+(1,0),\Delta}$, independently of $h^0$.
    \end{enumerate}
    The Gibbs property of $\nu^{\sv-\hom,1}_{\Delta}$ implies that $h^1$ defined as above has the law $\nu^{\sv-\hom,1}_{\Delta}$. Write $\PP$ for the joint law of $h^0$ and $h^1$.\\

    Let $\tilde{\gamma}_N$ be the circuit in $\cF$ whose closed interior is $F(\tilde{\gamma}_N) = F(\gamma_N)\cap F(\gamma_N+(1,0))$. Note that $\tilde{\gamma}$ is a deterministic function of $\gamma_N$. Observe that $h^0\ge 1$ on $\gamma_N$ and $h^1\le0$ on $\gamma_N+(1,0)$; this, plus the Lipschitz property of the six-vertex model, gives the fact that
    \bes
    h^1\le h^0 \ \  \text{ on $\tilde{\gamma}_N$}.
    \ees

    Now conditioning on the values of $h^0,h^1$ on $\tilde{\gamma}_N$, we see that by the comparison of boundary conditions Lemma \ref{lem:q:hf-fkg-markov}, $h^1|_{F(\tilde{\gamma}_N)} \preceq h^0|_{F(\tilde{\gamma}_N)}$. By Strassen's theorem, we can extend the coupling $\PP$ so that (again conditional on $\gamma_N$ existing) $h^0\ge h^1$ on $\Lambda_n$ almost surely. Hence $\PP[h^1|_{F(\tilde{\gamma}_N)} \leq h^0|_{F(\tilde{\gamma}_N)}] \ge \PP[\gamma_N \ \mathrm{exists}] >1-\eps$, which completes the proof that $h^1 \preceq h^0$ and gives the contradiction. Hence our assumption is false and $\xi^\bullet$ does not percolate almost surely. 
\end{proof}

We can finally conclude the proof of Theorem \ref{thm:Q:delocalisation}. Again we follow the proof of \cite[Theorem 1.3]{glazman-lammers} closely.

\begin{proof}[Proof of Theorem \ref{thm:Q:delocalisation}]
    We conclude from Lemma \ref{lem:q:xi-bullet-no-perc} that almost surely under $\nu^{\sv,\bplus}_\Delta$, each point $v\in\cF$ is surrounded by infinitely many disjoint circuits of $\xi^\bullet$ and infinitely many disjoint circuits of $\xi^\circ$. Further, by Lemma \ref{lem:q:6v+convergence}, since $\xi^\bullet$ does not percolate almost surely, we have $\nu^{\sv,\bplus}_\Delta=\nu^{\sv,\bminus}_\Delta:=\lim_{F\nearrow \cF}\nu^{\sv,\wplus}_{F,\Delta}$. Indeed, under $\nu^{\sv,\bplus}_\Delta$ there are infinitely many circuits of $\xi^{\bullet}$ which are each in $\xi^{\bminus}$ with probability $\tfrac{1}{2}$ independently; inside each of these circuits $\gamma$ the measure is $\nu^{\sv,\bminus}_{F(\gamma),\Delta}$. The stochastic domination Lemma \ref{lem:q:CBC}. Define 
    \bes
        \nu^{\sv}_\Delta:=\nu^{\sv,\bplus}_\Delta=\nu^{\sv,\bminus}_\Delta.
    \ees
    The stochastic domination from Lemma \ref{lem:q:CBC} gives that 
    \bes
        \nu^{\sv}_{F,\Delta}, \nu^{\sv,\bplus\wplus}_{F,\Delta}, \nu^{\sv,\wplus}_{F,\Delta}, \nu^{\sv,*}_{\cyl_{L,\beta},\Delta}, \nu^{\sv,*}_{\cyl_{L,\beta}^h,\Delta}
        \to\nu^\sv_\Delta,
    \ees
    as $F\nearrow(\ZZ+\tfrac{1}{2})\times\RR$ and as $L,\beta\to\infty$, for $*=\bullet,\circ$. Note that one can take $L,\beta\to\infty$ in either order or simultaneously. 
    It is straightforward to show that 
    \bes
    \nu^\sv_\Delta[\sigma^\bullet_{x_1}\sigma^\bullet_{x_2}]
    =\nu^\sv_\Delta[x_1\xleftrightarrow{(\xi^\circ)^*}x_2],
    \qquad
    \nu^\sv_\Delta[\sigma^\bullet_{x_1}\sigma^\bullet_{x_2}\sigma^\circ_{y_1}\sigma^\circ_{y_2}]
    =\nu^\sv_\Delta[x_1\xleftrightarrow{(\xi^\circ)^*}x_2, y_1\xleftrightarrow{(\xi^\bullet)^*}y_2],
    \ees
    by a standard colour-switching argument. Since there are no infinite clusters in $(\xi^\circ)^*$ or $(\xi^\bullet)^*$ almost surely under $\nu^\sv_\Delta$, both of the above tend to zero as $||x_1-x_2||$ and $||y_1-y_2||$ tend to $\infty$, and we have \eqref{eq:Q:6v-decay-correlations}.\\

    Finally we show delocalisation of the height function. The proof is identical to that of \cite[Theorem 1.3]{glazman-lammers}. Fix $x\in\cF$. We work with the measure $\nu^{\sv,\bplus\wplus}_{F,\Delta}$ with $F$ large enough so that $x\in F$. Let $\gamma_1$ be the largest circuit of $\xi^\circ$ surrounding (or passing through) $x$. For $i\ge1$, let $\gamma_{2i}$ be the largest circuit of $\xi^\bullet$ inside $\gamma_{2i-1}$, and let $\gamma_{2i+1}$ be the largest circuit of $\xi^\circ$ inside $\gamma_{2i}$; we stop when there are no more such circuits which surround or pass through $x$. 

    Write $\PP_F$ for the coupling between $\nu^{\sv,\bplus\wplus}_{F,\Delta}$ and $\nu^{\sv-\hom,0,1}_{F,\Delta}$. The definition of $h$ \eqref{eq:q:h-from-sigma} implies that $h$ is constant on each $\gamma_i$. We further have that 
    \be\label{eq:q:h-step-up-by-1}
        h(\gamma_{i+1})=h(\gamma_i)\pm1.
    \ee
    Indeed, say $\gamma_i\subset\xi^{\wplus}$ and let $\cC_i$ be the cluster of $\{\sigma^\circ=1\}$ containing $\gamma_i$ and $\cC_i'$ be the outermost cluster of $\{\sigma^\circ=-1\}$ surrounding $x$ and surrounded by $\cC_i$. The interface between $\cC_i$ and $\cC_i'$ is contained in $\om[\sigma^\circ]\subset\xi^\bullet$. The circuit $\gamma_{i+1}$ must therefore surround (or intersect) this interface and in particular is incident to $\cC_i$. By the definition of $h$ \eqref{eq:q:h-from-sigma}, $h$ is constant on $\cC_i$ and we have \eqref{eq:q:h-step-up-by-1}. The same argument works if $\gamma_i\subset\xi^{\wminus}$, etc.

    Since the sign of each cluster of $\xi$ is determined by a fair coin flip, we have for all $i$,
    \bes
        \PP_F[h(\gamma_{i+1})=h(\gamma_i)+1 \ | \ h(\gamma_i)]
        =
        \PP_F[h(\gamma_{i+1})=h(\gamma_i)-1 \ | \ h(\gamma_i)]
        =
        \frac{1}{2}.
    \ees   
    The law of $h(x)$ is therefore that of a simple random walk started at 0 and making $N_F$ steps, giving $\Var_F[h(x)]=\EE_F[N_F]$. Since $\nu^{\sv,\bplus\wplus}_{F,\Delta}\to\nu^\sv_\Delta$ and under $\nu^\sv_\Delta$ there are infinitely many alternating circuits of $\xi^\circ$ and $\xi^\bullet$ almost surely, we have that $\Var_F[h(x)]\to\infty$ as $F\to\infty$, and \eqref{eq:Q:delocalisation} holds. The argument to show that $\Var_{\nu^{\sv}_{\Delta}}[|h(x_1)-h(x_2)|]\to\infty$ as $|x_1-x_2|\to\infty$ is analogous; the circuits $\gamma_i$ are those that surround $x_2$ and do not surround $x_1$, and as $|x_1-x_2|\to\infty$, the number of these converges to infinity almost surely. This completes the proof of Theorem \ref{thm:Q:delocalisation}.

\end{proof}

\section{Results for the loop model and for XXZ with $\Delta\in[-1,0]$}\label{sec:quantum-loop-proofs}

\subsection{Overview of Theorem \ref{thm:main:quantum-loop}}

This section has precisely the same structure as Sections \ref{sec:1st-coupling}, \ref{sec:xxz-conv}, \ref{sec:new-coupling} and \ref{sec:8v-AT} combined, but applied to the space-time models. We take the delocalisation result, Theorem \ref{thm:Q:delocalisation}, and use it to prove Theorem \ref{thm:main:quantum-loop}, the result on the loop model, as well as Part 2 of Theorem \ref{thm:main-xxz}, the XXZ result. Note that the results on spin-spin correlations in the XXZ model follow from the results on the connection probabilities for the loop model via:
\bes
    \langle S^{(1)}_0 S^{(1)}_x(t) \rangle 
    =
    \frac{1}{4}\lim_{F\nearrow\cF}\nu^{\lp}_{\cF,\Delta}[(0,0)\leftrightarrow (x,t)];
\ees
this was originally proved by Ueltschi \cite{ueltschi}.\\

As in the earlier sections, we split Theorem \ref{thm:main:quantum-loop} into two propositions, Proposition \ref{prop:q:loop-conv-from-6v} and Proposition \ref{prop:q:loop-main-part-2} below. The former of these is proved directly by coupling the space-time six-vertex model and Ueltschi's loop model, akin to the first coupling between the six-vertex model and the mirror model in Section \ref{sec:1st-coupling}. The latter is proved by constructing the three remaining couplings, akin to those appearing in Section \ref{sec:new-coupling}: between the space-time six-vertex and Ashkin-Teller models, between the space-time Ashkin-Teller and eight-vertex models, and between the space-time eight-vertex model and Ueltschi's loop model. While the proofs have the same structure for the discrete models, there are various technical modifications needed.

\begin{proposition}\label{prop:q:loop-conv-from-6v}
    Fix $\Delta\in[-1,0]$, so $u\in[0,\tfrac{1}{2}]$. There exists an infinite-volume measure $\nu^{\lp}_\Delta$ on loop configurations on $\ZZ\times\RR$ such that as $F\nearrow (\ZZ+\tfrac{1}{2})\times\RR$, respectively as $\beta\to\infty$ and $L\to\infty$ either simultaneously or in either order,
        \be\label{eq:q:loop-measures-converge-repeat}
            \nu^{\lp}_{F,\Delta}, 
            \nu^{\lp,*}_{F,\Delta}, 
            \nu^{\lp,*}_{\cyl_{L,\beta},\Delta},
            \nu^{\lp,*}_{\cyl^h_{L,\beta},\Delta},
            \to \nu^{\lp}_\Delta,
        \ee
    weakly in $(\ul{\Om}^{\lp}_{\cF}, \FF^{\lp,\#}_{\cF})$ in the $w^\#$ topology, for $*=\bullet,\circ$. The measure $\nu^{\lp}_\Delta$ is tail-trivial and $\ZZ\times\RR$-invariant and ergodic. Further, for $(x_0,t_0)\in\cF$, we have that for all $\eps>0$,
        \be\label{eq:q:loop-connections-no-exp-decay-2}
            \sum_{n=1}^{\infty}n 
            \cdot 
            \lim_{V\nearrow\ZZ\times\RR}
            \nu^{\lp}_{F,\Delta}[(0,0)\leftrightarrow (nx_0,nt_0)]^{1-\eps} = \infty.
        \ee
\end{proposition}

\begin{proposition}\label{prop:q:loop-main-part-2}
    Fix $\Delta\in[-1,0]$, so $u\in[0,\tfrac{1}{2}]$. We have
        \be\label{eq:q:loop-connections-decay-to-0-repeat}
            \nu^{\lp}_{\Delta}[(0,0)\leftrightarrow (x,t)]
            \le 
            \lim_{F\nearrow\cF}
            \nu^{\lp}_{F,\Delta}[(0,0)\leftrightarrow (x,t)] \to 0
        \ee
    as $||(x,t)||\to\infty$. 

        If $u=\tfrac{1}{2}$, $\nu^{\lp}_\Delta$ is Gibbs, and in fact is the unique Gibbs measure and the unique thermodynamic limit, that is, for all $\chi'\in\ul{\Om}^\lp_{\cF}$, we have as $F\nearrow\cF$ that
    \be\label{eq:q:unique-loop-thermo-limit-repeat}
        \nu^{\lp,\chi'}_{F,\Delta}
        \to
        \nu^\lp_\Delta,
    \ee
    and further we have that 
	    \be\label{eq:q:xy-corr-decay-loop-repeat}
	    \lim_{F\nearrow \cF} 
        \nu^{\lp}_{F,\Delta}[(0,0)\leftrightarrow(x,t)] 
        \sim ||(x,t)||^{-\tfrac{1}{2}}.
	    \ee
\end{proposition}

\subsection{First coupling for the space-time models: loop model and six-vertex}
\begin{proof}[Proof of Proposition \ref{prop:q:loop-conv-from-6v}]
The proof is a straightforward adaptation of the proof of Proposition \ref{prop:mir-conv-from-6v}. In finite volume, take the space-time six-vertex model measure and couple it to Ueltschi's loop model measure as follows. 

\begin{lemma}\label{lem:q:couple-loop-and-sv}
    Let $\Delta\in[-1,1)$, so $u\in[0,1)$. 
    \begin{enumerate}
        \item Let $m\sim\nu^{\lp}_{F,\Delta}$. If one samples $\kappa$ by uniformly and independently orienting the loops of $m$, then $\kappa\sim\nu^{\sv}_{F,\Delta}$. 
        \item Let $\kappa\sim\nu^{\sv}_{F,\Delta}$. Sample $m$ according to the following:
        \be\label{eq:q:sample-loops-from-sv}
        \begin{cases}
            \mathrm{Sample} \ m_{\cross} \ \mathrm{as \ a \ PPP \ rate \ } u & \mathrm{on} \ \om_b[\kappa]; \\
            \mathrm{Sample} \ m_{\dbar} \ \mathrm{as \ a \ PPP \ rate \ } 1-u & \mathrm{on} \ \om_c[\kappa]; \\
            \mathrm{Set} \ x\in m_{\cross} \ \mathrm{w.p.} \ u, \ x\in m_{\dbar} \ \mathrm{w.p.} \ 1-u, & \mathrm{for} \ x\in\om_a[\kappa]. \\
        \end{cases}
        \ee
        Then $m\sim\nu^{\lp}_{F,\Delta}$.
        \item There exists an analogous coupling between $\nu^{\lp,\bullet}_{F,\Delta}$ and $\nu^{\sv,\bplus}_{F,\Delta}$, the only difference being that when sampling $(\sigma^\bullet,\sigma^\circ)$ from $\kappa$, one sets $\sigma^\bullet\equiv +1$ on all of $\partial^\bullet_F$ in part 1; in part 2 one sets $m_v\in m^\bullet$ deterministically for all $v\in\partial V$. One defines analogous couplings between the pair $\nu^{\lp,\circ}_{F,\Delta}$ and $\nu^{\sv,\wplus}_{F,\Delta}$, the pair $\nu^{\lp,*}_{\cyl_{L,\beta},\Delta}$ and $\mu^{\sv,*}_{\cyl_{L,\beta},\Delta}$ and the pair $\nu^{\lp,*}_{\cyl^h_{L,\beta},\Delta}$ and $\mu^{\sv,*}_{\cyl^h_{L,\beta},\Delta}$ for $*=\bullet,\circ$.
    \end{enumerate}
\end{lemma}
The proof of this lemma can we formulated similarly to that of Lemma \ref{lem:q:quantum-6v-triple} using Mecke's formula; we leave the details to the reader.\\

Take the infinite volume space-time six-vertex measure $\nu^{\sv}_\Delta$ and sample a loop configuration $m$ from it via \eqref{eq:q:sample-loops-from-sv}. Let $\nu^{\lp}_\Delta$ denote the resulting marginal on loop configurations $m$. One can employ the same proof as in Proposition \ref{prop:mir-conv-from-6v} (see the text around \eqref{eq:mir-conv-from-6v}) to show that 
\be\label{eq:q:loop-measures-converge-repeat-2}
    \nu^{\lp}_{F,\Delta}, 
    \nu^{\lp,*}_{F,\Delta}, 
    \nu^{\lp,*}_{\cyl_{L,\beta},\Delta},
    \nu^{\lp,*}_{\cyl^h_{L,\beta},\Delta}
    \to \nu^{\lp}_\Delta,
\ee
for $*=\bullet,\circ$. Indeed, for a local event $A$ of the loop configuration, and $F_A$ any domain containing the support of $A$, and for $\tau_F$ the sigma algebra generated by the space-time six-vertex configuration restricted to $F$, one has:
\bes
\begin{split}
    \nu^{\lp}_{F,\Delta}[A] 
    &=
    \nu^{\lp-\sv}_{F,\Delta}[ \nu^{\lp-\sv}_{F,\Delta}[A \ | \ \FF_{F^c_A}]]\\
    &=
    \nu^{\lp-\sv}_{F,\Delta}[ \nu^{\lp-\sv}_{\Delta}[A \ | \ \FF_{F^c_A}]]\\
    &\to 
    \nu^{\lp-\sv}_{\Delta}[ \nu^{\lp-\sv}_{\Delta}[A \ | \ \FF_{F^c_A}]]\\
    &=
    \nu^{\lp}_{\Delta}[A],
\end{split}
\ees
where the convergence is as $F\nearrow\cF$, $\nu^{\lp-\sv}_{F,\Delta}$ is the joint measure defined by Lemma \ref{lem:q:couple-loop-and-sv}, the second equality is due to the loop configuration being a local (random) function of the six-vertex configuration, and the convergence holds since $\nu^{\lp-\sv}_{\Delta}[A \ | \ \FF_{F^c_A}]$ is a local function of the six-vertex configuration. Here $\FF_{F^c_A}$ is the sigma algebra generated by the six-vertex configuration outside the set $F_A$.

Interleaving domains $F$ and $\tau_{x,t}F$ for a translation $\tau_{x,t}$ by $(x,t)\in\ZZ\times\RR$ as $F\nearrow\cF$, one finds that the limiting measure $\nu^{\lp}_\Delta$ is $\ZZ\times\RR$-translation invariant, and clearly is also invariant to reflections in the horizontal and vertical axes. We can employ the same proof proof from the discrete to use the fact that there are infinitely many circuits of $\xi^\bullet$ almost surely, as well as the domain Markov property of $\xi^\bullet$, to prove an analogue of \eqref{eq:mir-tail-triv} giving that $\nu^{\lp}_{\Delta}$ is tail-trivial. This plus translation invariance gain gives ergodicity.


It remains to prove \eqref{eq:q:loop-connections-no-exp-decay-2}. Again we adapt the proof of \cite{eng-lis}. Let $(x_0,t_0)\in\ZZ\times\RR$, and say $x_0,t_0>0$. Let $E_0$ be the set of points $\{(i-\tfrac{1}{2},0)\}_{1\le i \le x_0}$ unioned with the interval $\{x_0-\tfrac{1}{2}\}\times[0,t_0]$, and let $E_n=\tau_{(nx_0,nt_0)}E_0$. Let $M_{i,j}$ be the number of loops joining $E_i$ and $E_j$ and let $M_i$ be the number of distinct loops passing through $E_j$. Observe that for all $F$ containing $(-tfrac{1}{2},0)$, writing $W(-\tfrac{1}{2},0)$ for the sum of the winding of all loops around $(-\tfrac{1}{2},0)$, 
\bes
\begin{split}
    \nu^{\sv-\hom,0}_{F,\Delta}[|h(-\tfrac{1}{2},0)|]
    &\le 
    \nu^{\lp,\bullet}_{F,\Delta}[|W(-\tfrac{1}{2},0)|]\\
    &\le
    \sum_{i<0,j\ge0} \nu^{\lp,\bullet}_{F,\Delta}[M_{i,j}]\\
    &\le 
    \sum_{i<0,j\ge0} \nu^{\lp,\bullet}_{F,\Delta}[m_j \mathbbm{1}\{M_{i,j}>0\}]\\ 
    &\le 
    \sum_{i<0,j\ge0} \nu^{\lp,\bullet}_{F,\Delta}[m_j^q]^{1/q} \nu^{\lp,\bullet}_{F,\Delta}[E_i\leftrightarrow E_j]^{1/p}\\  
    &\le 
    C_{p,x_0,t_0} \sum_{i<0,j\ge0} \nu^{\lp,\bullet}_{F,\Delta}[E_i\leftrightarrow E_j]^{1/p},
\end{split}
\ees
where $p,q>1$ such that $\tfrac{1}{q}+\tfrac{1}{p}=1$. The fourth inequality is from Hölder's inequality, and the final inequality is due to the fact that $M_i$ is bounded above by $x_0$ plus the number of Poisson points of the loop model on the interval $\{ix_0-\tfrac{1}{2}\}\times[(i-1)t_0,it_0]$. This number of Poisson points is stochastically dominated by that of a pure Poisson point process of rate 2 (see \cite[Lemma 3.5]{bjornberg-ryan-dimerization}), so its $q^{th}$ moment is indeed finite (and dependent on $t_0$). Taking $V\nearrow\ZZ\times\RR$, the left hand side of the above tends to infinity by Theorem \ref{thm:Q:delocalisation}, and the translation invariance of $\lim_{V\nearrow\ZZ\times\RR}\nu^{\lp}_{F,\Delta}[E_0\leftrightarrow E_n]$ gives 
\bes
    \sum_{n=1}^\infty n \cdot 
    \lim_{V\nearrow\ZZ\times\RR}\nu^{\lp}_{F,\Delta}[E_0\leftrightarrow E_n]^{1/p} = \infty.
\ees
Finally, finite energy shows us, as in the discrete case, that $\nu^{\lp}_{F,\Delta}[E_0\leftrightarrow E_n] \ge C \nu^{\lp}_{F,\Delta}[(0,0)\leftrightarrow (nx_0,nt_0)]$, uniformly in $V$, and \eqref{eq:q:loop-connections-no-exp-decay-2} follows. This completes the proof of Proposition \ref{prop:q:loop-conv-from-6v}.
    
\end{proof}

\subsection{Convergence of the XXZ model in the range $\Delta\in[-1,0]$}\label{sec:xxz-conv-cts}

In this subsection we prove the convergence in part 2 of Theorem \ref{thm:main-xxz}, that is:
\bes
\begin{split}
    \langle \cdot \rangle &= 
    \lim_{L\to\infty}\lim_{\beta\to\infty} \langle \cdot \rangle_{\beta,\TT_L}
    =
    \lim_{\beta,L\to\infty} \langle \cdot \rangle_{\beta,\TT_L}^{\Psi^*_L}
    =
    \lim_{\beta,L\to\infty} \langle \cdot \rangle_{\beta,\Lambda_L},
    =
    \lim_{\beta,L\to\infty} \langle \cdot \rangle_{\beta,\Lambda_L}^{\Psi_L},
\end{split}
\ees
for $*=\bullet,\circ$ and where for all but the $\langle \cdot \rangle_{\beta,\TT_L}$ case, the limits in $\beta$ and $L$ can be taken in either order or simultaneously. In particular, we prove the convergence on any product of local observables of the form: 
\be\label{eq:q-convergence}
    \langle \prod_{k=1}^r A_k(t_k)\rangle_\Lambda 
    \to 
    \langle \prod_{k=1}^r A_k(t_k)\rangle,
\ee
where here we write $\langle\cdot\rangle_\Lambda$ for any of the finite volume, finite temperature states for which the convergence above holds. 
The proof is analogous to that in Section \ref{sec:xxz-conv}; this time we use the space-time six-vertex model and its percolation process $\xi^\bullet$. We fix $\Delta\in[-1,0]$ throughout and for conciseness we drop the subscript $\Delta$ from all measures and partition functions. \\

It suffices to prove the statement for $A_k(t_k)=E^{x_k}_{i_k^-,i_k^+}(t_k)$, where $E^{x_k}_{i_k^-,i_k^+}$ is the elementary 2x2 matrix with support at the single vertex $x_k$ and only the $(i_k^-,i_k^+)$ entry equal to 1. Let $X=\{(x_k,t_k)\}_{k=1}^r$. 
Write $\cZ^{\sv,\bullet}_{\cyl_{L,\beta}}$ for the partition function of the measure $\nu^{\sv,\bullet}_{\cyl^h_{L,\beta}}$ on the vertical cylinder $\cyl_{L,\beta}=\{-L+1,\dots,L\}\times[-\beta/2,\beta/2]$ (periodic in the horizontal direction). We formalise a measure with the disorders $\ul{i}$ at the points $X$:
\bes
\nu^{\sv,\bullet,X,\ul{i}}_{\cyl_{L,\beta}} [f(\kappa)]
=
\frac{1}
{\cZ^{\sv,\bullet,X,\ul{i}}_{\cyl_{L,\beta}}}
    \int\d \nu^{\poi}_{F,1}(\chi) \ 
    \sum_{\kappa\sim\chi}
    f(\kappa)
    \exp\left[u\left|\om^{\uparrow\uparrow}[\kappa]\right| + (1-u)\left|\om^{\uparrow\downarrow}[\kappa]\right|\right]
    \mathbbm{1}_{U_{X,\ul{i},\bullet}}
\ees
where $U_{X,\ul{i},\bullet}$ is the event that for each $(x_k,t_k)\in X$, the configuration $\kappa$ at $(x_k,t_k^{\pm})$ is $i^{\pm}_k$, (here $t_k^+$ is a point above and arbitrarily close to $t_k$; $t_k^-$ similar) and that $\kappa$ respects the boundary conditions $\bullet$ on the top and bottom of the cylinder: $\kappa(x,\beta/2)= \uparrow$ if and only if $\kappa(x+1,\beta/2)= \downarrow$ for all $x$ odd, and the same for the points $(x,-\beta/2)$. 

We have, analogously to \eqref{eq:disorders-mirrors}
\be\label{eq:q-disorders}
    \langle \prod_{k=1}^r E^{x_k}_{i_k^-,i_k^+}(t_k)\rangle^{\Psi^\bullet_L}_{\beta,\TT_L}
    =
    \frac{\cZ^{\sv,\bullet,X,\ul{i}}_{\cyl_{L,\beta}}}
    {\cZ^{\sv,\bullet}_{\cyl_{L,\beta}}}
    =
    \nu^{\lp,\bullet}_{\cyl_{L,\beta}}[2^{-l_{X}(m)} \mathbbm{1}_{\{m\sim(X,\ul{i})\}}],
\ee
where $l_{X}(m)$ is the number of loops of $m$ passing through $X$ and $m\sim(X,\ul{i})$ means that in the loops of $m$, the sinks of $(X,\ul{i})$ are connected to the sources. Here the sinks are points $(x_k,t_k^{\pm})$ oriented towards $(x_k,t_k)$, and sources are those oriented away from $(x_k,t_k)$. The proof of the second equality is identical to that of \eqref{eq:disorders-mirrors} using a switching bijection.

The same holds for $\bullet$ replaced by $\circ$, and analogous equalities hold for the pair 
$\langle \cdot \rangle_{\beta,\TT_L}$ and $\nu^{\sv}_{\TT_{L,\beta}}$,
the pair 
$\langle \cdot \rangle_{\beta,\Lambda_L}$ and $\nu^{\sv,*}_{\cyl^h_{L,\beta}}$, 
and the pair 
$\langle \cdot \rangle_{\beta,\Lambda_L}^{\Psi^*_L}$ and $\nu^{\sv,*}_{R_{L,\beta}}$, where $*=\bullet$ (resp. $*=\circ$) if $L$ is even (resp. odd) and recall $R_{L,\beta}$ is the rectangle $\cF\cap(\{-L+1,\dots,L\}\times[-\beta/2,\beta/2])$ and $\cyl^h_{L,\beta}$ is $R_{L,\beta}$ as the horizontal cylinder (periodic in the vertical direction). We will proceed using the identity \eqref{eq:q-disorders}; the proof is identical with the other boundary conditions. \\

The proof of convergence comes down to an analogue of Lemma \ref{lem:disorders-estimate}.

\begin{lemma}\label{lem:q-disorders-estimate}
    There exists a constant $C_{X}$ independent of $L,\beta$ such that for all $N\in\NN$ and all $L,\beta\ge N$,
    \be\label{eq:q:disorders-estimate}
        \left| 
        \frac{\cZ^{\sv,\bullet,X,\ul{i}}_{\cyl_{L,\beta}}}
        {\cZ^{\sv,\bullet}_{\cyl_{L,\beta}}}
        -
        \nu^{\sv,\bullet}_{\cyl_{L,\beta}}[f_N]
        \right|
        \le 
        C_{X}
        \nu^{\sv,\bullet}_{\cyl_{L,\beta}}[\eps_N],
    \ee
    where $f_N$ and $\eps_N$ are local functions of $(\sigma,\xi^\bullet)$, dependent only on the configuration in $\Lambda_N$, and 
    \bes
    \lim_{N\to\infty}\lim_{L,\beta\to\infty} \nu^{\sv,\bullet}_{\cyl_{L,\beta}}[\eps_N] =0.
    \ees
\end{lemma}

This lemma implies the convergence \eqref{eq:q-convergence}; the proof is identical to the discrete case after Lemma \ref{lem:disorders-estimate}.\\

We need to write the measure $\nu^{\sv,\bullet,X,\ul{i}}_{\cyl_{L,\beta}}$ when we also sample $\xi^\bullet$. As we did in the discrete \eqref{eq:sample-xi-from-arrows}, we can write the sets $\om[\sigma^\bullet]$ and $\om[\sigma^\circ]$ in terms of $\kappa$:
\be\label{eq:q:sample-xi-from-arrows}
\begin{split}
    \om[\sigma^\circ] &= \om_\horiz[\sigma^\circ] \cup \om_\vert[\sigma^\circ]; 
    \qquad
    \om_\horiz[\sigma^\circ]= (\chi\cap E^\circ_\vert ), 
    \qquad 
    \om_\vert[\sigma^\circ] = (\om^{\uparrow\uparrow}[\kappa]\cap E^\bullet_\vert)\\
    \om[\sigma^\bullet] &= \om_\horiz[\sigma^\bullet] \cup \om_\vert[\sigma^\bullet]; 
    \qquad
    \om_\horiz[\sigma^\bullet]= (\chi\cap E^\bullet_\vert ), 
    \qquad 
    \om_\vert[\sigma^\bullet] = (\om^{\uparrow\uparrow}[\kappa]\cap E^\circ_\vert)\\
\end{split}
\ee
where $\chi$ is the set of points of discontinuity of $\kappa$ (outside of $X$). Write $\om[\sigma]=\om[\sigma^\circ]\cup\om[\sigma^\bullet]$. Combining this with the sampling procedure for $\xi^\bullet$ \eqref{eq:q:sample-xi-black-from-6v}, and writing 
$\nu^{\poi}_{F,\sv}(\chi)=\nu^\poi_{F,1}(\chi)\otimes
    \nu^\poi_{V,1}(Y)$ with $Y^\bullet=Y\cap E^\bullet_\vert$, $Y^\circ=Y\cap E^\circ_\vert$, we have
\bes
\begin{split}
    \nu^{\sv,\bullet,X,\ul{i}}_{\cyl_{L,\beta}} [f(\kappa,\xi^\bullet)]
    &=
    \frac{1}
    {\cZ^{\sv,\bullet,X,\ul{i}}_{\cyl_{L,\beta}}}
    \int\d \nu^{\poi}_{F,\sv}(\chi,Y) \ 
    \sum_{\kappa,\xi^\bullet}
    \cW^{\sv}_{\chi,Y}(\kappa,\xi^\bullet)
    \mathbbm{1}_{U_{X,\ul{i},\bullet}}
    f(\kappa,\xi^\bullet)
\end{split}
\ees
where
\bes
\begin{split}
    \cW^{\sv}_{\chi,Y}(\kappa,\xi^\bullet)
    &=
    (1-2u)^{|Y^\bullet|}
    \mathbbm{1}_{\{Y^\bullet\subset E^\bullet_\vert\setminus\om[\sigma]\}}
    \mathbbm{1}_{\{Y^\circ\subset E^\circ_\vert\setminus\om[\sigma]\}}\\
    &\qquad\cdot
    \mathbbm{1}_{\{\kappa\}}
    e^{u\left|\om^{\uparrow\uparrow}[\kappa]\right| + (1-u)\left|\om^{\uparrow\downarrow}[\kappa]\right|}
    \mathbbm{1}_{\{(\xi^\bullet)^*_\horiz = Y^\bullet\cup \om_{\horiz}[\sigma^\bullet]\}}
    \mathbbm{1}_{\{\xi^\bullet_\horiz= Y^\circ\cup\om_\horiz[\sigma^\circ]\}}.
\end{split}
\ees
Let $\cO_{n,N}$ be the event that there exists a circuit of $\xi^\bullet$ contained in $\Lambda_N\setminus\Lambda_n$, surrounding $\Lambda_n$ (where $n$ is large enough so that $\Lambda_{n-2}$ contains $X$), and let $\gamma^\bullet_N$ be the outermost such circuit. 
One can prove an analogue of \eqref{eq:disorders-estimate-part-1}, that is, we have
\be\label{eq:q-disorders-estimate-part-1}
\begin{split}
    \frac{1}{\cZ^{\sv,\bullet}_{\cyl_{L,\beta}}}
    \int \d\nu^{\poi}_{F,\sv}(\chi,Y)
    \sum_{\substack{\kappa,\xi^\bullet}}
    \cW^{\sv}_{\chi,Y}(\kappa,\xi^\bullet)
    \mathbbm{1}_{U_{X,\ul{i},\bullet}}
    \mathbbm{1}_{\cO_{n,N}}
    =
    \nu^{\sv,\bullet}_{\cyl_{L,\beta}}
    [f_N], 
\end{split}
\ee
where $f_N = \mathbbm{1}_{\cO_{n,N}} \cZ^{\sv,\bullet,X,\ul{i}}_{(\gamma^\bullet_N)_{\inn}} / \cZ^{\sv,\bullet}_{(\gamma^\bullet_N)_{\inn}}$.
Note $(\gamma^\bullet_N)_{\inn}$ is the domain $F$ given by all points of $\cF$ on or inside of $\gamma^\bullet_N$, and $\mathbbm{1}_{\cO_{n,N}} 
    \cZ^{\sv,\bullet,X,\ul{i}}_{(\gamma^\bullet_N)_{\inn}} / 
    \cZ^{\sv,\bullet}_{(\gamma^\bullet_N)_{\inn}}$ is a function dependent only on the configuration in $\Lambda_N$.
The proof of \eqref{eq:q-disorders-estimate-part-1} is essentially identical to the discrete case \eqref{eq:disorders-estimate-part-1}; one uses Mecke's formula to condition on the location of the circuit $\gamma^\bullet_N$ and then one can essentially repeat \eqref{eq:disorders-estimate-part-1-working-1}. We leave the details to the reader. \\

It now suffices to prove that for all $\delta>0$, 
\be\label{eq:q-disorders-estimate-part-2}
\begin{split}
\frac{1}{\cZ^{\sv,\bullet}_{\cyl_{L,\beta}}}
    \int \d\nu^{\poi}_{F,\sv}(\chi,Y)
    \sum_{\substack{\kappa,\xi^\bullet}}
    \cW^{\sv}_{\chi,Y}(\kappa,\xi^\bullet)
    \mathbbm{1}_{U_{X,\ul{i},\bullet}}
    \mathbbm{1}_{\cO_{n,N}^c}
\le
C_X
\nu^{\sv,\bullet}_{\cyl_{L,\beta}}[\eps_N(\delta)] + \delta,
\end{split}
\ee
where $C_X$ is a constant independent of $N,L,\beta$ and where $\eps_N(\delta)$ is a local function dependent only on the configuration in $\Lambda_N$, and $\nu^{\sv,\bullet}_{\cyl_{L,\beta}}[\eps_N(\delta)] \to0$ as $L,\beta\to\infty$ and then $N\to\infty$.

The rest of this section is the proof of \eqref{eq:q-disorders-estimate-part-2}. The general strategy is the same as in the discrete, but the proof requires technical modifications. To start with, we need a technical lemma that bounds the probability of a given number of Poisson points of the six-vertex model, even under general boundary conditions or with disorders.

\begin{lemma}\label{lem:q:pp-dom}
    Consider the measure on arrow configurations $\nu^{\sv,\kappa_0,X,\ul{i}}_{F}$ with  $\Delta\in[-1,1)$ (so $u\in[0,1)$) on the domain $F$ (where $F$ can also be a cylinder) with boundary condition $\kappa_0$ and disorders $(X,\ul{i})$. Let $U\subset F$ be a finite union of intervals and let $U_k$ be the event that there are exactly $k$ Poisson points in $U$. Then
    \bes
        \nu^{\sv,\kappa_0,X,\ul{i}}_{F}[U_k]
        \le
        \frac{|U|^k}{k!}.
    \ees
\end{lemma}

\begin{proof}
We can rewrite the measure $\nu^{\sv,\kappa_0,X,\ul{i}}_{F}$ in a helpful way. Write $\chi$ as the sum of two Poisson point processes $\chi_\rightarrow$ and $\chi_\leftarrow$. We think of a point of $\chi_\rightarrow$ at $(x+\tfrac{1}{2},t)$ as a horizontal arrow from $(x,t)$ to $(x+1,t)$, and vice-versa for $\chi_\leftarrow$. A configuration $\vec{\chi}=(\chi_\rightarrow,\chi_\leftarrow)$ is alternating if along every vertical interval of $(\{x\}\times\RR)$ in $(V+(\tfrac{1}{2},0))\cup(V-(\tfrac{1}{2},0))$, the arrows of $\chi$ alternate between incoming and outgoing to the interval. An alternating configuration produces an arrow configuration $\kappa=\kappa(\vec{\chi}):(V+(\tfrac{1}{2},0))\cup(V-(\tfrac{1}{2},0))\to\{\uparrow,\downarrow\}$ by: at each point $x$ of $(V+(\tfrac{1}{2},0))\cup(V-(\tfrac{1}{2},0))$ there is an incoming arrow of $\chi$ above $x$ and an outgoing arrow below (or vice-versa); in either case, let $x$ be directed towards the outgoing arrow. We now define $\mathrm{alt}=\mathrm{alt}(\vec{\chi},\kappa_0,X,\ul{i})$ be the event the configuration is alternating and that $\kappa(\vec{\chi})=\kappa_0$ on the boundary and that $\kappa(\vec{\chi})$ respects the disorders $(X,\ul{i})$ (that is, $U_{X,\ul{i}}(\kappa(\vec{\chi}))$ occurs).
We have
    \bes
        \nu^{\sv,\kappa_0,X,\ul{i}}_{F}[f(\vec{\chi})]
        =
        \frac{1}{\cZ^{\sv,\kappa_0,X,\ul{i}}_{F}}
        \nu^\poi_{F,1}
        \left[\sum_{o}(\tfrac{1}{2})^{|\chi|} 
        f(\vec{\chi})\mathbbm{1}_{\mathrm{alt}(\vec{\chi},\kappa_0,X,\ul{i})} \cW^\sv_{\Delta}(\vec{\chi})\right],
    \ees
where the sum is over all $2^k$ possible orientations $o$ left or right of the Poisson points $\chi$ and 
$ \cW^\sv_{\Delta}(\vec{\chi})
=
\exp\left[u\left|\om^{\uparrow\uparrow}[\kappa(\vec{\chi})]\right| + (1-u)\left|\om^{\uparrow\downarrow}[\kappa(\vec{\chi})]\right|\right]$. Now using Mecke's formula \ref{lem:mecke},
\bes
    \begin{split}
        \nu^{\sv,\kappa_0,X,\ul{i}}_{F}[U_k]
        &=
        \frac{1}{\cZ^{\sv,\kappa_0,X,\ul{i}}_{F}}
        \nu^{\poi}_{F,1}
        \left[\sum_{o}(\tfrac{1}{2})^{|\chi|} 
        \mathbbm{1}_{U_k(\vec{\chi})}
        \mathbbm{1}_{\mathrm{alt}(\vec{\chi},\kappa_0,X,\ul{i})}
        \cW^{\sv}_{\Delta}(\vec{\chi})
        \right]\\
        &=
        \int \d\mu^{\odot k}(\eta)
        \sum_{o(\eta)}(\tfrac{1}{2})^{k}
        \frac{1}{\cZ^{\sv,\kappa_0,X,\ul{i}}_{F}}
        \nu^{\poi}_{F,1}
        \left[
        \mathbbm{1}_{U_0(\chi)}
        \sum_{o}(\tfrac{1}{2})^{|\chi|} 
        \mathbbm{1}_{\mathrm{alt}(\vec{\chi}\cup\vec{\eta},\kappa_0,X,\ul{i})}
        \cW^{\sv}_{\Delta}(\vec{\chi}\cup\vec{\eta})
        \right].
    \end{split}
\ees
Now using $\int \d\mu^{\odot k}(\eta)
        \sum_{o(\eta)}(\tfrac{1}{2})^{k} = \frac{|U|^k}{k!}$, the lemma follows from the bound
\bes
    \begin{split}
        \frac{1}{\cZ^{\sv,\kappa_0,X,\ul{i}}_{F}}
        \left[
        \mathbbm{1}_{U_0(\chi)}
        \sum_{o}(\tfrac{1}{2})^{|\chi|} 
        \mathbbm{1}_{\mathrm{alt}(\vec{\chi}\cup\vec{\eta},\kappa_0,X,\ul{i})}
        \cW^{\sv}_{\Delta}(\vec{\chi}\cup\vec{\eta})
        \right]
        &\le
        \frac{
        \nu^{\lp}_{F}[\mathbbm{1}_{\{m\sim \kappa_0,X,\ul{i},\vec{\eta}\}} 
        2^{-l(\kappa_0,X,\eta)} 
        ]}
        {\nu^{\lp}_{F}[\mathbbm{1}_{\{m\sim \kappa_0,X,\ul{i}\}} 
        2^{-l(\kappa_0,X)}]}\\
        &\le 1.
    \end{split}
\ees
Here $\nu^{\lp}_F$ is the free boundary conditions loop model measure on $F$. As well as $X$, we think of the boundary points $K_{\kappa_0}\subset \partial E$ where the height function changes according to $\kappa_0$ and the points of $\vec{\eta}$ as collections of sinks and sources, and  $m\sim \kappa_0, X,\ul{i}$ (resp. $m\sim \kappa_0, X,\ul{i},\vec{\eta}$) means that the loops of $m$ connect sinks to sources. Finally $l(\kappa_0,X,\eta)$ is the number of loops of $m$ passing through these points, and $l(\kappa_0,X)$ is the number passing through the points $X\cup K_{\kappa_0}$. 

The proof of the second to last equality a switching bijection with the same strategy as the proof of \eqref{eq:q-disorders}.
\end{proof}

We now turn to proving \eqref{eq:q-disorders-estimate-part-2}.
Let $\ell=\partial\Lambda_{n}+(\tfrac{1}{2},0)\subset\cF$ be the boundary of the box of size $n$.
Fix $\delta>0$. Since $(X,\ul{i})$ has the same number of sinks and sources, $\kappa\sim\nu^{\sv,\bullet,X,\ul{i}}_{\cyl_{L,\beta}}$ has a well-defined height function (up to overall shifts) along $\ell$, which we call $g$. 

By Lemma \ref{lem:q:pp-dom}, there is a large constant $K_\ell$ such that, uniformly in $L,\beta$, the probability that the height function $g$ changes values along $\ell$ more than $K_\ell$ times is less than $\delta$ (note the height changes deterministically $2n$ times along the horizontal (discrete) sides of $\ell$; it is along the vertical (continuous) sides that there is no deterministic bound).

We can then characterise $g$ along $\ell$ by the number of height changes, the points at which they are made, and the height change (+2 or -2) across each one. We can condition on the points at which the height function changes using Mecke's formula. We have that

\be\label{eq:q-disorders-estimate-part-3}
\begin{split}
    &\frac{1}{\cZ^{\sv,\bullet}_{\cyl_{L,\beta}}}
    \int \d\nu^{\poi}_{F,\sv}(\chi,Y)
    \sum_{\substack{\kappa,\xi^\bullet}}
    \cW^{\sv}_{\chi,Y}(\kappa,\xi^\bullet)
    \mathbbm{1}_{U_{X,\ul{i},\bullet}}
    \mathbbm{1}_{\cO_{n,N}^c}\\
    &\le
    \delta+ 
    \frac{1}{\cZ^{\sv,\bullet}_{\cyl_{L,\beta}}}
    \sum_{k=0}^{K_\ell} 
    \int \d\nu^{\poi}_{F,\sv}(\chi,Y)
    \sum_{\substack{\kappa,\xi^\bullet}}
    \cW^{\sv}_{\chi,Y}(\kappa,\xi^\bullet)
    \mathbbm{1}_{U_{X,\ul{i},\bullet}}
    \mathbbm{1}_{\cO_{n,N}^c}
    \mathbbm{1}_{\{k \ \mathrm{points \ of} \ \chi \ \mathrm{on} \ \ell\}}\\
    &=
    \delta+
    \frac{1}{\cZ^{\sv,\bullet}_{\cyl_{L,\beta}}}
    \sum_{k=0}^{K_\ell} 
    \int\d\mu^{\odot k}(\eta)
    \int \d\nu^{\poi}_{F,\sv}(\chi',Y)
    \sum_{\substack{\kappa,\xi^\bullet}}
    \cW^{\sv}_{\chi'\cup \eta,Y}(\kappa,\xi^\bullet)
    \mathbbm{1}_{U_{X,\ul{i},\bullet}}
    \mathbbm{1}_{\cO_{n,N}^c}
    \mathbbm{1}_{\{\chi'\cap \ell=\es\}}\\
    &=
    \delta+
    \frac{1}{\cZ^{\sv,\bullet}_{\cyl_{L,\beta}}}
    \sum_{k=0}^{K_\ell} 
    \int\d\mu^{\odot k}(\eta)
    \sum_{g\sim\eta}
    \int \d\nu^{\poi}_{F,\sv}(\chi',Y)
    \sum_{\substack{\kappa,\xi^\bullet}}
    \cW^{\sv}_{\chi'\cup \eta,Y}(\kappa,\xi^\bullet)
    \mathbbm{1}_{U_{X,\ul{i},\bullet,g}}
    \mathbbm{1}_{\cO_{n,N}^c}
    \mathbbm{1}_{\{\chi'\cap \ell=\es\}}\\
    &=
    \delta+
    \frac{1}{\cZ^{\sv,\bullet}_{\cyl_{L,\beta}}}
    \sum_{k=0}^{K_\ell} 
    \int\d\mu^{\odot k}(\eta)
    \sum_{g\sim\eta}
    \left[
    \int \d\nu^{\poi}_{\ell_{\inn},\sv}(\chi_{\inn}',Y_{\inn})
    \sum_{\substack{\kappa_{\inn},\xi_{\inn}^\bullet}}
    \cW^{\sv}_{\chi_{\inn}'\cup \eta,Y_{\inn}}(\kappa_{\inn},\xi_{\inn}^\bullet)
    \mathbbm{1}_{U_{X,\ul{i},g}}(\kappa_{\inn})
    \mathbbm{1}_{\cO_{n,N}^c}
    \mathbbm{1}_{\{\chi'_{\inn}\cap \ell=\es\}}
    \right] \\
    &\qquad\cdot
    \left[
    \int \d\nu^{\poi}_{\ell_{\outt},\sv}(\chi_{\outt}',Y_{\outt})
    \sum_{\substack{\kappa_{\outt},\xi_{\outt}^\bullet}}
    \cW^{\sv}_{\chi_{\outt}'\cup \eta,Y_{\outt}}(\kappa_{\outt},\xi_{\outt}^\bullet)
    \mathbbm{1}_{g,\bullet}(\kappa_{\outt})
    \mathbbm{1}_{\cO_{n,N}^c}(\xi^\bullet_{\outt})
    \mathbbm{1}_{\{\chi'_{\outt}\cap \ell=\es\}}
    \right]\\
    &=
    \delta+
    \frac{1}{\cZ^{\sv,\bullet}_{\cyl_{L,\beta}}}
    \sum_{k=0}^{K_\ell} 
    \int\d\mu^{\odot k}(\eta)
    \sum_{g\sim\eta}
    \cZ^{\sv,g,X,\ul{i}}_{\ell_{\inn}}
    \cZ^{\sv,\bullet,g}_{\cyl_{L,\beta}\setminus\ell_{\inn}}
    \nu^{\sv,\bullet,g}_{\cyl_{L,\beta}\setminus\ell_{\inn}}\left[
    \cO^c_{n,N}
    \right],
\end{split}
\ee

where $\eta$ is the set of points of $\chi$ on $\ell$, we use Mecke's formula in the first equality, the sum $g$ is over all height functions (up to overall shift) on $\ell$ compatible with $\eta$ ($\eta$ is exactly where $g$ changes value in the vertical direction) and the last second to last equality is the domain Markov property. The quantity $\cZ^{\sv,g,X,\ul{i}}_{\ell_{\inn}}$ is the partition function of the model on the interior of $\ell$, with height boundary condition $g$ and disorders $\ul{i}$ on $X$, and $\nu^{\sv,\bullet,g}_{\cyl_{L,\beta}\setminus\ell_{\inn}}$ is the measure on $\cyl_{L,\beta}\setminus\ell_{\inn}$ with boundary conditions $g$ on $\ell$ and $\bullet$ on $\partial\cyl_{L,\beta}$, and $\cZ^{\sv,\bullet,g}_{\cyl_{L,\beta}\setminus\ell_{\inn}}$ is its partition function.

Let us define $\nu^{\sv,\bullet,g}_{\cyl_{L,\beta}\setminus\ell_{\inn}}$ on height functions by fixing the lowest height achieved on $\ell$ by the boundary condition $g$ to be 0 (if the lowest height is even) or 1 (if it is odd). We claim that the height function is FKG. One can either prove this directly or by seeing the measure as a limit of discrete six-vertex model height functions; we do the latter in Lemmas \ref{lem:appendix:FKG-proofs} and \ref{lem:appendix-convergence}.

Let $D_{n,N}^{b_n}$ be the event that the height function $h\sim \nu^{\sv,\bullet,g}_{\cyl_{L,\beta}\setminus\ell_{\inn}}$ has a circuit of height $h\ge b_n$ in $\Lambda_N$ surrounding $\ell$. For $b_n$ larger than 2 plus the largest value taken by $g$ on $\ell$, we have that $D_{n,N}^{b_n}\subset O_{n,N}$. 

\begin{lemma}\label{lem:q-approx-h-circuit}
    There exist circuits $(\ell'_s)_{s\in[0,s_0]}$ surrounding $\ell$ and converging to a circuit $\ell_0$ as $s\to0$, such that $\nu^{\sv,\bullet,g}_{\cyl_{L,\beta}\setminus\ell_{\inn}}[h|_{\ell'_s}\le 1]>0$ for all $s\in(0,s_0]$ and
    \bes
        \nu^{\sv,\bullet,g}_{\cyl_{L,\beta}\setminus\ell_{\inn}}
        [h|_{\ell'_s} \in\{0,1\} \ | \ h|_{\ell'_s}\le1] \to 1
    \ees
    as $\ell'_s\to\ell_0$.
\end{lemma}

Given this lemma, the rest of the proof proceeds essentially as in the discrete setting. We have for $s$ small enough,
\bes
\begin{split}
    \nu^{\sv,\bullet,g}_{\cyl_{L,\beta}\setminus\ell_{\inn}}\left[
    \cO^c_{n,N}
    \right]
    &\le 
    \nu^{\sv,\bullet,g}_{\cyl_{L,\beta}\setminus\ell_{\inn}}\left[
    (D^{b_n}_{n,N})^c
    \right]\\
    &\le
    \nu^{\sv,\bullet,g}_{\cyl_{L,\beta}\setminus\ell_{\inn}}\left[
    (D^{b_n}_{n,N})^c \ | \ h|_{\ell'_s}\le1
    \right]\\
    &\le
    \nu^{\sv,\bullet,g}_{\cyl_{L,\beta}\setminus\ell_{\inn}}\left[
    (D^{b_n}_{n,N})^c \ | \ h|_{\ell'_s}\in\{0,1\}
    \right] +\delta\\
    &\le
    \nu^{\sv,\bullet}_{\cyl_{L,\beta}}\left[
    (D^{b_n}_{n,N})^c \ | \ h|_{\ell'_s}\in\{0,1\}
    \right]+\delta\\
    &\le
    C_{\ell}\nu^{\sv,\bullet}_{\cyl_{L,\beta}}\left[
    (D^{b_n}_{n,N})^c
    \right]+\delta,   
\end{split}
\ees
where the second inequality is due to the FKG inequality, the third is Lemma \ref{lem:q-approx-h-circuit}, the fourth is the domain Markov property and $C_l=\sup_{L,\beta}(\nu^{\sv,\bullet}_{\cyl_{L,\beta}}[h|_{\ell'}\in\{0,1\}])^{-1}$ is finite by the same argument as in the discrete. By the delocalisation of the space-time six-vertex model, Theorem \ref{thm:Q:delocalisation}, $\nu^{\sv,\bullet}_{\cyl_{L,\beta}}\left[(D^{b_n}_{n,N})^c\right]\to0$ as $L,\beta\to\infty$ and then $N\to\infty$. Finally, 
\bes
 \frac{1}{\cZ^{\sv,\bullet}_{\cyl_{L,\beta}}}
    \sum_{k=0}^{K_\ell} 
    \int\d\mu^{\odot k}(\eta)
    \sum_{g\sim\eta}
    \cZ^{\sv,g,X,\ul{i}}_{\ell_{\inn}}
    \cZ^{\sv,\bullet,g}_{\cyl_{L,\beta}\cap\ell_{\outt}}
    \le 
    \frac{\cZ^{\sv,\bullet,X,\ul{i}}_{\cyl_{L,\beta}}}{\cZ^{\sv,\bullet}_{\cyl_{L,\beta}}}
    \le 1.
\ees
so combining the last two equations with \eqref{eq:q-disorders-estimate-part-3} gives \eqref{eq:q-disorders-estimate-part-2}.
It now suffices to prove Lemma \ref{lem:q-approx-h-circuit}.

\begin{proof}[Proof of Lemma \ref{lem:q-approx-h-circuit}]
    Assume for ease of notation that we shift everything by $(-\tfrac{1}{2},0)$, so that $\ell$ is the boundary of the box $\Lambda_n$ intersected with $\ZZ\times\RR$ and height functions are defined on subsets of $\ZZ\times\RR$. Assume also for ease of notation that the points of $\eta$ (the points of discontinuity of $g$ on $\ell=\partial\Lambda_n$) are exactly at the points $(\pm n, i)$, $i=-(n-1),\dots,n-1$. The general case is essentially identical. Let $g_i^{\pm}$ be the value taken by $g$ on the interval $\{\pm n\}\times[i,i+1]$, for $i=-n,\dots,n-1$. Without loss of generality we can take $n$ to be even so that $g_i^{\pm}$ is always even. Let $D$ be the box $\Lambda_n$ unioned with the boxes $B_i^+=[n,n+g_i^+]\times[i,i+1]$ and $B_i^-=[-n-g_i^-,-n]\times[i,i+1]$ for $i=-n,\dots,n-1$. Let $\ell_0$ be the boundary of $D$.

    Now let $D(s)$ be defined as follows. For all $i=-n+1,\dots,n-1$, if $g_{i-1}^{\pm}>g_{i}^{\pm}$, let $E^+_{i}(s)=[n,n+g^+_{i-1}]\times[i,i+s]$, and if not, set $E^+_{i}(s)=[n,n+g^+_{i}]\times[i-s,i]$. Define $E^-_{i}(s)$ similarly. Let $C^+(s)=[-n-g^-_{n-1},n+g^+_{n-1}]\times[n,n+s]$ and $C^-(s) = [-n-g^-_{-n}, n+g^+_{-n}]\times[-n-s,n]$. Let $D_s$ be $D$, unioned with all the boxes $E^{\pm}_i(s)$, $i=-n+1,\dots,n-1$ and $C^{\pm}(s)$. Let $\ell_s$ be the boundary of $D(s)$. \\
    

    We prove that $\nu^{\sv,\bullet,g}_{\cyl_{L,\beta}\setminus\ell_{\inn}}[ \ h|_{\ell_s}\le 1]>0$ for all $s>0$. Let $I(s)$ be the union of all $E^{\pm}_i(s)$, $i=-n+1,\dots,n-1$, and $C^{\pm}(s)$. Observe that the size of $I(s)$ goes to 0 as $s$ goes to 0. Let $h$ be a height function configuration such that the following hold: 
    \begin{itemize}
        \item There are no points of discontinuity in $D(s)\setminus I(s)$. On every interval $(B^{\pm}_i\setminus I(s))\cap(\{j\}\times\RR)$, set the height to take value $g^{\pm}_i-(\pm j-n)$.  
        \item Every vertical interval $E^{+}_i(s)\cap(\{j\}\times\RR)$, $j=n+1,\dots,n+g^+_i$ has exactly one point $x^+_{i,j}$ of discontinuity, with $x^+_{i,j}$ strictly increasing in $j$ if $g_{i-1}^{\pm}>g_{i}^{\pm}$, and strictly decreasing otherwise. 
        \item On $C^+(s)$, the height function on the bottom of the rectangle is determined by the above and by $g$. Start with a height function $h_0$ which is constant on each column and equal to the values on the bottom of the rectangle. If this height function takes its maximal value on columns $J\subset\ZZ$, insert a point of discontinuity on these columns at vertical coordinate $n+y_1$ such that the height function drops by 2 above $n+y_1$. This induces a new height function $h_1$ constant on columns on the smaller rectangle $C^+(s)\cap(\ZZ\times[n+y_1,n+s])$. Repeat the same procedure with points of discontinuity with vertical coordinates $n+y_1, n+y_2,\dots, n+y_r$, with $0<y_1<\cdots<y_r<s$ until the height function $h_r$ only takes values $0$ or $1$.
        \item Do the same in the downwards direction for $C^-(s)$. 
    \end{itemize}
    Let $W_{\min}$ be the event that a height function $h$ so defined occurs in $(\ell_s)_{\inn}$. It is straightforward to check that on $W_{\min}$, $h|_{\ell_s}\in\{0,1\}$, that $\nu^{\sv,\bullet,g}_{\cyl_{L,\beta}\setminus\ell_{\inn}}[W_{\min}]>0$, and moreover one can prove by a direct computation that
    \bes
        \nu^{\sv,\bullet,g}_{\cyl_{L,\beta}\setminus\ell_{\inn}}[W_{\min}]
        \ge
        \frac{C|I(s)|^k}{k!}
        =
        \frac{C's^k}{k!},
    \ees
    where $k$ is the number of Poisson points which appear in $I(s)$ under $W_{\min}$. In fact, one can prove that $W_{\min}=U_k\cap \{ h|_{\ell_s}\in\{0,1\}\}$, where $U_{k}$ be the event that $I(s)$ contains exactly $k$ Poisson points. The value $k$ is the minimal number of Poisson points one can have in $I(s)$ in order to achieve $h|_{\ell_s}\in\{0,1\}$. For $j\ge0$ let $U_{\ge j}$ denote the event that $I(s)$ contains at least $j$ Poisson points (so $\{h|_{\ell_s}\in\{0,1\}\subset U_{\ge k}\}$). 
    Now
    \bes
    \begin{split}
        \nu^{\sv,\bullet,g}_{\cyl_{L,\beta}\setminus\ell_{\inn}}
        [h|_{\ell'_s} \in\{0,1\} \ | \ h|_{\ell'_s}\le1]
        &\ge
        \frac{
        \nu^{\sv,\bullet,g}_{\cyl_{L,\beta}\setminus\ell_{\inn}}
        [W_{\min}]
        }
        {
        \nu^{\sv,\bullet,g}_{\cyl_{L,\beta}\setminus\ell_{\inn}}
        [h|_{\ell'_s}\le1]
        }\\
        &=
        1-
        \frac{
        \nu^{\sv,\bullet,g}_{\cyl_{L,\beta}\setminus\ell_{\inn}}
        [U_{\ge k+1} \cap \{h|_{\ell'_s}\le1\}]
        }
        {
        \nu^{\sv,\bullet,g}_{\cyl_{L,\beta}\setminus\ell_{\inn}}
        [h|_{\ell'_s}\le1]
        }\\
        &\ge
        1-
        \frac{\sum_{j>k}\tfrac{C''s^j}{j!}
        }
        {
        \nu^{\sv,\bullet,g}_{\cyl_{L,\beta}\setminus\ell_{\inn}}
        [W_{\min}]
        }\\
        &\ge
        1-
        \frac{\sum_{j>k}\tfrac{C''s^j}{j!}
        }
        {
        \tfrac{C's^k}{k!}
        }
    \end{split}
    \ees
    which tends to 1 as $s\to0$. Here we used Lemma \ref{lem:q:pp-dom} in the second inequality. This completes the proof of Lemma \ref{lem:q-approx-h-circuit} and therefore the proof of \eqref{eq:q-disorders-estimate-part-2}, and so completes the proof of Lemma \ref{lem:q-disorders-estimate}. As noted above, this completes the proof of the convergence \eqref{eq:q-convergence}.
\end{proof}

Finally, we note that one can prove a mixing statement of the following form. For all $\Delta\in[-1,0]$ and for local observables $A,B$, we have:
\bes
    \langle A \cdot \tau_xB(t) \rangle
    \to 0
\ees
as $||(x,t)||\to\infty$. The proof is a combination of the approach in the discrete in Section \ref{sec:extremality} and the tools used above to prove convergence using the continuous models. We leave the details to the reader.

This completes the proof of Parts 1 and 2 of Theorem \ref{thm:main:quantum-loop} and of Part 2 of Theorem \ref{thm:main-xxz}.

\section{The space-time loop, eight-vertex and Ashkin-Teller models}\label{sec:quantum-AT-8v}

In the remaining sections we proceed to the second half of Theorem \ref{thm:main:quantum-loop}, the proof of Proposition \ref{prop:q:loop-main-part-2}. We mirror what we did in Section \ref{sec:new-coupling}, and take the delocalisation result, Theorem \ref{thm:Q:delocalisation}, and pass it through three couplings: to the space-time Ashkin-Teller model, then to the space-time eight-vertex model, and finally to Ueltschi's loop model. Again we do this in reverse order, starting with the space-time version of our new coupling, between the eight-vertex model and Ueltschi's loop model.

\subsection{The (self-dual) space-time Ashkin-Teller and eight-vertex models}

Both of these models are models of two interacting (space-time) Ising models. Recall $\cF^\bullet=(2\ZZ+\tfrac{1}{2})\times\RR$ and $\cF^\circ = (2\ZZ+\tfrac{3}{2})\times\RR$; for a domain $F$ let $F^\bullet=F\cap \cF^\bullet$ and $F^\circ=F \cap \cF^\circ$. 

Let $\cE^\bullet_\horiz$ be the set of line segments $e$ from $(x-(1,0),t)\in \cF^\bullet$ to $(x+(1,0),t)\in\cF^\bullet$ (think of this as horizontal edges in the black ``lattice''. Meanwhile define $\cE^\bullet_\vert$ to be simply $F^\bullet$ (think of this as the vertical edges in the black ``lattice''). Define $\cE^\circ_\horiz$, $\cE^\circ_\vert$ similarly. As in the discrete, there is a natural bijection between $\cE^\bullet_\horiz$ and $\cE^\circ_\vert$; for $e\in\cE^\bullet_\horiz$ the line segment from $(x-(1,0),t)\in \cF^\bullet$ to $(x+(1,0),t)\in\cF^\bullet$, let $e^*$ be the midpoint $(x,t)\in \cE^\circ_\vert$ (think of this as its ``dual edge''). Similarly $\cE^\circ_\horiz$ and $\cE^\bullet_\vert$ are in bijection.

For $F=(V,F)$ a domain, let $E^\bullet_\vert=V\cap\cE^\bullet_\vert$ and let $E^\circ_\horiz$ be the image of $E^\bullet_\vert$ under the bijection above. Similarly define $E^\circ_\vert$ and $E^\bullet_\horiz$. \\


For the formal definition of a space-time Ising model, see Section \ref{sec:delocalisation}. Let $\ul{\Om}_F^{\ising}$ be the set of functions $\sigma:F\to\{\pm1\}$ which are right continuous in each copy of $\RR$ and have finitely many points of discontinuity in any compact interval. 

For an Ising model configuration $\tau\in\ul{\Om}_{F^\bullet}^{\ising}$ on $F^\bullet$, define $\om[\tau]$ to be the domain walls of $\tau$, that is, $\om[\tau]=\om_\horiz[\tau]\cup\om_\vert[\tau]$, where $\om_\horiz[\tau]\subset E^\circ_\horiz$ are the edges whose midpoints are the points of discontinuity of $\tau$, and $\om_\vert[\tau]\subset E^\vert[\tau]$ are the points $(x,t)$ such that $\tau_{(x+(1,0),t)}\neq \tau_{(x-(1,0),t)}$. If $\tau\in\ul{\Om}_{F^\circ}^{\ising}$, define $\om[\tau]$ the same but with black and white exchanged.\\

We can now define the (self-dual) space-time Ashkin-Teller model. Configurations are pairs of Ising model configurations, both on the same space - in this case, the ``black faces'' $F^\bullet$. Write $\ul{\Om}^{\AT}_{F^\bullet} =\ul{\Om}_{F^\bullet}^{\ising} \times \ul{\Om}_{F^\bullet}^{\ising}$ and define the sigma algebra $\FF_{F^\bullet}^{\AT}$ to be that generated by the $w^{\#}$ topology \cite[Section A2.6]{daley-verejones-1}. For a domain $F$, we define the measure (with free boundary conditions) to be
\be\label{eq:q:AT-measure}
\begin{split}
    \d\nu^{\AT}_{F^\bullet,\Delta}(\tau,\tau')
    &\propto
    \d\nu^{\poi}_{E^\circ_\horiz,\frac{1}{2}}(\chi_1)
    \d\nu^{\poi}_{E^\circ_\horiz,\frac{1}{2}-u}(\chi_2)
    \mathbbm{1}_{\{\chi_1 = \om_\horiz[\tau]\triangle\om_\horiz[\tau']\}}
    \mathbbm{1}_{\{\chi_2 = \om_\horiz[\tau]\cap\om_\horiz[\tau']\}}\\
    &\qquad\qquad\cdot
    e^{
    2u|\om_\vert[\tau]\triangle\om_\vert[\tau']|
    }
    e^{
    2|\om_\vert[\tau]^c\cap\om_\vert[\tau']^c|
    },
\end{split}
\ee
where the sum is over $\tau,\tau'\in\ul{\Om}^{\ising}_{F^\bullet}$, satisfying the given conditions (almost surely a finite sum). When we write $\om_\vert[\tau]^c$ we mean $E^\circ_\vert\setminus\om_\vert[\tau]$.

We define the model $\nu^{\AT,+,\f}_{F,\Delta}$ to be the model above, conditioned so that $\tau=+1$ on $\partial F^\bullet=\partial F \cap F^\bullet$, and the models $\nu^{\AT,\f,+}_{F,\Delta}$ and $\nu^{\AT,+,+}_{F,\Delta}$ similarly (the events one conditions on have positive probability). Note that when $\Delta=0$ ($u=\tfrac{1}{2}$), $\tau$ and $\tau'$ are independent (critical Ising model spins on $F^\bullet$).\\

We also define the space-time (self-dual) eight-vertex model. Configurations are pairs $(\pi^\bullet,\pi^\circ)$ where $\pi^\bullet$ is an Ising model configuration on the ``black faces'' $F^\bullet$ and $\pi^\circ$ an Ising model configuration on the ``white faces'' $F^\circ$. Write $\ul{\Om}^{\eiv}_{F^\bullet} =\ul{\Om}_{F^\bullet}^{\ising} \times \ul{\Om}_{F^\circ}^{\ising}$ and define the sigma algebra $\FF_{F^\bullet}^{\eiv}$ to be that generated by the $w^{\#}$ topology \cite[Section A2.6]{daley-verejones-1}. As in the discrete, it is the dual in one spin of the Ashkin-Teller model. The measure is
\be\label{eq:q:8v-measure}
\begin{split}
    \d\nu^{\eiv}_{F,\Delta}(\pi^\bullet,\pi^\circ)
    \propto
    \d\nu^{\poi}_{V,u}(\chi_1)
    \d\nu^{\poi}_{V,1-u}(\chi_2)
    \cdot
    \mathbbm{1}_{\{\chi_1=\om[\pi^\bullet]\cap\om[\pi^\circ]\}}
    \mathbbm{1}_{\{\chi_2=\om_\horiz[\pi]\setminus\om_\vert[\pi] \}}
    e^{
    -|\om_\vert[\pi]|
    },
\end{split}
\ee
where the sum is over $\pi^\bullet\in\ul{\Om}^{\ising}_{F^\bullet}$, $\pi^\circ\in\ul{\Om}^{\ising}_{F^\circ}$ satisfying the given conditions (this is almost surely a finite sum).

We define the model $\nu^{\eiv,\bplus}_{F,\Delta}$ to be the model above, conditioned so that $\pi^\bullet=+1$ on $\partial F^\bullet$, and the models $\nu^{\eiv,\wplus}_{F,\Delta}$ and $\nu^{\eiv,\bplus\wplus}_{F,\Delta}$ similarly (the events one conditions on have positive probability). 
Note that when $\Delta=0$ ($u=\tfrac{1}{2}$), $\pi^\bullet$ and $\pi^\circ$ are independent (critical Ising model spins on $F^\bullet$, $F^\circ$ respectively).\\

\begin{theorem}[The space-time Ashkin-Teller and eight-vertex models]\label{thm:q:8v-AT}
    Let $\Delta\in[-1,0]$, so $u\in[0,\tfrac{1}{2}]$.
    \begin{enumerate}
        \item In the self-dual space-time Ashkin-Teller model with parameter $u$, the measures $\nu^{\AT}_{F^\bullet,u}$, $\nu^{\AT;+,f}_{F^\bullet,u}$, $\nu^{\AT;f,+}_{F^\bullet,u}$, $\nu^{\AT;+,+}_{F^\bullet,u}$ converge weakly in $(\ul{\Om}^{\AT}_{\cF}, \FF^{\AT,\#}_{\cF})$ as $F^\bullet\nearrow\cF^\bullet$ to a common limit $\nu^{\AT}_{\Delta}$, which is Gibbs, $\cF^\bullet$ translation-invariant and ergodic, and extremal, and satisfies 
        \be\label{eq:q:AT-decay-correlations}
        \begin{split}
            \lim_{||x_1-x_2||\to\infty}& \nu^\AT_{\Delta}[\tau_{x_1}\tau_{x_2}]=0.\\
        \end{split}
        \ee
        Further if $u=\tfrac{1}{2}$, then $\nu^{\AT}_{\Delta}$ is the unique Gibbs measure and the unique thermodynamic limit, that is, for all $(\tau_0,\tau'_0)\in\ul{\Om}^\AT_{\cF^\bullet}$, we have the weak convergence
        \be
        	\lim_{F^\bullet \nearrow \cF^\bullet} 
            \nu^{\AT;(\tau_0,\tau'_0)}_{F^\bullet,u} = \nu_{\Delta}^\AT.
        \ee
        \item In the (self-dual) space-time eight-vertex model with parameter $\Delta$, the measures $\nu^{\eiv}_{F,\Delta}$,   $\nu^{\eiv;\bplus}_{F,\Delta}$, $\nu^{\eiv;\wplus}_{F,\Delta}$, $\nu^{\eiv;\bplus,\wplus}_{F,\Delta}$ converge weakly in $(\ul{\Om}^{\eiv}_{\cF}, \FF^{\eiv,\#}_{\cF})$ as $V\nearrow\ZZ^2$ to a common limit $\nu^{\eiv}_{\Delta}$, which is Gibbs, $\ZZ\times\RR$ translation-invariant and ergodic, and extremal, and satisfies 
        \be\label{eq:q:8v-decay-correlations}
        \begin{split}
            \lim_{||x_1-x_2||\to\infty}& \nu^\eiv_{\Delta}[\pi^\bullet_{x_1}\pi^\bullet_{x_2}]=0;\\
            \lim_{\substack{||x_1-x_2||\to\infty \\ ||y_1-y_2||\to\infty}}&
            \nu^\eiv_{\Delta}[\pi^\bullet_{x_1}\pi^\bullet_{x_2}
            \pi^\circ_{y_1}\pi^\circ_{y_2}]=0.
        \end{split}
        \ee
        Further if $u=\tfrac{1}{2}$, then $\nu^{\eiv}_{\Delta}$ is the unique Gibbs measure and the unique thermodynamic limit, that is, for all $\pi'\in\ul{\Om}^\eiv_{\ZZ\times\RR}$, we have the weak convergence
        \be
        	\lim_{V \nearrow \ZZ\times\RR} \nu^{\eiv;\pi'}_{F,\Delta} = \nu_{\Delta}^\eiv.
        \ee
    \end{enumerate}
\end{theorem}

\subsection{The new coupling: space time loop and eight-vertex models}

\begin{proof}[Proof of Proposition \ref{prop:q:loop-main-part-2}]

The space-time loop model \eqref{eq:quantum-loop-measure} can be coupled with the space-time eight-vertex model as in the discrete coupling of Proposition  \ref{prop:mir-8v-coupling}. Given a $\pi\sim \nu^{\eiv}_{F,\Delta}$, one samples a loop configuration $m$ as follows. Let $A(\pi)\subset E_\horiz$ be the points not in $\om_{\horiz}[\om]$ where both endpoints of the horizontal edge lie in $\om_{\vert}[\pi]$ or both do not. Sample a Poisson point process $Y$ of rate 1 on $A(\pi)$. Then:

\be\label{eq:q:sample-loop-from-8v}
\begin{cases}
        x \in m_{\cross} \ \mathrm{deterministically}  
        & \mathrm{if} \ x\in\om_{\horiz}[\pi], \ 
        \pi_{x+(1,0)}\neq\pi_{x-(1,0)}\\
        x \in m_{\dbar} \ \mathrm{deterministically}  
        & \mathrm{if} \ x\in\om_{\horiz}[\pi], \ 
        \pi_{x+(1,0)}=\pi_{x-(1,0)}\\
        x \in m_{\cross} (\mathrm{resp.} \ m_{\dbar}) \ \mathrm{w.p.} \ u \ (\mathrm{resp.} \ 1-u)  
        & \mathrm{if} \ x\in Y\\
        x \notin m & \mathrm{otherwise}
\end{cases}
\ee

\begin{proposition}[New coupling of space-time loop and eight-vertex models]
Let $\Delta\in[-1,1)$, so $u=\tfrac{1}{2}(\Delta+1)\in[0,1)$. 
\begin{enumerate}
        \item Take $m\sim \nu^{\lp}_{F,\Delta}$ and independently colour the loops red with probability $\tfrac{1}{2}$ each. Let $\pi=(\pi^\bullet,\pi^\circ)\in \Om^{\eiv}_{V}$ be defined by setting $\pi_{f_0}=\pm 1$ independently and uniformly at some fixed point $f_0\in F$, and then let the value of $\pi$ change if and only if one crosses a red loop. Then $\pi\sim \nu^{\eiv}_{F,\Delta}$. 
        \item Take $\pi\sim \nu^{\eiv}_{F,\Delta}$ and sample a loop configuration $m$ according to the rules \eqref{eq:q:sample-loop-from-8v}. Then $m\sim\nu^{\mir}_{V,\ul{p}}$. 
        \item There exists an analogous coupling between $\nu^{\eiv,\bplus,\wplus}_{F,\Delta}$ and $\nu^{\lp}_{F}$, the only difference being that one sets $\pi^\bullet\equiv \pi^\circ\equiv +1$ on $\partial F$ instead of the random $\pi_{f_0}$. Similarly, there exists an analogous coupling between $\nu^{\eiv,\bplus}_{F,\Delta}$ and $\nu^{\lp,\bullet}_{F,\Delta}$, where one fixes only $\pi^\bullet\equiv1$ on $\partial F^\bullet$ and in sampling $m$ one fixes $\partial F\subset m^\bullet$; similar for $\nu^{\eiv,\wplus}_{F,\Delta}$ and $\nu^{\lp,\circ}_{F,\Delta}$.
    \end{enumerate}
\end{proposition}

The proof of the coupling is straightforward and analogous to the proof of Proposition \ref{prop:mir-8v-coupling} in the discrete. Also as in the discrete case, we obtain the equality
\be\label{eq:q:loop-connections}
\begin{split}
    \nu^{\mathrm{loop}}_{V,\Delta}[e_1\leftrightarrow e_2]
    &=
    \nu^{\eiv}_{V,\Delta}[
    \pi^\bullet_{x_1}\pi^\circ_{y_1}
    \pi^\bullet_{x_2}\pi^\circ_{y_2}
    ]
    ,
\end{split}
\ee
where $e_1,e_2\in E$ and $x_i,y_i$ are the two points of $F^\bullet$ $F^\circ$ respectively lying either side of $e_i$. The proof is identical to that of Lemma \ref{lem:8v-connection-probs}.

The rest of the proof proceeds as in the discrete setting. Take the infinite volume eight-vertex measure $\nu^\eiv_{K,W}$ described in Theorem \ref{thm:q:8v-AT}. Sample loops according to the procedure in \eqref{eq:q:sample-loop-from-8v}. One can prove that the marginal $\mu$ on mirrors must be $\nu^{\lp}_{\Delta}$. Indeed, use the coupling above and the local sampling of mirrors given spins and apply the same proof as in the proof of Proposition \ref{prop:q:loop-conv-from-6v} to show that $\nu^\lp_{F,\Delta}$, $\nu^{\lp;\bullet}_{F,\Delta}$, and $\nu^{\lp;\circ}_{F,\Delta}$ all converge to $\mu$. By Proposition \ref{prop:q:loop-conv-from-6v}, these measures already converge to $\nu^\lp_{\Delta}$.\\

Combining \eqref{eq:q:loop-connections} and \eqref{eq:q:8v-decay-correlations} and a similar proof to \eqref{eq:mir-connections-finite-inf-vol}, one obtains \eqref{eq:q:loop-connections-decay-to-0-repeat}.


It remains to prove that for $\Delta=0$, $\nu_{\Delta}^\lp$ is the unique Gibbs measure and the unique thermodynamic limit. The proof again is identical to in the discrete: for $m_0\in\Om^\lp_{\ZZ\times\RR}$, one considers $\nu^{\lp;m_0}_{F,\Delta}$ and forms a coupling with an eight-vertex model by independently and uniformly colouring the loops one of two colours. This marginal on the eight-vertex model must converge to the unique thermodynamic limit $\nu^\eiv_\Delta$ as $V\nearrow\ZZ\times\RR$ by Theorem \ref{thm:q:8v-AT}, and since any local event $A$ in the loop model is a local function of the eight-vertex configuration, we have $\nu^{\lp;m_0}_{F,\Delta}[A]\to\nu^\lp_\Delta[A]$.

\end{proof}


\subsection{Remaining couplings for space-time models: eight-vertex, Ashkin-Teller and six-vertex models}

It now remains to prove the results for the space-time eight-vertex and Ashkin Teller models, Theorem \ref{thm:q:8v-AT}. The structure of this section mirrors that of Section \ref{sec:8v-AT}.\\

We first define the space-time version of the (self-dual) Ashkin-Teller random-cluster model. Configurations are given by a pair of bond percolation configurations $(\eta^\bullet,\eta^\circ)$ - $\eta^\bullet$ on $E^\bullet$ and $\eta^\circ$ on $E^\circ$. Recall the definitions of $\eta^\bullet_\vert$ and $\eta^\bullet_\horiz$; we define $\eta_\vert=\eta^\bullet_\vert\cup\eta^\circ_\vert$ and similarly $\eta_\horiz=\eta^\bullet_\horiz\cup\eta^\circ_\horiz$. The measure with free boundary conditions is given by:
\be\label{eq:q:8vRC-measure}
\begin{split}
    \d\nu^{\eivRC,0,0}_{F,\Delta}(\eta^\bullet,\eta^\circ)
    &\propto
    \d\nu^{\poi}_{F,u}(\chi_1)
    \d\nu^{\poi}_{F,2u}(\chi_2)
    \d\nu^{\poi}_{F,1-2u}(\chi_3)\\
    &\qquad\cdot
    \mathbbm{1}_{\{\chi_1=(\eta^\bullet)^*\cap(\eta^\circ)^*\}}
    \mathbbm{1}_{\{\chi_2=\eta^\bullet\cap\eta^\circ\}}
    \mathbbm{1}_{\{\chi_3\cap E^\bullet_\vert=\eta^\circ_\horiz\setminus\eta^\bullet_\vert\}}
    \mathbbm{1}_{\{\chi_3\cap E^\circ_\vert=\eta^\bullet_\horiz\setminus\eta^\circ_\vert\}}\\
    &\qquad\cdot 
    e^{-uL(\eta_\vert\setminus\eta_\horiz)}2^{k(\eta^\bullet)+k(\eta^\circ)}.
\end{split}
\ee
Note that the factor $e^{-uL(\eta_\vert\setminus\eta_\horiz)}$ is cosmetic, since the length of this set is the entire volume of the domain $F$, as it is the only part of the process which is not a Poisson point process. We leave it in as it makes the couplings later more transparent. 

Define the measure $\nu^{\eivRC,1,0}_{F,\Delta}$ by redefining $k(\om^\bullet)$ so that all clusters touching the boundary are counted as one, and define  $\nu^{\eivRC,0,1}_{F,\Delta}$, and $\nu^{\eivRC,1,1}_{F,\Delta}$ similarly.

\begin{proposition}\label{prop:q:8vRC}
    Let $\Delta\in[-1,0]$, so $u=\tfrac{1}{2}(\Delta+1)\in[0,\tfrac{1}{2}]$.
        The self-dual space-time eight-vertex Random Cluster model with parameter $u$ has a unique Gibbs measure $\nu^{\eivRC}_\Delta$ on $\ZZ\times\RR$, which is also $\ZZ\times\RR$ translation-invariant and ergodic, and extremal. It is the weak limit of the measures $\nu^{\eivRC;\clubsuit}_{F,\Delta}$, for $\clubsuit\in\{0,1\}^2$ as $V\nearrow\ZZ\times\RR$, and is supported on configurations where $\eta^\bullet$, $\eta^\circ$, $(\eta^\bullet)^*$, and $(\eta^\circ)^*$ have no infinite clusters. A similar statement holds for the ordinary Ashkin-Teller Random Cluster model.
\end{proposition}

Recall the orderings \eqref{eq:orderings-8vRC}, that is, that 
$\eta_1\le\eta_2$ if $\eta_1^\bullet(e)\le \eta_2^\bullet(e)$ and $\eta_1^\circ(e)\le \eta_2^\circ(e)$ and $\eta_1\tilde{\le}\eta_2$ if 
$\eta_1^\bullet(e)\le \eta_2^\bullet(e)$ and $\eta_1^\circ(e)\ge \eta_2^\circ(e).$

\begin{lemma}\label{lem:q:FKG-8vRC}
    Let $\Delta\in[-1,0]$, so $u\in[0,\tfrac{1}{2}]$. For all boundary conditions $\eta_0$, we have under the measure $\nu_{F,\Delta}^{\eivRC,\eta_0}$, $\eta$ is FKG with respect to the ordering $\tilde{\le}$. In the special case $u=\tfrac{1}{2}$, $\eta^\bullet$ and $\eta^\circ$ are independent (they are both quantum random cluster models with $q=2$), and also are FKG with respect to the ordering $\le$. 
\end{lemma}

Again there are two ways one can prove this - either from first principles, or by seeing that the random cluster measure is a limit of discrete eight-vertex random cluster measures. As with the FKG for the six-vertex percolation model $\xi$, we take the latter approach; see Lemmas \ref{lem:appendix:FKG-proofs} and \ref{lem:appendix-convergence}. 

As in the discrete, the FKG property gives the convergence to infinite volume measures: for $\Delta\in[-1,0]$, $\nu_{F,\Delta}^{\eivRC,1,0} \to \nu_{\Delta}^{\eivRC,1,0}$ and $\nu_{F,\Delta}^{\eivRC,0,1} \to \nu_{\Delta}^{\eivRC,0,1}$, and the stochastic domination 
\bes
    \nu_\Delta^{\eivRC,0,1} \ \tilde{\prec} \ \nu \ \tilde{\prec} \ \nu_\Delta^{\eivRC,1,0}
\ees
for any infinite volume Gibbs measure $\nu$, where $\tilde{\prec}$ is stochastic domination with respect to the ordering $\tilde{\le}$. 

We will prove the following analogy of \eqref{eq:8v-8vRC-connections}:
\be\label{eq:q:8v-8vRC-connections}
	\nu^{\eivRC;1,0}_{\Delta}[x\xleftrightarrow{\eta^\bullet} \infty] 
	=
	\nu^{\eiv;\bplus}_{\Delta}[\pi^\bullet_x]
	=
	\nu^{\AT;+,+}_{\Delta}[\tau_x]
    =0
\ee
which by the above equality and standard percolation arguments gives that $\nu_\Delta^{\eivRC,0,1} = \nu_\Delta^{\eivRC,1,0}$ and uniqueness of the random cluster Gibbs measure. \\

We now couple the space-time eight-vertex, Ashkin-Teller and ATRC measures in analogy with Lemma \ref{lem:8v-AT-coupling}. 
Given $\pi\sim\nu^\eiv_{F,\Delta}$, let $\eta$ be sampled as follows.
\be\label{eq:q:sample-8vRC-from-8v}
\begin{split}
&\text{On $E^\bullet_\vert$:}
\begin{cases}
    \text{set} \ (\eta^\bullet,\eta^\circ)=(0,0) \ \text{deterministically} 
    \ & \ \text{on $\om_\horiz[\pi^\bullet]\cap\om_\vert[\pi^\circ]$},\\
    \text{set} \ (\eta^\bullet,\eta^\circ)=(0,0) \ \text{w.p. $\tfrac{u}{1-u}$}, 
    \ (\eta^\bullet,\eta^\circ)=(0,1) \ \text{o.w.}
    \ & \ \text{on $\om_\horiz[\pi^\bullet]\setminus\om_\vert[\pi^\circ]$},\\
    \text{sample $(\eta^\bullet,\eta^\circ)=(0,0)$ as a PPP rate $u$}, \\
    \qquad \text{set} \ (\eta^\bullet,\eta^\circ)=(1,0) \ \text{o.w.}
    \ & \ \text{on $\om_\horiz[\pi^\circ]\setminus\om_\vert[\pi^\bullet]$},\\
    \text{sample $(\eta^\bullet,\eta^\circ)=(0,0)$ as a PPP rate $u$}, \\
    \qquad \text{sample $(\eta^\bullet,\eta^\circ)=(0,1)$ as a PPP rate $1-2u$}, \\
    \qquad \text{sample $(\eta^\bullet,\eta^\circ)=(1,1)$ as a PPP rate $2u$}, \\
    \qquad \text{set} \ (\eta^\bullet,\eta^\circ)=(1,0) \ \text{o.w.}
    \ & \ \text{on $\om_\horiz[\pi^\bullet]^c\cap\om_\vert[\pi^\circ]^c$};
\end{cases}\\
&\text{On $E^\circ_\vert$:}
\begin{cases}
    \text{do the same with black and white exchanged.} 
\end{cases}
\end{split}
\ee
Here again w.p. stands for ``with probability'' and o.w. stands for ``otherwise'', and PPP stands for ``Poisson point process''. For $(\tau,\tau')\sim\nu^\AT_{G^\bullet,\Delta}$, sample $(\om,\om')$ as follows. On $E^\bullet_\vert$:
\be\label{eq:q:sample-ATRC-from-AT-1}
\begin{split}
&\text{On $E^\bullet_\vert$:}
\begin{cases}
    \text{set} \ (\om,\om')=(0,0) \ \text{deterministically} 
    \ & \ \text{on $\om_\horiz[\tau]\cap\om_\horiz[\tau']$},\\
    \text{set} \ (\om,\om')=(1,0) \ \text{w.p. $\tfrac{u}{2}$}, 
    \ (\om,\om')=(0,0) \ \text{o.w.}
    \ & \ \text{on $\om_\horiz[\tau']\setminus\om_\horiz[\tau]$},\\
    \text{set} \ (\om,\om')=(0,1) \ \text{w.p. $\tfrac{u}{2}$}, 
    \ (\om,\om')=(0,0) \ \text{o.w.}
    \ & \ \text{on $\om_\horiz[\tau]\setminus\om_\horiz[\tau']$},\\
    \text{sample $(\om,\om')=(0,0)$ as a PPP rate $\tfrac{1}{2}-u$}, \\
    \qquad \text{sample $(\om,\om')=(1,0)$ as a PPP rate $u$}, \\
    \qquad \text{sample $(\om,\om')=(0,1)$ as a PPP rate $u$}, \\
    \qquad \text{set} \ (\om,\om')=(1,1) \ \text{o.w.}
    \ & \ \text{on $\om_\horiz[\tau]^c\cap\om_\horiz[\tau']^c$};
\end{cases}\\
& \text{On $E^\circ_\vert$}:
\begin{cases}
    \text{set} \ (\om,\om')=(0,0) \ \text{deterministically} 
    \ & \ \text{on $\om_\vert[\tau]\cap\om_\vert[\tau']$},\\
    \text{sample $(\om,\om')=(1,0)$ as a PPP rate $2u$}, \\
    \qquad \text{set} \ (\om,\om')=(0,0) \ \text{o.w.}, 
    \ & \ \text{on $\om_\vert[\tau']\setminus\om_\vert[\tau]$},\\
    \text{sample $(\om,\om')=(0,1)$ as a PPP rate $2u$}, \\
    \qquad \text{set} \ (\om,\om')=(0,0) \ \text{o.w.}, 
    \ & \ \text{on $\om_\vert[\tau]\setminus\om_\vert[\tau']$},\\
    \text{sample $(\om,\om')=(1,0)$ as a PPP rate $2u$}, \\
    \qquad \text{sample $(\om,\om')=(0,1)$ as a PPP rate $2u$}, \\
    \qquad \text{sample $(\om,\om')=(1,1)$ as a PPP rate $2-4u$}, \\
    \qquad \text{set} \ (\om,\om')=(0,0) \ \text{o.w.}
    \ & \ \text{on $\om_\vert[\tau]^c\cap\om_\vert[\tau']^c$}.
\end{cases}
\end{split}
\ee

\begin{lemma}[Coupling eight-vertex, Ashkin-Teller and ATRC]\label{lem:q:coupling-8v-AT}
Let $\Delta\in[-1,0]$, so $u=\tfrac{1}{2}(\Delta+1)\in[0,\tfrac{1}{2}]$.
\begin{enumerate}
    \item If one takes $\eta\sim\nu_{F,\Delta}^{\eivRC,0,0}$ and assigns spins $+$ or $-$ independently and uniformly to every cluster of $\eta^\bullet$ and $\eta^\circ$, the resulting spin configuration $\pi$ has law $\nu_{F,\Delta}^{\eiv}$. If one assigns spins independently and uniformly to the clusters of $\eta^\bullet$ and $(\eta^\circ)^*$, with boundary clusters of $(\eta^\circ)^*$ always getting $+1$, the resulting spin configuration $(\tau,\tau')$ has law $\nu_{F,\Delta}^{\AT,\f,+}$. 
    \item If one takes $\pi\sim\nu_{F,\Delta}^{\eiv}$ and samples a percolation configuration $\eta$ according to the rule \eqref{eq:q:sample-8vRC-from-8v}, or if one takes $(\tau,\tau')\sim\nu_{F,\Delta}^{\AT,\f,+}$ and samples a percolation configuration $(\om,\om')$ according to the rule \eqref{eq:q:sample-ATRC-from-AT-1} and sets $(\eta^\bullet,\eta^\circ)=(\om,(\om')^*)$, then $\eta$ has law $\nu_{F,\Delta}^{\eivRC,0,0}$.
    \item Analogous statements holds for the triples of measures: $\nu_{F,\Delta}^{\eivRC,0,1}$, $\nu_{F,\Delta}^{\eiv,\wplus}$, $\nu_{F,\Delta}^{\AT,\f,\f}$, and 
    $\nu_{F,\Delta}^{\eivRC,1,0}$, $\nu_{F,\Delta}^{\eiv,\bplus}$, $\nu_{F,\Delta}^{\AT,+,+}$, and
    $\nu_{F,\Delta}^{\eivRC,1,1}$, $\nu_{F,\Delta}^{\eiv,\bplus,\wplus}$, $\nu_{F,\Delta}^{\AT,+,\f}$.
    \item Further,
    \be\label{eq:q:8v-8vRC-connections-full}
    \begin{split}
    \nu^{\eivRC,1,0}_{\Delta}[x\xleftrightarrow{\eta^\bullet} \infty] 
	&=
	\nu^{\eiv,\bplus}_{\Delta}[\pi^\bullet_x]
	=
	\nu^{\AT,+,+}_{\Delta}[\tau_x].
    \end{split}
    \ee
\end{enumerate}
\end{lemma}

\begin{proof}
It is helpful to write the eight-vertex and Ashkin-Teller measures as follows:
\be\label{eq:q:8v-measure-renorm}
\begin{split}
    \d\nu^{\eiv}_{F,\Delta}(\pi^\bullet,\pi^\circ)
    \propto
    \d\nu^{\poi}_{F,u}(\chi_1)
    \d\nu^{\poi}_{F,1-u}(\chi_2)
    \cdot
    \mathbbm{1}_{\{\chi_1=\om[\pi^\bullet]\cap\om[\pi^\circ]\}}
    \mathbbm{1}_{\{\chi_2=\om_\horiz[\pi]\setminus\om_\vert[\pi]\}}
    e^{
    |F\setminus\om_\vert[\pi]|
    },
\end{split}
\ee
and
\be\label{eq:q:AT-measure-renorm}
\begin{split}
    \d\nu^{\AT}_{F^\bullet,\Delta}(\tau,\tau')
    &\propto
    \d\nu^{\poi}_{F^\bullet,\frac{1}{2}}(\chi_1)
    \d\nu^{\poi}_{F^\bullet,\frac{1}{2}-u}(\chi_2)
    \mathbbm{1}_{\{\chi_1 = \om_\horiz[\tau]\triangle\om_\horiz[\tau']\}}
    \mathbbm{1}_{\{\chi_2 = \om_\horiz[\tau]\cap\om_\horiz[\tau'] \}}
     \\
    &\qquad\cdot
    e^{
    \tfrac{1}{2}|\om_\horiz[\tau]^c\cap\om_\horiz[\tau']^c|
    }
    e^{
    u|\om_\vert[\tau]\triangle\om_\vert[\tau']|
    }
    e^{
    -u|\om_\vert[\tau]\cap\om_\vert[\tau']|
    }
    e^{
    (2-u)|\om_\vert[\tau]^c\cap\om_\vert[\tau']^c|
    }
    ,
\end{split}
\ee
which are obtained from \eqref{eq:q:8v-measure} and \eqref{eq:q:AT-measure} by renormalising. It is a straightforward exercise (and very similar to other Edwards-Sokal type couplings) to prove that if $\pi\sim\nu^{\eiv}_{F,\Delta}$ \eqref{eq:q:8v-measure} and one samples $\eta$ according to \eqref{eq:q:sample-8vRC-from-8v}, then $\eta$ is distributed according to $\nu^{\eivRC}_{F,\Delta}$ \eqref{eq:q:8vRC-measure}. Further it is straightforward to prove that taking $\eta\sim\nu^{\eivRC}_{F,\Delta}$ and assigning spins $\pi$ to its clusters independently and uniformly gives $\pi\sim\nu^{\eiv}_{F,\Delta}$. 

For a space-time percolation $\om$ on a domain $F$, we have the Euler relation
\be\label{eq:euler}
    k(\om) + |\om_\horiz| - |I_{F}| = 
    k(\om^*) + |\om^*_\horiz| - |I_{F^*}|,
\ee
where $|I_{F}|$ is the number of intervals making up $F$ and $|I_{F^*}|$ the number making up $F^*$. Similar relations can be found in \cite[Lemma 3.3]{aiz-nacht} and \cite[Equation 2.4.4]{bjornberg-thesis}. Using this, one can show that if $\eta\sim\nu^{\eivRC}_{F,\Delta}$ then $(\om,\om')=(\eta^\bullet,(\eta^\circ)^*)$ is distributed as:
\be\label{eq:q:8vRC-measure-2}
\begin{split}
    \d\nu^{\eivRC,0,0}_{F,\Delta}(\eta^\bullet,\eta^\circ)
    &=
    \d\nu^{\poi}_{E^\bullet_\vert,\tfrac{1}{2}-u}(\chi_1)
    \d\nu^{\poi}_{E^\bullet_\vert,u}(\chi_2)
    \d\nu^{\poi}_{E^\bullet_\horiz,2u}(\chi_3)
    \d\nu^{\poi}_{E^\bullet_\horiz,2-4u}(\chi_3)\\
    &\qquad\cdot
    \mathbbm{1}_{\{\chi_1=E^\bullet_\vert\setminus(\om\cup\om')\}}
    \mathbbm{1}_{\{\chi_2=\om_\vert\triangle\om'_\vert\}}
    \mathbbm{1}_{\{\chi_3=\om_\horiz\triangle\om'_\horiz\}}
    \mathbbm{1}_{\{\chi_4=\om_\horiz\cap\om'_\horiz\}}\\
    &\qquad\cdot 
    e^{-u| \om_\vert\cap\om'_\vert|}
    e^{-u|E^\bullet_\horiz\setminus(\om\cup\om')|}
    2^{k(\om)+k(\om')}.
\end{split}
\ee
The rest of the proof of the lemma is a similar Edwards-Sokal type exercise, with part 4 being a standard colour-switching argument.
\end{proof}

Finally we do the same for the Ashkin-Teller and six-vertex models. We have already described how to sample $\xi^\bullet$ from a six-vertex configuration $\xi$ in \eqref{eq:quantum-6v-triple}, so we describe how to sample it from the Ashkin-Teller model. Given $(\tau,\tau')\sim\nu^\AT_{G^\bullet,\Delta}$, sample $\xi^\bullet\in\Om^\perc_{E^\bullet}$ as follows.
\be\label{eq:q:sample-xi-from-AT}
\begin{split}
&\text{On} \ E^\bullet_\vert:
\begin{cases}
    \text{set} \ \xi^\bullet=0 \ \text{deterministically} 
    \ & \ \text{on $\om_\horiz[\tau]\cup\om_\horiz[\tau'] $},\\
    \text{sample $\xi^\bullet=0$ as a PPP rate $\tfrac{1}{2}-u$}, 
    \ \text{set} \ \xi^\bullet=1 \ \text{o.w.}
    \ & \ \text{on $\om_\horiz[\tau]^c\cap\om_\horiz[\tau']^c$};
\end{cases}\\
&\text{On} \ E^\circ_\vert:
\begin{cases}
    \text{set} \ \xi^\bullet=0 \ \text{deterministically} 
    \ & \ \text{on $\om_\vert[\tau]\cup\om_\vert[\tau'] $},\\
    \text{sample $\xi^\bullet=1$ as a PPP rate $2$}, 
    \ \text{set} \ \xi^\bullet=0 \ \text{o.w.}
    \ & \ \text{on $\om_\vert[\tau]^c\cap\om_\vert[\tau']^c$}.
\end{cases}
\end{split}
\ee

\begin{lemma}[Coupling Ashkin-Teller and six-vertex]\label{lem:q:coupling-AT-6v}
Let $\Delta\in[-1,0]$, so $u=\tfrac{1}{2}(\Delta+1)\in[0,\tfrac{1}{2}]$.
\begin{enumerate}
    \item If one takes $(\sigma^\bullet,\sigma^\circ,\xi^\bullet)\sim\nu_{F,\Delta}^\sv$ as in Lemma \ref{lem:q:quantum-6v-triple} and samples a spin configuration $\tau$ by assigning $+$ or $-$ independently and uniformly to every cluster of $(\xi^\bullet)$, then $(\tau,\tau\sigma^\bullet)\sim\nu_{F,\Delta}^{\AT,+,\f}$, that is, $\tau$ has the law of a single spin of Ashkin-Teller, while $\sigma^\bullet$ has the law of the product of Ashkin-Teller spins.
    \item If one takes $(\tau,\tau')\sim\nu_{F,\Delta}^{\AT,+,\f}$ and samples a percolation configuration $\xi^\bullet$ according to the rule \eqref{eq:q:sample-xi-from-AT}, and then further samples a spin configuration $\sigma^\circ$ by assigning $\pm1$ to each cluster of $(\xi^\bullet)^*$, then $(\tau\tau',\sigma^\circ)\sim\nu^\sv_{F,\Delta}$.
    \item Analogous statements holds for the triples of measures: $\nu_{F,\Delta}^{\eivRC,0,1}$, $\nu_{F,\Delta}^{\eiv,\wplus}$, $\nu_{F,\Delta}^{\AT,\f,\f}$, and 
    $\nu_{F,\Delta}^{\eivRC,1,0}$, $\nu_{F,\Delta}^{\eiv,\bplus}$, $\nu_{F,\Delta}^{\AT,+,+}$, and
    $\nu_{F,\Delta}^{\eivRC,1,1}$, $\nu_{F,\Delta}^{\eiv,\bplus,\wplus}$, $\nu_{F,\Delta}^{\AT,+,\f}$.
    \item One has that in infinite volume,
    \bes
    	\nu^{\AT;++}_{\Delta}[\tau_x] 
    	=
    	\nu^{\sv}_{\Delta}[x \xleftrightarrow{\xi^\bullet} \infty]
    	=
    	0.
    \ees
\end{enumerate}
\end{lemma}

\begin{proof}
    Again the proof is an straightforward exercise; we note the helpful identities: the marginal of $\nu^\sv_{F,\Delta}$ on $\sigma^\bullet,\xi^\bullet$ can be written as:
    \be\label{eq:quantum-6v-double}
    \begin{split}
        \d\nu^\sv_{F,\Delta}(\sigma^\bullet,\xi^\bullet) = 
        &
        \d\nu^\poi_{E^\bullet_\vert,1}(\chi_1)
        \d\nu^\poi_{E^\bullet_\vert,1-2u}(\chi_2) 
        \d\nu^\poi_{E^\circ_\vert,1}(\chi_3)\\
        & \qquad\cdot
        \mathbbm{1}_{\{\sigma^\bullet\perp\xi^\bullet\}}
        \mathbbm{1}_{\{\chi_1=\om_{\horiz}[\sigma^\bullet]\}}
        \mathbbm{1}_{\{\chi_2=(\xi^\bullet)^*_\horiz=E^\bullet_\vert\setminus \xi^\bullet\}}
        \mathbbm{1}_{\{\chi_3=\xi^\bullet_{\horiz}\}}\\
        & \qquad \cdot
        e^{u\left|\xi^\bullet_\vert\right|}
        e^{u\left|\om_\vert[\sigma^\bullet]\right|}
        e^{-u\left|E^\circ_\vert\setminus
        (\om_\vert[\sigma^\bullet]\cup\xi^\bullet_\horiz)\right|}
        2^{k((\xi^\bullet)^*)}
    \end{split}
    \ee
    which, using the Euler relation \eqref{eq:euler}, can be rewritten (upon renormalising) as:
    \be\label{eq:quantum-AT-double}
    \begin{split}
        \d\nu^\sv_{F,\Delta}(\sigma^\bullet,\xi^\bullet) = 
        &
        \d\nu^\poi_{E^\bullet_\vert,\tfrac{1}{2}}(\chi_1)
        \d\nu^\poi_{E^\bullet_\vert,\tfrac{1}{2}-u}(\chi_2) 
        \d\nu^\poi_{E^\circ_\vert,2}(\chi_3)\\
        & \qquad\cdot
        \mathbbm{1}_{\{\sigma^\bullet\perp\xi^\bullet\}}
        \mathbbm{1}_{\{\chi_1=\om_{\horiz}[\sigma^\bullet]\}}
        \mathbbm{1}_{\{\chi_2=(\xi^\bullet)^*_\horiz=E^\bullet_\vert\setminus \xi^\bullet\}}
        \mathbbm{1}_{\{\chi_3=\xi^\bullet_{\horiz}\}}\\
        & \qquad \cdot
        e^{u\left|\xi^\bullet_\vert\right|}
        e^{u\left|\om_\vert[\sigma^\bullet]\right|}
        e^{-u\left|E^\circ_\vert\setminus
        (\om_\vert[\sigma^\bullet]\cup\xi^\bullet_\horiz)\right|}
        2^{k(\xi^\bullet)}.
    \end{split}
    \ee
    Setting $\tau\tau'=\sigma^\bullet$ and sampling $\tau$ independently and uniformly on all clusters of $\xi$ gives the Ashkin-Teller measure \eqref{eq:q:AT-measure}; on the other hand taking $(\tau,\tau')\sim\nu^{\AT}_{F^\bullet,\Delta}$, sampling $\xi^\bullet$ according to \eqref{eq:q:sample-xi-from-AT} and setting $\sigma^\bullet=\tau\tau'$, one finds that $(\sigma^\bullet,\xi^\bullet)$ are distributed as \eqref{eq:quantum-AT-double}. 

    The final identity in part 4 is again a standard colour-switching argument, and the fact that it is 0 is a consequence of Theorem \ref{thm:Q:delocalisation}, the delocalisation of the space-time six-vertex model.
\end{proof}

\begin{proof}[Proof of Theorem \ref{thm:q:8v-AT}]

We can now conclude this section. Lemmas \ref{lem:q:coupling-8v-AT} and \ref{lem:q:coupling-AT-6v} together prove Proposition \ref{prop:q:8vRC}. Finally, we use Proposition \ref{prop:q:8vRC} to prove Theorem \ref{thm:q:8v-AT}. The proof is essentially the same as in the discrete. The weak limits of all four 8vRC measures exist and are equal by Proposition \ref{prop:8v-AT-RC} (write $\nu^{\eivRC}_{\Delta}$ for this common measure). Let $\nu^{\eiv}_{\Delta}$ be the measure on eight-vertex spins given by assigning $\pm1$ uniformly and independently to every cluster of $\eta^\bullet$ and every cluster of $\eta^\circ$, where $\eta\sim\nu^{\eivRC}_{\Delta}$. By standard arguments using monotonicity (see \cite[Theorem 4.91]{grimmett-RC}), it follows that the weak limits of $\nu^{\eiv}_{F,\Delta}$, $\nu^{\eiv;\bplus}_{F,\Delta}$, $\nu^{\eiv;\wplus}_{F,\Delta}$, $\nu^{\eiv;\bplus,\wplus}_{F,\Delta}$ all exist and are equal to $\nu^{\eiv}_{\Delta}$. The limit is Gibbs (since the interaction is local), $(\RR\times\ZZ)$ translation-invariant (since $\nu^{\eivRC}_{\Delta}$ is), and ergodic; together these properties imply extremality (see e.g. \cite[Section 14]{georgii}). By the coupling and another standard colour-switching argument, we have
\be\label{eq:q:8v-8vRC-2pt-4pt}
\begin{split}
	\nu^\eiv_{\Delta}[\pi^\bullet_{x_1}\pi^\bullet_{x_2}]
	&=
	\nu^{\eivRC}_{\Delta}[x_1\xleftrightarrow{\eta^\bullet} x_2];\\
	\nu^\eiv_{\Delta}[\pi^\bullet_{x_1}\pi^\bullet_{x_2}
	\pi^\circ_{y_1}\pi^\circ_{y_2}]
	&=
	\nu^{\eivRC}_{\Delta}[x_1\xleftrightarrow{\eta^\bullet} x_2,
	y_1\xleftrightarrow{\eta^\circ} y_2],
\end{split}
\ee
and the right hand side of both of these equations tends to 0 as $||x_1-x_2||_1,||y_1-y_2||_1\to\infty$. Finally, when $\Delta=0$, $\eta^\bullet$ and $\eta^\circ$ are independent and distributed according to critical Ising models. A proof that the space-time Ising model has a unique thermodynamic limit follows essentially the same method as \cite[Theorem 11.3]{grimmett-RC}. This completes the proof of Theorem \ref{thm:q:8v-AT}.
\end{proof}

\subsection{Exact decay rate for $\Delta=0$}
Finally, we finish the proof of Part 2 of Theorem \ref{thm:xxz-decay-rates}. For $\Delta=0$, if $(\pi^\bullet,\pi^\circ)\sim \nu^{\eiv}_{\Delta}$ then $\pi^\bullet$ and $\pi^\circ$ are independent and both distributed as critical space-time Ising models on $\ZZ\times\RR$ (translated so they lie on $\cF^\bullet$ and $\cF^\circ$, respectively). Writing $y_{(a,b)}$ for the point of $\cF^\bullet$ adjacent to $(a,b)\in\ZZ\times\RR$, and $z_{(a,b)}$ for the point of $\cF^\circ$, we have
\bes
\begin{split}
    \langle S^{(1)}_0 S^{(1)}_x(t) \rangle 
    &=
    \frac{1}{4}\lim_{F\nearrow\cF}\mu^{\lp}_{\cF,\Delta}[(0,0)\leftrightarrow (x,t)]\\
	&=
	\frac{1}{4}\nu^\eiv_{\Delta}[\pi^\bullet_{(y_{(0,0)})}\pi^\bullet_{y_{(x,t)}}
	\pi^\circ_{z_{(0,0)}}\pi^\circ_{z_{(x,t)}}]\\
    &=
    \frac{1}{4}\nu^{\mathrm{Ising}}_{\beta_c}[\pi^\bullet_{(y_{(0,0)})}\pi^\bullet_{y_{(x,t)}}]\cdot
	\nu^{\mathrm{Ising}}_{\beta_c}[\pi^\circ_{z_{(0,0)}}\pi^\circ_{z_{(x,t)}}]\\
    &\sim
    (||(x,t)||^{-1/4})^2.
\end{split}
\ees
Here the second equality is \eqref{eq:q:loop-connections} and the decay is proved in \cite[Theorem 3.9]{li-mahfouf-ising}.

\appendix
\section{FKG inequalities}\label{appendix:FKG}

In this section we prove the three FKG inequalities for space-time measures that we use in Sections \ref{sec:delocalisation} and \ref{sec:quantum-loop-proofs}. 

\begin{lemma}\label{lem:appendix:FKG-proofs}\textcolor{white}{.}
    \begin{enumerate}
        \item Let $\Delta\le0$, so $u\le\tfrac{1}{2}$. The triple $(\sigma^\bullet,\xi^{\bplus},-\xi^{\bminus})$ satisfies the FKG inequality under the measure $\nu^\sv_{F,\Delta}$, as well as with the boundary conditions $\bplus$ and/or $\wplus$, , and also under the measures $\nu^{\sv,*}_{\cyl_{L,\beta},\Delta}$, $\nu^{\sv,*}_{\cyl_{L,\beta}^h,\Delta}$, for $*=\bullet,\circ$. The same holds for the triple $(\sigma^\circ,\xi^{\wplus},-\xi^{\wminus})$.
        \item Let $\Delta\le0$, so $u\le\tfrac{1}{2}$. Let $h\sim\nu^{\sv-\hom,\tilde{h}}_{F,\Delta}$ \eqref{eq:q:6v-hom-measure} or $h\sim\nu^{\sv,\bullet,g}_{\cyl_{L,\beta}\setminus\ell_{\inn}}$. Then $h$ satisfies the FKG inequality.
        \item Let $\Delta\in[-1,0]$, so $u\in[0,\tfrac{1}{2}]$. For all boundary conditions $\eta_0$, we have under the measure $\nu_{F,\Delta}^{\eivRC,\eta_0}$, $\eta$ is FKG with respect to the ordering $\tilde{\le}$. In the special case $u=\tfrac{1}{2}$, $\eta^\bullet$ and $\eta^\circ$ are independent (they are both quantum random cluster models with $q=2$), and are also FKG with respect to the ordering $\le$. 
    \end{enumerate}
\end{lemma}

Our method of proof is to show that these measures are limits of discrete measures for which the FKG inequalities hold. 

\begin{lemma}\label{lem:appendix-convergence}\textcolor{white}{.}
    \begin{enumerate}
        \item Let $\Delta\le0$, so $u\le\tfrac{1}{2}$. The measure $\nu^\sv_{F,\Delta}$, is the limit (weakly in the $w^\#$ topology on $(\ul{\Om}^\sv_V,\FF^{\sv,\#}_V)$) of a sequence of discrete measures $\mu^{\sv}_{V_\eps,a_\eps,b_\eps,c_\eps}$ (under which $(\sigma^\bullet,\xi^{\bplus},-\xi^{\bminus})$ is FKG) as $\eps\to0$, where $\Delta(a_\eps,b_\eps,c_\eps)=\Delta$. The same holds for the other boundary conditions in Lemma \ref{lem:appendix:FKG-proofs}.
        \item Let $\Delta\le0$, so $u\le\tfrac{1}{2}$. The measure $\nu^{\sv-\hom,\tilde{h}}_{F,\Delta}$ is the limit (weakly in the $w^\#$ topology on $(\ul{\Om}^\sv_V,\FF^{\sv,\#}_V)$) of a sequence of measures $\mu^{\sv-\hom,\tilde{h}_\eps}_{V_\eps,a_\eps,b_\eps,c_\eps}$ (under which $(\sigma^\bullet,\xi^{\bplus},-\xi^{\bminus})$ is FKG) as $\eps\to0$, where $\Delta(a_\eps,b_\eps,c_\eps)=\Delta$. The same holds for the measure $\nu^{\sv,\bullet,g}_{\cyl_{L,\beta}\setminus\ell_{\inn}}$.
        \item Let $\Delta\in[-1,0]$, so $u\in[0,\tfrac{1}{2}]$. The measure $\nu_{F,\Delta}^{\eivRC,\eta_0}$, is the limit (weakly in the $w^\#$ topology on $(\ul{\Om}^\eivRC_V,\FF^{\eivRC,\#}_V)$) of a sequence of discrete measures $\mu^{\eivRC,(\eta_0)_{\eps}}_{V_\eps,K_\eps, W_\eps}$ (under which $\eta$ is FKG with respect to the ordering $\tilde{\le}$) as $\eps\to0$.
    \end{enumerate}
\end{lemma}

Lemma \ref{lem:appendix-convergence} implies that for all increasing continuity events $A,B$ for the limiting measure (events whose boundary in the $w^\#$ topology has zero measure), the FKG inequality holds. 

Any increasing event can be approximated by a sequence of increasing continuity events. Taking the six-vertex example, for $A$ increasing and $k\in\RR_{>0}$, let $A_k=\{x\in\ul{\Om}^\sv_F \ : \ \exists y\le x, \ \d^\#(y,A^c)\ge k\}$, where $A^c=\ul{\Om}^\sv_{F}$. Here $\d^\#$ is the metric defining the $w^\#$ topology (see Section A2.6 in \cite{daley-verejones-1}. The sets $\partial A_k$ are all disjoint and so only at most countably many can have positive probability under $\nu^\sv_{F,\Delta}$, so there exists a sequence $A_{k_n}$, $k_n\to\infty$ which are all continuity events. Since $A_{k_n}\nearrow A$, the convergence $\nu^\sv_{F,\Delta}[A_{k_n}]\to\nu^\sv_{F,\Delta}[A]$ holds, and the FKG inequality for general increasing events follows. It therefore suffices to prove Lemma \ref{lem:appendix-convergence} in order to prove Lemma \ref{lem:appendix:FKG-proofs}.\\ 

We now begin defining the measures $\mu^{\sv}_{V_\eps,a_\eps,b_\eps,c_\eps}$, $\mu^{\sv-\hom,\tilde{h}_\eps}_{V_\eps,a_\eps,b_\eps,c_\eps}$ and $\mu^{\eivRC,(\eta_0)_{\eps}}_{V_\eps,K_\eps, W_\eps}$. We define $V_\eps$ first.

Recall the definitions of $\cF, \cF^\bullet$ and $\cF^\circ$ from the beginning of Section \ref{sec:delocalisation}. We define discrete lattices $\ZZ^2_\eps$, $\LL^\bullet_\eps$, $\LL^\circ_\eps$ which will approximate these continuous objects. Let $\ZZ^2_\eps$ be the usual $\ZZ^2$ lattice, rotated $45^\circ$, rescaled so that the centres of each pair of neighbouring black faces is distance 2 apart, shifted so that the centres of all black faces all lie in $(2\ZZ+\tfrac{1}{2})\times\RR$ (with one at $(\tfrac{1}{2},0)$), and finally rescaled only in the vertical direction so that the centre of each black face is distance $\eps$ from the one above it. Define $\LL^\bullet_\eps$ and $\LL^\circ_\eps$ to be the lattices given by black and white faces, respectively, with edges given by nearest neighbours. 

Let $F=(V,F)$ be a domain in $\cF$. Let $V_\eps$ be the collection of vertices of $\ZZ^2_\eps$ lying inside $V$. As in Section \ref{sec:1st-coupling}, let $E_\eps$ be the set of edges of $\ZZ^2_\eps$ incident to at least one vertex in $V_\eps$, and let $F_\eps$ be the set of faces of $\ZZ^2_\eps$ incident to at least one vertex of $V_\eps$. Let $F^\bullet_\eps$, $E^\bullet_\eps$, $F^\circ_\eps$ and $E^\circ_\eps$ be defined as in Section \ref{sec:1st-coupling}, along with $\partial F^\bullet_\eps$, $\partial F^\circ_\eps$, $\partial E^\bullet_\eps$ and $\partial E^\circ_\eps$. \\

We now define the parameters. Consider the six-vertex model parameters: for $\Delta=-\cos(\zeta)$ with $\zeta\in(0,\pi)$ (so $\Delta\in(-1,1)$), $r>0$ and $\theta\in(0,\pi)$,
\be\label{eq:appendix:six-vertex-weights}
a \sin\tfrac{\zeta}{2} = r \sin((1-\tfrac{\theta}{\pi})\zeta),
\qquad
b \sin\tfrac{\zeta}{2} = r \sin(\tfrac{\theta\zeta}{\pi})
\qquad
c = 2r\cos\tfrac{\zeta}{2}.
\ee
and for $\Delta=-1$,
\be\label{eq:appendix:six-vertex-weights-2}
a = 2r \tfrac{\pi-\theta}{\pi},
\qquad
b = 2r \tfrac{\theta}{\pi}
\qquad
c = 2r.
\ee

Setting $a+b+c=2$ (which comes from $p_{\NE}+p_{\NW}+p_{\es}=1$) amounts to fixing a value of $r$, and then setting $a=\eps$, one finds after come computation that
\be\label{eq:a:6v-weights}
\begin{split}
	a&=p_\es+p_\W = \eps,\\
	b&=p_\es+p_\N = 1-(1-u)\eps + O(\eps^2),\\
	c&=p_\W+p_\N = 1-u\eps + O(\eps^2),
\end{split}
\ee
where we write $p_\W$ for a mirror oriented east-west in the new rotated lattice and $p_\N$ for a mirror oriented north-south. Note that 
\be\label{eq:a:mirror-weights}
\begin{split}
	p_\W &= (1-u)\eps + O(\eps^2),\\
	p_\es&= u\eps + O(\eps^2),\\
	p_\N &= 1-\eps + O(\eps^2).
\end{split}
\ee

We will proceed by ignoring the $O(\eps^2)$ terms; one can show that they make no difference to the convergence.

Write $E^\bullet_{\eps,\horiz}$ for the horizontal edges in $E^\bullet_\eps$, with $E^\bullet_{\eps,\vert}$, $E^\circ_{\eps,\horiz}$, $E^\circ_{\eps,\vert}$ similar. \\

We will prove Part 1 of Lemma \ref{lem:appendix-convergence}; the other parts are proved by similar techniques.
\begin{proof}[Proof of Part 1 of Lemma \ref{lem:appendix-convergence}]
From \eqref{eq:sampling-xi} the marginal of $(\sigma^\bullet,\sigma^\circ,\xi^\bullet)$ under $\mu^\sv_{V_\eps,u}$ is given by
\bes
\begin{split}
    \mu^\sv_{V_\eps,u}[(\sigma^\bullet,\sigma^\circ,\xi^\bullet)]
    &\propto
    \prod_{v\in E^\bullet_{\eps,\vert} = E^\circ_{\eps,\vert}}
    \left(\tfrac{a}{c}\right)^{\mathbbm{1}\{v\in\om_\horiz[\sigma^\bullet]\}}
    \left(\tfrac{b}{c}\right)^{\mathbbm{1}\{v\in\xi^\bullet_\vert\}}
    \left(1-\tfrac{b}{c}\right)^{\mathbbm{1}\{v\in(\xi^\bullet)^*\setminus\om[\sigma^\bullet]\}}\\
    &\qquad\cdot
    \prod_{v\in E^\circ_{\eps,\vert} = E^\bullet_{\eps,\vert}}
    \left(\tfrac{b}{c}\right)^{\mathbbm{1}\{v\in\om_\vert[\sigma^\bullet]\}}
    \left(\tfrac{a}{c}\right)^{\mathbbm{1}\{v\in\xi^\bullet_\horiz\}}
    \left(1-\tfrac{a}{c}\right)^{\mathbbm{1}\{v\in(\xi^\bullet)^*\setminus\om[\sigma^\bullet]\}}\\
    &\qquad\cdot
    \mathbbm{1}\{\sigma^\bullet\perp\xi^\bullet\} 
    \mathbbm{1}\{\sigma^\circ\perp(\xi^\bullet)^*\}\\
\end{split}
\ees
which we rewrite as
\be\label{eq:a:6v-binomial-measure}
\begin{split}
    &\prod_{v\in  E^\bullet_{\eps,\vert} = E^\circ_{\eps,\vert}}
    \left(\tfrac{\eps}{1+(1-u)\eps}\right)^{\mathbbm{1}\{v\in\chi^\bullet\}}
    \left(\tfrac{1-(1-u)\eps}{1+(1-u)\eps}\right)^{\mathbbm{1}\{v\in\xi^\bullet_\vert\}}
    \left(\tfrac{(1-2u)\eps}{1+(1-u)\eps}\right)^{\mathbbm{1}\{v\in(\xi^\bullet)^*\setminus\chi^\bullet\}}\\
    &\qquad\cdot
    \prod_{v\in E^\circ_{\eps,\vert} = E^\bullet_{\eps,\vert}}
    \left(\tfrac{\eps}{1-u\eps}\right)^{\mathbbm{1}\{v\in\xi^\bullet_\horiz\}}
    \left(1-\tfrac{\eps}{1-u\eps}\right)^{\mathbbm{1}\{v\notin\xi^\bullet\}}
    \cdot
    \prod_{v\in  E^\bullet_{\eps,\vert} = E^\circ_{\eps,\vert}}
    \left(\tfrac{1+(1-u)\eps}{1-u\eps}\right)\\
    &\qquad\cdot
    \sum_{\substack{\sigma^\bullet\in{\Om}^\ising_{F^\bullet} \\ \sigma^\circ\in{\Om}^\ising_{F^\circ}}}
    \mathbbm{1}_{\{\chi^\bullet=\om_{\horiz}[\sigma^\bullet]\}}
    \mathbbm{1}_{\{\sigma^\bullet\perp \xi^\bullet, \ \sigma^\circ\perp (\xi^\bullet)^*\}}\\
    &\qquad\cdot
    \prod_{v\in  E^\bullet_{\eps,\vert} = E^\circ_{\eps,\vert}}  
    \left(1+\tfrac{2u\eps}{1-(1+u)\eps}\right)^{\mathbbm{1}\{v\in\om_\vert[\sigma^\bullet]\}}\\
\end{split}
\ee
Here we have substituted the values for $a,b,c$ in terms of $\eps,u$, and the final factor in the second line
$\prod_{v\in E^\bullet_{\eps,\vert} = E^\circ_{\eps,\vert}}
    \left(\tfrac{1+(1-u)\eps}{1-u\eps}\right)$ 
is taken out to make the three factors in the first line sum to 1. The variable $\chi^\bullet$ is a dummy variable which we put in to draw comparison with the $\chi^\bullet$ in the space-time six-vertex measure. Indeed, recall from Lemma \ref{lem:q:quantum-6v-triple} that 
\be\label{eq:a:quantum-6v-triple}
\begin{split}
    \d\nu^\sv_{F,\Delta}(\sigma^\bullet,\sigma^\circ,\xi^\bullet) = 
    &
    \d\nu^\poi_{V^\bullet,1}(\chi^\bullet)
    \d\nu^\poi_{V^\bullet,1-2u}((\xi^\bullet)^*\setminus\chi^\bullet) 
    \d\nu^\poi_{V^\circ,1}(\xi^\bullet_\horiz)\\
    & \qquad\cdot
    \sum_{\substack{\sigma^\bullet\in\ul{\Om}^\ising_{F^\bullet} \\ \sigma^\circ\in\ul{\Om}^\ising_{F^\circ}}}
    \mathbbm{1}_{\{\chi^\bullet=\om_{\horiz}[\sigma^\bullet]\}}
    \mathbbm{1}_{\{\sigma^\bullet\perp \xi^\bullet, \ \sigma^\circ\perp (\xi^\bullet)^*\}}\\
    & \qquad \cdot
     \exp[2u |\om_\vert[\sigma^\bullet]|]
\end{split}
\ee

Let us describe what is going on heuristically. On $E^\bullet_{\eps,\vert} $, $\om_\horiz[\sigma^\bullet]$ occurs with probability $\tfrac{\eps}{1+(1-u)\eps}$, and $(\xi^\bullet)^*\setminus\om[\sigma^\bullet]$ occurs with probability $\tfrac{(1-2u)\eps}{1+(1-u)\eps}$, so in the limit they'll behave as Poisson point processes with rates 1 and $1-2u$ respectively. On  $E^\circ_{\eps,\vert} $, $\xi^\bullet_\horiz$ occurs with probability $\tfrac{\eps}{1-u\eps}$, so will behave as a Poisson point process of rate 1. The factor 
$\prod_{v\in E^\bullet_{\eps,\vert} = E^\circ_{\eps,\vert}}
    \left(\tfrac{1+(1-u)\eps}{1-u\eps}\right)$ 
is not dependent on the configuration, so can be taken to be part of the partition function. The factor 
$\prod_{v\in E^\bullet_{\eps,\vert} = E^\circ_{\eps,\vert}}  
    \left(1+\tfrac{2u\eps}{1-(1+u)\eps}\right)^{\mathbbm{1}\{v\in\om_\vert[\sigma^\bullet]\}}$
will become $\exp[2u |\om_\vert[\sigma^\bullet]|]$ in the limit.\\

To prove the convergence rigorously, we first define the discrete process as in the framework of the continuous one. Each horizontal edge in $\xi^\bullet_\horiz$ is naturally an element of $\cE^\bullet_\horiz$. For a configuration $(\sigma^\bullet,\sigma^\circ,\xi^\bullet)$ on $V_\eps$, extend the configuration to be defined on all of $E^\circ_\vert=E^\bullet_\horiz$ as follows: if $v\in\xi^\bullet_\horiz$ then keep it as such, and if not, let it take the value (that is, either in $\om_\vert[\sigma^\bullet]$ or in $(\xi^\bullet)^*\setminus\om[\sigma^\bullet]$) of the closest point of $^\circ_{\eps}$ (if there are two such points, let it take the one to the right in the copy of $\RR$). Similarly we extend the configuration to all of $E^\bullet_\vert=E^\circ_\horiz$ as follows. For each point in $\om_\horiz[\sigma^\bullet]$ or $(\xi^\bullet)^*\setminus\om[\sigma]$, keep it as such, and every other point set it to be in $\xi^\bullet_\vert$. Finally, with $\om[\sigma^\bullet]$ and $\xi^\bullet$ well defined as space-time bond percolations, define $\sigma^\bullet$ as the space-time Ising configuration which agrees with $\sigma^\bullet$ on the discrete points $F^\eps$ and has domain walls $\om[\sigma]$, and define $\sigma^\circ$ similarly so that it is constant on clusters of $(\xi^\bullet)^*$.

Note that with this setup, we can rewrite the factor
\be\label{eq:a:rewrite-factor}
    \prod_{v\in  E^\bullet_{\eps,\vert} = E^\circ_{\eps,\vert}}  
    \left(1+\tfrac{2u\eps}{1-(1+u)\eps}\right)^{\mathbbm{1}\{v\in\om_\vert[\sigma^\bullet]\}}
    =
    \left(1+\tfrac{2u\eps}{1-(1+u)\eps}\right)^
    {\tfrac{1}{\eps} |\om_\vert[\sigma^\bullet]|}.
\ee

    Define $\nu^{\sv,\poi}_{F,\Delta}$ to be the product of the Poisson point processes appearing in the space-time six vertex measure \eqref{eq:a:quantum-6v-triple}, and let $\nu^{\sv,\mathrm{bin}}_{V_\eps,u}$ be the measure given by only the first two lines (the binomial parts) of \eqref{eq:a:6v-binomial-measure}. It is clear that $\nu^{\sv,\mathrm{bin}}_{V_\eps,u}$ converges to $\nu^{\sv,\poi}_{F,\Delta}$. Let $F(\chi^\bullet,\xi^\bullet)$ be the final two lines in \eqref{eq:a:quantum-6v-triple} and let $F_\eps(\chi^\bullet,\xi^\bullet)$ be the final two lines in \eqref{eq:a:6v-binomial-measure}, using \eqref{eq:a:rewrite-factor}. We claim it suffices to show the partition functions converge, that is,
    \be\label{eq:a:partition-fn-convergence}
        Z^{\sv}_{V_\eps,u}
        :=
        \nu^{\sv,\mathrm{bin}}_{V_\eps,u}[F_\eps(\chi^\bullet,\xi^\bullet)]
        \to
        \nu^{\sv,\poi}_{F,\Delta}[F(\chi^\bullet,\xi^\bullet)]
        =:
        \ul{Z}^\sv_{F,\Delta}.
    \ee
    Certainly $F_\eps(\chi^\bullet,\xi^\bullet)\to F(\chi^\bullet,\xi^\bullet)$ pointwise, they are all bounded by $K^{\# \ \mathrm{poisson \ points}}$ for some large $K$, which is integrable under all the measures. The convergence \eqref{eq:a:partition-fn-convergence} follows by the dominated convergence theorem. Then we have for any bounded, continuous function $G$ that 
    \be
        \nu^{\sv,\mathrm{bin}}_{V_\eps,u}[
        \tfrac{1}{Z^{\sv}_{V_\eps,u}}
        F_\eps(\chi^\bullet,\xi^\bullet)
        G(\chi^\bullet,\xi^\bullet)]
        \to
        \nu^{\sv,\poi}_{F,\Delta}[
        \tfrac{1}{\ul{Z}^{\sv}_{F,\Delta}}
        F(\chi^\bullet,\xi^\bullet)
        G(\chi^\bullet,\xi^\bullet)],
    \ee
    by the same argument that gave us \eqref{eq:a:partition-fn-convergence}, and the convergence in the lemma is proved. The FKG inequality follows as a simple corollary.
\end{proof}

As a final word about the triple $(\sigma^\bullet,\sigma^\circ,\xi^\bullet)$, in \cite{glazman-lammers}, the authors do not consider the discrete model on the cylinder. Since we use that the space-time measure on the cylinder satisfies the FKG inequality in the proof of Theorem \ref{thm:main-xxz}, we need to check that the FKG inequality holds for the discrete measure here. We define the measure $\mu^{\sv,\bullet}_{\cyl_{L,M}^{\mathrm{diag}},a,b,c}$ by setting $V$ in the rotated by 45 degrees lattice $\ZZ^2_1$ to be the interior of rectangular circuit $\gamma\subset E^\bullet_1$ of width $L$ and height $M$, and setting boundary conditions given by $\bullet$ on the top and bottom rows, and periodic boundary conditions on the sides.

\begin{lemma}\label{lem:a:FKG-cylinder}
    For $c\ge a,b$ and all $L,M>0$, under the measures  $\mu^{\sv,\bullet}_{\cyl_{L,M}^{\mathrm{diag}},a,b,c}$ and $\mu^{\sv,\bullet}_{\cyl_{L,M},a,b,c}$ the triple $(\sigma^\bullet,\xi^{\bplus},\xi^{\bminus})$ (and the triple $(\sigma^\circ,\xi^{\wplus},\xi^{\wminus})$) satisfies the FKG inequality. 
\end{lemma}
\begin{proof}
    The proof of the lemma is identical to that of \cite[Proposition 4.9]{glazman-lammers}. In particular, one can observe that on the cylinder with the boundary conditions $\bullet$ on the top and bottom of the cylinder, the height function is well defined (up to a constant shift), so one does not have to make the same adjustments as on the torus in \cite[Section 5]{glazman-lammers}.
\end{proof}

\printbibliography

@article{HDJS-infAT,
    author = {Yuan Huang and Youjin Deng and Jesper Lykke Jacobsen and Jesús Salas},
    title = "{The Hintermann – Merlini – Baxter – Wu and the infinite-coupling-limit Ashkin–Teller models}",
    journal = {Nuclear Physics B},
    volume = {868},
    number = {2},
    pages = {492-538},
    year = {2013},
}

@article{karol1,
      title="{On condensation properties of Bethe roots associated with the XXZ chain}", 
      author={Karol K. Kozlowski},
    journal = {Communications in Mathematical Physics},
    volume = {357},
    year = {2018},
}

@article{mccoy68xy,
  title = "{Spin Correlation Functions of the XY Model}",
  author = {McCoy, Barry},
  journal = {Physical Review},
  volume = {173},
  issue = {2},
  pages = {531--541},
  numpages = {0},
  year = {1968},
  month = {Sep},
  publisher = {American Physical Society},
}

@article{chayes-shtengel,
    author = {Lincoln Chayes and Kirill Shtengel},
    title = "{Vertex models and random labyrinths:
phase diagrams for ice-type vertex
models}",
    journal = {Journal of Statistical Mechanics},
    year = {2005},
}

@article{dckmo-6v-deloc,
   title="{Delocalization of the height function of the six-vertex model}",
   volume={26},
   number={11},
   journal={Journal of the European Mathematical Society},
   author={Duminil-Copin, Hugo and Karrila, Alex M. and Manolescu, Ioan and Oulamara, Mendes},
   year={2024},
}

@article{dc-tass-renorm,
  title="{Renormalization of Crossing Probabilities in the Planar Random-Cluster Model}",
  author={Hugo Duminil-Copin and Vincent Tassion},
  journal={Moscow Mathematical Journal},
  year={2019},
}

@article{dc-li-mano-FK,
   title="{Universality for the random-cluster model on isoradial graphs}",
   volume={23},
   journal={Electronic Journal of Probability},
   author={Duminil-Copin, Hugo and Li, Jhih-Huang and Manolescu, Ioan},
   year={2018},
}

@article{ryan-orth,
    author = {Ryan, Kieran},
    title = "{On a Class of Orthogonal-Invariant Quantum Spin Systems on the Complete Graph}",
    journal = {International Mathematics Research Notices},
    volume = {2023},
    number = {7},
    pages = {6078-6131},
    year = {2023},
}

@article{chel-hong-izy,
   title="{Conformal invariance of spin correlations in the planar Ising model}",
   volume={181},
   number={3},
   journal={Annals of Mathematics},
   author={Dmitry Chelkak and Clément Hongler and Konstantin Izyurov},
   year={2015},
}

@article{li-manhattan-21,
author = {Linjun Li},
title = "{On the Manhattan pinball problem}",
volume = {26},
journal = {Electronic Communications in Probability},
publisher = {Institute of Mathematical Statistics and Bernoulli Society},
pages = {1 -- 11},
year = {2021},
}

@article{glazman-lammers,
    title="{Delocalisation and Continuity in 2D: Loop $O(2)$, Six-Vertex, and Random-Cluster Models}",
    author={Alexander Glazman and Piet Lammers},
    year={2025},
    journal={Communications in Mathematical Physics},
    volume={406},
}

@article{burton-keane,
    title="{Density and uniqueness in percolation}",
    author={Robert Burton and Michael Keane},
    year={1989},
    journal={Communications in Mathematical Physics},
    volume={121},
}

@article{glaz-peled,
   title="{On the transition between the disordered and antiferroelectric phases of the 6-vertex model}",
   volume={28},
   journal={Electronic Journal of Probability},
   publisher={Institute of Mathematical Statistics},
   author={Glazman, Alexander and Peled, Ron},
   year={2023},
   }

@article{mccoy-wu,
    author = {Barry McCoy and Tai Tsun Wu},
    title = "{Ising model correlation functions: difference equations
and applications to gauge theory}",
    journal = {Nonlinear integrable systems—classical theory and
quantum theory},
    year = {1983}
}

@article{eng-lis,
    title="{An Elementary Proof of Phase Transition in the Planar XY Model}",
    author={Engelenburg, Diederik van and Marcin Lis},
    year={2022},
    journal={Communications in Mathematical Physics},
    volume={399}
}

@article{koz-sid,
    title="{Lower bound for the escape probability in the Lorentz Mirror Model on the lattice}",
    author={Gady Kozma and Vladas Sidoravicius},
    year={2013},
    journal={Israel Journal of Mathematics},
    volume={209},
}

@book{israel-book,
    title="{Convexity in the theory of lattice gases}",
    author={Robert Israel},
    year={1979},
    publisher={Princeton University Press},
}

@article{nienhuis-rietman,
      title="{A Solvable Model for Intersecting Loops}", 
      author={Bernard Nienhuis and Ronald Rietman},
      year={1993},
      eprint={hep-th/9301012},
      archivePrefix={arXiv},
}

@article{martins-nienhuis-rietman,
  title = {Intersecting Loop Model as a Solvable Super Spin Chain},
  author = {Marcio Jose Martins and Bernard Nienhuis and Ronald Rietman},
  journal = {Physical Review Letters},
  volume = {81},
  pages = {504--507},
  year = {1998},
}

@article{lis-lopez-heeney,
      title="{The Ising magnetisation field and the Gaussian free field}", 
      author={Tomás Alcalde López and Lorca Heeney and Marcin Lis},
      year={2026},
      eprint={2602.05886},
      archivePrefix={arXiv},
}

@article{ray-spinka-2,
   title="{Finitary codings for gradient models and a new graphical representation for the six‐vertex model}",
   volume={61},
   number={1},
   journal={Random Structures $\&$ Algorithms},
   author={Ray, Gourab and Spinka, Yinon},
   year={2021},
   pages={193–232}
}

@article{dober,
author = {Moritz Dober},
title = "{On antiferromagnetic regimes in the Ashkin–Teller model}",
volume = {30},
journal = {Electronic Journal of Probability},
pages = {1 -- 30},
year = {2025},
}

@article{chayes-mckellar-winn,
year = {1998},
volume = {31},
number = {45},
author = {Lincoln Chayes and Danica McKellar and Brandy Winn},
title = "{Percolation and Gibbs states multiplicity for ferromagnetic Ashkin-Teller models on $\ZZ^2$}",
journal = {Journal of Physics A: Mathematical and General},
}

@article{ueltschi,
    title="{Random loop representations for quantum spin systems}",
    author={Daniel Ueltschi},
    year={2013},
    journal={Journal of Mathematical Physics},
    volume={54},
    number={8},
}

@book{daley-verejones-1,
    title="{An Introduction to the Theory of Point Processes, Volume I: Elementary Theory and Methods}",
    author={Daryl J. Daley and David Vere-Jones},
    year={2008},
    publisher={Springer},
}

@book{daley-verejones-2,
    title="{An Introduction to the Theory of Point Processes, Volume II: General Theory and Structure}",
    author={Daryl J. Daley and David Vere-Jones},
    year={2008},
    publisher={Springer},
}

@article{bjornberg-thesis,
    title="{Graphical representations of Ising and Potts models}",
    author={Jakob E. Björnberg},
    year={2009},
    journal={PhD Thesis},
}

@article{dc--raoufi-tassion,
    title="{Sharp phase transition for the random-cluster and Potts models via decision trees}",
    author={Hugo Duminil-Copin and Aran Raoufi and Vincent Tassion},
    year={2019},
    journal={Annals of Mathematics},
    volume={189},
    number={1},
    pages={75-99},
}

@article{mccoy-perk-shrock-1,
title = "{Time-dependent correlation functions of the transverse Ising chain at the critical magnetic field}",
journal = {Nuclear Physics B},
volume = {220},
number = {1},
pages = {35-47},
year = {1983},
author = {Barry McCoy and Jacques Perk and Robert Shrock},
}

@article{mccoy-perk-shrock-2,
title = "{Correlation functions of the transverse Ising chain at the critical field for large temporal and spatial separations}",
journal = {Nuclear Physics B},
volume = {220},
number = {3},
pages = {269-282},
year = {1983},
author = {Barry  McCoy and Jacques Perk and Robert Shrock},
}

@article{aoun-dober-glazman,
      title="{Phase diagram of the Ashkin-Teller model}", 
      author={Yacine Aoun and Moritz Dober and Alexander Glazman},
      year={2024},
      journal={Communications in Mathematical Physics},
      volume={405},
      number={2},
      pages={75-99},
}

@article{ott-lammers,
   title="{Delocalisation and absolute-value-FKG in the solid-on-solid model}",
   volume={188},
   journal={Probability Theory and Related Fields},
   author={Lammers, Piet and Ott, Sébastien},
   year={2023}, }

@article{aiz-nacht,
author = {Michael Aizenman and Bruno Nachtergaele},
journal = "Communications in Mathematical Physics",
number = "1",
pages = "17--63",
publisher = "Springer",
title = "Geometric {A}spects of {Q}uantum {S}pin {S}tates",
volume = "164",
year = "1994"
}

@article{toth,
   title="{Improved Lower Bound on the Thermodynamic
Pressure of the Spin 1/2 Heisenberg
Ferromagnet}",
   volume={28},
   journal={Letters in Mathematical Physics},
   publisher={Kluwer Academic Publisher},
   author={Bálint Tóth},
   year={1993},
}

@article{lis-deloc,
   title="{On Delocalization in the Six-Vertex Model}",
   volume={383},
   number={2},
   journal={Communications in Mathematical Physics},
   author={Lis, Marcin},
   year={2021}, }

@article{koehler-schindler-tassion,
    title="{Crossing probabilities for planar percolation}",
    author={Laurin Köhler-Schindler and Vincent Tassion},
    year={2023},
    journal={Duke Mathematical Journal},
    volume={172},
    number={4},
    pages={809-838},
}

@book{lindvall,
    title="{Lectures on the Coupling Method}",
    author={Torgny Lindvall},
    year={1992},
    publisher={Wiley},
}

@book{georgii,
    title="{Gibbs Measures and Phase Transitions}",
    author={Hans-Otto Georgii},
    year={1988},
    publisher={De Gruyter},
}

@book{grimmett-RC,
    title="{The Random Cluster Model}",
    author={Geoffrey Grimmet},
    year={2006},
    publisher={Springer-Verlag},
}

@book{friedli-velenik,
    title="{Statistical Mechanics of Lattice Gases - A Concrete Mathematical Introduction}",
    author={Sacha Friedli and Yvan Velenik},
    year={2017},
    publisher={Cambridge University Press},
}

@article{lis-spins,
author = {Marcin Lis},
title = {{Spins, percolation and height functions}},
volume = {27},
journal = {Electronic Journal of Probability},
pages = {1-21},
year = {2022},
}

@article{sheffield-thesis,
author = {Scott Sheffield},
title = "{Random Surfaces}",
volume = {304},
journal = {Astérisque},
year = {2005},
}

@InProceedings{hugo-notes,
author={Hugo Duminil-Copin},
title={{Lectures on the Ising and Potts Models on the Hypercubic Lattice}},
booktitle="Random Graphs, Phase Transitions, and the Gaussian Free Field",
year="2020",
publisher="Springer International Publishing",
pages="35--161",
}

@article{XY-ising-jullien-fields,
author = {Rémi Jullien and J.N. Fields},
title = "{Equivalence between a spin 1/2 dimerized XY chain and two independent Ising chains in transverse fields}",
journal = {Physics Letters A},
volume = {69},
number = {3},
pages = {214-216},
year = {1978},
}

@article{sheffield-square-ice,
      title={Delocalization of uniform graph homomorphisms from $\mathbb{Z}^2$ to $\mathbb{Z}$}, 
      author={Nishant Chandgotia and Ron Peled and Scott Sheffield and Martin Tassy},
    journal = {Communications in Mathematical Physics},
    volume = {387},
    pages = {621–647},
    year = {2021},
}

@article{dc-square-ice,
      title={Logarithmic variance for the height function of square-ice}, 
      author={Hugo Duminil-Copin and Matan Harel and Benoit Laslier and Aran Raoufi and Gourab Ray},
      year={2022},
    journal = {Communications in Mathematical Physics},
    volume = {396},
    pages = {867–902},
}

@article{dc-etal-exp-decay,
author = {Hugo Duminil-Copin and Ron Peled and Wojciech Samotij and Yinon Spinka},
title = "{Exponential Decay of Loop Lengths in the Loop O(n) Model with Large n}",
journal = {Communications in Mathematical Physics},
volume = {349},
number = {3},
pages = {777-817},
year = {2017},
}

@article{bjornberg-ryan-dimerization,
    title="{Dimerization in $O(n)$-invariant quantum spin chains}",
    author={Jakob Björnberg and Kieran Ryan},
    year={2025},
    eprint={2506.06103},
    archivePrefix={arXiv},
}

@article{li-mirror-poly,
      title="{Polynomial bound for the localization length of Lorentz mirror model on the 1D cylinder}", 
      author={Linjun Li},
      year={2020},
      eprint={2010.05900},
      archivePrefix={arXiv},
}

@article{ryan-mirrors,
author = {Kieran Ryan},
title = "{The Manhattan and Lorentz Mirror Models: A Result on the Cylinder with Low Density of Mirrors}",
journal = {Journal of Statistical Physics},
volume = {185},
number = {2},
year = {2021},
}

@article{elboim-et-al-mirrors,
      title="{Diffusivity of the Lorentz mirror walk in high dimensions}", 
      author={Dor Elboim and Antoine Gloria and Felipe Hernández},
      year={2025},
      eprint={2505.01341},
      archivePrefix={arXiv},
}

@article{dc-etal-bethe-ansatz-exposition,
   title="{The Bethe ansatz for the six-vertex and XXZ models: An exposition}",
   volume={15},
   journal={Probability Surveys},
   publisher={Institute of Mathematical Statistics},
   author={Duminil-Copin, Hugo and Gagnebin, Maxime and Harel, Matan and Manolescu, Ioan and Tassion, Vincent},
   year={2018},
}

@article{pfister-velenik-AT,
   title="{Random-cluster representation of the Ashkin-Teller model}",
   volume={88},
   journal={Journal of Statistical Physics},
   author={Charles-Edouard Pfister and Yvan Velenik},
   year={1997},
}

@book{last-penrose,
    title = "{Lectures on the Poisson process}",
    author={Gunter Last and Mathew Penrose},
    publisher={Cambridge University Press},
    year={2017},
    }

@article{dc-sid-tass-cont,
   title="{Continuity of the Phase Transition for Planar Random-Cluster and Potts Models with ${1 \le q \le 4}$}",
   volume={349},
   journal={Communications in Mathematical Physics},
   author={Hugo Duminil-Copin and Vladas Sidoravicius and Vincent Tassion},
   year={2017},
}

@misc{dc-discont-q>4,
      title="{Discontinuity of the phase transition for the planar random-cluster and Potts models with $q>4$}", 
      author={Hugo Duminil-Copin and Maxime Gagnebin and Matan Harel and Ioan Manolescu and Vincent Tassion},
      year={2017},
   volume={4},
   journal={Annales Scientifiques de l’École Normale Supérieure},
}

@article{li-mahfouf-ising,
    title="{Conformal Invariance in the Quantum Ising Model}",
    author={Jhih-Huang Li and Rémy Mahfouf},
    year={2021},
    eprint={2112.04811},
    archivePrefix={arXiv},
}

@article{dc-etal-GFF-convergence,
    title="{Gaussian free field convergence of the six-vertex model
with $-1\le \Delta \le -\tfrac{1}{2}$}",
    author={Hugo Duminil-Copin and Karol Kajetan Kozlowski and Piet Lammers and Ioan Manolescu},
    year={2026},
    eprint={2603.06268},
    archivePrefix={arXiv},
}

@article{affleck-lieb-XXZ,
   title="{A Proof of Part of Haldane's Conjecture on Spin Chains}",
   volume={12},
   journal={Letters in Mathematical Physics},
   author={Ian Affleck and Elliot Lieb},
   year={1986},
}

@article{ADCW,
   title="{Dimerization and N\'eel order in different quantum spin chains
    through a shared loop representation}",
   volume={21},
   journal={Annales Henri Poincaré},
   author={Michael Aizenman and Hugo Duminil-Copin and Simone Warzel},
   year={2020},
}

@article{ASW,
   author={Francisco Alcaraz and Silvio Salinas and Walter Wreszinski},
   title="{Anisotropic Ferromagnetic Quantum Domains}",
   volume={75},
   journal={Physical Review Letters},
   year={1995},
}

@article{araki-matsui-XY,
   author={Huzihiro Araki and Taku Matsui},
   title="{Ground states of the XY-model}",
   volume={101},
   journal={Communications in Mathematical Physics},
   year={1985},
}

@article{BMNU,
   author={Jakob Bj\"ornberg and Peter M\"uhlbacher and Bruno Nachtergaele and Daniel Ueltschi},
   title="{Dimerization in quantum spin chains with $O(n)$ symmetry}",
   volume={387},
   journal={Communications in Mathematical Physics},
   year={2021},
}

@article{dc-free-energy,
      title="{On the six-vertex model's free energy}", 
      author={Hugo Duminil-Copin and Karol Kajetan Kozlowski and Dmitry Krachun and Ioan Manolescu and Tatiana Tikhonovskaia},
      year={2022},
      journal={Communications in Mathematical Physics},
      volume={395},
}

@article{kit-maillet-terras-corr-fns,
      title="{Correlation functions of the XXZ Heisenberg spin-1/2 chain in a magnetic field}", 
      author={Nikolai Kitanine and Jean-Michel Maillet and Véronique Terras},
      year={2000},
      journal={Nuclear Physics B},
      volume={567},
}

@article{kit-maillet-terras-form-factors,
      title="{Form factors of the XXZ Heisenberg spin-1/2 finite chain}", 
      author={Nikolai Kitanine and Jean-Michel Maillet and Véronique Terras},
      year={1999},
      journal={Nuclear Physics B},
      volume={554},
}

@article{jimbo-miki-miwa-nakayashiki,
   title="{Correlation functions of the XXZ model for $\Delta <- 1$}",
   volume={168},
   author={Jimbo, Michio and Miki, Kei and Miwa, Tetsuji and Nakayashiki, Atsushi},
   year={1992},
   pages={256–263} 
}

@article{jimbo-miwa,
   title="{Quantum KZ equation with |q| = 1 and correlation functions of the XXZ model in the gapless regime}",
   volume={29},
   journal={Journal of Physics A: Mathematical and General},
   author={Jimbo, Michio and Miwa, Tetsuji},
   year={1996},
   pages={2923–2958} 
}

@article{KKMST,
   title="{Algebraic Bethe ansatz approach to the asymptotic behavior of correlation functions}",
   volume={2009},
   journal={Journal of Statistical Mechanics: Theory and Experiment},
   author={Kitanine, Nikolai and Kozlowski, Karol and Maillet, Jean-Michel and Slavnov, Nikita A and Terras, Véronique},
   year={2009},
}

@article{karol2,
   title="{Large-Distance and Long-Time Asymptotic Behavior of the Reduced Density Matrix in the Non-Linear Schrödinger Model}",
   volume={16},
   journal={Annales Henri Poincaré},
   author={Kozlowski, Karol},
   year={2015},
}

@article{karol3,
   title="{On the thermodynamic limit of form factor expansions of dynamical correlation functions in the massless regime of the XXZ spin 1/2 chain}",
   volume={59},
   journal={Journal of Mathematical Physics},
   author={Kozlowski, Karol},
   year={2018},
}

@article{karol4,
   title="{On singularities of dynamic response functions in the massless regime of the XXZ spin-1/2 chain}",
   volume={62},
   journal={Journal of Mathematical Physics},
   author={Kozlowski, Karol},
   year={2021},
}

@article{karol5,
   title="{Long-distance and large-time asymptotic behaviour of dynamic correlation functions in the massless regime of the XXZ spin-1/2 chain}",
   volume={60},
   journal={Journal of Mathematical Physics},
   author={Kozlowski, Karol},
   year={2019},
}

@article{luther-peschel,
      title="{Calulation of critical exponents in two dimensions from quantum field theory in one dimension}", 
      author={Alan Luther and Ingo Peschel},
      year={1975},
      journal={Physical Review B},
      volume={12},
}

@article{haldane-1,
      title="{General relation of correlation exponents and spectral properties of one dimensional Fermi systems: Application to the anisotropic s=1/2 Heisenberg chain}", 
      author={Duncan Haldane},
      year={1980},
      journal={Physical Review Letters B},
      volume={45},
}

@article{haldane-2,
title = "{Demonstration of the “Luttinger liquid” character of Bethe-ansatz-soluble models of 1-D quantum fluids}",
journal = {Physics Letters A},
volume = {81},
pages = {153-155},
year = {1981},
author = {Duncan Haldane}
}

@article{haldane-3,
year = {1981},
volume = {14},
author = {Duncan Haldane},
title = "{'Luttinger liquid theory' of one-dimensional quantum fluids. I. Properties of the Luttinger model and their extension to the general 1D interacting spinless Fermi gas}",
journal = {Journal of Physics C: Solid State Physics},
}

@article{lukyanov-zamadolo,
title = "{Exact expectation values of local fields in the quantum sine-Gordon model}",
journal = {Nuclear Physics B},
volume = {493},
pages = {571-587},
year = {1997},
author = {Sergei Lukyanov and Alexander Zamolodchikov},
}

@article{lukyanov,
  title = "{Correlation amplitude for the XXZ spin chain in the disordered regime}",
  author = {Lukyanov, Sergei},
  journal = {Physical Review B},
  volume = {59},
  year = {1999},
}

@article{lukyanov-terras,
title = "{Long-distance asymptotics of spin–spin correlation functions for the XXZ spin chain}",
journal = {Nuclear Physics B},
volume = {654},
pages = {323-356},
year = {2003},
author = {Sergei Lukyanov and Véronique Terras},
}

@article{gottstein-werner,
    title="{Ground states of the infinite q-deformed Heisenberg ferromagnet}",
    author={Carl-Tesla Gottstein and Reinhard Werner},
    year={1995},
    eprint={cond-mat/9501123},
    archivePrefix={arXiv},
}

@article{dorlas-samsonov,
      title="{On the thermodynamic limit of the 6-vertex model}", 
      author={Teunis C. Dorlas and Maxim Samsonov},
      year={2009},
      eprint={0903.2657},
      archivePrefix={arXiv},
}

@article{yang-yang-1,
   author={Chen Ning Yang and Chen Ping Yang},
   title="{One dimensional chain of anisotropic spin-spin interactions. I. Proof of Bethe’s
hypothesis}",
   volume={150},
   journal={Physical Review},
   year={1966},
}

@article{koma-nacht-XXZ-ferro,
   author={Tohru Koma and Bruno Nachtergaele},
   title="{The Complete set of ground states of the ferromagnetic XXZ chains}",
   volume={2},
   journal={Advances in Theoretical and Mathematical Physics},
   year={1998},
}

@book{KBI-book,
    title="{Quantum Inverse Scattering Method and Correlation Functions}",
    author={Nikolai Bogoliubov and Anatoly Izergin and Vlamdimir Korepin},
    year={1993},
    publisher={Cambridge University Press},
}

@article{lieb-schultz-mattis,
   author={Elliot Lieb, Theodore Schultz, Daniel Mattis},
   title="{Two Soluble Models of an Antiferromagnetic Chain}",
   volume={16},
   journal={Annals of Physics},
   year={1961},
}

@article{matsui-XXZ-antiferro,
   author={Taku Matsui},
   title="{On the absence of non-periodic ground states for the antiferromagnetic XXZ model}",
   volume={253},
   journal={Communications in Mathematical Physics},
   year={2005},
}

@article{matsui-XXZ-ferro,
   author={Taku Matsui},
   title="{On Ground States of the One-Dimensional Ferromagnetic XXZ Model}",
   volume={37},
   journal={Letters in Mathematical Physics},
   year={1996},
}

@article{nacht-uelt,
   author={Bruno Nachtergaele and Daniel Ueltschi},
   title="{A direct proof of dimerization in a family of SU (n)-invariant quantum spin chains}",
   volume={107},
   journal={Letters in Mathematical Physics},
   year={2017},
}

@book{sutherland-book,
    title="{Beautiful Models: 70 Years of Exactly Solved Quantum Many-Body Problems}",
    author={Bill Sutherland},
    year={2004},
    publisher={IOP Publishing},
}

@book{takahashi-book,
    title="{Thermodynamics of One-Dimensional Solvable Models}",
    author={Minoru Takahashi},
    year={1999},
    publisher={Cambridge University Press},
}

@book{tasaki-book,
    title="{Physics and mathematics of quantum many-body systems}",
    author={Hal Tasaki},
    year={2020},
    publisher={Springer},
}

\end{document}